\documentclass[11pt]{report}
\usepackage{geometry,enumerate,latexsym,amsmath,amssymb,amsthm,booktabs,setspace,graphicx,url,mathtools}
\usepackage{enumitem}
\usepackage{tabu} 
\usepackage{float} 
\usepackage{color}
\usepackage[colorlinks=true,linkcolor=blue,urlcolor  = blue, citecolor=blue]{hyperref}
\usepackage{MnSymbol}
\usepackage{multirow}
\usepackage{pbox}
\usepackage{bussproofs}
\usepackage{appendix}
\usepackage{chngcntr}			
\counterwithout{footnote}{chapter}	
\addtocontents{toc}{\protect\pagestyle{empty}}

\usepackage{array}
\newcolumntype{L}[1]{>{\raggedright\let\newline\\\arraybackslash\hspace{0pt}}m{#1}}
\newcolumntype{C}[1]{>{\centering\let\newline\\\arraybackslash\hspace{0pt}}m{#1}}
\newcolumntype{R}[1]{>{\raggedleft\let\newline\\\arraybackslash\hspace{0pt}}m{#1}}
\usepackage{arydshln}
\usepackage{listings}
\definecolor{mygreen}{RGB}{28,172,0} 
\definecolor{mylilas}{RGB}{170,55,241}
\newcounter{Problem}[chapter]

\makeatletter
\renewcommand*{\p@Problem}{%
  \expandafter\p@@Problem
}
\newcommand*{\p@@Problem}[1]{%
 #1%
}
\makeatother

\newenvironment{Problem}[1][]{%
	\setcounter{Problem}{\value{equation}}
  \refstepcounter{Problem}
  \ifx\\#1\\%
  \else
    \label{#1}%
  \fi
  \begin{equation}%
}{
  \end{equation}}

\newtheorem{Definition}{Definition}
\newtheorem{Conjecture}{Conjecture}
\newtheorem{Theorem}{Theorem}
\newtheorem{Corollary}{Corollary}
\newtheorem{Lemma}{Lemma}

\newtheorem{Proposition}{Proposition}
\numberwithin{Theorem}{chapter}
\numberwithin{Corollary}{chapter}
\numberwithin{Definition}{chapter}
\numberwithin{Conjecture}{chapter}
\numberwithin{Lemma}{chapter}
\numberwithin{Algorithm}{chapter}
\numberwithin{Proposition}{chapter}
\numberwithin{equation}{chapter}

\usepackage{etoolbox}
\makeatletter
\patchcmd{\hyper@makecurrent}{%
    \ifx\Hy@param\Hy@chapterstring
        \let\Hy@param\Hy@chapapp
    \fi
}{%
    \iftoggle{inappendix}{
        \@checkappendixparam{chapter}%
        \@checkappendixparam{section}%
        \@checkappendixparam{subsection}%
        \@checkappendixparam{subsubsection}%
        \@checkappendixparam{paragraph}%
        \@checkappendixparam{subparagraph}%
    }{}%
}{}{\errmessage{failed to patch}}

\newcommand*{\@checkappendixparam}[1]{%
    \def\@checkappendixparamtmp{#1}%
    \ifx\Hy@param\@checkappendixparamtmp
        \let\Hy@param\Hy@appendixstring
    \fi
}
\makeatletter

\newtoggle{inappendix}
\togglefalse{inappendix}

\apptocmd{\appendix}{\toggletrue{inappendix}}{}{\errmessage{failed to patch}}
\apptocmd{\subappendices}{\toggletrue{inappendix}}{}{\errmessage{failed to patch}}

\newcommand{\CVaRnormX}[2][]{\llangle \mathbf{x} \rrangle^{#1}_{#2}}

\newcommand{\E}{\mathbb{E}}
\newcommand{\xinRn}{\mathbf{x} \in \mathbb{R}^n}
\newcommand{\xinR}[1]{\mathbf{x} \in \mathbb{R}^{#1}}
\newcommand{\LnormX}[2][]{\vert \vert \mathbf{x} \vert \vert^{#1}_{#2}}
\newcommand{\LnormDot}[2][]{\vert \vert \cdot \vert \vert^{#1}_{#2}}
\newcommand{\DnormX}[1]{\vert \vert \vert \mathbf{x} \vert \vert \vert_{#1}}
\newcommand{\defeq}{\vcentcolon=}
\DeclareMathOperator{\Ker}{Ker}
\newcommand{\al}[1]{\alpha_{#1}}

\usepackage{enumitem}
\setlist{nosep}

\begin{document}

\begin{titlepage}
\setstretch{1.3}
\vspace*{.5em}
\center
\textbf{The School of Mathematics\\}
\vspace*{1em}
\begin{figure}[H]
\centering
\includegraphics[width=150pt]{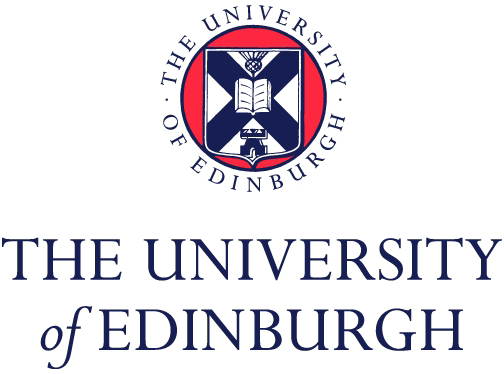}
\end{figure}
\vspace{2em}
\textbf{\Huge{Conditional Value-at-Risk: Theory and Applications}}\\[2em]
\textbf{\LARGE{by}}\\
\vspace{2em}
\textbf{\LARGE{Jakob Kisiala}}\\
\textbf{s1301096}\\
\vspace{6.5em}
\Large{Dissertation Presented for the Degree of\\
MSc in Operational Research}\\
\vspace{6.5em}
\Large{August 2015}\\
\vspace{3em}
\Large{Supervised by\\Dr Peter Richt\'{a}rik}
\end{titlepage}


\begin{titlepage}
\setstretch{1}
\thispagestyle{empty}
\vspace*{1em}
\begin{abstract}
This thesis presents the Conditional Value-at-Risk concept and combines an analysis that covers its application as a risk measure and as a vector norm. For both areas of application the theory is revised in detail and examples are given to show how to apply the concept in practice.

In the first part, CVaR as a risk measure is introduced and the analysis covers the mathematical definition of CVaR and different methods to calculate it. Then, CVaR optimization is analysed in the context of portfolio selection and how to apply CVaR optimization for hedging a portfolio consisting of options. The original contributions in this part are an alternative proof of Acerbi's Integral Formula in the continuous case and an explicit programme formulation for portfolio hedging.

The second part first analyses the Scaled and Non-Scaled CVaR norm as new family of norms in $\mathbb{R}^n$ and compares this new norm family to the more widely known $L_p$ norms. Then, model (or signal) recovery problems are discussed and it is described how appropriate norms can be used to recover a signal with less observations than the dimension of the signal. The last chapter of this dissertation then shows how the Non-Scaled CVaR norm can be used in this model recovery context. The original contributions in this part are an alternative proof of the equivalence of two different characterizations of the Scaled CVaR norm, a new proposition that the Scaled CVaR norm is piecewise convex, and the entire \autoref{chapter:Recovery_using_CVaR}. Since the CVaR norm is a rather novel concept, its applications in a model recovery context have not been researched yet. Therefore, the final chapter of this thesis might lay the basis for further research in this area.
\end{abstract}


\renewcommand{\abstractname}{Acknowledgements}
\begin{abstract}
First of all, I would like to thank my supervisor Peter Richt\'{a}rik, whose valuable feedback and ideas improved the quality of this thesis considerably. He inspired me to broaden my horizon and study topics which went beyond the syllabus. Furthermore, I would like to thank all the teaching staff who enabled me to learn a lot during my master studies.

I would also like to mention my classmates who made this year a memorable experience beyond the class room. Especially Wendy, who was always a beam of sunshine in this often cloudy and rainy city.
\end{abstract}
\vfill
\end{titlepage}


\begin{titlepage}
\setstretch{1.5}
\thispagestyle{empty}
\vspace*{1em}
\begin{center}
\textbf{\LARGE{Own Work Declaration}}\\
\end{center}
\vspace*{1em}
I declare that this thesis was composed by myself and
that the work contained therein is my own, except where
explicitly stated otherwise in the text.
\center
\vspace*{12em}
\parbox{5cm}{\centering Edinburgh, 21 August 2015\hrule
\strut \centering\footnotesize Place, Date} \hfill\parbox{4cm}{\hrule
\strut \centering\footnotesize Jakob Kisiala}
\vfill
\end{titlepage}

\begin{titlepage}
~
\end{titlepage}

\pagenumbering{gobble}
\tableofcontents
\clearpage
\listoffigures
\listoftables
\clearpage
\setstretch{1}
\bibliographystyle{abbrv}
\pagenumbering{arabic}
\setcounter{page}{1}

%
%

\chapter{Introduction}\label{chapter:intro}

This chapter presents the motivation for this thesis, gives the outline of the following chapters, and states the original contributions of the thesis.

Note that are no dedicated chapters covering a literature review or to establish notation. Rather, the literature is reviewed and notation is established in each chapter and section where it is appropriate.


\section{Motivation of the Thesis} \label{sec:intro-Motivation}

In financial risk management, especially with practitioners, Value-at-Risk (VaR) is a widely used risk measure because its concept is easily understandable and it focusses on the down-side, i.e. tail risk. A possible definition is given by Choudhry: ``VaR is a measure of market risk. It is the maximum loss which can occur with [$(\alpha \times 100)$] \% confidence [...]'' \cite[p. 30]{2006Choudhry_Intoduction_to_VaR}.

However, despite its wide use, VaR is \emph{not} a coherent risk measure. The concept of a \emph{coherent risk measure} was introduced by Artzner et al. in \cite{Artzner1999_Coherent_Measures_of_Risk}. They formulated that a risk measure $\rho$ is coherent if it satisfies the following axioms (see \autoref{sec:CVaR_RM-Coherent_Risk_Measures} for details):
\begin{itemize}
	\item Monotonicity
	\item Translation equivariance
	\item Subadditivity
	\item Positive Homogeneity
\end{itemize}
VaR is only coherent when the underlying loss distribution is normal, otherwise it lacks subadditivity. Other disadvantages of the VaR measure are that it does not give any information about potential losses in the $1 - \alpha$ worst cases and that calculating VaR optimal portfolios can be difficult, if not impossible \cite[p. 1444]{Rockafellar2002_CVaR_for_general_loss_distributions}.

The Conditional Value-at-Risk (CVaR) is closely linked to VaR, but provides several distinct advantages. In fact, in settings where the loss is normally distributed, CVaR, VaR, and Minimum Variance (Markowitz) optimization give the same optimal portfolios \cite[p. 29]{Rockafellar2000_Optimization_of_CVaR}. The advantages of CVaR become apparent when the loss distribution is not normal or when the optimization problem is high-dimensional:  CVaR is a coherent risk measure for any type of loss distribution.

Furthermore, in settings where an investor wants to form a portfolio of different assets, the portfolio CVaR can be optimized by a computationally efficient, linear minimization problem, which simultaneously gives the VaR at the same confidence level as a by-product. On the other hand, it is difficult to form VaR optimal portfolios, as is these settings VaR is difficult to calculate. This computationally efficient way to optimize the portfolio CVaR can also be transferred to hedging problems, in which an investment decision has been taken, but adjustments are possible so that the downside risk of the investment can be reduced. For example, \cite{Albrecht2015_Tail_Risk_Hedging}, \cite{Bardou2013_CVaR_Hedging_Using_Quatization_Algorithm}, \cite{Topaloglou2002_CVaR_Models_with_selective_hedging}, and \cite{Xue2015_Optimal_Inventory_and_Hedging_Decisions} used CVaR optimization to hedge risk, each one in a different setting. \\

What is more remarkable, is that the CVaR concept (which was developed as a financial risk measure) can be abstracted to form a new family of norms in $\mathbb{R}^n$. The Scaled and (Non-Scaled) CVaR norm can then be used as alternatives to the widely established family of $L_p$ norms. Moreover, by choosing suitable $\alpha$, the CVaR norm is equivalent to the $L_1$ and $L_{\infty}$ norm.

Having this new CVaR norm also opens up new opportunities in Big Data optimization, particularly in model or signal recovery problems. In these problems, it is the goal to reconstruct a model or signal of dimension $p$ when less than $p$ observations are available. This can be achieved by exploiting the structure of particular signals and solving a norm minimization problem using an appropriate norm. Particularly the $L_1$ and $L_{\infty}$ norm are used for two different types of models, and having the CVaR norm as another norm in $\mathbb{R}^n$ could recover further types of signals and models. To the best knowledge of the author, no research has been undertaken so far to use the CVaR norm in model recovery problems, so this might be another area of research to consider in the future.


\section{Outline of the Thesis} \label{sec:intro-Outline}

This thesis consists of 7 main chapters (not counting the introduction and conclusion), which concentrate on two main areas: First, the use of CVaR as a risk measure and second, the characteristics of the CVaR norm with an outlook on possible future applications. For both areas, an extensive analysis on the theory of CVaR and the CVaR norm is given, before showing how this theory can be applied in practice. \\

\autoref{chapter:CVaR_as_risk_measure} introduces the concept of CVaR as a risk measure for a univariate loss distribution. It starts by showing how VaR and CVaR are related to each other. Then, the notion of a coherent risk measure is introduced and it is shown why VaR is not coherent. \autoref{sec:CVaR_RM-Closer_Analysis_of_CVaR} then examines the mathematical definition of CVaR and shows how the CVaR can be calculated using the Convex Combination Formula. The chapter finishes by showing an alternative way to calculate CVaR, namely using Acerbi's Integral Formula.

\autoref{chapter:CVaR_Portfolio_Optimization} moves from univariate to multivariate loss distributions. These loss distributions arise in portfolio optimization problems, where there are different assets, each with their own loss distribution and the investor's loss depends on his investment decision into each asset. \autoref{sec:CVaR_PO_Markowitz} discusses the first model that was introduced to optimize a portfolio with regards to risk (the Markowitz Model, which aims to reduce the portfolio variance). Identifying the shortcomings of the Markowitz Model gives the motivation for the next model that is considered, i.e. the Rockafellar and Uryasev Model, which optimizes the portfolio CVaR. The analysis extends the results of the CVaR analysis in the univariate case to the multivariate case and gives a linear optimization programme that minimizes the CVaR of a portfolio. This section also shows that the Markowitz Model and Rockafellar and Uryasev Model lead to the same optimal portfolio if the loss of all assets in the portfolio is normally distributed. \autoref{sec:CVaR_PO_Examples} then gives two numerical examples to demonstrate the results that were established in this chapter. First, it is shown that in certain cases CVaR and Mean-Variance optimization indeed give the same portfolio, before demonstrating that for non-normal loss distributions CVaR optimization gives a less risky portfolio that Mean-Variance optimization.

Next, \autoref{chapter:Portfolio_Hedging} shows how the CVaR optimization problem can be used to hedge tail losses from a previous investment decision. In this particular example, a scenario based on real world data is created. Simplifying assumptions are made to focus on the hedging procedure instead of the technical implementation of the hedge. For the scenario, a trader's portfolio is to be adjusted, so that the CVaR of the portfolio is minimized. Since it is an option portfolio (for which the risk manager needs a daily estimate on the portfolio variance) \autoref{sec:Portfolio_Hedging-Background_on_Options} and \autoref{sec:Portfolio_Hedging-Background_on_Risk_Management} give the necessary finance and risk management background. \autoref{sec:Portfolio_Hedging-Forming_Strangle} briefly describes how the portfolio is formed before \autoref{sec:Portfolio_Hedging-Hedging_Strangle} explains the hedging procedure, including an explicit formulation of the hedging problem. The portfolio risk before and after hedging are compared and it is shown how the hedging procedure can improve the risk profile of the portfolio. \\

Moving away from the financial context, \autoref{chapter:CVaR_Norms} introduces two norms that are based on CVaR: the Scaled CVaR norm $C_{\alpha}^S$, and the (Non-Scaled) CVaR norm $C_{\alpha}$. For both norms, two different yet equivalent characterizations are given. \autoref{sec:CVaR_Norms-Properties} then describes the properties of each norm and especially shows how their properties with regards to the parameter $\alpha$ are fundamentally different. Since these norms are fairly novel and standard algorithms to calculate them are not yet implemented in MATLAB, \autoref{sec:CVaR_Norms-Computational_Efficiency} examines the computational efficiency of calculating the two norms, $C_{\alpha}^S$ and $C_{\alpha}$, using the two different characterizations for each.

To give a better understanding of  $C_{\alpha}^S$ and $C_{\alpha}$, they are both compared to the more familiar family of $L_p$ norms in \autoref{chapter:Comparison_to_other_vector_norms}. First $C_{\alpha}^S$ is compared to $L_p^S$ norms before the $C_{\alpha}$ is analysed with regards to the parameter $\alpha$ and its proximity to $L_p$ norms.

\autoref{chapter:Model_Recovery_using_Atomic_Norms} then gives a possible application of the CVaR norm in an optimization context: model recovery using atomic norms. In model (or signal) recovery the goal is to reconstruct a $p$-dimensional model (or signal) with $n$ random measurements, such that $n < p$. For a recovery to be successful, the model must have a certain structure that can be exploited by a corresponding \emph{atomic norm}.  \autoref{sec:MRuAN-Background} provides the background on atomic norms and convex geometry (e.g. the notions of tangent and normal cones) that is needed to explore the usefulness of the CVaR norm in this setting. \autoref{sec:MRuAN-Recovery_Conditions} states the necessary recovery conditions, more precisely the number of random measurements needed to ensure that a $p$-dimensional model can be recovered from $n$ measurements. The number of measurements $n$ is derived by using \emph{Gaussian Widths}, which are quite difficult to compute directly. Therefore, \autoref{sec:MRuAN-Gaussian_Widths} states some properties of Gaussian Widths that might prove useful when establishing a bound on $n$.

The final chapter, \autoref{chapter:Recovery_using_CVaR}, is completely original in the sense that it explores how the CVaR norm can be used in the context of model recovery problems. To the best knowledge of the author, no research in this particular area has been carried out before. Unfortunately, due to the limited scope of this thesis, the analysis could not be completed. Rather, this chapter should show areas of further research, with pointers towards what could be analysed in more detail. \autoref{section:RuCVaR-Atomic_CVaR_Norm} contains a conjecture about the set of atoms of the CVaR norm for a certain $\alpha$. A proposition based on the conjecture is proven, but due to the limited scope of this dissertation, the conjecture could not be proven in full. Still, a numerical experiment was carried out to identify the atoms of the CVaR norm in $\mathbb{R}^4$ and this experiment provides further evidence that the conjecture is true. \autoref{section:RuCVaR-Gaussian_Width} is rather short, showing how a bound on the number of measurements $n$ can be derived if expressions are available for the tangent or normal cone with respect to the atoms of CVaR norm. Some numerical experiments were performed to recover simple signals using the CVaR norm in \autoref{section:RuCVaR-Experiments}. The results are not impressive, as the experiments were limited to a certain $\alpha$ and only few special cases of signals. Analysing model recovery using the CVaR norm further could lead to different set ups, for which the results could be better.


\section{Original Contributions of the Thesis} \label{sec:intro-Contributions}

First of all, to the best knowledge of the author, this thesis is the first piece of work that analyses CVaR as a risk measure and the CVaR norm (including possible applications) in a unified way. There is an abundance of papers on CVaR, CVaR portfolio optimization, and further applications of CVaR as a risk measure. However, there is little research on the CVaR norm and no research on the application of the CVaR norm in the context of model recovery.

A large part of this thesis presents results of other papers. Even with established concepts, the author aims to present them in such a way that the concepts are easily understandable. Also, most plots in this paper were reproduced independently to confirm the results of other authors. But throughout the paper several original contributions are made, either by presenting new proofs to existing propositions, or by stating new propositions / conjectures. In detail, the original contributions are:
\begin{itemize}
	\item \autoref{sec:CVaR_RM-Acerbi_new_proof}: A new proof of Acerbi's Integral Formula (first proposed in \cite{Acerbi2002_Coherence_of_ES}) to calculate CVaR is given.
	\item \autoref{sec:CVaR_PO_Markowitz}: Although this is a standard result, the author proves independently why portfolio diversification reduces risk (when measured by standard deviation). The reason to give an independent proof is that the standard introductory financial literature only shows this result for $N=2$ assets, while this thesis shows this result for $N \geq 2$ assets.
	\item \autoref{sec:Portfolio_Hedging-Hedging_Strangle}: Although hedging using CVaR optimization was discussed by Rockafellar and Uryasev in \cite{Rockafellar2000_Optimization_of_CVaR}, they never explicitly formulated the optimization programme. This thesis clearly defines the variables and states the problem for a CVaR optimal hedge of a portfolio of options.
	\item \autoref{subsec:CVaR_Norms-Scaled-Alternative}: This subsection introduces a second, equivalent characterization of the Scaled CVaR norm, which was proposed by Pavlikov and Uryasev in \cite{pavlikov2014_CVaR_Norm_and_applications}. The original contribution of this thesis is an alternative proof of the equivalence of the two different characterizations.
	\item \autoref{prop:Scaled_CVaR_Norm_piecewise_convex}: The piecewise convexity of the Scaled CVaR norm is a new and original proposition of this thesis, to the best knowledge of the author.
	\item \autoref{sec:CVaR_Norms-Computational_Efficiency}: To the best knowledge of the author, the computational efficiency of different algorithms to calculate the Scaled and Non-Scaled CVaR norm has not been investigated before.
	\item \autoref{section:RuCVaR-Atomic_CVaR_Norm}: To the best knowledge of the author, the atoms (i.e. the extreme points of the unit ball) of the CVaR norm have never been explicitly stated before. This section conjectures the set of atoms of the CVaR norm for a specific $\alpha$. It shows that for different $\alpha$ the unit ball of the CVaR norm looks different, and finally a numerical experiment is performed to provide evidence for the conjecture in $\mathbb{R}^4$.
	\item \autoref{section:RuCVaR-Experiments}: To the best knowledge of the author, the CVaR norm has never been analysed in the context of model recovery problems. This section performs some numerical recovery experiments to see how suitable $C_{\alpha}$ would be recover a special type of signal. Because of the close link between the CVaR norm and the $L_1$ and $L_{\infty}$ norms, it is also investigated how well the CVaR norm performs in signal recovery problems when compared to these two $L_p$ norms.
\end{itemize}

\clearpage

%
%

\chapter{Conditional Value-at-Risk as a Risk Measure}\label{chapter:CVaR_as_risk_measure}

This chapter introduces the concept of CVaR (building on the VaR concept) in the way that it was first introduced - a financial risk measure. In \autoref{sec:CVaR_RM-Basic_Notions} the mathematical definitions of VaR and CVaR are given, followed by an intuitive description of their properties and interactions. \autoref{sec:CVaR_RM-Coherent_Risk_Measures} presents the axioms that must be satisfied for a risk measure to be considered \emph{coherent}. Specifically, an example is shown to prove that VaR is not subadditive - whereas for the same example, CVaR is subadditive. Finally, \autoref{sec:CVaR_RM-Closer_Analysis_of_CVaR} explores the CVaR concept in more detail, giving different algorithms and optimization programmes to calculate the CVaR of a given loss distribution in a variety of settings. \autoref{sec:CVaR_RM-Acerbi} states Acerbi's Integral Formula to calculate CVaR and gives an alternative proof of the formula.


\section{Basic Notions in the VaR / CVaR Framework} \label{sec:CVaR_RM-Basic_Notions}

Since losses are random variables, some statistical measures need to be introduced to cover the basics for latter sections and chapters, especially the ones concerning portfolio optimization (\autoref{chapter:CVaR_Portfolio_Optimization} and \autoref{chapter:Portfolio_Hedging}).

\begin{Definition}[{\cite[p. 17]{2012Kaltenbach_Statistics} Expectation}] \label{def:expectation}
	The \emph{expectation}, sometimes called \emph{expected value} or \emph{mean}, of a random variable $X$ is defined as
	\begin{align}
		\E[X] & \defeq \int \limits_{- \infty}^{\infty} x f(x) dx & \text{in the continous case} \label{eqn:expectation_continuous}
	\end{align}
	or
	\begin{align}
		\E[X] & \defeq \sum \limits_{k = - \infty}^{\infty} k P(X = k) & \text{in the discrete case,} \label{eqn:expectation_discrete}
	\end{align}
	where $f(x)$ is the probability density function of $X$ and $P(X = k)$ is the probability mass function $X$.
\end{Definition}
The expectation is often denoted by the letter $\mu$, such that $\mu = \E[X]$.\footnote{Many texts apply the distinction to use $\mu$ for the population mean and $\hat{\mu}$ for the sample mean. Although the expectation of the loss variable $X$ is actually a sample mean, this dissertation will use the notation $\mu$ when talking about the expectation of losses.} $\E[X]$ provides information about the distribution of $X$; informally it can be described as the centre value around which possible values of $X$ disperse \cite[p. 17]{2012Kaltenbach_Statistics}.

\begin{Definition}[{\cite[p. 18]{2012Kaltenbach_Statistics} Variance}] \label{def:variance}
	The \emph{variance} of a random variable $X$ is defined as
	\begin{align}
		\text{\emph{Var}~}(X) & \defeq \E \left[ \left(X - \E [X] \right)^2 \right] . \label{eqn:variance}
	\end{align}
\end{Definition}
The variance is often denoted as $\sigma^2$.\footnote{Again, many texts apply a distinction between the population variance $\sigma^2$ and the sample variance $s^2$. As in the case with the expectation, this dissertation will use the notation $\sigma^2$ when talking about the variance of losses.} Since the variance is hard to interpret as it is given in square units, the standard deviation (denoted $\sigma = \sqrt{\text{Var}(X)}$) is often used. It does not contain additional information, but is easier to interpret as $\sigma$ is given in the same units as $\mu$ \cite[p. 18]{2012Kaltenbach_Statistics}.

The standard deviation $\sigma$ (or variance $\sigma^2$) measures how strongly $X$ is dispersed around $\mu$. Small values of $\sigma$ indicate that $X$ is concentrated strongly around $\mu$, while large values of $\sigma$ mean that values of $X$ further away from $\mu$ (in either direction) are more likely. \\

Another important concept throughout this dissertation is \emph{Covariance}.
\begin{Definition}[{\cite[p. 21]{2012Kaltenbach_Statistics} Covariance}] \label{def:covariance}
	The \emph{covariance} of two random variables $X_1$ and $X_2$ is defined as
	\begin{align}
		\text{\emph{Cov}~}(X_1, X_2) & \defeq \E \left[ \left(X_1 - \E [X_1] \right) \left( X_2 - \E [X_2] \right)  \right] . \label{eqn:covariance}
	\end{align}
\end{Definition}
 Covariance measures how strongly the variable $X_1$ varies together with $X_2$ (and vice versa). As a special case, $\text{Cov}(X,X) = \text{Var} (X)$. Also, if $X_1$ and $X_2$ are independent, their covariance is 0 \cite[p. 21]{2012Kaltenbach_Statistics}. As in the case with variance, the covariance is hard to interpret, as its unit is the product of the respective units of $X_1$ and $X_2$. Therefore, another measure for dependency that is derived from the covariance and variance is commonly used to express how strongly $X_1$ and $X_2$ vary together - it is called the \emph{correlation coefficient}:
\begin{Definition}[{\cite[p. 22]{2012Kaltenbach_Statistics} Correlation Coefficient}] \label{def:correlation}
	The \emph{correlation coefficient} of two random variables $X_1$ and $X_2$ is defined as
	\begin{align}
		\rho_{12} & \defeq \frac{\text{\emph{Cov}~}(X_1, X_2)}{\sqrt{\text{\emph{Var}~}(X_1)} \sqrt{\text{\emph{Var}~}(X_2)}}. \label{eqn:correlation}
	\end{align}
\end{Definition}

$\rho$ always takes values between -1 and 1 and is therefore easier to interpret than covariance. If $\vert \rho_{12} \vert$ is close to 1, then there is a strong dependence between $X_1$ and $X_2$ \cite[p. 22]{2012Kaltenbach_Statistics}. \\

As pointed out in the introduction, Value-at-Risk (VaR) is the maximum loss that will not be exceeded at a given confidence level. This gives the following mathematical definition of VaR:
\begin{Definition} [{\cite[week 8, p. 5]{Richtarik2015_OMF_Lecture} Value-at-Risk (VaR)}]\label{def:VaR}
	Let $X$ be a random variable representing loss. Given a parameter $0 < \alpha < 1$, the $\alpha$-VaR of $X$ is
	\begin{equation} \label{eqn:VaR}
		\text{\emph{VaR}}_{\alpha} (X) \defeq \min \{ c : P (X \leq c) \geq \alpha \} \,.
	\end{equation}
\end{Definition}
\noindent Given \autoref{def:VaR}, VaR can have several equivalent interpretations \cite[week 8, p. 5]{Richtarik2015_OMF_Lecture}:
\begin{itemize}
	\item $\text{VaR}_{\alpha} (X)$ is the \emph{minimum} loss that will not be exceeded with probability $\alpha$.
	\item $\text{VaR}_{\alpha} (X)$ is the $\alpha$-quantile of the distribution of $X$.
	\item $\text{VaR}_{\alpha} (X)$ is the \emph{smallest} loss in the $(1 - \alpha) \times 100$\% worst cases.
	\item $\text{VaR}_{\alpha} (X)$ is the \emph{highest} loss in the $\alpha \times 100$\% best cases.
\end{itemize}
~

The general definition of CVaR is given in \autoref{sec:CVaR_RM-Closer_Analysis_of_CVaR}. At this point, only the CVaR definition for continuous random variables will be given to create a more intuitive introduction into the topic. For continuous $X$, the Conditional Value-at-Risk is the expected loss, \emph{conditional} on the fact that the loss exceeds the VaR at the given confidence level:
\begin{Definition} [{\cite[week 8, p. 13]{Richtarik2015_OMF_Lecture} Conditional Value-at-Risk (CVaR) in the continuous case}]\label{def:CVaR}
	Let $X$ be a continuous random variable representing loss. Given a parameter $0 < \alpha < 1$, the $\alpha$-CVaR of $X$ is
	\begin{equation} \label{eqn:CVaR}
		\text{\emph{CVaR}}_{\alpha} (X) \defeq \E [ X \mid X \geq \text{\emph{VaR}}_{\alpha} (X) ] .
	\end{equation}
\end{Definition}

Alternative names for CVaR found in the literature are \emph{Average Value-at-Risk}, \emph{Expected Shortfall}, or \emph{Tail Conditional Expectation}, although some authors make subtle distinctions between their definitions \cite[week 8, p. 13]{Richtarik2015_OMF_Lecture}.

\autoref{fig:VaR_CVaR_Explanation} shows the VaR and CVaR for a specific continuous random variable $X$. The cumulative distribution function of $X$ can be used to find $\text{VaR}_{\alpha} (X)$, and $\text{VaR}_{\alpha} (X)$ can be used in turn to calculate $\text{CVaR}_{\alpha} (X)$. \footnote{An alternative approach to find VaR and CVaR is shown in \autoref{theorem:CVaR_Optimization}}

\begin{figure}[H]
	\centering
	\includegraphics[width = 0.9 \textwidth]{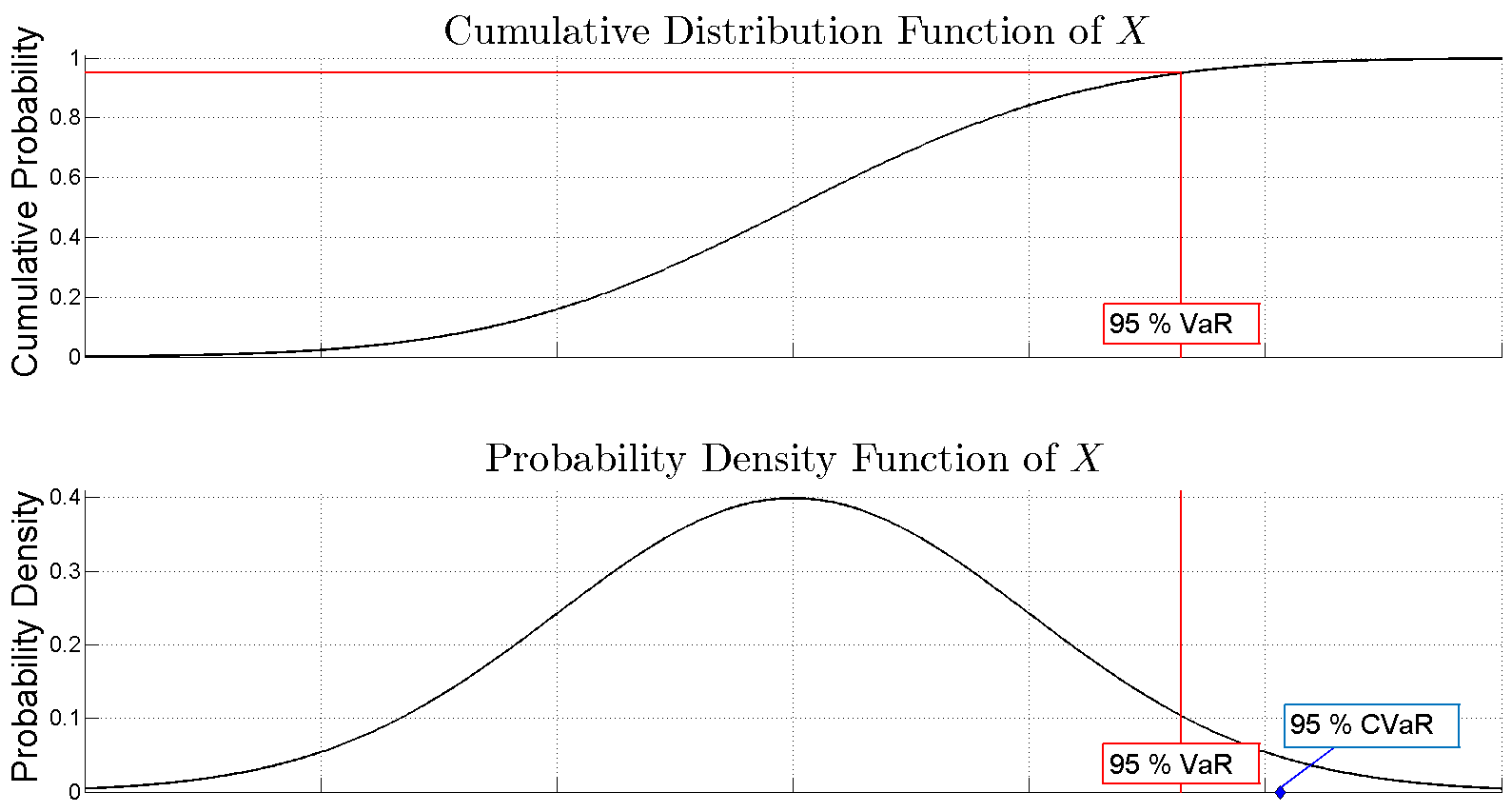}
	\caption{$\text{VaR}_{\alpha}$ and $\text{CVaR}_{\alpha}$ of a random variable $X$ representing loss.}
	\label{fig:VaR_CVaR_Explanation}
\end{figure}


\section{Coherent Risk Measures} \label{sec:CVaR_RM-Coherent_Risk_Measures}

Artzner et al. analysed risk measures in \cite{Artzner1999_Coherent_Measures_of_Risk} and stated a set of properties / axioms that should be desirable for any risk measure. Any risk measure which satisfies these axioms is said to be \emph{coherent}. The four axioms they stated are \emph{Monotonicity}, \emph{Translation equivariance}, \emph{Subadditivity}, and \emph{Positive Homogeneity}. For the definitions of all axioms, $X$ and $Y$ are random variables representing loss, $c \in \mathbb{R}$ is a scalar representing loss, and $\rho$ is a risk function, i.e. it maps the random variable $X$ (or $Y$) to $\mathbb{R}$, according to the risk associated with $X$ (or $Y$).

\begin{Definition}[{\cite[p. 210]{Artzner1999_Coherent_Measures_of_Risk} Monotonicity}] \label{def:Monotonicity}
	A risk measure $\rho$ is \emph{monotone}, if for all $X$, $Y$:
	\begin{equation} \label{eqn:Monotonicity}
		X \leq Y \Rightarrow \rho (X) \leq \rho(Y) .
	\end{equation}
\end{Definition}

\begin{Definition}[{\cite[p. 209]{Artzner1999_Coherent_Measures_of_Risk} Translation Equivariance}] \label{def:Translation_Equivariance}
	A risk measure $\rho$ is \emph{translation equivariant}, if for all $X$, $c$:
	\begin{equation} \label{eqn:Translation_Equivariance}
		\rho (X + c) = \rho(X) + c .
	\end{equation}
\end{Definition}

\begin{Definition}[{\cite[p. 209]{Artzner1999_Coherent_Measures_of_Risk} Subadditivity}] \label{def:Subadditivity}
	A risk measure $\rho$ is \emph{subadditive}, if for all $X$, $Y$:
	\begin{equation} \label{eqn:Subadditivity}
		\rho (X + Y) \leq \rho(X) + \rho(Y) .
	\end{equation}
\end{Definition}

\begin{Definition}[{\cite[p. 209]{Artzner1999_Coherent_Measures_of_Risk} Positive Homogeneity}] \label{def:Positive_Homogeneity}
	A risk measure $\rho$ is \emph{positively homogeneous}, if for all $X$, $\lambda \geq 0$:
	\begin{equation} \label{eqn:Positive_Homogeneity}
		\rho (\lambda X) = \lambda \rho(X) .
	\end{equation}
\end{Definition}

Speaking in a more intuitive way, the above axioms (\autoref{def:Monotonicity} - \autoref{def:Positive_Homogeneity}) can be interpreted as follows \cite[week 8, p. 10 f.]{Richtarik2015_OMF_Lecture}:
\begin{itemize}
	\item \textbf{Monotonicity}: Higher losses mean higher risk.
	\item \textbf{Translation Equivariance}: Increasing (or decreasing) the loss increases (decreases) the risk by the same amount.
	\item \textbf{Subadditivity}: Diversification decreases risk.
	\item \textbf{Monotonicity}: Doubling the portfolio size doubles the risk.			
\end{itemize}~

VaR fails to meet the subadditivity axiom (\autoref{def:Subadditivity}) and is therefore criticized for not being a coherent risk measure. A simple example shows this \cite[week 8, p. 19]{Richtarik2015_OMF_Lecture}:

Consider two possible investments, $A$ and $B$, which have the loss profile shown in \autoref{table:example_loss_profile1}. There are three different scenarios $\xi_1, \xi_2, \xi_3$, each with associated probability $p(\xi_i)$.
\begin{table}[H]
	\centering
	$\begin{tabu}{| c | r r r |}
		\hline
		 & \xi_1 & \xi_2 & \xi_3 \\
		 p(\xi_i) & 0.04 & 0.04 & 0.92 \\
		 \hline
		 A & 1000 & 0 & 0 \\
		 B & 0 & 1000 & 0 \\
		 \hline
	\end{tabu}$
	\caption{Losses for investments $A$ and $B$ under three scenarios.}
	\label{table:example_loss_profile1}
\end{table}
Using \autoref{eqn:VaR} to calculate the VaR at the 95 \% confidence level for investments in $A$, $B$, and $A + B$ gives
\begin{align*}
	\text{VaR}_{0.95} (A) = \min \{ c : P (A \leq c) \geq 0.95 \} & = 0 && ( P(A \leq 0) = 0.96 ) \,,\\
	\text{VaR}_{0.95} (B) = \min \{ c : P (B \leq c) \geq 0.95 \} & = 0 & &( P(B \leq 0) = 0.96 ) \,, \text{~and} \\
	\text{VaR}_{0.95} (A + B) = \min \{ c : P (A + B \leq c) \geq 0.95 \} & = 1000 \,.
\end{align*}
In this example, $\text{VaR}_{0.95} (A + B) \not \leq \text{VaR}_{0.95} (A) + \text{VaR}_{0.95} (B)$, hence VaR is not subadditive according to \autoref{def:Subadditivity}. Therefore, it is not a coherent risk measure in the sense of Artzner et al.

Acerbi and Tasche proved in \cite{Acerbi2002_Coherence_of_ES} that CVaR in satisfies the above axioms and is therefore a coherent risk measure.\footnote{To be precise: In \cite{Acerbi2002_Coherence_of_ES} Acerbi and Tasche defined \emph{Expected Shortfall} (ES) and \emph{CVaR} slightly differently. In the paper, they first proved that ES is a coherent risk measure and later proved that ES is identical to CVaR.} Using the previous example together with \autoref{eqn:CVaR_Convex_Combination_Formula} of \autoref{prop:CVaR_Convex_Combination_Formula} gives
\begin{align*}
	\text{CVaR}_{0.95} (A) & = 800 && (\lambda = 0.2,  \text{CVaR}_{0.95}^+ (A) = 1000) \,, \\
	\text{CVaR}_{0.95} (B) & = 800 && (\lambda = 0.2,  \text{CVaR}_{0.95}^+ (B) = 1000)  \,, \text{~and}\\
	\text{CVaR}_{0.95} (A + B) & = 1000 && (\lambda = 1,  \text{CVaR}_{0.95}^+ (A + B) = 0) \,.
\end{align*}
which shows that subadditivity holds for CVaR, as $\text{CVaR}_{0.95} (A + B) = 1000 \leq \text{CVaR}_{0.95} (A) + \text{CVaR}_{0.95} (B) = 1600$.


\section{Closer Analysis of CVaR} \label{sec:CVaR_RM-Closer_Analysis_of_CVaR}

Analysing CVaR in a wider context, one can derive CVaR from the \emph{generalized $\alpha$-tail distribution} of a random variable $X$ (which represents loss). This is what Rockafellar and Uryasev did in \cite{Rockafellar2002_CVaR_for_general_loss_distributions}. While \cite{Rockafellar2002_CVaR_for_general_loss_distributions} focused on general distributions, their previous work in \cite{Rockafellar2000_Optimization_of_CVaR} concerned the CVaR of continuous loss distributions. This section will present the results of both papers in a unified way, for discrete as well as for continuous loss distributions.\\

Suppose that $X$ is the loss distribution, and that $F_X (z)$ is the cumulative distribution function of $X$, i.e. $F_X (z) = P (X \leq z)$. Then the generalized $\alpha$-tail distribution of is defined as \cite[week 8, p. 15]{Richtarik2015_OMF_Lecture}
\begin{equation} \label{eqn:generalized_alpha_tail}
F^{\alpha}_X (z) \defeq \left\{
\begin{array}{l r}
	0, & \text{when~} z < \text{VaR}_{\alpha} (X) \\
	\frac{F_X (z) - \alpha}{1 - \alpha}, & \text{when~} z \geq \text{VaR}_{\alpha} (X)
\end{array}
\right. .
\end{equation}
Now, if $X^{\alpha}$ is the random variable whose cumulative distribution function is $F^{\alpha}_X$ (\autoref{eqn:generalized_alpha_tail}), then the CVaR is defined as
\begin{align}
	\text{CVaR}_{\alpha} (X) &\defeq \E [X^{\alpha}] , \label{def:CVaR_general_definition}
\end{align}
which leads to \autoref{def:CVaR} in the continuous case ($\text{CVaR}_{\alpha} (X) = \E [ X \mid X \geq \text{VaR}_{\alpha} (X) ]$), but is different for the discrete case \cite[week 8, p. 15]{Richtarik2015_OMF_Lecture}.

For discrete or non-continuous loss distributions, Rockafellar and Uryasev proposed to calculate CVaR as a weighted average, also called the \emph{Convex Combination Formula}. To apply the Convex Combination Formula, one needs the $\text{VaR}_{\alpha}$ and $\text{CVaR}^+_{\alpha}$ of $X$, where $\text{CVaR}^+_{\alpha} (X)$ is the expected loss \emph{strictly} greater than the $\text{VaR}_{\alpha} (X)$, i.e.,
\begin{equation} \label{eqn:CVaR_plus}
	\text{CVaR}_{\alpha}^+ (X) \defeq \E [X \mid X > \text{VaR}_{\alpha} (X)] . 
\end{equation}

\begin{Proposition} [{\cite[p. 1452]{Rockafellar2002_CVaR_for_general_loss_distributions} CVaR as a weighted average / Convex Combination Formula}] \label{prop:CVaR_Convex_Combination_Formula}
	Let	$\Psi$ be cumulative probability of $\text{VaR}_{\alpha} (X)$, i.e. $ \Psi = F_X (\text{\emph{VaR}}_{\alpha} (X) )$ and define $\lambda$ as 
	\begin{equation*}
		\lambda \defeq \frac{\Psi - \alpha}{1 - \alpha} \,,
	\end{equation*}
	for $0 \leq \alpha < 1$. We then have:
	\begin{equation} \label{eqn:CVaR_Convex_Combination_Formula}
		\text{\emph{CVaR}}_{\alpha} (X) = \lambda \text{\emph{VaR}}_{\alpha} (X) + (1 - \lambda) \text{\emph{CVaR}}_{\alpha}^+ (X) .
	\end{equation}
\end{Proposition}
Note that \autoref{prop:CVaR_Convex_Combination_Formula} is valid for all loss distributions, including continuous ones. From \autoref{prop:CVaR_Convex_Combination_Formula} it follows that $\text{CVaR}_{\alpha}$ dominates $\text{VaR}_{\alpha}$, i.e. $\text{CVaR}_{\alpha} \geq \text{VaR}_{\alpha}$. In fact, $\text{CVaR}_{\alpha} > \text{VaR}_{\alpha}$, unless $\text{VaR}_{\alpha}$ is the maximum loss possible \cite[p. 1452]{Rockafellar2002_CVaR_for_general_loss_distributions}. Another result to emphasize is that the representation of CVaR by \autoref{eqn:CVaR_Convex_Combination_Formula} is rather surprising. As shown earlier, VaR is not a coherent risk measure (see \autoref{sec:CVaR_RM-Coherent_Risk_Measures}) and, in fact, neither is $\text{CVaR}^+$ \cite[week 8, p. 16]{Richtarik2015_OMF_Lecture}. However, both these incoherent risk measures are combined in the Convex Combination Formula to yield CVaR, which is coherent and therefore has many advantageous properties \cite[p. 1452]{Rockafellar2002_CVaR_for_general_loss_distributions}.

To provide a better understanding of the Convex Combination Formula (\autoref{eqn:CVaR_Convex_Combination_Formula}), an example of a discrete loss distribution will be presented. The losses $y_i$ with associated probabilities are given in \autoref{table:CVaR_Convex_Combination_Formula_Example}.
\begin{table}[H]
	\centering
	\begin{tabu}{| *{7}{c |} }
		\hline
		 i & 1 & 2 & 3 & 4 & 5 & 6\\
		 \hline
		 $y_i$ & 100 & 200 & 400 & 800 & 900 & 1000 \\
		 $P ( Y = y_i )$ & 0.1 & 0.2 & 0.5 & 0.18 & 0.01 & 0.01 \\
		 \hline
	\end{tabu}
	\caption{Discrete loss distribution of a random variable $Y$.}
	\label{table:CVaR_Convex_Combination_Formula_Example}
\end{table}

Now assume the 95 \% CVaR is to be determined. Since $F_Y(400) = P( Y \leq 400) = 0.8$ and $F_Y(800) = P( Y \leq 800) = 0.98$, it follows that $\text{VaR}_{0.95} (Y) = \min \{ c: P(Y \leq c) \geq 0.95\} = 800$ and $\lambda = \frac{0.98 - 0.95}{1 - 0.95} = \frac{3}{5}$. Also, $\text{CVaR}_{0.95}^+ (Y)$ can be calculated as $\frac{1}{2} \times 900 + \frac{1}{2} \times 1000 = 950$. Hence, applying \autoref{eqn:CVaR_Convex_Combination_Formula} gives
$$ \text{CVaR}_{0.95} (Y) = \frac{3}{5} \times 800 + \frac{2}{5} \times 950 = 860 .$$ \\


\section{Acerbi's Integral Formula} \label{sec:CVaR_RM-Acerbi}

Another way to express CVaR is to use Acerbi's integral formula.
\begin{Proposition}[{\cite[p. 329]{2014Chatterjee_Practical_Methods_FERM} Acerbi's Integral Formula for CVaR}] \label{prop:Acerbis_integral_formula}
	The CVaR of a random variable $X$, which represents loss, at the confidence level $\alpha$ can be expressed as
	\begin{equation} \label{eqn:Acerbis_integral_formula}
		\text{\emph{CVaR}}_{\alpha} (X) = \frac{1}{1 - \alpha} \int \limits_{\alpha}^{1} \text{\emph{VaR}}_{\beta}(X) \, d \beta .
	\end{equation}
\end{Proposition}
Hence, $\text{CVaR}_{\alpha}$ can also be interpreted as the average $\text{VaR}_{\beta}$ for $\beta \in [\alpha , 1]$ \cite[week 8, p. 33]{Richtarik2015_OMF_Lecture}. To demonstrate how \autoref{eqn:Acerbis_integral_formula} is applied, an example with a uniform loss distribution will be given. For this example, assume that the loss is distributed continuously and uniformly between 0 and 100, i.e., $X \sim U(0,100)$. Thus, $f_X (z) = \frac{1}{100}$ for $0 \leq z \leq 100$ and 0 elsewhere. The VaR at confidence level $\beta$ is given as $\text{VaR}_{\beta} (X) = 100 \times \beta$. Then the CVaR at confidence level $\alpha$ can be calculated as
\begin{align*}
	\text{CVaR}_{\alpha} (X) =& \frac{1}{1 - \alpha} \int \limits_{\alpha}^{1} \text{VaR}_{\beta} (X) \, d \beta = \frac{1}{1 - \alpha} \int \limits_{\alpha}^{1} 100 \times \beta \, d \beta \\
	=& \frac{100}{1 - \alpha} \left[ \frac{1}{2} \beta^2 \right]_{\alpha}^{1} = 50 \times (1 + \alpha) .
\end{align*}
So in this example, the 90 \% CVaR would be $\text{CVaR}_{0.9} (X) = 50 \times (1 + 0.9) = 95$.

\subsection{A New Proof of Acerbi's Integral Formula} \label{sec:CVaR_RM-Acerbi_new_proof}

Although Acerbi and Tasche proved \autoref{prop:Acerbis_integral_formula} in \cite[p. 1492]{Acerbi2002_Coherence_of_ES}, another proof will be given here. Two reasons for this alternative proof are, first, that Acerbi used different definitions in his paper, and second, to show how the result can be derived in another way. To the best knowledge of the author, this alternative proof has not been published before. However, the proof given here only holds for \emph{continuous} random variables and therefore lacks the generality of Acerbi's proof.

For this alternative proof, the probability density function of the generalized $\alpha$-tail distribution is needed, which can be derived from \autoref{eqn:generalized_alpha_tail} as $f^{\alpha}_X (z) = \frac{d}{dz} F^{\alpha}_X (z)$, i.e.,
\begin{equation} \label{eqn:generalized_alpha_tail_pdf}
f^{\alpha}_X (z) = \left\{
\begin{array}{l r}
	0, & \text{when~} z < \text{VaR}_{\alpha} (X) \\
	\frac{f_X (z)}{1 - \alpha}, & \text{when~} z \geq \text{VaR}_{\alpha} (X)
\end{array}
\right. .
\end{equation}

\begin{proof}
	(Continuous case only) Starting from the very basic definition of CVaR given in \autoref{def:CVaR_general_definition}, one can use integration by substitution to arrive at \autoref{eqn:Acerbis_integral_formula}:
	\begin{align*}
		\text{CVaR}_{\alpha} (X) =& \E [X^{\alpha}] \\
		=& \int \limits_{- \infty}^{\infty} z f_X^{\alpha}(z) dz \\
		=& \int \limits_{- \infty}^{\text{VaR}_{\alpha}(X)} z f_X^{\alpha}(z) dz + \int \limits_{\text{VaR}_{\alpha}(X)}^{\infty} z f_X^{\alpha}(z) dz . \\
	\end{align*}
	Using the definition of $f_X^{\alpha}(z)$ given in \autoref{eqn:generalized_alpha_tail_pdf}, the above equality simplifies to
	\begin{align*}
		\text{CVaR}_{\alpha} (X) =& \int \limits_{\text{VaR}_{\alpha}(X)}^{\infty} z \frac{f_X (z)}{1 - \alpha} dz . \\
	\end{align*}
	Now, one can define a new variable $\beta$, such that $\beta = F_X (z)$. Differentiating $\beta$ with respect to $z$ gives
	\begin{equation}
		\frac{d}{dz} \beta = f_X (z) \Longleftrightarrow f_X (z) dz = d \beta . \notag
	\end{equation}
	Furthermore, since $X$ is continuous, there is a one-to-one relationship between $\beta$ and $z$ and by \autoref{eqn:VaR}, $z$ can be expressed as $z = \text{VaR}_{\beta} (X)$. So substituting $\beta = F_X (z)$, $z = \text{VaR}_{\beta} (X)$, and adjusting the limits of the integral ($F_X ( \text{VaR}_{\alpha} (X) ) = \alpha$ and $F_X ( \infty ) = 1$) yields
	\begin{align*}
		\text{CVaR}_{\alpha} (X) =& \frac{1}{1 - \alpha} \int \limits_{\alpha}^{1} \text{VaR}_{\beta}(X) \, d \beta \,,
	\end{align*}	
	which completes the proof.
\end{proof}

%
%

\chapter{Portfolio Optimization Using CVaR} \label{chapter:CVaR_Portfolio_Optimization}

While \autoref{chapter:CVaR_as_risk_measure} introduced the CVaR concept for univariate random distributions, the concept can be extended to multivariate random distributions or random vectors as well. This will be done here with a focus on portfolio optimization, i.e. investment decisions where the investor is able to invest his funds in more than one asset. First, \autoref{sec:CVaR_PO_Markowitz} gives an introduction into portfolio optimization by presenting the first model that has been developed to improve decision making for portfolio investments \cite{Markowitz1952_Portfolio_Selection}, namely the \emph{Markowitz} or \emph{Mean Variance Model}. Then, \autoref{sec:CVaR_PO_Min_CVaR} introduces the \emph{CVaR Model} that has been developed by Rockafellar and Uryasev in \cite{Rockafellar2000_Optimization_of_CVaR}. It will also be explained why the CVaR Model is preferable to the Markowitz Model with regards to risk management. And finally, numerical examples will be given in \autoref{sec:CVaR_PO_Examples} to show how the two models can be applied in practice.\\

Before beginning with the first section, some notation will be established for the concepts that are used throughout this chapter and the rest of the dissertation.

First of all, the investor can invest in $N$ different assets. His investment decision can be represented mathematically by a \emph{decision vector} $\mathbf{x} \in S \subseteq \mathbb{R}^N$. Here, $S$  represents the feasible set for investment decisions.\footnote{For example, $S$ could have the unit budget constraint $\sum_i x_i = 1$, or a concentration risk constraint $x_j \leq 0.3 \sum_i x_i \, \forall j \leq N$. In the case of the budget unit constraint, $x_3 = 0.3$ means that 30 \% of available funds should be invested in asset number 3.}

To define the set of admissible portfolios $S$ for this chapter, the investor only has two constraints: He cannot short sell any assets and his decision needs to satisfy the unit budget constraint. With these considerations, the set of admissible portfolios $S$ which consists of $N$ assets can be as
\begin{equation} \label{eqn:admissible_portfolio_S}
	S = \left\{ \xinR{N} : x_i \geq 0  \,\, \forall \,\, i \in \{1, 2, \dots, N\} \,, \sum \limits_{i = 1}^N x_i = 1 \right\} .
\end{equation} ~\\

Also, the returns of each asset are random. Therefore, the losses can be expressed by a random loss vector $\mathbf{r} \in \mathbb{R}^N$,\footnote{Here, the losses are the negative values of returns. Hence, a negative $r_i$ means that asset $i$ is giving the investor a profit.} so that $r_i$ is a random variable that is distributed according to the loss distribution of the $i$th asset. Note that $r_i$ and $r_j$ for $i \not = j$ do not need to have the same distribution. Furthermore, $r_i$ and $r_j$ can be correlated (and in most cases are), which is why portfolio optimization is concerned with multivariate loss distributions.

So the loss $X$  that an investor can experience is a random variable that depends on the (random) losses of each asset and also on the investment in each asset, so that $X = X(\mathbf{x}, \mathbf{r})$.

For the following considerations, the investor demands a minimum expected return. Taking $\mathbf{r}$ as the vector of random losses, $\mathbf{x}$ the vector of investment decisions, and labelling the minimum required return $R$, the minimum expected return constraint can be formulated as
\begin{align} 
	\mathbf{x}^T \widehat{\mathbf{r}} & \leq - R \,, \label{eqn:minimum_expected_return}
\end{align}
where $\widehat{\mathbf{r}} = \E [ \mathbf{r} ]$.

\section{Mean Variance Optimization (Markowitz Model)} \label{sec:CVaR_PO_Markowitz}

Before modern portfolio theory was introduced by Markowitz in 1952 (\cite{Markowitz1952_Portfolio_Selection}), investment decisions were mostly made by an investor's belief.\footnote{Even after Markowitz's paper was published it took several decades to be adapted by the financial industry because computers did not have the necessary power to perform the calculations.} Although the expected return and variance of a single asset could be calculated, investors were not able to form optimal portfolios, i.e. assign their funds in such a way that the whole portfolio had preferable characteristics \cite{Wolf2015_FRM_Lecture}.

The most important contribution of \cite{Markowitz1952_Portfolio_Selection} is that it is favourable to diversify a portfolio because this will reduce the portfolio's standard deviation (risk) as long as the correlation between assets is less than 1. This result can be shown by a portfolio of $N$ assets \cite[p. 32]{Wolf2015_FRM_Lecture}.

Assume that an investor can buy $N$ assets, with expected returns $\widehat{r}_1 \,, \dots \,, \widehat{r}_N$ and variance $\sigma^2_1 \,, \dots \,, \sigma^2_N$. Assigning $x_i$ of his funds to the $i$th asset, the investor can expect a return of
\begin{align*}
	\E [\mathbf{x}^T \mathbf{r}] =& \sum \limits_{i=1}^N x_i \times \widehat{r}_i \,,
\end{align*}
which is the weighted average of expected asset returns. However, the risk for the investor can be lower than the weighted average of asset risks. To show this, the \emph{covariance matrix} $\mathbf{\Sigma} \in \mathbb{R}^{N \times N}$ of the random loss vector $\mathbf{r}$ will be introduced. $\mathbf{\Sigma}$ is defined as \cite[week 3, p. 11]{Richtarik2015_OMF_Lecture}
\begin{equation}
	\mathbf{\Sigma} \defeq \begin{bmatrix}
		\text{Var}(r_1) & \text{Cov}(r_1, r_2) & \cdots & \text{Cov}(r_1, r_N)  \\
		\text{Cov}(r_2, r_1)  & \text{Var}(r_2) & \cdots & \text{Cov}(r_2, r_N) \\
		\vdots & \vdots & \ddots & \vdots \\
		\text{Cov}(r_N, r_1) & \text{Cov}(r_N, r_2) & \cdots & \text{Var}(r_N) \\
		\end{bmatrix} \,,
		\notag
\end{equation}
where $\text{Var}(r_i) = \sigma^2_i$ was defined in \autoref{eqn:variance}. Using \autoref{eqn:correlation}, $\text{Cov}(r_i, r_j)$ can be expressed as 
\begin{equation}
	\text{Cov}(r_i, r_j) = \rho_{ij} \sigma_i \sigma_j \,, \notag
\end{equation}
which leads to the expression below. This expression is a standard result in financial literature but has been derived independently by the author:\footnote{In the standard financial literature, e.g. \cite{2010Bodie_Investments}, this result is usually derived for $N = 2$ assets but not for $N > 2$.}
\begin{align*}
	\sigma (\mathbf{x}^T \mathbf{r}) = \sqrt{\text{Var}(\mathbf{x}^T \mathbf{r})} =& \sqrt{\mathbf{x}^T \mathbf{\Sigma} \mathbf{x} } \\
	=& \sqrt{\sum \limits_{i = 1}^N x^2_i \sigma^2_i + \sum \limits_{i = 1}^{N-1} \sum \limits_{j = i+1}^N 2 \rho_{ij}x_i x_j \sigma_i \sigma_j} \\
	=& \sqrt{\sum \limits_{i = 1}^N x^2_i \sigma^2_i + \sum \limits_{i = 1}^{N-1} \sum \limits_{j = i+1}^N 2 x_i x_j \sigma_i \sigma_j - \sum \limits_{i = 1}^{N-1} \sum \limits_{j = i+1}^N 2 (1 - \rho_{ij}) x_i x_j \sigma_i \sigma_j  }\\
	=& \sqrt{ \left( \sum \limits_{i = 1}^N x_i \sigma_i \right)^2 - \sum \limits_{i = 1}^{N-1} \sum \limits_{j = i+1}^N 2 (1 - \rho_{ij}) x_i x_j \sigma_i \sigma_j } \\
	\leq& \sqrt{ \left( \sum \limits_{i = 1}^N x_i \sigma_i \right)^2} = \sum \limits_{i = 1}^N x_i \sigma_i \,,
\end{align*}
for $\mathbf{x} \in S$. The above inequality is strict whenever $\rho_{ij} < 1 \text{~for~} i \not = j$, meaning that the portfolio risk (given by the standard deviation) is less than the weighted average of asset risks whenever the assets are not perfectly correlated (which is usually the case).

Using Markowitz's findings, a quadratic programme can be formulated to find a minimum variance portfolio. Including the constraint given by \autoref{eqn:minimum_expected_return}, the programme can give the investor a portfolio which offers the required minimum return at the lowest possible risk. The inputs for the model are $\widehat{\mathbf{r}}$, the expected returns of assets $1, \dots, N$ and $\mathbf{\Sigma}$, the covariance matrix. Usually these inputs have to be estimated and one possibility of estimating the entries of the covariance matrix is given in \autoref{sec:Portfolio_Hedging-Background_on_Risk_Management} but a further discussion on parameter estimation is beyond the scope of this dissertation.

\begin{Definition}[{\cite[week 3, p. 15]{Richtarik2015_OMF_Lecture} Minimum Variance Portfolio}] \label{def:minimum_variance_portfolio}
	A minimum variance portfolio in the sense of \cite{Markowitz1952_Portfolio_Selection} is a portfolio which can be formed by solving
	\begin{Problem}[problem:MVO_optimization]
		\left.
		\begin{array}{rll}
			\min \limits_{\mathbf{x}} & \mathbf{x}^T \mathbf{\Sigma} \mathbf{x} \\
			\text{s.t.} & \mathbf{x}^T \widehat{\mathbf{r}} \leq - R\\
			& \mathbf{x} \in S
		\end{array}
		\right\} \,,
	\end{Problem}
	where $\mathbf{\Sigma}$ is the covariance matrix of the random loss vector $\mathbf{r}$, $\widehat{\mathbf{r}} = \E [ \mathbf{r} ]$, and $S$ is the set of admissible portfolios.
\end{Definition}

Since a covariance matrix $\mathbf{\Sigma}$ is always positive definite \cite[week 3, p. 13]{Richtarik2015_OMF_Lecture}, \autoref{problem:MVO_optimization} is a convex optimization problem. It has therefore either a unique solution or is infeasible. The only situation under which  \autoref{problem:MVO_optimization} becomes infeasible is when the required expected return is higher than any single expected return of the $N$ assets under consideration. \\

To see how the portfolio risk changes for different expected returns, one can solve \autoref{problem:MVO_optimization} for different values of $R$ (expected minimum return) and calculate the resulting portfolio risk (standard deviation). These risk/return pairs can be used to draw the \emph{efficient frontier}, which is ``a graph of the lowest possible [risk] that can be attained for a given portfolio expected return'' \cite[p. 220]{2010Bodie_Investments}.

For a sample portfolio of three assets with expected returns and covariance matrix
\begin{align*}
	\widehat{\mathbf{r}} =& \begin{bmatrix}
		-0.1073 \\
		-0.0737 \\
		-0.0627
		\end{bmatrix} & \text{and} &&
	\mathbf{\Sigma} =& \begin{bmatrix} 
				0.02778 & 0.00387 & 0.00021 \\
				0.00387 & 0.01112 & -0.00020 \\
				0.00021 & -0.00020 & 0.00115
				\end{bmatrix} \,,
\end{align*}
the efficient frontier is shown in \autoref{fig:Efficient_Frontier}.

\begin{figure}[H]
	\centering
	\includegraphics[width = 0.9 \textwidth]{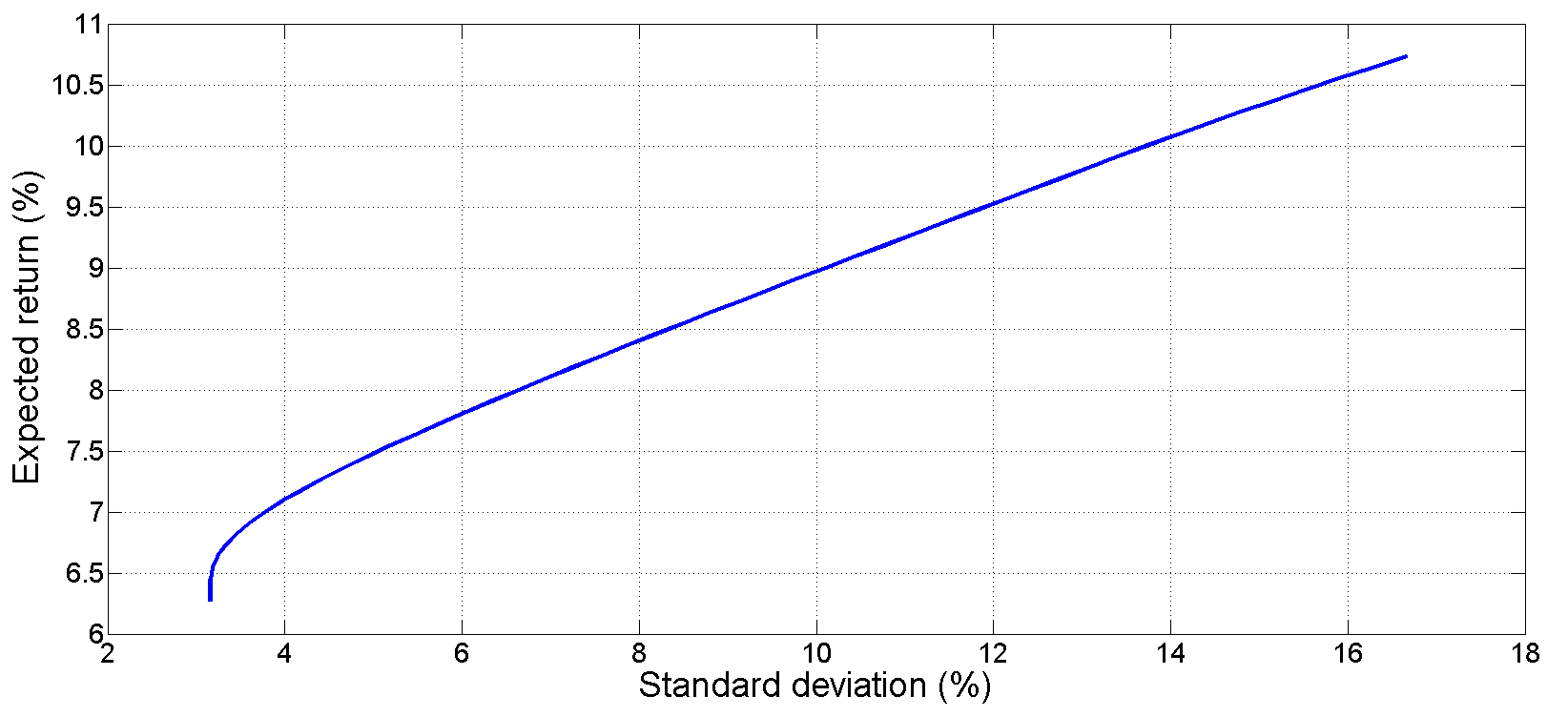}
	\caption{Efficient frontier for a sample portfolio.}
	\label{fig:Efficient_Frontier}
\end{figure}

Because of the quadratic term in the objective function of \autoref{problem:MVO_optimization}, an investor can increase his expected portfolio return with little additional risk if the portfolio has a low standard deviation to begin with. For example, increasing the expected return from 6.5 to 7 \% only increases the standard deviation by 0.6 \%. However, the more expected return an investor demands, the higher the increase in risk. Increasing the expected return from 9.5 to 10 \% requires an additional risk of 1.7 \%.

It is possible to form a portfolio with a risk/return profile that lies below the efficient frontier. However, it is not possible to form a portfolio whose risk/return profile is above or to the left of the efficient frontier in \autoref{fig:Efficient_Frontier} \cite[p. 220]{2010Bodie_Investments}.

\section{CVaR Optimization (Rockafellar and Uryasev Model)} \label{sec:CVaR_PO_Min_CVaR}

Despite revolutionizing risk management at its time, the Markowitz Model has some drawbacks regarding risk management. Two important disadvantages arise because it measures the risk in terms of variance of the portfolio:
\begin{enumerate}
	\item Variance is only a useful risk measure for normally (or symmetrically) distributed losses. Since variance is measured in either direction, tail losses arising from skewed loss distributions are not taken in account.
	\item Variance is not a coherent risk measure as it is not monotone.
\end{enumerate} ~

The first argument is illustrated in the second scenario of \autoref{sec:CVaR_PO_Examples}, while the second argument can easily be shown by an example: Consider two random variables (both representing loss) which are normally distributed, but with different $\mu$ and $\sigma^2$: $X \sim N(\mu_X = 0, \sigma_X^2 = 2)$ and $Y \sim N(\mu_Y = 10, \sigma_Y^2 = 1)$. The probability that $X$ is bigger than $Y$ is insignificantly small. To be precise, $P(Y \leq X) = 3.9 \times 10^{-9}$. Hence, it is nearly impossible that the loss of $X$ will exceed the loss of $Y$. However, $X$ has a higher variance than $Y$, i.e. $\text{Var}(X) = 2 \geq \text{Var}(Y) = 1$, and would therefore be considered riskier if the risk were measured by the variance. \\

Because of this, it is preferable for a risk manager to optimize the portfolio with regards to CVaR than with regards to variance. Rockafellar and Uryasev proposed a linear programme in \cite{Rockafellar2000_Optimization_of_CVaR} to optimize the CVaR of a portfolio. They also proved that under certain conditions the CVaR optimization will give the same optimal portfolio as the minimum variance optimization. The rest of this section introduces their notation and presents their results.\footnote{Although this section follows the outline of \cite{Rockafellar2000_Optimization_of_CVaR}, the expressions are closer aligned with \cite[week 8]{Richtarik2015_OMF_Lecture}.}

To derive later results, Rockafellar and Uryasev labelled the cumulative distribution function of losses $\Psi (\mathbf{x},c)$, so that for any given decision $\mathbf{x} \in S$, random asset losses $\mathbf{r} \in \mathbb{R}^n$, and loss distribution $X(\mathbf{x}, \mathbf{r})$,
	\begin{align}
		\Psi (\mathbf{x},c) &= F_X (c) = P ( X(\mathbf{x}, \mathbf{r}) \leq c) & \text{in the general case, and} \label{eqn:cdf_Losses_general} \\
	\Psi (\mathbf{x},c) &= F_X (c) = \int \limits_{\mathbf{r}: X(\mathbf{x},\mathbf{r}) \leq c} p(\mathbf{r}) d \mathbf{r} & \text{in the continuous case,} \label{eqn:cdf_Losses_continuous}
	\end{align}
where $p(\mathbf{r})$ in \autoref{eqn:cdf_Losses_continuous} is the pdf for a continuous $\mathbf{r}$. The function $\Psi (\mathbf{x},c)$ can be interpreted as the probability that the losses do not exceed threshold $c$.

Continuing with the notation of $\Psi (\mathbf{x},c)$ as the threshold of losses, $\text{VaR}_{\alpha}$ and $\text{CVaR}_{\alpha}$ of an investment decision $\mathbf{x}$ can be then written as
\begin{align}
	\text{VaR}_{\alpha} ( \mathbf{x} ) &= \text{VaR}_{\alpha} (  X(\mathbf{x}, \mathbf{r}) ) = \min \{ c : \Psi (\mathbf{x},c) \geq \alpha \} \label{eqn:VaR_in_terms_of_Psi} \text{, and}\\
	\text{CVaR}_{\alpha} ( \mathbf{x} ) &= \text{CVaR}_{\alpha} (  X(\mathbf{x}, \mathbf{r}) ) = \E_{\mathbf{r}} [  X(\mathbf{x}, \mathbf{r}) \mid  X(\mathbf{x}, \mathbf{r}) \geq \text{\emph{VaR}}_{\alpha} ( \mathbf{x} )] \label{eqn:CVaR_in_terms_of_Psi}.
\end{align}
Rockafellar and Uryasev characterized \autoref{eqn:VaR_in_terms_of_Psi} and \autoref{eqn:CVaR_in_terms_of_Psi} in terms of a function
\begin{equation} \label{eqn:phi_alpha}
	\phi_{\alpha} ( \mathbf{x}, c) \defeq c + \frac{1}{1 - \alpha} \E \left[ ( X( \mathbf{x}, \mathbf{r} ) - c)^+ \right],
\end{equation}
where $\E \left[ \cdot \right]$ is the expectation and $(t)^+ = \max \{ 0, t \}$. Based on \autoref{eqn:phi_alpha}, they formulated \autoref{theorem:CVaR}, the most important result of \cite{Rockafellar2000_Optimization_of_CVaR}.
\begin{Theorem}[{\cite[p. 24]{Rockafellar2000_Optimization_of_CVaR}}] \label{theorem:CVaR}
	As a function of $c$, $\phi_{\alpha} ( \mathbf{x}, c)$ is convex and continuously differentiable. The $\text{\emph{CVaR}}_{\alpha}$ of the loss associated with any $\mathbf{x} \in S$ can be determined from the formula
	\begin{equation} \label{eqn:CVaR_Theorem-CVaR_from_Psi}
		\text{\emph{CVaR}}_{\alpha} (\mathbf{x}) = \min_{c \in \mathbb{R}} \phi_{\alpha} ( \mathbf{x}, c) .
	\end{equation}
	Furthermore, let $\Phi_{\alpha}^* (\mathbf{x}) \defeq \arg \min_c  \phi_{\alpha} ( \mathbf{x}, c)$, i.e. $\Phi_{\alpha}^* (\mathbf{x})$ is the set of minimizers of $\phi_{\alpha} ( \mathbf{x}, c)$. Then
	\begin{equation} \label{eqn:CVaR_Theorem-VaR_from_Psi}
		\text{\emph{VaR}}_{\alpha} (\mathbf{x}) = \min \{ c: c \in \Phi_{\alpha}^* (\mathbf{x}) \} .
	\end{equation}	
	And following from \autoref{eqn:CVaR_Theorem-CVaR_from_Psi} and \autoref{eqn:CVaR_Theorem-VaR_from_Psi}, the following equation always holds:
	\begin{equation} \label{eqn:CVaR_Theorem-CVaR_Final_Expression}
		\text{\emph{CVaR}}_{\alpha} (\mathbf{x}) = \phi_{\alpha} ( \mathbf{x}, \text{\emph{VaR}}_{\alpha} (\mathbf{x}) ) .
	\end{equation}
\end{Theorem}
The proof of \autoref{theorem:CVaR} is given in the appendix of \cite{Rockafellar2000_Optimization_of_CVaR}. Based on \autoref{theorem:CVaR}, Rockafellar and Uryasev stated another theorem, which is useful for the computational calculation to find a CVaR optimal portfolio $\mathbf{x}^* \in S$.
\begin{Theorem}[{\cite[p. 25 f.]{Rockafellar2000_Optimization_of_CVaR}}] \label{theorem:CVaR_Optimization}
	Let $S$ be a convex set of feasible decisions $\mathbf{x}$ and assume that $X(\mathbf{x}, \mathbf{r})$ is convex in $\mathbf{x}$. Then minimizing the $\text{\emph{CVaR}}_{\alpha}$ of the loss associated with decision $\mathbf{x} \in S$ is equivalent to minimizing $\phi_{\alpha} ( \mathbf{x}, c)$ over all $(\mathbf{x}, c) \in S \times \mathbb{R}$, in the sense that
	\begin{Problem}[eqn:CVaR_Optimization_1]
		\min \limits_{\mathbf{x} \in S} \text{\emph{CVaR}}_{\alpha} (\mathbf{x}) = \min \limits_{(\mathbf{x}, c) \in S \times \mathbb{R}} \phi_{\alpha} ( \mathbf{x}, c) \,,
	\end{Problem}
	where, moreover, a pair $(\mathbf{x}^*, c^*)$ achieves the right hand side minimum if and only if $\mathbf{x}^*$ achieves the left hand side minimum and $c^* \in \Phi_{\alpha}^* (\mathbf{x})$. Therefore, in circumstances where the interval $\Phi_{\alpha}^* (\mathbf{x})$ reduces to a single point (as is typical), the minimization of $\phi_{\alpha} ( \mathbf{x}, c)$ produces a pair $(\mathbf{x}^*, c^*)$ such that $\mathbf{x}^*$ minimizes the $\text{\emph{CVaR}}_{\alpha}$ and $c^*$ gives the corresponding $\text{\emph{VaR}}_{\alpha}$.
\end{Theorem}

\autoref{theorem:CVaR_Optimization} not only gives a way to express the CVaR minimization problem in a tractable form, but also allows to calculate $\text{CVaR}_{\alpha}$ without having to calculate $\text{VaR}_{\alpha}$ first, as would have been the case with \autoref{def:CVaR}. More remarkably, finding the CVaR by using \autoref{theorem:CVaR_Optimization}, gives the corresponding VaR as a by-product \cite[p. 25 f.]{Rockafellar2000_Optimization_of_CVaR}.

Applying \autoref{theorem:CVaR_Optimization}  with \autoref{eqn:phi_alpha}, the investment decision $\mathbf{x}$ that minimizes the Conditional Value-at-Risk of a portfolio at the confidence level $\alpha$ can be expressed as \cite[week 8, p. 21]{Richtarik2015_OMF_Lecture}
\begin{Problem}[eqn:CVaR_Definition_Tractable]
	\min \limits_{\mathbf{x} \in S} \text{CVaR}_{\alpha} (\mathbf{x}) = \min \limits_{\mathbf{x} \, \in S, c \, \in \mathbb{R}} \left( c + \frac{1}{1 - \alpha} \E \left[ ( X(\mathbf{x}, \mathbf{r}) - c )^+ \right] \right) .
\end{Problem}

To provide a better understanding of how to solve \autoref{eqn:CVaR_Definition_Tractable}, a one-dimensional example will be given, i.e. there is only asset with a univariate, discrete loss distribution. Since there is only one asset to consider, $\mathbf{x} = [1]$. Because of this, it is not the goal in this example to find the optimal portfolio composition, but rather to find the VaR and CVaR using \autoref{theorem:CVaR_Optimization}. The asset has the loss distribution of $Y$ given in \autoref{table:CVaR_Convex_Combination_Formula_Example}. The table is reproduced below for convenience.
\begin{table}[H]
	\centering
	\begin{tabu}{| *{7}{c |} }
		\hline
		 i & 1 & 2 & 3 & 4 & 5 & 6\\
		 \hline
		 $y_i$ & 100 & 200 & 400 & 800 & 900 & 1000 \\
		 $P ( Y = y_i )$ & 0.1 & 0.2 & 0.5 & 0.18 & 0.01 & 0.01 \\
		 \hline
	\end{tabu}
\end{table}

For this asset, the function $\phi_{\alpha} ( \mathbf{x}, c) = c + \frac{1}{1 - \alpha} \E \left[ ( X( \mathbf{x}, \mathbf{r} ) - c)^+ \right]$ will be drawn against $c$ to find $\text{CVaR}_{\alpha} (\mathbf{x}) = \min \limits_{c \in \mathbb{R}} \phi_{\alpha} ( \mathbf{x}, c)$ graphically. The graph of $\phi_{\alpha} ( \mathbf{x}, c)$ for $\alpha = 0.95$ is shown in \autoref{fig:phic_vs_c}.

\begin{figure}[H]
	\centering
	\includegraphics[width = 0.9 \textwidth]{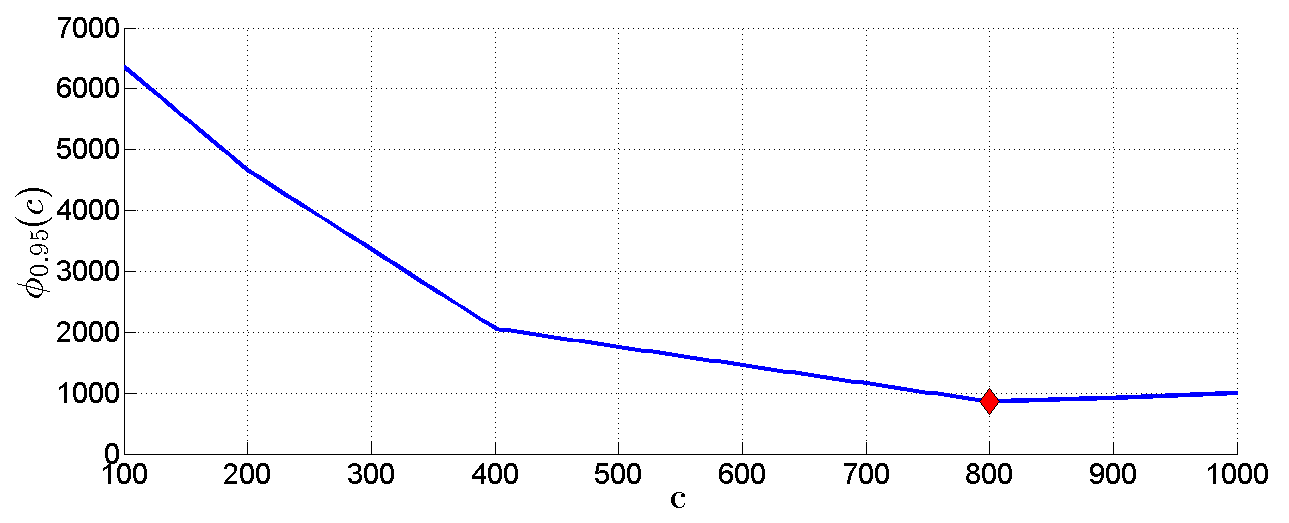}
	\caption{Function value $\phi_{0.95} (c)$ of $Y$ for different values of $c$.}
	\label{fig:phic_vs_c} 
\end{figure}

The graph shows that the minimum of $\phi_{\alpha} ( \mathbf{x}, c)$ occurs at $c^* = 800$. Thus, $ \min \limits_{c \in \mathbb{R}} \phi_{\alpha} ( \mathbf{x}, c) = \phi_{\alpha} ( \mathbf{x}, 800) = 860$. Hence, by \autoref{theorem:CVaR_Optimization}, it follows that $\text{VaR}_{0.95} = 800$ and $\text{CVaR}_{0.95} = 860$, which agrees with the results of the Convex Combination Formula in \autoref{sec:CVaR_RM-Closer_Analysis_of_CVaR} as expected. Another characteristic to point out is that $\phi_{\alpha} ( \mathbf{x}, c)$ has ``kinks'' at points $y_i \,, i = 1, \dots, 6$ \cite[week 8, p. 22]{Richtarik2015_OMF_Lecture}.

\autoref{eqn:CVaR_Definition_Tractable} is still difficult to evaluate if the loss distribution $X$ is continuous. One remedy is to use Monte Carlo Sampling to draw $K$ i.i.d. samples of the loss vector $\mathbf{r}$ ($\mathbf{r}_k \,, k \in \{1, 2, \dots, K\}$) from the distribution of $\mathbf{r}$, so that  \autoref{eqn:CVaR_Definition_Tractable} can be written in a tractable LP form \cite[week 8, p. 29]{Richtarik2015_OMF_Lecture}. Adding constraint \ref{eqn:minimum_expected_return} to ensure a minimum expected return for the investor, the tractable LP form of the optimization problem is given as
\begin{Problem}[problem:CVaR_Definition_in_LP_including_min_return]
\left.
\begin{array}{rll}
	\min \limits_{c, \mathbf{z}} & c + \frac{1}{K (1 - \alpha)} \sum \limits_{k = 1}^K z_k \\
	\text{s.t.} & z_k \geq \mathbf{x}^T \mathbf{r}_k - c & \text{for~} k \in \{1, \dots , K\} \\
	& z_k \geq 0 & \text{for~} k \in \{1, \dots , K\} \\
	& \mathbf{x}^T \widehat{\mathbf{r}} \leq - R\\
	& \mathbf{x} \in S
\end{array}
\right\} .
\end{Problem}

Another interesting link between mean variance and CVaR optimization was established in \cite{Rockafellar2000_Optimization_of_CVaR} as well. Rockafellar and Uryasev proposed that under certain conditions, \autoref{problem:MVO_optimization} and \autoref{eqn:CVaR_Definition_Tractable} give the same optimal portfolio.
\begin{Proposition}[{\cite[p. 29]{Rockafellar2000_Optimization_of_CVaR}}] \label{prop:MVO_CVaR_common_optimal_portfolio}
	Suppose that the loss associated  with each $\mathbf{x}$ is normally distributed as holds when $\mathbf{r}$ is normally distributed. If $\alpha \geq 0.5$ and the constraint \ref{eqn:minimum_expected_return} is active at solutions to \autoref{problem:MVO_optimization} and \autoref{eqn:CVaR_Optimization_1}, then the solutions to those problems are the same; a common portfolio $\mathbf{x}^*$ is optimal by both criteria.
\end{Proposition}

This means that under the conditions stated in the proposition, it is possible to find the minimum variance portfolio by finding the minimum CVaR portfolio. \autoref{prop:MVO_CVaR_common_optimal_portfolio} will be explored in the first scenario of \autoref{sec:CVaR_PO_Examples}.

\section{Numerical Examples}  \label{sec:CVaR_PO_Examples}

This section gives numerical examples for finding minimum CVaR portfolios. More precisely, the CVaR criterion will be compared to the minimum variance criterion (as formulated by Markowitz in \cite{Markowitz1952_Portfolio_Selection}, see \autoref{def:minimum_variance_portfolio}) and two scenarios will be given to show the effect of the criterion on the portfolio composition. The first scenario is adapted from \cite{Rockafellar2000_Optimization_of_CVaR} and concerns normally distributed losses. The second scenario is a theoretical construct with a positively skewed loss distribution. \\

\subsection*{First Scenario: Normally Distributed Losses}

This scenario serves to display the proposition by Rockafellar and Uryasev that for certain conditions the minimum variance optimization and CVaR optimization give the same optimal portfolio $\mathbf{x}^*$:

In the example from \cite[p. 29 ff.]{Rockafellar2000_Optimization_of_CVaR}, three assets ($N = 3$) are available: The S\&P 500 index ($x_1$), long-term US government bonds ($x_2$), and a portfolio of small cap stocks ($x_3$). The expected return of each asset and their covariance matrix is given in \autoref{table:CVaR_Risk_measure-mean_asset_returns} and \autoref{table:CVaR_Risk_measure-portfolio_covariance}, respectively.

\begin{table}[H]
	\centering
	\begin{tabular}{| >{$}c<{$} l | r |}
	\hline
	& \textbf{Asset} & \textbf{Mean Loss} \\
	\hline
	x_1 & S\&P 500 & - 0.0101110 \\
	x_2 & Gov. bond & - 0.0043532 \\
	x_3 & Small Cap & - 0.0137058 \\
	\hline
	\end {tabular}
	\caption{Mean Asset Losses of S\&P, Government Bonds, and Small Cap.}
	\label{table:CVaR_Risk_measure-mean_asset_returns}
\end{table}

\begin{table}[H]
	\centering
	\begin{tabular}{| >{$}c<{$} l | c c c |}
	\hline
	\multicolumn{2}{|c |}{\textbf{Covariance}} & $x_1$ & $x_2$ & $x_3$ \\
	\multicolumn{2}{|c |}{\textbf{Matrix}} & S\&P 500 & Gov. bond & Small Cap \\
	\hline
	x_1 & S\&P 500 & 0.00324625 & 0.00022983 & 0.00420395 \\
	x_2 & Gov. bond & 0.00022983 & 0.00049937 & 0.00019247 \\
	x_3 & Small Cap & 0.00420395 & 0.00019247 & 0.00764097 \\
	\hline
	\end {tabular}
	\caption{Covariance Matrix of S\&P, Government Bonds, and Small Cap.}
	\label{table:CVaR_Risk_measure-portfolio_covariance}
\end{table}

Using the CVX package in MATLAB, the minimum variance portfolios (MV opt) and minimum CVaR portfolios (CVaR opt) are calculated for expected minimum returns of 0.6\%, 0.9\%, and 1.1\%. To calculate the minimum CVaR portfolio for $\alpha = 0.95$, 100,000 Monte Carlo simulations were run to estimate the loss distribution. The results are given in \autoref{table:MVO_CVaR_Optimal_Portfolios_different_R}.
\begin{table}[H]
	\centering
	\begin{tabular}{| c *{3}{| r r} |}
	\hline
	\parbox{2.5cm}{\vspace{3pt} \centering \textbf{Required} \\ \textbf{return} \vspace{3pt}}&  \multicolumn{2}{|c }{\textbf{0.6 \%}} &  \multicolumn{2}{|c }{\textbf{0.9 \%}} &  \multicolumn{2}{|c |}{\textbf{1.1 \%}}\\
	\hline
	\hline
Portfolio: & MV opt & $\text{CVaR}_{0.95}$ opt & MV opt & $\text{CVaR}_{0.95}$ opt & MV opt & $\text{CVaR}_{0.95}$ opt \\
	\hline
S \& P	&	17.54	\%	&	17.28	\%	&	34.19	\%	&	34.82	\%	&	45.15	\%	&	46.20	\%	\\
Gov. Bonds	&	75.65	\%	&	75.75	\%	&	37.18	\%	&	36.93	\%	&	11.58	\%	&	11.52	\%	\\
Small Cap	&	6.81	\%	&	6.97	\%	&	28.64	\%	&	28.25	\%	&	43.27	\%	&	43.18	\%	\\
	\hline
	\end {tabular}
	\caption{Minimum Variance and Minimum CVaR portfolios for different required returns.}
	\label{table:MVO_CVaR_Optimal_Portfolios_different_R}
\end{table}

Comparing the two portfolios for different levels of required return, one can see that their compositions only vary slightly (although they should be identical). The reason they are not completely identical is because the minimum variance portfolio was computed analytically, while Monte Carlo simulations were used to calculate the CVaR optimal portfolio. Otherwise, they can be considered identical, as was stated in \autoref{prop:MVO_CVaR_common_optimal_portfolio}.

\subsection*{Second Scenario: Positively Skewed Loss Distribution}

In this subsection, the effect of the portfolio selection criterion is analysed when the loss distributions are not normal. Therefore, two further characteristics are needed to describe their distribution They are named \emph{skewness} and \emph{kurtosis}, respectively:
\begin{Definition}[{\cite[p. 22]{2012Kaltenbach_Statistics} Skewness}] \label{def:skewness}
	The \emph{skewness} of a random variable $X$ is defined as
	\begin{align}
		\text{\emph{skew}~}(X) & \defeq \E \left[ \left( \frac{X - \mu}{\sigma} \right)^3 \right] \,. \label{eqn:skewness}
	\end{align}
\end{Definition}
\begin{Definition}[{\cite[p. 22]{2012Kaltenbach_Statistics} Kurtosis}] \label{def:kurtosis}
	The \emph{kurtosis}\footnote{Some texts subtract 3 from the fourth central (normalized) moment when they define the kurtosis - so that the normal distribution has a kurtosis of 0. This convention is \emph{not} followed in this dissertation.} of a random variable $X$ is defined as
	\begin{align}
		\text{\emph{kurt}~}(X) & \defeq \E \left[ \left( \frac{X - \mu}{\sigma} \right)^4 \right] \,. \label{eqn:kurtosis}
	\end{align}
\end{Definition}

A skewness of 0 means that the distribution of $X$ is symmetrical about its mean $\mu$, while a negative skewness indicates that values of X below $\mu$ are more likely and a positive skewness means that values of X greater than $\mu$ are more probable. Kurtosis measure how the variance is affected by extreme deviations from the mean. A high kurtosis shows that a high variance is caused by few extreme deviations from the mean $\mu$ \cite[p. 22 f.]{2012Kaltenbach_Statistics}.\\

In this scenario, four assets will be considered (called Index, Bonds, Mid Cap, Emerging Markets Stocks) and the following assumptions will be made:
\begin{itemize}
	\item The loss distributions of the four assets are independent of each other, i.e. their correlations are 0.
	\item The loss distributions of the first three assets have the same mean and variance as in the previous scenario. The fourth assets has higher mean and variance than the previous three.
	\item The minimum variance and minimum CVaR portfolios are formed the same way as in the previous scenario.
	\item Two cases will be considered: In the first case, all single loss distributions are normal, i.e. they have skewness 0. In the second case, all loss distributions are positively skewed, i.e. high losses are more likely than high profits.
\end{itemize}

The first assumption is highly theoretical, as in any real world setting there exists at least some correlation. However, uncorrelated assets are very favourable in portfolio diversification as this reduces the combined variance significantly. The second and third assumption create a link between this scenario and the previous one. Hence, the effects can be better compared. Finally, the fourth assumption should show the dangers of using minimum variance optimization in the cases where losses are not normally distributed. The first case (in which losses are normally distributed) serves as a benchmark portfolio for the second case with skewed loss distributions.

The loss distributions will be characterized by their mean, variance, skewness, and kurtosis (see \autoref{table:Loss_Distributions_Scenario2}). The implementation of these random losses in MATLAB will be done with the function \href{http://uk.mathworks.com/help/stats/pearsrnd.html}{\emph{pearsrnd}} and the loss distributions for the single assets in both cases are shown in \autoref{app_diagrams:Loss_distributions_single_assets_scenario2}. 

\begin{table}[H]
	\centering
	\begin{tabular}{| >{$}c<{$} l | c c c c c |}
	\hline
	\multicolumn{2}{|c |}{\textbf{Distribution}} & \multicolumn{2}{|c }{~} & \multicolumn{2}{c }{skewness} & \\
	\multicolumn{2}{|c |}{\textbf{Parameters}} & $\mu$ & $\sigma^2$ & case 1 & case 2 & kurtosis \\
	\hline
	x_1 & Index & - 0.0101110 & 0.00324625 & 0 & 0.7 & 3 \\
	x_2 & Bonds &  - 0.0043532 & 0.00049937 & 0 & 0.7 & 3\\
	x_3 & Mid Cap & - 0.0137058 & 0.00764097 & 0 & 0.7 & 3\\
	x_4 & EMS & -0.018 & 0.01 & 0 & 0.7 & 3 \\
	\hline
	\end {tabular}
	\caption{Characterization of loss distributions used in second scenario.}
	\label{table:Loss_Distributions_Scenario2}
\end{table}

For all simulations and both cases, a minimum return of $- 0.006$ was required. For both cases (no skewness and skewness = 0.7), the minimum variance optimal portfolio is the same, while the minimum CVaR portfolio differs: In both cases, even with normally distributed losses, it is different from the minimum variance portfolio. In the first case the portfolio is different because the minimum return constraint is not active. It differs more strongly in the case of skewed distributions, as the CVaR optimization programme (\autoref{problem:CVaR_Definition_in_LP_including_min_return}) takes the skewness of the losses into account when forming the optimal portfolio, while the minimum variance programme (\autoref{problem:MVO_optimization}) does not. The respective optimal portfolios are shown in \autoref{table:MVO_CVaR_Optimal_Portfolios_scenario2} below.

\begin{table}[H]
	\centering
	\begin{tabular}{| c *{2}{| r r} |}
	\hline
	\textbf{Case} &  \multicolumn{2}{|c }{\textbf{1, skewness = 0}} &  \multicolumn{2}{|c |}{\textbf{2, skewness = 0.7}} \\
	\hline
	\hline
	Portfolio: & MV opt & $\text{CVaR}_{0.95}$ opt & MV opt & $\text{CVaR}_{0.95}$ opt  \\
	\hline
	Index	& 12.12 \% & 13.34 \% & 12.12 \% & 14.36	\% \\
	Bonds	& 78.80 \% & 75.36 \% & 78.80 \% & 72.87	\% \\
	Mid Cap & 5.15 \%	& 6.15 \% & 5.15 \% & 6.95	\% \\
	EMS & 3.93 \% & 5.15 \% & 3.93 \% & 5.82 \% \\
	\hline
	\end {tabular}
	\caption{Minimum Variance and Minimum CVaR portfolios for scenario 2.}
	\label{table:MVO_CVaR_Optimal_Portfolios_scenario2}
\end{table}

Although the loss distributions for both optimal portfolios are very similar in both cases (see \autoref{app_diagrams:Loss_distributions_portfolios_scenario2}), the CVaR optimal portfolio shows a better performance for the 100,000 simulations. Among other performance and risk measures, \emph{Expected Loss (EL)} will also be considered. The definition of EL is given below.
\begin{Definition}[{\cite[p. 23]{Fragniere2015_FRM_Lecture} Expected Loss (EL)}] \label{def:expected_loss}
	Let $X$ be a random variable representing loss. The expected loss of $X$ is defined as
	\begin{equation} \label{eqn:expected_loss}
		\text{\emph{EL}} (X) = \E [ X \mid X \geq 0 ] .
	\end{equation}
\end{Definition}
Hence, the expected loss is the average loss, given that there is a loss. In this sense EL is similar to CVaR but with the difference that the condition for the expectation is different. A summary of several performance and risk indicators for both optimal portfolios is given in \autoref{table:Performance_of_Optimal_Portfolios_scenario2}.

\begin{table}[H]
	\centering
	\begin{tabular}{| c *{2}{| r r} |}
	\hline
	\textbf{Case} &  \multicolumn{2}{|c }{\textbf{1, skewness = 0}} &  \multicolumn{2}{|c |}{\textbf{2, skewness = 0.7}} \\
	\hline
	\hline
	Portfolio: & MV opt & $\text{CVaR}_{0.95}$ opt & MV opt & $\text{CVaR}_{0.95}$ opt  \\
	\hline
	Expected Return $\mu$ & -0.0061 & -0.0064 &-0.0061 & -0.0064 \\
	Standard Deviation $\sigma$ & 0.0198 & 0.0199 & 0.0198 & 0.0199 \\
	Expected Loss & 0.0137 & 0.0137 & 0.0148 & 0.0147 \\
	$\text{VaR}_{0.95}$ & 0.0265 & 0.0263 & 0.0265 & 0.0263 \\
	$\text{CVaR}_{0.95}$ & 0.0347 & 0.0345 & 0.0398 & 0.0393 \\
	\hline
	\end {tabular}
	\caption{Performance and risk indicators of optimal portfolios for scenario 2.}
	\label{table:Performance_of_Optimal_Portfolios_scenario2}
\end{table}

\autoref{table:Performance_of_Optimal_Portfolios_scenario2} shows that the performance and risk measures for each optimal portfolio and each different case. In both cases, the investor can expect a higher profit when using a CVaR optimal portfolio. The standard deviation of returns is slightly higher for the CVaR optimal portfolio than for the minimum variance portfolio (0.0199 vs. 0.0198). However, for all other risk measures that were considered, the CVaR optimal portfolio has lower or equal risk than the minimum variance portfolio (to 4 decimal places). Hence, in this setting it would be favourable for the investor to use the CVaR optimal portfolio, as he can achieve a higher return with the same or less risk if he uses either of EL, VaR, or CVaR as the risk measure.

\clearpage

%
%

\chapter{Portfolio Hedging using CVaR} \label{chapter:Portfolio_Hedging}

\autoref{chapter:CVaR_as_risk_measure} stated the definition of CVaR, explained its properties and \autoref{sec:CVaR_PO_Min_CVaR} gave a computationally tractable optimization programme to calculate CVaR optimal investment portfolios, for which corresponding examples were given in \autoref{sec:CVaR_PO_Examples}. In \cite[p. 32 ff.]{Rockafellar2000_Optimization_of_CVaR}, Rockafellar and Uryasev (later followed by other authors, e.g. \cite{Albrecht2015_Tail_Risk_Hedging}, \cite{Bardou2013_CVaR_Hedging_Using_Quatization_Algorithm}, \cite{Topaloglou2002_CVaR_Models_with_selective_hedging}, and \cite{Xue2015_Optimal_Inventory_and_Hedging_Decisions}) expanded the use of CVaR to hedge against potential losses that arise from a previous investment decision. A possible scenario for this application is when a trader entered a position only looking at potential gains but disregarding possible losses. The risk manager might then intervene to hedge against the potential losses, i.e. minimizing the trader's risk while still maintaining acceptable potential gains.

This chapter will start by introducing the basic notions of options and financial risk management methods in \autoref{sec:Portfolio_Hedging-Background_on_Options} and \autoref{sec:Portfolio_Hedging-Background_on_Risk_Management}, followed by applying the hedging procedure that Rockafellar and Uryasev used\footnote{The example used was taken from \cite[p. 172 ff.]{Mauser1999_Beyond_VaR}.} to call and put options on Google and Yahoo traded on 21 July 2015.\footnote{The ticker symbols for the underlying equity are NASDAQ:GOOGL and NASDAQ:YHOO.} Based on the available data as of 21 July 2015, two \emph{strangles} are formed and described in \autoref{sec:Portfolio_Hedging-Forming_Strangle}, while the subsequent hedging procedure is described and applied in \autoref{sec:Portfolio_Hedging-Hedging_Strangle}.

\section{Background on Options} \label{sec:Portfolio_Hedging-Background_on_Options}

In \autoref{chapter:CVaR_Portfolio_Optimization}, investments in an index fund, bonds and equity were considered when forming the portfolio. These securities are basic investment possibilities, which are easy to understand as their payoff is directly linked to their market value. This means that if the price of a common share of Google rises (or falls) by 1 \%, an investor who invested all his funds into Google shares makes a profit (or loss) of 1 \% as well. 

Derivatives, such as call and put options,\footnote{Other derivative securities are for example futures or swaps. For more information on those and other derivatives please refer to \cite{2012Hull_Options_Futures}.} are ``securities whose prices are determined by, or 'derive [sic] from,' the prices of other securities'' \cite[p. 678]{2010Bodie_Investments}. Since these prices do not need to depend linearly on the price of the underlying, their payoff profile can be more complicated than the payoff of bonds or equity.

\begin{Definition}[{\cite[p. 679]{2010Bodie_Investments} Call Option}] \label{def:call_option}
	A \emph{call option} gives its holder the right to \emph{purchase} an asset for a specified price, called \emph{strike price}, on the specified expiration date.\footnote{This is known as a \emph{European} option. American options can be exercised at any time before the expiration date.}
\end{Definition}
\begin{Definition}[{\cite[p. 690]{2010Bodie_Investments} Put Option}] \label{def:put_option}
	A \emph{put option} gives its holder the right to \emph{sell} an asset for a specified price, called \emph{strike price}, on the specified expiration date.
\end{Definition}

For stock options, one option contact gives the holder to the right to buy (call option) or sell (put option) 100 shares at the specified priced \cite[p. 199]{2012Hull_Options_Futures}.\footnote{In the following example, only stock options will be considered} For any type of option, four basic positions can be taken (these positions can be combined to give more complex option strategies, e.g. a spread or a strangle) \cite[p. 197]{2012Hull_Options_Futures}:
\begin{enumerate}
	\item A long position in a call option (i.e. buying a call option)
	\item A short position in a call option (i.e. selling a call option)
	\item A long position in a put option (i.e. buying a put option)
	\item A short position in a put option (i.e. selling a put option)
\end{enumerate} ~

The payoff and profit profiles for each of the four basic option positions are given in \autoref{fig:payoff_profile_call} and \autoref{fig:payoff_profile_put} below.

\begin{figure}[H]
	\centering
	\includegraphics[width = \textwidth]{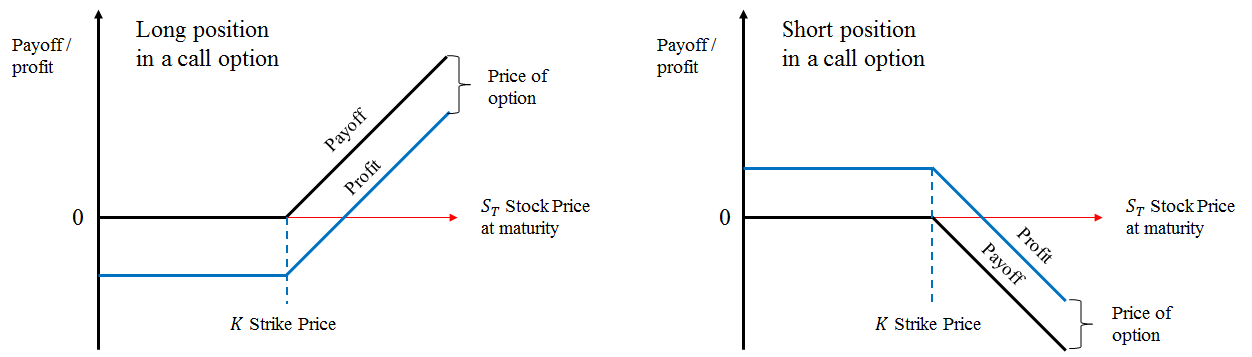}
	\caption{Reproduced from \cite[p. 198]{2012Hull_Options_Futures}, payoff and profit profile for a call option.}
	\label{fig:payoff_profile_call}
\end{figure}

Denoting $K$ the strike price, $S_T$ the price of the underlying stock at maturity, and $p_C$ the price of the call, the payoff and profit of a long position in a call option can be expressed as \cite[p. 198]{2012Hull_Options_Futures}
\begin{align}
	\text{Payoff}_{\text{Long Call}} &= \max \{ S_T - K, 0\} \label{eqn:payoff_long_call} \\
	\text{Profit}_{\text{Long Call}} &= \max \{ S_T - K, 0\} - p_C \label{eqn:profit_long_call} 	
\end{align}
The payoff and profit for a short position are the negatives of \autoref{eqn:payoff_long_call} and \autoref{eqn:profit_long_call} and can be expressed as \cite[p. 198]{2012Hull_Options_Futures}
\begin{align}
	\text{Payoff}_{\text{Short Call}} &= \min \{K - S_T, 0\} \label{eqn:payoff_short_call} \\
	\text{Profit}_{\text{Short Call}} &= \min \{K - S_T, 0\} + p_C \label{eqn:profit_short_call} 	
\end{align}

\begin{figure}[H]
	\centering
	\includegraphics[width = \textwidth]{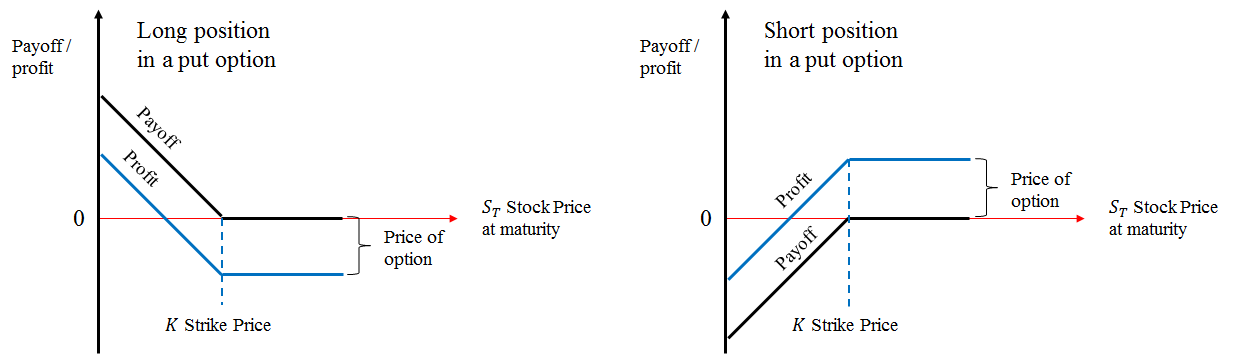}
	\caption{Reproduced from \cite[p. 198]{2012Hull_Options_Futures}, payoff and profit profile for a put option.}
	\label{fig:payoff_profile_put}
\end{figure}

Using the same expressions as before and denoting the price of the put as $p_P$, the payoff and profit for a long put position can be expressed as \cite[p. 198]{2012Hull_Options_Futures}
\begin{align}
	\text{Payoff}_{\text{Long Put}} &= \max \{ K - S_T, 0\} \label{eqn:payoff_long_put} \\
	\text{Profit}_{\text{Long Put}} &= \max \{ K - S_T, 0\} - p_P \label{eqn:profit_long_put} 	
\end{align}
while the payoff and profit for a short put are \cite[p. 198]{2012Hull_Options_Futures}
\begin{align}
	\text{Payoff}_{\text{Short Put}} &= \min \{S_T - K, 0\} \label{eqn:payoff_short_put} \\
	\text{Profit}_{\text{Short Put}} &= \min \{S_T - K, 0\}+ p_P \label{eqn:profit_short_put}
\end{align}

Hence, the bounds for profits and losses are quite different between call and put options. While a trader has no upper bound on possible profits from a long call, the losses for a short call are unbounded as well. On the hand, profits and losses are bounded for both positions, long and short, in put options. \\

As mentioned previously, the four basic positions can be combined in a variety of ways to create many different payoff profiles.\footnote{For a more detailed description of option trading strategy, please refer to \cite[p. 234 ff.]{2012Hull_Options_Futures}.} In this dissertation, only a \emph{strangle} will be considered.
\begin{Definition}[{\cite[p. 248]{2012Hull_Options_Futures} Sale of a Strangle}] \label{def:strangle}
	In the sale of a \emph{strangle}, sometimes called a \emph{top vertical combination}, the investors sells a European put and a European call option with the same expiration date, but different strike prices ($K_{\text{Put}} < K_{\text{Call}}$).
\end{Definition}
The payoff and profit profile from the sale of a strangle is shown in \autoref{fig:payoff_profile_strangle}. It is an easy to construct strategy and suitable for investors who feel that large stock price movements are unlikely. The profit from the sale of strangle is constant if the stock price at maturity is between the two strike prices, i.e. $K_{\text{Put}} \leq S_T \leq K_{\text{Call}}$. However potential losses are unlimited if the stock price rises above $K_{\text{Call}}$ because of the short call position \cite[p. 248]{2012Hull_Options_Futures}.

\begin{figure}[H]
	\centering
	\includegraphics[width = 0.7 \textwidth]{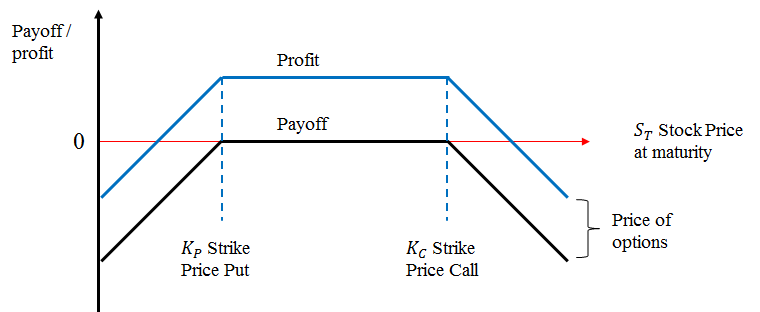}
	\caption{Reproduced from \cite[p. 249]{2012Hull_Options_Futures}, payoff and profit profile for the sale of a strangle.}
	\label{fig:payoff_profile_strangle}
\end{figure}

\section{Background on Financial Risk Management} \label{sec:Portfolio_Hedging-Background_on_Risk_Management}

When managing the risk of an option trader's portfolio, it is crucial to have the most up to date estimates for the variance (or standard deviation / volatility\footnote{Volatility is just another term for standard deviation that is commonly used in finance.}) and covariance of the underlying stock's price movements. Just prices constantly change, so does the volatility of the price changes. In periods of economic stability, huge price fluctuations are unlikely so the volatility is low - while in times of uncertainty price fluctuations are more common.

Hence, it might be unsuitable to estimate the variance and covariance using \autoref{def:variance} and \autoref{def:covariance} with the entire historic data. To estimate the market risk\footnote{Market risk is the risk that is caused by the uncertainty of price changes.}, practitioners tend to use running averages or exponentially weighted moving averages to estimate the current volatility of an asset because this places more importance on recent observations of price fluctuations \cite[p. 16]{Wolf2015_FRM_Lecture}.

This section describes how to calculated the daily EWMA estimates for the variance and covariance and how to scale the variance if the holding period of a portfolio is longer than one day. The following variables will be used in the definitions:
\begin{table}[H]
	\centering
	\begin{tabular}{r l}
		$t$: & the day of the estimation \\
		$r_{x,t}$: & the natural log of the daily return of an asset $x$ \\
		& from $t-1$ to $t$, i.e. $\ln \left( \frac{\text{Price}_{x,t} - \text{Price}_{x,t-1}}{\text{Price}_{x,t-1}}\right)$
	\end{tabular}	
\end {table}
The natural log of returns is used instead of the regular returns, because the distribution of log returns is better fitted by the normal distribution than the regular return. And at the same time, log returns usually have a correlation with regular returns of close to 1 \cite[p. 12]{Wolf2015_FRM_Lecture}.

\begin{Definition}[{\cite[p. 16]{Wolf2015_FRM_Lecture} EWMA of Variance}] \label{def:EWMA_variance}
	The daily variance of the returns of an asset $x$ using an exponentially weighted moving average with parameter $\lambda$ is estimated by the formula
	\begin{equation}
		\text{\emph{Var}}_{t} (x) \defeq \lambda \text{\emph{Var}}_{t-1} (x) + (1 - \lambda) r_{x,t-1}^2 . \label{eqn:EWMA_variance}
	\end{equation}
\end{Definition}
Hence, the variance of any given day is estimated by using the variance estimate of the previous day and the natural log of observed returns of the previous day. To apply \autoref{eqn:EWMA_variance}, two parameters must be set: the variance estimate of day 0 and $\lambda$. If the estimates have been calculated for a long enough horizon, $\text{Var}_{0} (x)$ is of little importance so it can be set equal to 0. In practice, risk managers usually set $\lambda = 0.94$, as this provides a good balance between the volatility estimates of recent and historic data \cite[p. 16 ff.]{Wolf2015_FRM_Lecture}.

\begin{Definition}[{\cite[p. 25]{Wolf2015_FRM_Lecture} EWMA of Covariance}] \label{def:EWMA_covariance}
	The daily covariance between the returns of an asset $x$ and an asset $y$ using an exponentially weighted moving average with parameter $\lambda$ is estimated by the formula
	\begin{equation}
		\text{\emph{Cov}}_{t} (x,y) \defeq \lambda \text{\emph{Cov}}_{t-1} (x,y) + (1 - \lambda) r_{x,t-1} r_{y,t-1}. \label{eqn:EWMA_covariance}
	\end{equation}
\end{Definition}
Again, two parameters must be set to apply \autoref{eqn:EWMA_covariance}: $\text{Cov}_{0} (x,y)$ and $\lambda$. Using the same arguments as before, they should be set to $\text{Cov}_{0} (x,y) = 0$ and $\lambda = 0.94$ \cite[p. 25]{Wolf2015_FRM_Lecture}. \\

If the portfolio is held for longer than one day, the variance and covariance estimates need to be scaled to estimate the risk over the entire holding period. Assuming that returns follow a random walk, the variance and covariance over a $n$ day holding period (denoted $\text{Var}_{t}^n (x)$ and $\text{Cov}_{t}^n (x)$, respectively) are given as \cite[p. 13]{Wolf2015_FRM_Lecture}
\begin{align}
	\text{Var}_{t}^n (x) &= n \times \text{Var}_{t} (x) \label{eqn:EWMA_variance_ndays} \text{, and} \\
	\text{Cov}_{t}^n (x,y) &= n \times \text{Cov}_{t} (x,y) \label{eqn:EWMA_covariance_ndays}.
\end{align}

\section{Forming a Strangle} \label{sec:Portfolio_Hedging-Forming_Strangle}

As described in the introduction, one scenario where CVaR hedging can be used is the adjustment of a trader's portfolio to protect the trading firm against unlikely, but very high losses. For this scenario the following set-up is given and the following assumptions are made:
\begin{itemize}
	\item The date and time is 22 July 2015, 9 PM New York time (before US markets open).
	\item The trader only trades in call and put options on Google (NASDAQ:GOOGL) and Yahoo (NASDAQ:YHOO) which are expiring on 24 July 2015.
	\item The trader builds his position and does not change until the option contract expire, i.e. the holding time is 3 trading days.
	\item Only options with strike prices for which the open interest is greater than 200 will be considered.
	\item There is no bid-ask spread, i.e. options can be bought and sold at the same price. \footnote{Usually, the price to buy (ask) is higher than the price to sell (bid). Here, the price of an option is the average between ask and bid price.}
	\item There are no transaction costs.
	\item All data is taken from \href{https://www.google.co.uk/finance}{Google Finance UK}.
	\item The trader believes that high price movements are unlikely, he will build a pure strangle with Google options and a strangle with additional positions with Yahoo options. The additional positions on Yahoo are because the trader believes that an upward movement of Yahoo's share price is more likely than a downward movement.
\end{itemize}~

To be more precise, the trader believes that at the market closing on 24 July 2014, the share price of Yahoo will be between USD 37.5 and 42.5, while the share price of Google will be between USD 665 and 730. Based on the trader's positions, the payoff and profit profile for different prices of Yahoo and Google at maturity is shown in \autoref{fig:strangle_googl_yhoo_payoff_profile_before}. More detailed information about option prices is given in \autoref{app_table:option_prices_yhoo} and \autoref{app_table:option_prices_googl}, while the trader's positions are given in \autoref{app_table:trader_positions_before}.

\begin{figure}[H]
	\centering
	\includegraphics[width = 0.9 \textwidth]{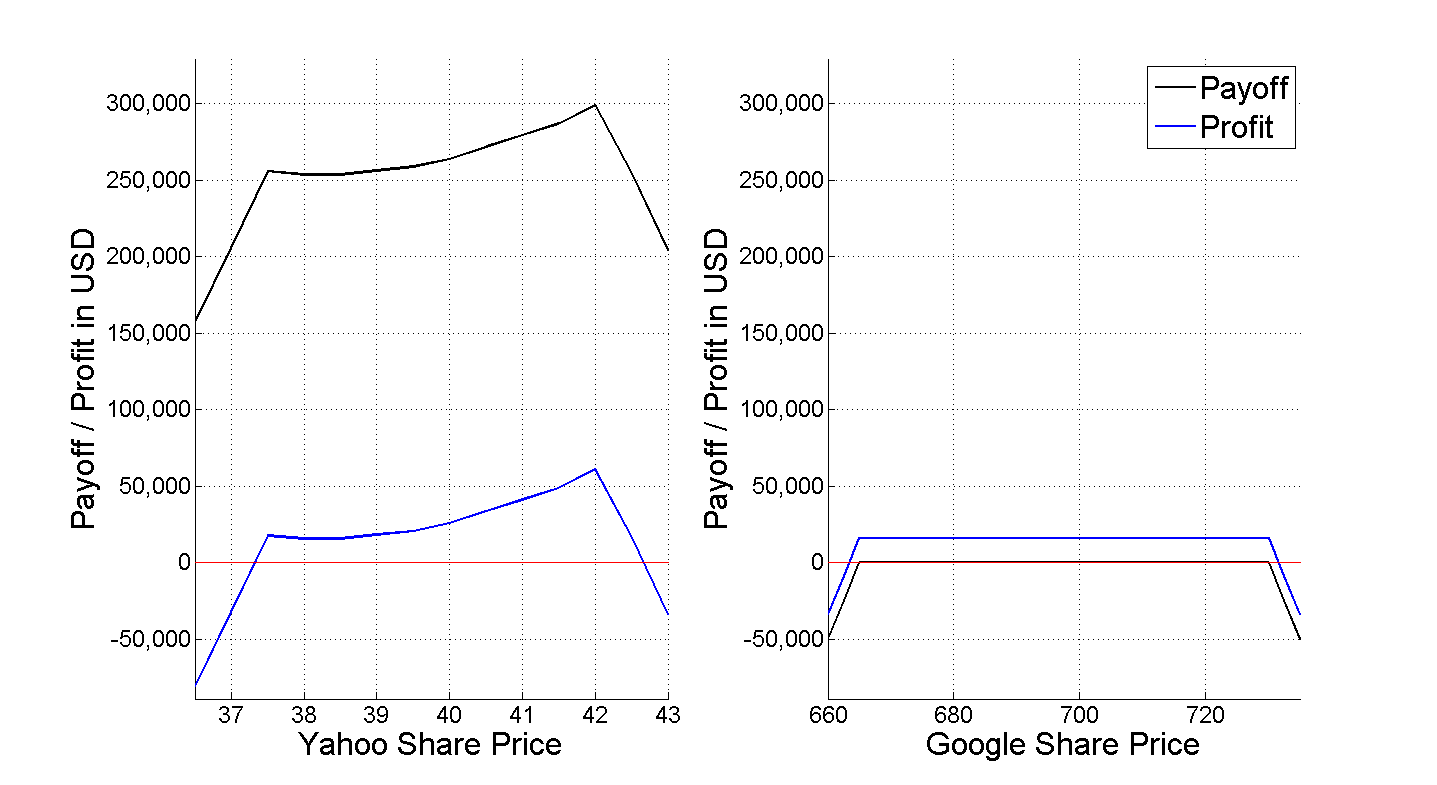}
	\caption{Profit profiles for (unhedged) Google and Yahoo strangles at maturity.}
	\label{fig:strangle_googl_yhoo_payoff_profile_before}
\end{figure}

Hence, if Google's share price closes within the trader's expectations on 24 July, the trader will make a constant profit. If Yahoo's share price closes within the expectations, the trader will also make a profit, but the profit will be highest if the share price closes at USD 42. However, the trader will suffer severe losses if the share prices close outside of his expectation, as can be seen at the left and right edges of the profit profiles in \autoref{fig:strangle_googl_yhoo_payoff_profile_before}.

\section{Hedging Against a Strangle} \label{sec:Portfolio_Hedging-Hedging_Strangle}

To perform the risk assessment of the trader's positions, the variance and covariance of Yahoo's and Google's share price movements need to be estimated. Using the daily share price movements over the last year, together with \autoref{eqn:EWMA_variance} and \autoref{eqn:EWMA_covariance} gives the following covariance matrix \footnote{As noted before, $\lambda$ is chosen to be 0.94 and the initial estimates for the variance and covariance are 0} for daily price movements:
\begin{align*}
	\Sigma^{1} &= \begin{bmatrix}
		0.00021176 & 0.00010049 \\
		0.00010049 & 0.00017589
		\end{bmatrix} \,,
\end{align*}
where $\Sigma^{1}_{1,1}$ is the variance for Yahoo's and $\Sigma^{1}_{2,2}$ is the variance for Google's share price movements.

Since the trader will hold the portfolio for 3 days, $\Sigma^{1}$ needs to be multiplied by 3 to give the variance and covariance estimates for the whole holding period (see \autoref{eqn:EWMA_variance_ndays} and \autoref{eqn:EWMA_covariance_ndays}). This gives the following covariance matrix for all subsequent risk assessments:
\begin{align}
	\Sigma &= \begin{bmatrix}
		0.00063528	& 0.00030147 \\
		0.00030147	& 0.00052767
		\end{bmatrix} \,. \label{eqn:covariance_mat_est_trader}
\end{align}

The remainder of this section mostly follows the hedging procedure used by Rockafellar and Uryasev in \cite{Rockafellar2000_Optimization_of_CVaR}. However, the optimization programme used to determine the CVaR optimal hedge was never stated in \cite{Rockafellar2000_Optimization_of_CVaR}, so the explicit formulation of \autoref{problem:hedging_using_CVaR} (together with \autoref{table:CVaR_Hedging_Variables_in_LP})  is an original contribution of this thesis.\\

With the initial prices of Yahoo and Google at USD 39.73 and 695.35, respectively, on the morning of July 22 and the variance estimates given in $\Sigma$, one can calculate the probability that the share prices will be outside the trader's beliefs. Denoting the share prices at maturity of the options as $S_{T,y}$ and $S_{T,g}$, these probabilities can be expressed as
\begin{align*}
	P (S_{T,y} < 37.5) + P (S_{T,y} > 42.5) &= 0.016 \text{~, and}\\
	P (S_{T,g} < 665) + P (S_{T,g} > 730) &= 0.044 \,.
\end{align*}

Hence, there is a high probability that the trader will be correct in his assumption. Taking the risk analysis a little further, 20,000 simulations\footnote{A higher number of simulations could not be performed as the PC ran out of memory for a CVX programme with more than 20,000 simulations.} of share price developments were run (taking into account the correlation between Yahoo and Google share price movements). For each of the 20,000 scenarios the trader's loss was calculated. The loss distribution of the simulations is shown in \autoref{fig:Trader_Loss_Distribution_before} and several risk metrics are given in \autoref{table:Risk_Metrics_unhedged_hedged}.
\begin{figure}[H]
	\centering
	\includegraphics[width = 0.9 \textwidth]{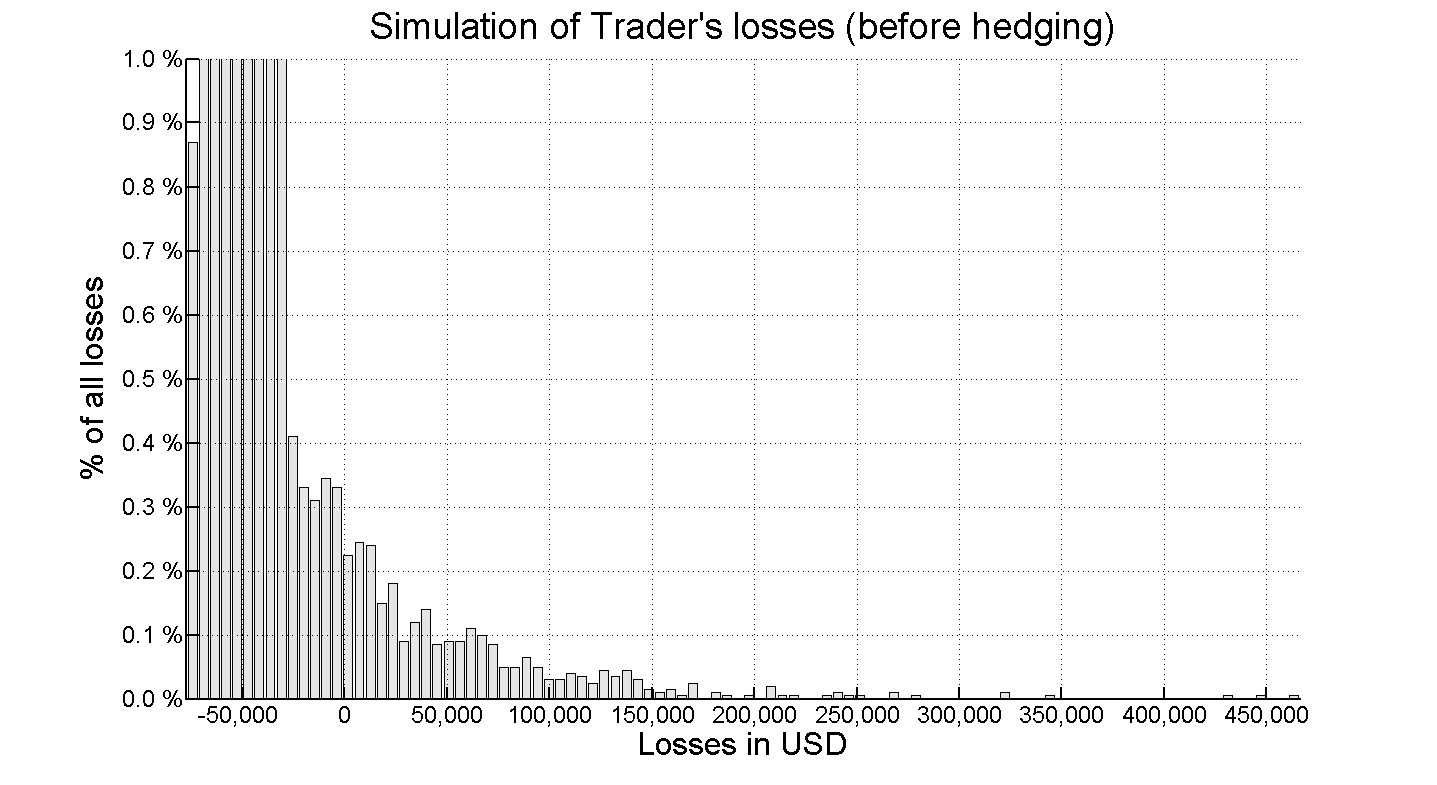}
	\caption{Histogram of trader's (unhedged) portfolio losses from 20,000 simulations.}
	\label{fig:Trader_Loss_Distribution_before}
\end{figure}

Only in very few simulations (2.6 \%) the trader actually makes a loss. Quantifying the Value-at-Risk also gives a positive assessment of the positions, as $\text{VaR}_{0.95} = -31,441$, meaning that with 95 \% probability, the trader makes at least a profit of USD 31,440. However, the tail risk is not taken into account. Since the profits are bounded, but losses are unlimited (see profit profile in \autoref{fig:strangle_googl_yhoo_payoff_profile_before}), it is impossible to say how much the trader can expect to lose using VaR alone. Actually, the the 95 \% CVaR over all simulations is USD 22,458. This means that in the 5 \% worst cases, the trader can expect to lose this much.

To hedge against the tail losses, one can modify \autoref{problem:CVaR_Definition_in_LP_including_min_return} and define a linear programme that computes a CVaR optimal portfolio, starting from the trader's positions (given in \autoref{app_table:trader_positions_before}). The variables used in the programme are shown in \autoref{table:CVaR_Hedging_Variables_in_LP}.\footnote{Note that the trader's positions (denoted $\mathbf{x}$) are now given in number of contracts instead of percentages (which was done in \autoref*{chapter:CVaR_Portfolio_Optimization}).}
\begin{table}[H]
	\centering
	\begin{tabular}{| >{$} c <{$} | >{$} c <{$} | p{10cm} |}
		\hline
		\textbf{Variable} & \textbf{Dimension} & \textbf{Description} \\
		\hline
		N_y, N_g & 1 & Number of strike prices for Yahoo / Google options \\
		\mathbf{k}^{y} & N_y \times 1 & Strike Prices for Yahoo call / put options \\
		\mathbf{k}^{g} & N_g \times 1 & Strike Prices for Google call / put options \\
		\mathbf{p}^{C,y}, \mathbf{p}^{P,y}& N_y \times 1 & Prices to buy / sell Yahoo call / put options \\
		\mathbf{p}^{C,g}, \mathbf{p}^{P,g} & N_g \times 1 & Prices to buy / sell Google call / put options \\
		\mathbf{x}^{C,y}, \mathbf{x}^{P,y} & N_y \times 1 & Trader's positions in Yahoo call / put options \\
		\mathbf{x}^{C,g}, \mathbf{x}^{P,g} & N_g \times 1 & Trader's positions in Google call / put options \\		
		\mathbf{y}^{C,y}, \mathbf{y}^{P,y} & N_y \times 1 & Hedging adjustments for Yahoo call / put options \\
		\mathbf{y}^{C,g}, \mathbf{y}^{P,g} & N_g \times 1 & Hedging adjustments for Google call / put options \\
		\mathbf{a}^{C,y}, \mathbf{a}^{P,y} & N_y \times 1 & Maximum position adjustments in the hedge using Yahoo call / put options \\
		\mathbf{a}^{C,g}, \mathbf{a}^{P,g} & N_y \times 1 & Maximum position adjustments in the hedge using Google call / put options \\
		M & 1 & Number of price simulations \\
		\mathbf{S} & M \times 2 & Simulated share prices at maturity for Yahoo and Google \\
		\mathbf{PO}^{C,y}, \mathbf{PO}^{P,y} & M \times N_y & The payoff for call / put options in Yahoo, by simulated share price and strike price of the option \\
		\mathbf{PO}^{C,g}, \mathbf{PO}^{P,g} & M \times N_g & The payoff for call / put options in Google, by simulated share price and strike price of the option \\
		\text{cost}^y, \text{cost}^g & 1 & Cost for building the trader's position \\
		spc & 1 & $spc = 100$; The number of shares covered by 1 option contract \\
		\hline
	\end{tabular}
	\caption{Variables used in LP to calculate CVaR optimal hedge.}
	\label{table:CVaR_Hedging_Variables_in_LP}
\end{table}

The advantage of using CVaR optimization for hedging is that all positions can be adjusted simultaneously with relatively little computing power as the problem formulation is a linear programme (compared to pure VaR optimization methods). However, in hedging the general profile of the trader's positions should be maintained and only the risk reduced. Therefore, the changes (denoted by $\mathbf{y}$) cannot be arbitrarily large, and the maximum possible adjustment for each position is given by the $\mathbf{a}$ vectors. \cite[p. 33 f.]{Rockafellar2000_Optimization_of_CVaR}

Also, the payoffs $\mathbf{PO}$ can be calculated before running the optimization programme (but after the scenarios were simulated). Their entries are
\begin{align*}
	PO^{C,y}_{i,j} &= \max \{ S_{i,1} - k^{y}_j ,0 \}  && \text{~for~} i \in \{1, \dots, M\}, j \in \{1, \dots, N_y\}  \,,\\
	PO^{P,y}_{i,j} &= \max \{ k^{y}_j - S_{i,1},0 \}  && \text{~for~} i \in \{1, \dots, M\}, j \in \{1, \dots, N_y\} \,, \\
	PO^{C,g}_{i,j} &= \max \{ S_{i,2} - k^{g}_j ,0 \}  && \text{~for~} i \in \{1, \dots, M\}, j \in \{1, \dots, N_g\} \,, \text{~and} \\
	PO^{P,g}_{i,j} &= \max \{ k^{g}_j - S_{i,2},0 \}  && \text{~for~} i \in \{1, \dots, M\}, j \in \{1, \dots, N_g\} \,.
\end{align*}

Hence, the hedging problem using CVaR optimization can be formulated as
\begin{Problem}[problem:hedging_using_CVaR]
\left.
\begin{array}{ r *{3}{l} }
	&\min \limits_{c, \mathbf{z}} & c + \frac{1}{M (1 - \alpha)} \sum \limits_{m = 1}^M z_m \\ [9pt]
	\text{s.t.} & - a^{C,y}_i \leq& y^{C,y}_i \leq a^{C,y}_i & \text{for~} i \in \{1, \dots , N_y\} \\ [3pt]
	& - a^{P,y}_i \leq& y^{P,y}_i \leq a^{P,y}_i & \text{for~} i \in \{1, \dots , N_y\} \\ [3pt]	
	& - a^{C,g}_i \leq& y^{C,g}_i \leq a^{C,g}_i & \text{for~} i \in \{1, \dots , N_g\} \\ [3pt]
	& - a^{P,g}_i \leq& y^{P,g}_i \leq a^{P,g}_i & \text{for~} i \in \{1, \dots , N_g\} \\ [3pt]
	& \mathbf{PO}^y =& \left[ \mathbf{PO}^{C,y} \left( \mathbf{x}^{C,y} + \mathbf{y}^{C,y} \right) \right.\\
	& & \left. + \mathbf{PO}^{P,y} \left( \mathbf{x}^{P,y} + \mathbf{y}^{P,y} \right) \right] \times spc \\ [3pt]
	& \mathbf{PO}^g =& \left[ \mathbf{PO}^{C,g} \left( \mathbf{x}^{C,g} + \mathbf{y}^{C,g} \right) \right. \\
	& & \left. + \mathbf{PO}^{P,g} \left( \mathbf{x}^{P,g} + \mathbf{y}^{P,g} \right) \right] \times spc \\ [3pt]
	& \text{adjCost}^y =& \left[ \sum \limits_{i = 1}^{N_y} p^{C,y}_i \times y^{C,y}_i  + \sum \limits_{i = 1}^{N_y} p^{P,y}_i \times y^{P,y}_i \right] \times spc \\ [15pt]
	& \text{adjCost}^g =& \left[ \sum \limits_{i = 1}^{N_g} p^{C,g}_i \times y^{C,g}_i  + \sum \limits_{i = 1}^{N_g} p^{P,g}_i \times y^{P,g}_i \right] \times spc \\ [15pt]
	& z_m \geq& \text{adjCost}^y + \text{adjCost}^g + \text{cost}^y \\ [3pt]
	& &+ \text{cost}^g - \left[ PO^y_m + PO^g_m \right]& \text{for~} m \in \{1, . , M\} \\ [3pt]
	& z_m \geq& 0 & \text{for~} m \in \{1, . , M\} \\ [3pt]
\end{array}
\right\} \,.
\end{Problem}

Hedging the trader's portfolio using \autoref{problem:hedging_using_CVaR} with $a^{P,y}_i = a^{C,y}_i = 50 \text{~for~} i \in \{1, \dots , N_y\}$ and $a^{P,g}_i = a^{C,g}_i = 5 \text{~for~} i \in \{1, \dots , N_g\}$ yields the payoff / profit profile shown in \autoref{fig:strangle_googl_yhoo_payoff_profile_after} and the loss distribution \autoref{fig:Trader_Loss_Distribution_after}. The exact composition of the hedged portfolio is shown in \autoref{app_table:trader_positions_after_yhoo} and \autoref{app_table:trader_positions_after_googl}.

\begin{figure}[H]
	\centering
	\includegraphics[width = 0.9 \textwidth]{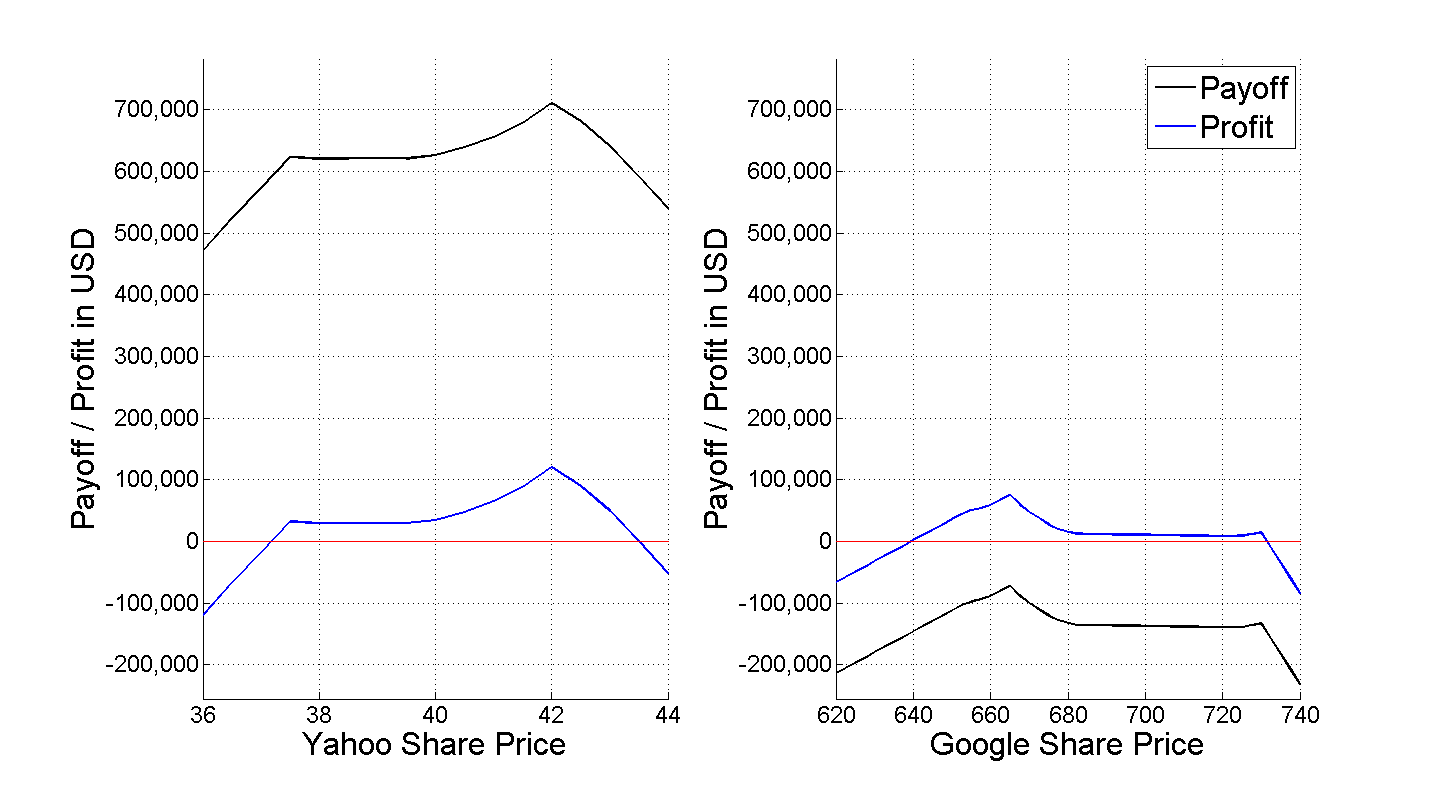}
	\caption{Profit profiles for hedged Google and Yahoo strangles at maturity.}
	\label{fig:strangle_googl_yhoo_payoff_profile_after}
\end{figure}

After hedging, the profit profile for Yahoo options only changed slightly. The most noticeable change is that the graph is mostly scaled, that is, the profit for any given share price is about twice as high as for the unhedged portfolio. Still, the highest profit will be achieved when the share price of Yahoo is at USD 42.

The pure strangle that was formed by options on Google changed its shape more noticeably. While the profit was mostly constant in the unhedged portfolio, there is now a clear peak at $S_{T,g} = 665$. While USD 665 was the trader's assumed lower bound for the final share price, it is now the share price at which the maximum profit will be achieved. Also, the trader will make a profit as long as Google's share price closes above USD 640. This adjustment can be explained by the correlation between Yahoo's and Google's share price movements. As they are positively correlated, a drop in Yahoo's share price will be compensated in the trader's portfolio by the positions in Google options and vice versa.

\begin{figure}[H]
	\centering
	\includegraphics[width = 0.9 \textwidth]{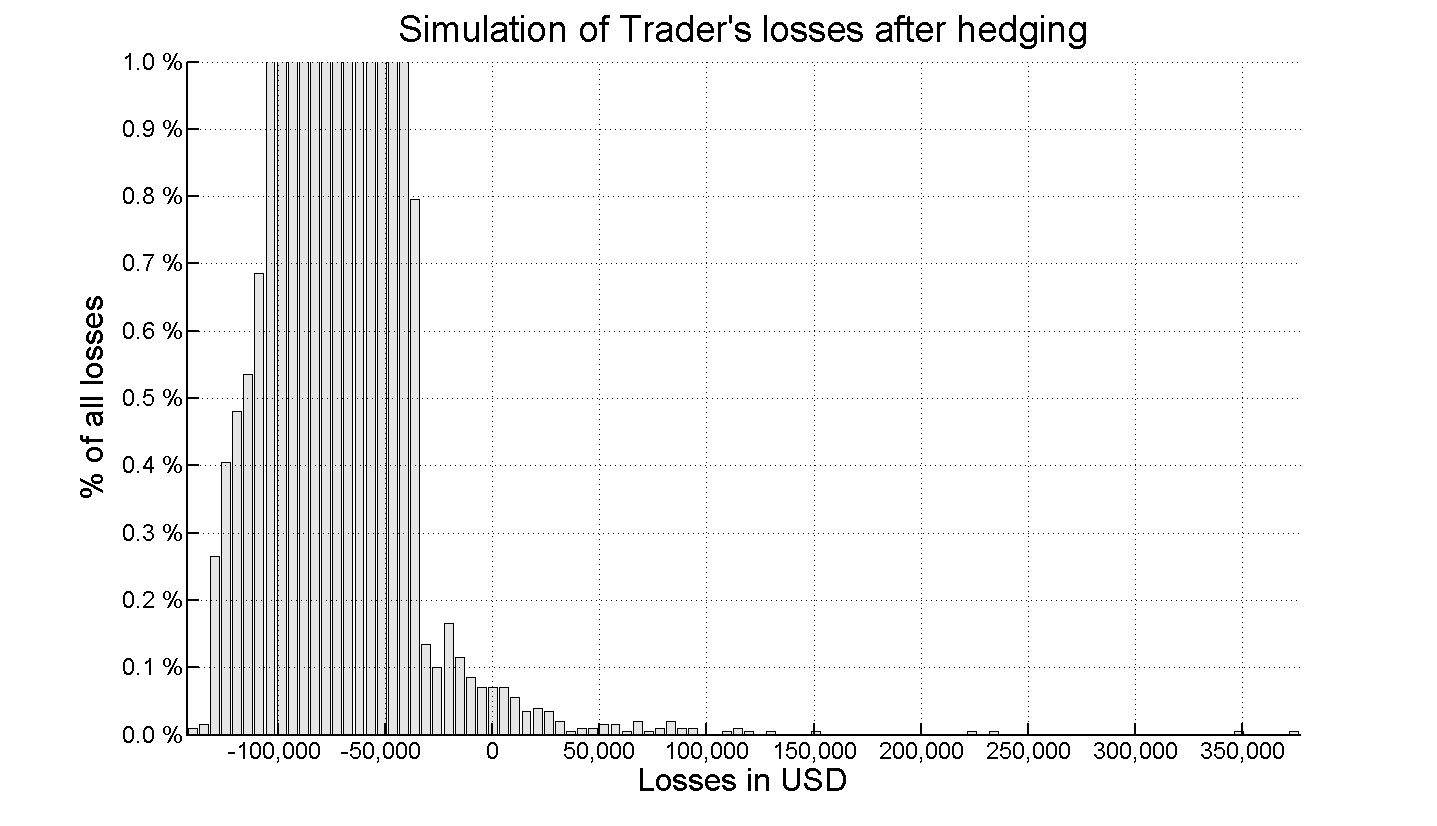}
	\caption{Histogram of trader's hedged portfolio losses from 20,000 simulations.}
	\label{fig:Trader_Loss_Distribution_after}
\end{figure}

The loss distribution is also much more favourable, as there much less losses and also higher profits can be realized than with the unhedged portfolio. A summary of main risk metrics is given in \autoref{table:Risk_Metrics_unhedged_hedged} below.

\begin{table}[H]
	\centering
	\begin{tabular}{| l | r | r | }
		\hline
		\textbf{Metric} & \textbf{Original Portfolio} & \textbf{Hedged Portfolio} \\
		\hline
		Mean Loss	&	-38,882 	&	-54,910 	\\
		Min Loss	&	-77,072 	&	-142,556 	\\
		Max Loss	&	 466,221 	&	 376,638 	\\
		Probability of Loss	&	2.62 \%	&	0.48 \%	\\
		95 \% VaR	&	-31,441 	&	-39,648 	\\
		95 \% CVaR	&	 22,458 	&	-27,911 	\\
		\hline
	\end{tabular}
	\caption{Risk metrics for the original and hedged option portfolio.}
	\label{table:Risk_Metrics_unhedged_hedged}
\end{table}

As table \autoref{table:Risk_Metrics_unhedged_hedged} demonstrates, the hedged portfolio performs better than the original in any of the 6 metrics under consideration. The portfolio has a higher expected profit and lower probability of generating a loss. Also, the 95 \% VaR is lower (meaning that the minimum profit in the 95 \% best cases is higher than for the original portfolio). Most notably however, is the fact that the hedged portfolio has a negative 95 \% CVaR. The means that even in the 5 \% worst cases, the trader can expect a profit of USD 27,911. Still, losses are possible as can be seen in \autoref{fig:Trader_Loss_Distribution_after}, but they are far less likely and less severe than for the original portfolio. \\

To conclude this chapter, it needs to be emphasized that the given example (although relying on real world data) is only demonstrating how to apply CVaR optimization when trying to hedge a portfolio. The hedging effect shown here is astonishing, but can barely be reproduced in an actual trading environment for several reasons. First, the original portfolio was just an example, it has not been optimized with regards to profit maximization. For a more balanced portfolio, the effects of hedging would be less extreme. Also, the prices were simplified, enabling to buy and sell at the same price, without any transaction costs. Introducing ask and bid prices, as well as transaction costs would decrease the profit and hence increase possible losses. Third, the trader and risk manager could buy and sell unlimited quantities of any option. In reality the offer and demand for any given option is limited. Finally, all other simplifying assumption would make it hard to reproduce the same results in a real world setting, e.g. that the assumption that the trader holds the portfolio until the maturity of the options or that the volatility would remain constant over the holding period.

\clearpage

%
%

\chapter{Conditional Value-at-Risk as a Norm} \label{chapter:CVaR_Norms}

In the previous chapters, CVaR was introduced as a risk measure, which was the original intention of CVaR. Applications to portfolio optimization and hedging were also explored. In more recent research, Pavlikov and Uryasev (\cite{pavlikov2014_CVaR_Norm_and_applications}) abstracted the concept of CVaR to a more general interpretation, so that it can also be used to define a family of norms in $\mathbb{R}^n$. Pavlikov and Uryasev proposed two norms: a scaled CVaR norm (denoted $C^S_{\alpha}$), and a non-scaled CVaR norm (denoted $C_{\alpha}$, later simply referred as \emph{CVaR Norm}), which only differ by a factor.

This chapter first presents the two different and equivalent definitions that Pavlikov and Uryasev used to define the $C^S_{\alpha}$ norm, and how the $C^S_{\alpha}$ and $C_{\alpha}$ norms are related to one another by a multiplying factor. \autoref{sec:CVaR_Norms-Properties} presents some of the norm properties that were identified by Pavlikov and Uryasev in \cite{pavlikov2014_CVaR_Norm_and_applications}, enriched by some original ideas of the author. \autoref{sec:CVaR_Norms-Computational_Efficiency} introduces algorithms to computationally evaluate the different CVaR norms ($C^S_{\alpha}$ and $C_{\alpha}$). Algorithms are derived for both equivalent definition of each CVaR norm and the computational efficiency of each algorithm is evaluated.


\section{Scaled CVaR Norm} \label{sec:CVaR_Norms-Scaled}

The scaled CVaR norm of the vector $\xinRn$ is denoted by $\CVaRnormX[S]{\alpha}$, where $\alpha$ is a parameter in the range $0 \leq \alpha \leq 1$. The first way to define $\CVaRnormX[S]{\alpha}$ is given in \autoref{subsec:CVaR_Norms-Scaled-Definition} below, while an alternative characterization is given in \autoref{subsec:CVaR_Norms-Scaled-Alternative}.

\subsection{Definition} \label{subsec:CVaR_Norms-Scaled-Definition}

\begin{Definition}[{\cite[p. 3f.]{pavlikov2014_CVaR_Norm_and_applications} Component-wise Scaled CVaR Norm}]
	\label{def:Scaled_CVaR_Component_Wise}
	Let the absolute values of the components of vector $\mathbf{x} \in \mathbb{R}^n$ be ordered in \emph{ascending} order, i.e., $\mid x_{(1)} \mid \, \leq \, \mid x_{(2)} \mid \, \leq \, \dots \, \leq \, \mid x_{(n)} \mid$.\\
	For $\alpha_j = \frac{j}{n}, j = 0, \dots , n-1$, the scaled CVaR norm $\CVaRnormX[S]{\alpha}$  of vector $\mathbf{x}$ with parameter $\alpha_j$ is defined as
	\begin{equation}\label{eqn:Scaled_CVaR_Component_Wise_1}
		\llangle \mathbf{x} \rrangle^S_{\alpha_j} \defeq \frac{1}{n-j} \sum \limits_{i = j+1}^{n} \mid x_{(i)} \mid .
	\end{equation}
	For $\alpha$ such that $\alpha_j < \alpha < \alpha_{j+1}$, $j = 0, \dots, n-2$, the scaled CVaR norm $\CVaRnormX[S]{\alpha}$ equals the weighted average of $\llangle \mathbf{x} \rrangle^S_{\alpha_j}$ and $\llangle \mathbf{x} \rrangle^S_{\alpha_{j+1}}$, i.e.,
	\begin{equation} \label{eqn:Scaled_CVaR_Component_Wise_2}
		\CVaRnormX[S]{\alpha} \defeq \mu \CVaRnormX[S]{\alpha_j} + (1 - \mu) \CVaRnormX[S]{\alpha_{j+1}},
	\end{equation}
	where
	\begin{align*}
		\mu &= \frac{\left( \alpha_{j+1} - \alpha \right) \left( 1 - \alpha_j \right)}{\left( \alpha_{j+1} - \alpha_j \right) \left( 1 - \alpha \right) }.
	\end{align*}
	And finally, for $\alpha$ such that $\frac{n-1}{n} < \alpha \leq 1$,
	\begin{equation}\label{eqn:Scaled_CVaR_Component_Wise_3}
		\CVaRnormX[S]{\alpha} \defeq \max_i \mid x_i \mid .
	\end{equation}
\end{Definition}

To illustrate the scaled CVaR norm, $\CVaRnormX[S]{\alpha}$ will be calculated for a vector $\xinR{4}$ and the unit ball of $\xinR{2}$ will be drawn, both for different values of $\alpha$. For $\mathbf{x} = [ 10, -14, 2, -9]^T$,
\begin{equation}
	\begin{array}{r l l}
		\CVaRnormX[S]{0} &= \frac{1}{4} \left( \vert 2 \vert + \vert -9 \vert + \vert 10 \vert + \vert -14 \vert \right) &= 8.75 \,, \\
		\\
		\CVaRnormX[S]{0.25} &= \frac{1}{3} \left( \vert -9 \vert + \vert 10 \vert + \vert -14 \vert \right) &= 11 \,, \\
		\\
		\CVaRnormX[S]{0.5} &= \frac{1}{2} \left( \vert 10 \vert + \vert -14 \vert \right) &= 12 \,, \text{~and} \\
		\\
		\CVaRnormX[S]{0.75} &= \vert -14 \vert &= 14 \,.
	\end{array} \notag
\end{equation}

Note that by \autoref{eqn:Scaled_CVaR_Component_Wise_3}, $\CVaRnormX[S]{\alpha} = 14$ for all $\alpha > 0.75$ as well. To calculate $\CVaRnormX[S]{\frac{1}{3}}$, $\mu$ must be calculated first to use \autoref{eqn:Scaled_CVaR_Component_Wise_2}. Since $0.25 < \mu < 0.5$,
\begin{align*}
	\mu =& \frac{\left( \frac{1}{2} - \frac{1}{3} \right) \left( 1 - \frac{1}{4} \right)}{\left( \frac{1}{2} - \frac{1}{4} \right) \left( 1 - \frac{1}{3} \right)} = \frac{3}{4} .
\end{align*}
Hence, $\CVaRnormX[S]{\frac{1}{3}} = \mu \CVaRnormX[S]{0.25} + \left( 1 - \mu \right) \CVaRnormX[S]{0.5} = \frac{3}{4} 11 + \frac{1}{4} 12$, so that $\CVaRnormX[S]{\frac{1}{3}} = 11.25$.

For $\xinR{2}$, the unit balls of $\CVaRnormX[S]{\alpha}$ for $\alpha \in \left\{ 0, 0.1, 0.25, 0.4 , 0.5 \right\}$ are shown below in \autoref{fig:CSalpha_unit_balls_own_example}.
\begin{figure}[H]
	\centering
	\includegraphics[width = 0.9 \textwidth]{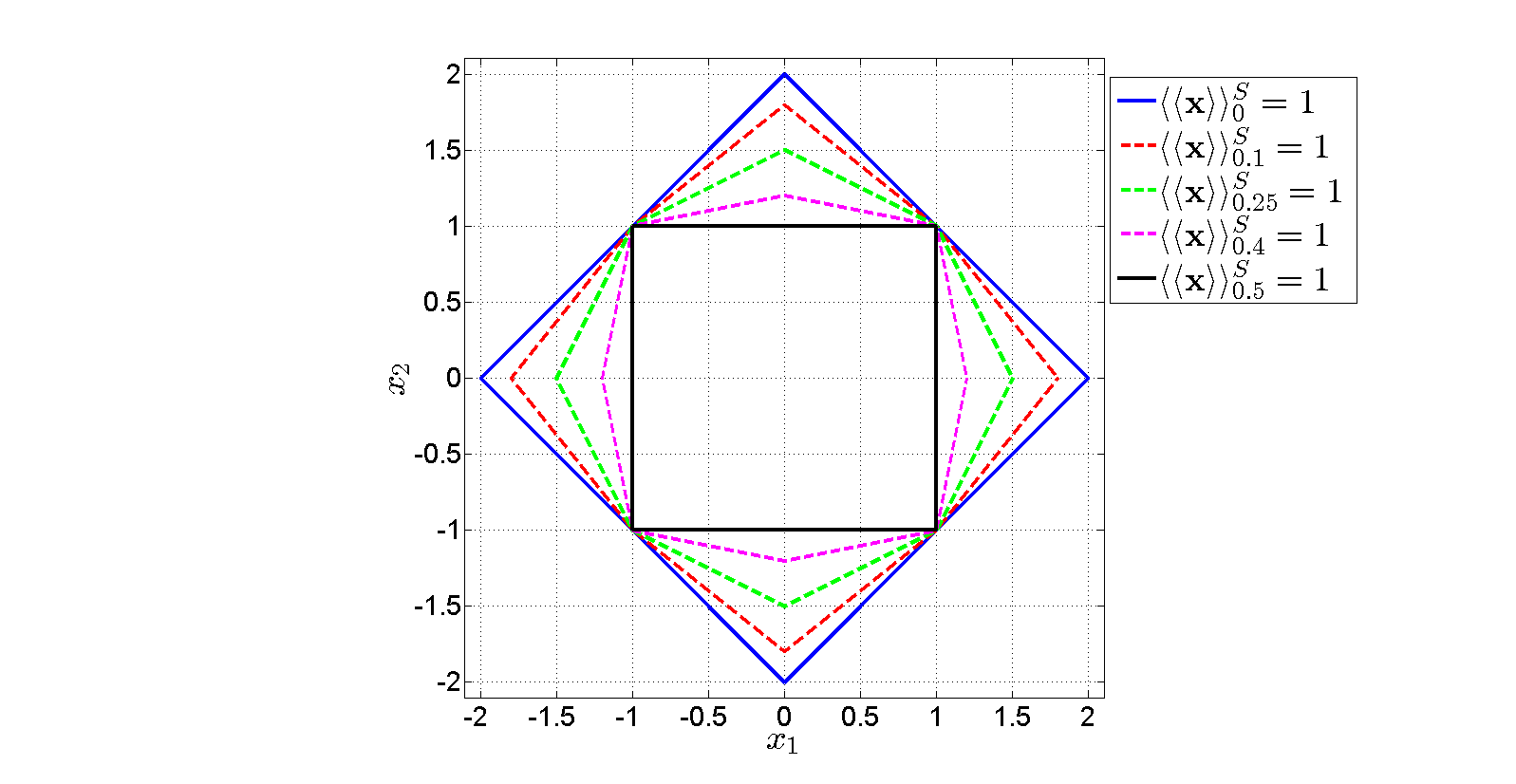}
	\caption{Unit balls of $\CVaRnormX[S]{\alpha}$ for $\xinR{2}$ and different values of $\alpha$.}
	\label{fig:CSalpha_unit_balls_own_example}
\end{figure}

\subsection{Alternative Characterization (Including a New Proof)} \label{subsec:CVaR_Norms-Scaled-Alternative}

Alternatively, the vector $\xinRn$ can be associated with a random variable $X$ with the set of possible outcomes $\{ \vert x_1 \vert \,, \vert x_2 \vert \,, \dots \,, \vert x_n \vert \}$, each of which is equally likely. Then the scaled CVaR norm can be derived from the CVaR definition itself (see \autoref{eqn:CVaR_Definition_Tractable}). That is, the scaled CVaR norm $\CVaRnormX[S]{\alpha}$ is equal to $\text{CVaR}_{\alpha} (X)$ as defined in \autoref{eqn:CVaR_Theorem-CVaR_from_Psi}.
\begin{Proposition}[{\cite[p. 6f.]{pavlikov2014_CVaR_Norm_and_applications} Alternative Characterization of the Scaled CVaR Norm}]
	\label{prop:Scaled_CVaR_based_on_Definition}
	For every $\xinRn$, $0 \leq \alpha < 1$, and $c \in \mathbb{R}^n$,
	\begin{align}
		\CVaRnormX[S]{\alpha} &= \min \limits_{c \in \mathbb{R}} \left( c + \frac{1}{n(1 - \alpha)} \sum \limits_{i=1}^{n} \left( \mid x_i \mid - c \right)^+ \right) \text{~, and} \label{eqn:Scaled_CVaR_based_on_Definition1} \\
		\CVaRnormX[S]{1} &= \max_i \mid x_i \mid \label{eqn:Scaled_CVaR_based_on_Definition2} \,.
	\end{align}
\end{Proposition}

Although \autoref{prop:Scaled_CVaR_based_on_Definition} has been proven by Pavlikov and Uryasev in \cite[p. 9ff.]{pavlikov2014_CVaR_Norm_and_applications}, a novel proof will be presented here to show how the proof of \autoref{prop:Scaled_CVaR_based_on_Definition} can be derived in a different way. To the best knowledge of the author this novel proof has not been published before.

In their proof, Pavlikov and Uryasev showed that for the function $f(c) \defeq c + \frac{1}{n (1 - \alpha)} \sum_{i=1}^n \left[ \vert x_i \vert - c \right]^+$ it follows that $\vert x_{(j+1)} \vert \in \arg \min_c f(c)$. They used this result together with \autoref{eqn:Scaled_CVaR_based_on_Definition1} to manipulate the alternative characterization of the scaled CVaR norm so that it was equal to \autoref{def:Scaled_CVaR_Component_Wise}. The novel proof has two steps. First, it will be shown that when interpreting $\xinRn$ as the distribution of a discrete random variable $X$, the right hand side of both, \autoref{eqn:Scaled_CVaR_based_on_Definition1} and \autoref{eqn:Scaled_CVaR_based_on_Definition2}, are an expression for $\text{CVaR}_{\alpha} (X)$. In the second step, it will be shown that $\text{CVaR}_{\alpha}(X)$ can be expressed by the Convex Combination Formula (\autoref{eqn:CVaR_Convex_Combination_Formula}) so that it is equivalent to $\CVaRnormX[S]{\alpha}$ in \autoref{def:Scaled_CVaR_Component_Wise}.

\begin{proof}
	Let $\xinRn$ describe the distribution of a discrete random variable $X$, so that the possible values of $X$ are $\vert x_i \vert$ for $i \in \{1, \dots, n\}$, with $P(X = \vert x_i \vert) = \frac{1}{n}$. Then for $0 \leq \alpha < 1$, the right hand side of \autoref{eqn:Scaled_CVaR_based_on_Definition1} is equivalent to
	\begin{align*}
		 &\min \limits_{c \, \in \mathbb{R}} \left( c + \frac{1}{n(1 - \alpha)} \sum \limits_{i=1}^n ( \vert x_i \vert - c )^+ \right) \\
		 =& \min \limits_{c \, \in \mathbb{R}} \left( c + \frac{1}{1 - \alpha} \E \left[ ( X - c )^+ \right] \right) \\
		 =& \text{CVaR}_{\alpha} (X) \,,
	\end{align*}
	where the last line follows from \autoref{eqn:CVaR_Definition_Tractable}. And by \autoref{eqn:CVaR}, $\max_i \vert x_i \vert = \text{CVaR}_{1} (X)$. \\
	
	To determine the $\alpha$ CVaR of $X$ by the Convex Combination Formula (\autoref{eqn:CVaR_Convex_Combination_Formula}), three cases need to be considered. The first case is $\alpha = \alpha_j = \frac{j}{n} \,, j \in \{0, 1, \dots , n-1\}$, the second case is $\al{j} < \alpha < \al{j+1} \,, j \in \{0, 1, \dots , n-2 \}$, and the third and last case is $\frac{n-1}{n} < \alpha \leq 1$. For all three cases the absolute values of the components of $\mathbf{x}$ should be ordered in ascending order, such that $\vert x_{(1)} \vert \leq \vert x_{(2)} \vert \leq \dots \leq \vert x_{(n)} \vert$. Also, for the special case $\alpha = 0$, $\vert x_{(0)} \vert \defeq 0$ is introduced.
	
	In the first case, i.e., $\alpha = \alpha_j = \frac{j}{n} \,, j \in \{0, 1, \dots , n-1\}$, $\text{VaR}_{\alpha} (X)$, $\text{CVaR}_{\alpha}^+ (X)$, and $\lambda$ are
	\begin{align*}
		\text{VaR}_{\al{j}} (X) =& \vert x_{(j)} \vert \,, & \text{CVaR}_{\al{j}}^+ (X) =& \frac{1}{n - j} \sum \limits_{i = j+1}^n \vert x_{(i)} \vert \,, & \text{and~} \lambda=& \frac{\al{j} - \al{j}}{1 - \alpha}=0 \,,
	\end{align*}
	so that the CVaR can be expressed as
	\begin{equation} \label{eqn:Scaled_CVaR_norm_proof1}
		\text{CVaR}_{\al{j}} (X) = \frac{1}{n - j} \sum \limits_{i = j+1}^n \vert x_{(i)} \vert \,,
	\end{equation}
	which equals $\CVaRnormX[S]{\al{j}}$ by \autoref{eqn:Scaled_CVaR_Component_Wise_1}.
	
	In the second case, i.e., $\al{j} < \alpha < \al{j+1} \,, j \in \{0, 1, \dots , n-2 \}$, $\text{VaR}_{\alpha} (X)$, $\text{CVaR}_{\alpha}^+ (X)$, and $\lambda$ are
	\begin{align*}
		\text{VaR}_{\alpha} (X) =& \vert x_{(j+1)} \vert \,, & \text{CVaR}_{\alpha}^+ (X) =& \frac{1}{n - (j+1)} \sum \limits_{i = j+2}^n \vert x_{(i)} \vert \,, & \text{and~} \lambda=& \frac{\al{j+1} - \alpha}{1 - \alpha} \,,
	\end{align*}
	so that the CVaR can be expressed as
	\begin{equation} \label{eqn:Scaled_CVaR_norm_proof2}
		\text{CVaR}_{\alpha} (X) = \frac{\al{j+1} - \alpha}{1 - \alpha} \vert x_{(j+1)} \vert + \left( 1 -  \frac{\al{j+1} - \alpha}{1 - \alpha} \right) \frac{1}{n - (j+1)} \sum \limits_{i = j+2}^n \vert x_{(i)} \vert \,.
	\end{equation}	
	To show that \autoref{eqn:Scaled_CVaR_norm_proof2} equals \autoref{eqn:Scaled_CVaR_Component_Wise_2}, \autoref{eqn:Scaled_CVaR_Component_Wise_2} needs to be manipulated, so that
	\begin{align}
		\CVaRnormX[S]{\alpha} =& \mu \CVaRnormX[S]{\alpha_j} + (1 - \mu) \CVaRnormX[S]{\alpha_{j+1}} \notag \\
		\notag \\
		=& \mu \frac{1}{n-j} \sum \limits_{i = j+1}^{n} \vert x_{(i)} \vert + \left(1 - \mu \right) \frac{1}{n-(j+1)} \sum \limits_{i = j+2}^{n} \vert x_{(i)} \vert \notag \\
		\notag \\
		=& \mu \frac{1}{n-j} \vert x_{(j+1)} \vert + \mu \frac{1}{n-j} \sum \limits_{i = j+2}^{n} \vert x_{(i)} \vert + \frac{1}{n-(j+1)} \sum \limits_{i = j+2}^{n} \vert x_{(i)} \vert - \mu \frac{1}{n-(j+1)} \sum \limits_{i = j+2}^{n} \vert x_{(i)} \vert \notag \\
		\notag \\
		& - \frac{\al{j+1} - \alpha}{1 - \alpha} \frac{1}{n - (j+1)} \sum \limits_{i = j+2}^{n} \vert x_{(i)} \vert + \frac{\al{j+1} - \alpha}{1 - \alpha} \frac{1}{n - (j+1)} \sum \limits_{i = j+2}^{n} \vert x_{(i)} \vert \notag \\
		\notag \\
		=& \mu \frac{1}{n-j} \vert x_{(j+1)} \vert + \left( 1 - \frac{\al{j+1} - \alpha}{1 - \alpha} \right) \frac{1}{n - (j+1)} \sum \limits_{i = j+2}^{n} \vert x_{(i)} \vert \notag \\
		\notag \\
		& + \left(  \mu \frac{1}{n-j} - \mu \frac{1}{n-(j+1)} +\frac{\al{j+1} - \alpha}{1 - \alpha} \frac{1}{n - (j+1)} \right) \sum \limits_{i = j+2}^{n} \vert x_{(i)} \vert \notag \\
		\notag \\
		=& \frac{\al{j+1} - \alpha}{1 - \alpha} \vert x_{(j+1)} \vert + \left( 1 -  \frac{\al{j+1} - \alpha}{1 - \alpha} \right) \frac{1}{n - (j+1)} \sum \limits_{i = j+2}^n \vert x_{(i)} \vert \,. \label{eqn:Scaled_CVaR_norm_proof2b}
	\end{align}
	The last step follows because
	\begin{align*}
		\mu \frac{1}{n-j} =& \frac{\left( \alpha_{j+1} - \alpha \right) \left( 1 - \alpha_j \right)}{\left( \alpha_{j+1} - \alpha_j \right) \left( 1 - \alpha \right) } \frac{1}{n-j} \\
		=& \frac{\left( \alpha_{j+1} - \alpha \right) \left( 1 - \frac{j}{n} \right)}{\left( \frac{j+1}{n} - \frac{j}{n} \right) \left( 1 - \alpha \right) (n-j)} = \frac{\al{j+1} - \alpha}{1 - \alpha} \,,
	\end{align*}
	and
	\begin{align*}
		\mu \frac{1}{n-j} - \mu \frac{1}{n-(j+1)} +\frac{\al{j+1} - \alpha}{1 - \alpha} \frac{1}{n - (j+1)} =& 0 \,.
	\end{align*}
	Comparing \autoref{eqn:Scaled_CVaR_norm_proof2b} and \autoref{eqn:Scaled_CVaR_norm_proof2} shows that $\text{CVaR}_{\alpha} (X) = \CVaRnormX[S]{\alpha}$ for $\al{j} < \alpha < \al{j+1} \,, j \in \{0, 1, \dots , n-2 \}$.
	
	The last step is to show that $\text{CVaR}_{\alpha} (X) = \CVaRnormX[S]{\alpha}$ for $\frac{n-1}{n} < \alpha \leq 1$, which is trivial, as $\text{CVaR}_{\alpha} (X) = \max_i \vert x_i \vert = \CVaRnormX[S]{\alpha}$ in this case. This follows from \autoref{eqn:Scaled_CVaR_Component_Wise_3} and because $\text{CVaR}_{\alpha} (X) = \text{VaR}_{\alpha} (X)$, when $\text{VaR}_{\alpha} (X)$ is the maximum loss possible \cite[p. 1452]{Rockafellar2002_CVaR_for_general_loss_distributions}, which is the case for $\frac{n-1}{n} < \alpha \leq 1$.
	
	So both, \autoref{def:Scaled_CVaR_Component_Wise} and the right hand side of \autoref{eqn:Scaled_CVaR_based_on_Definition1} and \autoref{eqn:Scaled_CVaR_based_on_Definition2} in \autoref{prop:Scaled_CVaR_based_on_Definition} are equal to $\text{CVaR}_{\alpha} (X)$, and hence must be equivalent.
\end{proof}


\section{Non-Scaled CVaR Norm} \label{sec:CVaR_Norms-Non_scaled_CVaR_Norm}

The non-scaled CVaR norm (also called \emph{CVaR norm}) is obtained by multiplying the scaled CVaR norm by a factor. This norm will have more significance in the following chapters.

\subsection{Definition} \label{subsec:CVaR_Norms-NonScaled-Definition}

The non-scaled CVaR norm is obtained by multiplying the scaled CVaR norm by the factor $n (1 - \alpha)$, i.e.,
\begin{equation} \label{eqn:Non_Scaled-Scaled_CVaR_norm}
	\CVaRnormX{\alpha} \defeq n(1 - \alpha) \cdot \CVaRnormX[S]{\alpha} \,.
\end{equation}
The non-scaled CVaR norm will be called \emph{CVaR norm} from here on for simplicity. \\

Algorithms for calculating the scaled CVaR norm and CVaR norm will be implemented computationally and their efficiency will be compared in \autoref{sec:CVaR_Norms-Computational_Efficiency}. Since the algorithms will be based on the definitions of the norms, it is computationally more efficient to calculate the CVaR norm from an algorithm based on \autoref{def:CVaR_Component_Wise} than based on \autoref{eqn:Non_Scaled-Scaled_CVaR_norm} as this eliminates two calculation steps: first scaling by $n - j$ and then multiplying by $n(1 - \alpha)$. Hence, the following definition of the CVaR norm will be used.

\begin{Definition}[{\cite[p. 14f.]{pavlikov2014_CVaR_Norm_and_applications} Component-wise CVaR Norm}]
	\label{def:CVaR_Component_Wise}
	Let the absolute values of the components of vector $\mathbf{x} \in \mathbb{R}^n$ be ordered in \emph{ascending} order, i.e. $\mid x_{(1)} \mid \, \leq \, \mid x_{(2)} \mid \, \leq \, \dots \, \leq \, \mid x_{(n)} \mid$.\\
	For $\alpha_j = \frac{j}{n}, j = 0, \dots , n-1$, the CVaR norm $\CVaRnormX{\alpha}$  of vector $\mathbf{x}$ with parameter $\alpha_j$ is defined as
	\begin{equation}\label{eqn:CVaR_Component_Wise_1}
		\CVaRnormX{\alpha} \defeq \sum \limits_{i = j+1}^{n} \mid x_{(i)} \mid .
	\end{equation}
	For $\alpha$ such that $\alpha_j < \alpha < \alpha_{j+1}$, $j = 0, \dots, n-2$, the CVaR norm $\CVaRnormX{\alpha}$ equals the weighted average of $\CVaRnormX{\alpha_j}$ and $\CVaRnormX{\alpha_{j+1}}$, i.e.
	\begin{equation} \label{eqn:CVaR_Component_Wise_2}
		\CVaRnormX{\alpha} \defeq \lambda \CVaRnormX{\alpha_j} + (1 - \lambda) \CVaRnormX{\alpha_{j+1}},
	\end{equation}
	where
	\begin{align*}
		\lambda &= \frac{\alpha_{j+1} - \alpha}{\alpha_{j+1} - \alpha_j}.
	\end{align*}
	And finally, for $\alpha$ such that $\frac{n-1}{n} < \alpha < 1$,
	\begin{equation}\label{eqn:CVaR_Component_Wise_3}
		\CVaRnormX{\alpha} \defeq n(1 - \alpha) \CVaRnormX{\alpha_{n-1}} = n(1 - \alpha) \max_i \mid x_i \mid .
	\end{equation}
\end{Definition}

Again, some examples will be given to gain a better familiarity with the CVaR norm. The examples are the same as in \autoref{subsec:CVaR_Norms-Scaled-Definition}. For $\mathbf{x} = [ 10, -14, 2, -9]^T$,
\begin{equation}
	\begin{array}{r l l}
		\CVaRnormX{0} &= \vert 2 \vert + \vert -9 \vert + \vert 10 \vert + \vert -14 \vert &= 35 \,, \\
		\\
		\CVaRnormX{0.25} &= \vert -9 \vert + \vert 10 \vert + \vert -14 \vert &= 33 \,, \\
		\\
		\CVaRnormX{0.5} &= \vert 10 \vert + \vert -14 \vert &= 24 \,, \text{~and} \\
		\\
		\CVaRnormX{0.75} &= \vert -14 \vert &= 14 \,.
	\end{array} \notag
\end{equation}

In contrast to $\CVaRnormX[S]{\alpha}$, $\CVaRnormX{\alpha} \not = \CVaRnormX{0.75}$ for $\alpha > 0.75$, as, for example, $\CVaRnormX{0.9} = 4 (1 - 0.9) \cdot 14 = 5.6$. And to calculate $\CVaRnormX{\frac{1}{3}}$, $\lambda$ must be calculated first to use \autoref{eqn:CVaR_Component_Wise_2}. Since $0.25 < \lambda < 0.5$,
\begin{align*}
	\lambda =& \frac{\frac{1}{2} - \frac{1}{3}}{\frac{1}{2} - \frac{1}{4}} = \frac{2}{3} .
\end{align*}
Hence, $\CVaRnormX{\frac{1}{3}} = \lambda \CVaRnormX{0.25} + \left( 1 - \lambda \right) \CVaRnormX{0.5} = \frac{2}{3} 33 + \frac{1}{3} 24$, so that $\CVaRnormX{\frac{1}{3}} = 30$.

For $\xinR{2}$, the unit balls of $\CVaRnormX{\alpha}$ for $\alpha \in \left\{ 0, 0.1, 0.25, 0.4 , 0.5 \right\}$ are shown below in \autoref{fig:Calpha_unit_balls_own_example}.
\begin{figure}[H]
	\centering
	\includegraphics[width = 0.9 \textwidth]{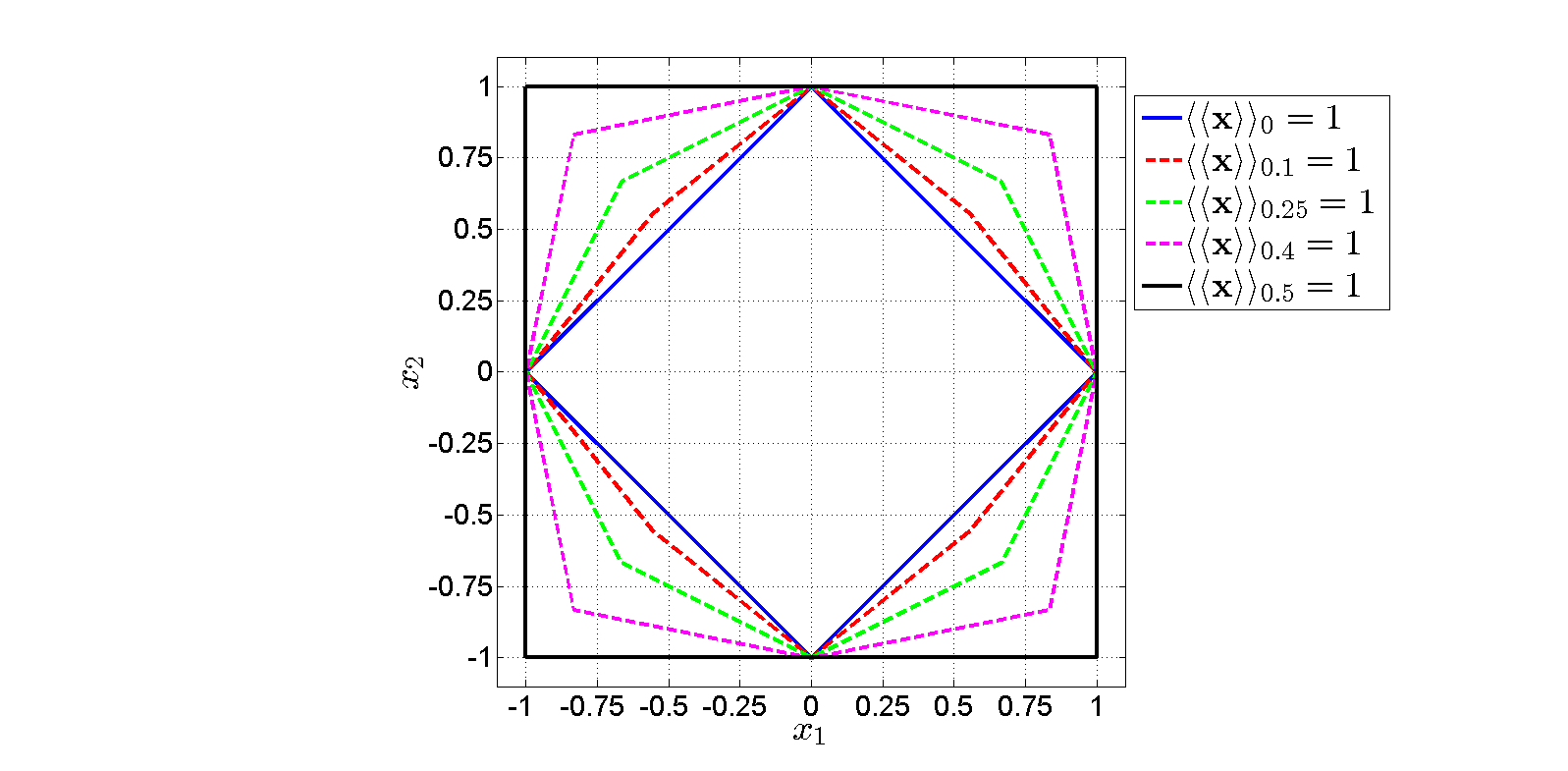}
	\caption{Unit balls of $\CVaRnormX{\alpha}$ for $\xinR{2}$ and different values of $\alpha$.}
	\label{fig:Calpha_unit_balls_own_example}
\end{figure}

\subsection{Alternative Characterization} \label{subsec:CVaR_Norms-NonScaled-Alternative}

Alternatively, the CVaR norm can be obtained by solving the following minimization (using \autoref{eqn:Non_Scaled-Scaled_CVaR_norm} and \autoref{prop:Scaled_CVaR_based_on_Definition}).
\begin{Proposition}[{\cite[p. 16]{pavlikov2014_CVaR_Norm_and_applications} CVaR Norm based on CVaR Definition}]
	\label{prop:CVaR_based_on_Definition}
	For $0 \leq \alpha < 1$,
	\begin{equation} \label{eqn:CVaR_Norm_based_on_Definition}
		\CVaRnormX{\alpha} = \min_c \left( n (1 - \alpha) c + \sum \limits_{i=1}^{n} \left( \mid x_i \mid - c \right)^+ \right) \,.
	\end{equation}
\end{Proposition}

Writing \autoref{prop:CVaR_based_on_Definition} as an LP, i.e.,
\begin{Problem}[problem:CVaR_Norm_LP]
\left.
\begin{array}{rrll}
	\CVaRnormX{\alpha} =& \min \limits_{c} & n (1 - \alpha) c + \sum \limits_{i = 1}^n z_i \\
	&\text{s.t.} & z_i \geq \vert x_i \vert - c & \text{for~} i \in \{1, \dots , n\} \\
	&& z_i \geq 0 & \text{for~} i \in \{1, \dots , n\}
\end{array}
\right\} ,
\end{Problem}
one can use the strong duality theory of LP to obtain an equivalent definition of the CVaR norm \cite[p. 5]{Gotoh2015_Two_Pairs_Of_Polyhedral_Norms_Versus_Lp_Norms}. This alternative definition can be expressed as
\begin{Problem}[eqn:CVaR_Dual_Definition_in_LP]
\left.
\begin{array}{rll}
	\max & \sum \limits_{i = 1}^n \vert x_i \vert q_i \\
	\text{s.t.} & \sum \limits_{i = 1}^n q_i = n (1 - \alpha) & \text{for~} i \in \{1, \dots , n\} \\
	& 0 \leq q_i \leq 1 & \text{for~} i \in \{1, \dots , n\}	
\end{array}
\right\} ,
\end{Problem} ~\\
which is the \emph{continuous knapsack problem}.

The knapsack problem is a standard integer programming problem. Suppose that there is a decision to make on whether to use any of $n$ items, each of which has a benefit $b_i$ and a cost $c_i$ for $i \in \{ 1, 2, \dots , n\}$. The goal is to maximize total benefit with a constraint on the total costs, $C$. The only additional constraint of the knapsack problem is that  the decision variables $q_i$ must be 0 or 1, i.e., an item is used completely or not at all - which makes it an integer programming problem \cite[p. 524]{2004_Winston_OR}. Hence, the knapsack problem can be formulated as
\begin{Problem}[problem:knapsack_original]
\left.
\begin{array}{rll}
	\max \limits_{\mathbf{q}} & \sum \limits_{i = 1}^n b_i q_i \\
	\text{s.t.} & \sum \limits_{i = 1}^n c_i q_i \leq C \\
	& q_i \in \{0 ,1 \} & \text{for~} i \in \{1, \dots , n\}
\end{array}
\right\} .
\end{Problem} ~\\
Changing the integer constraint ($q_i \in \{0 ,1 \}$) to a linear constraint ($0 \leq q_i \leq 1$)  and changing the inequality of the first constraint to an equality transforms the knapsack problem into the \emph{continuous} knapsack problem, which is a linear programming problem. In the continuous knapsack problem it is possible to use fractions of any item, making it easier and more straightforward to solve (see \autoref{def:CVaR_based_on_Dual_CVaR_Problem}). The parameters between \autoref{problem:knapsack_original} and \autoref{eqn:CVaR_Dual_Definition_in_LP} are linked in such a way that $b_i = \vert x_i \vert$, $c_i = 1$ for $i \in \{ 1, \dots , n \}$, and $C = n (1 - \alpha)$.

The optimal objective value of \autoref{eqn:CVaR_Dual_Definition_in_LP} is another equivalent definition of the CVaR norm (since strong duality holds). The optimal objective value of \autoref{eqn:CVaR_Dual_Definition_in_LP} can be found by a greedy algorithm, the result of which is stated below. \footnote{The greedy algorithm (stated in \autoref{def:CVaR_based_on_Dual_CVaR_Problem}) can be interpreted as follows: The knapsack has a limit of $n(1 - \alpha)$ and each vector component $ \vert x_i \vert$ has the same weight. Pack as much of $\vert x_{(1)} \vert$ (the component with highest magnitude) into the knapsack. If the component completely fits into the knapsack (i.e. $q_i = 1$), start packing the component of next highest magnitude. As soon as the knapsack is full, stop. Fractional values for $q_i$ are allowed.} 
\begin{Proposition}[{\cite[p. 6]{Gotoh2015_Two_Pairs_Of_Polyhedral_Norms_Versus_Lp_Norms} CVaR Norm based on dual formulation of CVaR definition}]
	\label{def:CVaR_based_on_Dual_CVaR_Problem}
	Let the absolute values of the components of vector $\mathbf{x} \in \mathbb{R}^n$ be ordered in \emph{descending} order, i.e. $\mid x_{(1)} \mid \, \geq \, \mid x_{(2)} \mid \, \geq \, \dots \, \geq \, \mid x_{(n)} \mid$. Then
	\begin{equation} \label{eqn:CVaR_based_on_Dual_CVaR_Problem}
		\CVaRnormX{\alpha} = \sum \limits_{i = 1}^{ \lfloor n ( 1 - \alpha) \rfloor} \vert x_{(i)} \vert + \left( n (1 - \alpha) - \lfloor n ( 1 - \alpha) \rfloor \right) \vert x_{(\lfloor n ( 1 - \alpha) \rfloor + 1)} \vert .
	\end{equation}
\end{Proposition}

In \autoref{def:CVaR_based_on_Dual_CVaR_Problem}, the absolute values of the components of $\mathbf{x}$ are ordered in descending order, which contrasts the original definition of the CVaR norm in \autoref{def:CVaR_Component_Wise}. This is done so that the equivalence between \autoref{eqn:CVaR_based_on_Dual_CVaR_Problem} and the D-norm given in \autoref{def:D_norm} will become apparent (see \autoref{subsec:sec:CVaR_Norms-Properties-NonScaled}).


\section{CVaR Norm Properties} \label{sec:CVaR_Norms-Properties}

Any function $\rho : \mathbb{R}^n \rightarrow \mathbb{R}$ satisfies the following properties is a norm on $\mathbb{R}^n$ \cite[p. 20]{prugovecki1981_Hilbert_Space_Theory}:
\begin{enumerate}[label= \roman*)]
	\item $\rho ( \mathbf{x} ) \geq 0 \, \forall \mathbf{x} \in \mathbb{R}^n$
	\item $\rho ( \lambda \mathbf{x} ) = \, \mid \lambda \mid \rho ( \mathbf{x} ), \forall \mathbf{x} \in \mathbb{R}^n, \forall \lambda \in \mathbb{R}$
	\item $\rho ( \mathbf{x} + \mathbf{y} ) \leq \rho ( \mathbf{x} ) + \rho ( \mathbf{y} ), \forall \mathbf{x}, \mathbf{y} \in \mathbb{R}^n$
	\item $\rho ( \mathbf{x} ) = 0 \Rightarrow \mathbf{x} = 0 $
\end{enumerate} ~

The scaled CVaR norm and CVaR norm both satisfy these properties. The proof is given in \cite{pavlikov2014_CVaR_Norm_and_applications}. Hence, it is justified to call these objects \emph{norms}.

\subsection{Properties of the Scaled CVaR Norm}

Pavlikov and Uryasev showed that the scaled CVaR norm $C^S_{\alpha}$ is a non-decreasing function of the parameter $\alpha$.
\begin{Proposition}[{\cite[p. 7]{pavlikov2014_CVaR_Norm_and_applications}}] \label{prop:Scaled_CVaR_Norm_non-decreasing}
	For a vector $\mathbf{x} \in \mathbb{R}^n$ and $0 \leq \alpha_1 \leq \alpha_2 \leq 1$,
	\begin{align*}
		\CVaRnormX[S]{\alpha_1} &\leq \CVaRnormX[S]{\alpha_2}.
	\end{align*}
\end{Proposition}

Another property, which to the best knowledge of the author has not been published or proven before, is that the scaled CVaR norm is piecewise convex in $\al{}$ within each interval $[\al{j}, \al{j+1}]$.
\begin{Proposition} \label{prop:Scaled_CVaR_Norm_piecewise_convex}
	For any vector $\xinRn$, and $\alpha \in [\frac{j}{n}, \frac{j+1}{n}] \,, j = 0, 1, \dots n-1$ the scaled CVaR norm $\CVaRnormX[S]{\alpha}$ is convex in $\alpha$, i.e.,
	\begin{align*}
		\CVaRnormX[S]{\lambda \al{1} + (1 - \lambda) \al{2}} \leq& \lambda \CVaRnormX[S]{\al{1}} + (1 - \lambda) \CVaRnormX[S]{\al{2}}
	\end{align*}
for all $\al{1}, \al{2} \in \left[ \frac{j}{n}, \frac{j+1}{n} \right], j = 0, 1, \dots, n-1 $ and $\lambda \in [0,1]$.
\end{Proposition}

\begin{proof}
For $\al{} \in (\frac{n - 1}{n}, 1]$ the proof of \autoref{prop:Scaled_CVaR_Norm_piecewise_convex} is obvious, as $\CVaRnormX[S]{\alpha}$ is constant for these values of $\alpha$.

To show that $\CVaRnormX[S]{\alpha}$ is piecewise convex in $\alpha$ within each interval $[\frac{j}{n}, \frac{j+1}{n}]\,, j = 0, 1, \dots n-2$, \autoref{def:Scaled_CVaR_Component_Wise} can be used, together with the following notation:\\
Suppose that $\al{1}, \al{2} \in [\al{j}, \al{j+1}]$, $t = \lambda \al{1} + (1 - \lambda) \al{2}\,, \lambda \in [0,1]$, and $\al{1}, \al{2}, \al{j}$ and $\al{j+1}$ are labelled $a,b,c,d$ in such a way that 
\begin{equation} \notag
	0 \leq a = \al{j} \leq b \leq t \leq c \leq d = \al{j+1} \leq \frac{n-1}{n} .
\end{equation}
Then $\CVaRnormX[S]{\lambda \al{1} + (1 - \lambda) \al{2}} = \CVaRnormX[S]{t}\,,  \CVaRnormX[S]{\al{1}}$ and $ \CVaRnormX[S]{\al{2}}$ can be written as
\begin{align}
	\CVaRnormX[S]{t} =& \mu_0 \CVaRnormX[S]{a} + \left( 1 -\mu_0 \right) \CVaRnormX[S]{d} & \text{with~} \mu_0 =& \frac{(d - t) (1 - a)}{(d-a) (1 - t)} , \label{eqn:proof1_mu0} \\
	\notag \\
	\CVaRnormX[S]{\al{1}} =& \mu_1 \CVaRnormX[S]{a} + \left( 1 - \mu_1 \right) \CVaRnormX[S]{d} & \text{with~} \mu_1 =& \frac{(d - b) (1 - a)}{(d-a) (1 - b)} \text{, and} \label{eqn:proof1_mu1}\\
	\notag \\
	\CVaRnormX[S]{\al{2}} =& \mu_2 \CVaRnormX[S]{a} + \left( 1 - \mu_2 \right) \CVaRnormX[S]{d} &  \text{with~} \mu_2 =&  \frac{(d - c) (1 - a)}{(d - a) (1 - c)} . \label{eqn:proof1_mu2}
\end{align}
Hence, it needs to be shown that $\CVaRnormX[S]{t} \leq \lambda \CVaRnormX[S]{\al{1}} + (1 - \lambda) \CVaRnormX[S]{\al{2}}$, i.e.
\begin{align*}
	 \mu_0 \CVaRnormX[S]{a} + \left( 1 -\mu_0 \right) \CVaRnormX[S]{d} \leq & \lambda \left[ \mu_1 \CVaRnormX[S]{a} + \left( 1 - \mu_1 \right) \CVaRnormX[S]{d} \right] \\ 
	 & + (1-\lambda) \left[ \mu_2 \CVaRnormX[S]{a} + \left( 1 - \mu_2 \right) \CVaRnormX[S]{d} \right] .
\end{align*}
Rearranging $\CVaRnormX[S]{a}$ and $\CVaRnormX[S]{d}$ leaves to prove that
\begin{align*}
	0 \leq& \left( \lambda \mu_1 + (1 - \lambda) \mu_2 - \mu_0 \right) \CVaRnormX[S]{a} \\
	&+ \left( \lambda (1 - \mu_1) + (1 - \lambda) (1 - \mu_2) - (1 - \mu_0) \right) \CVaRnormX[S]{d} \\
	\Longleftrightarrow 0 \leq & \left( \mu_2 + \lambda \mu_1 - \lambda \mu_2 - \mu_0 \right) \CVaRnormX[S]{a} \\
	& + \left(\mu_0 + \lambda \mu_2 - \lambda \mu_1 - \mu_2 \right) \CVaRnormX[S]{d} \\
	\Longleftrightarrow 0 \leq & \left(\mu_0 + \lambda \mu_2 - \lambda \mu_1 - \mu_2 \right) \left( \CVaRnormX[S]{d} - \CVaRnormX[S]{a} \right) .
\end{align*}

By \autoref{prop:Scaled_CVaR_Norm_non-decreasing}, since $d \geq a \Rightarrow \CVaRnormX[S]{d} - \CVaRnormX[S]{a} \geq 0$. Hence, to complete the proof, it must be shown that $\mu_0 + \lambda \mu_2 - \lambda \mu_1 - \mu_2 \geq 0$ for all $0 \leq a = \al{j} \leq b \leq t \leq c \leq d = \al{j+1} \leq \frac{n-1}{n}$ and $\lambda \in [0,1]$. Using expressions \ref{eqn:proof1_mu0}, \ref{eqn:proof1_mu1} and \ref{eqn:proof1_mu2} and eliminating the common $\frac{1-a}{d-a}$ term yields:
\begin{equation}
	0 \leq \mu_0 + \lambda \mu_2 - \lambda \mu_1 - \mu_2 = \frac{d - t}{1 - t} + \lambda \frac{d - c}{1 - c} -\lambda \frac{d - b}{1 - b} - \frac{d - c}{1 - c} \Leftrightarrow \notag
\end{equation}
\begin{equation} \label{eqn:proof1_0_leq_expression}
	\begin{array}{rl}	
	0 \leq &(d - t)(1 - b)(1 - c) + \lambda (d - c) (1 - b) (1 - t) \\
	&- \lambda (d - b)(1 - c)(1 - t) - (d - c)(1 - b)(1 - t) .
	\end{array}
\end{equation}

Substituting $t = \lambda b + (1 - \lambda) c$ into \autoref{eqn:proof1_0_leq_expression}, expanding all brackets and summarizing the terms gives
\begin{equation} \notag
	\begin{array}{rl}	
	0 \leq & \lambda \left( b^2 - b^2 d + c^2 - c^2 d + 2bcd - 2bc \right) \\
	& + \lambda^2 \left( b^2 d - b^2 + c^2 d - c^2 + 2bc - 2bcd \right) ,
	\end{array}
\end{equation}
which simplifies to
\begin{equation} \label{eqn:proof1_final_expression}
	0 \leq \lambda \left( 1 - \lambda \right) \left( 1 - d \right) \left( c - b \right)^2 .
\end{equation}
\autoref{eqn:proof1_final_expression} holds for all  $0 \leq a = \al{j} \leq b \leq t \leq c \leq d = \al{j+1} \leq \frac{n-1}{n}$ and $\lambda \in [0,1]$, which completes the proof.
\end{proof}

To illustrate \autoref{prop:Scaled_CVaR_Norm_piecewise_convex}, $\CVaRnormX[S]{\alpha}$ is drawn against $\alpha$ for four different $\mathbf{x}$ in \autoref{fig:piecewise_convexity_example}. Depending on the components of $\mathbf{x}$, the convexity is more or less pronounced in the graphs.

\begin{figure}[H]
	\centering
	\includegraphics[width = 0.95 \textwidth]{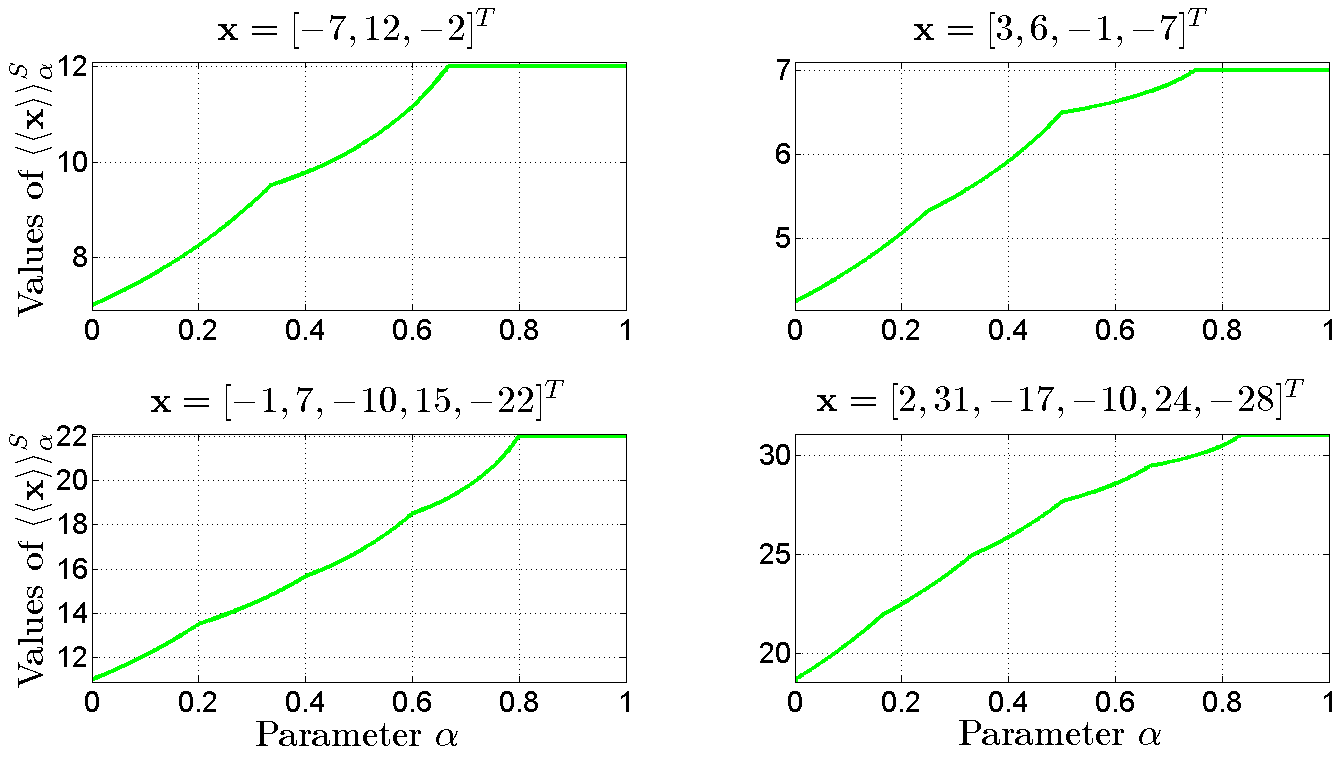}
	\caption{Scaled CVaR norm $C^S_{\alpha}$ against $\alpha$ for different $\mathbf{x}$.}
	\label{fig:piecewise_convexity_example}
\end{figure}

To show that $\CVaRnormX[S]{\alpha}$ is not convex over the whole interval $[0,1]$ consider $\mathbf{x} = [-7, 12, -2]$, whose scaled CVaR norm is shown in the top left graph of \autoref{fig:piecewise_convexity_example}. Taking $\al{1} = 0.2$, $\al{2} = 0.4$, and $\lambda = \frac{1}{3}$ gives $\al{t} = \lambda \al{1} + (1-\lambda) \al{2} = \frac{1}{3}$ and
\begin{align*}
	\CVaRnormX[S]{0.2} =& \frac{33}{4} = 8.25 \,,\\
	\CVaRnormX[S]{0.4} =& \frac{88}{9} \approx 9.78 \,, \text{~and} \\
	\CVaRnormX[S]{\frac{1}{3}} =& \frac{19}{2} = 9.5 \,.
\end{align*}
Hence, $\CVaRnormX[S]{\al{t}} = \CVaRnormX[S]{\frac{1}{3}} = 9.5 \not \leq \lambda \CVaRnormX[S]{0.2} + (1 - \lambda) \CVaRnormX[S]{0.4} = \frac{1}{3} \frac{33}{4} + \frac{2}{3} \frac{88}{9} \approx 9.27$. Therefore, $\CVaRnormX[S]{\alpha}$ is only piecewise convex, but not over the whole interval $[0,1]$. This is also apparent from the plots themselves.

\subsection{Properties of the CVaR Norm} \label{subsec:sec:CVaR_Norms-Properties-NonScaled}

While the scaled CVaR norm is a non-decreasing function of the parameter $\alpha$ (see \autoref{prop:Scaled_CVaR_Norm_non-decreasing}), the CVaR norm shows different properties:
\begin{Proposition}[{\cite[p. 15]{pavlikov2014_CVaR_Norm_and_applications}}] \label{prop:CVaR_Norm_non-increasing_concave_pwl}
	For $\mathbf{x} \in \mathbb{R}^n$, the CVaR norm $\CVaRnormX{\alpha}$ is a non-increasing, concave, piecewise-linear function of the parameter $\alpha$. 
\end{Proposition}

Furthermore, the CVaR norm $C_{\alpha}$ coincides with the D-norm, which is defined below.
\begin{Definition}[{\cite[p. 513]{bertsimas2004_Robust_Optimization_under_general_norms} D-Norm}] \label{def:D_norm}
	For $\xinRn$ and parameter $\kappa \in [1,n]$, the D-norm $\DnormX{\kappa}$ is defined as
	\begin{align*}
		\DnormX{\kappa} &\defeq \max_{S,t} \left( \sum \limits_{i \in S} \vert x_i \vert + (\kappa - \lfloor \kappa \rfloor) \vert x_t \vert \right),
	\end{align*}
	where $N = \{1, \dots, n\} \,, S \subseteq N \,, \vert S \vert \leq \lfloor \kappa \rfloor \,,$ and $t \in S \setminus N$.
\end{Definition}

The D-norm is used in robust optimization as an alternative to the $L_2$ norm for describing an uncertainty set using a norm. The D-norm has advantages such as the guarantee of feasibility independent of uncertainty distributions and a flexibility in trade off between robustness and performance \cite[p. 40]{2014Yang_DRO}. A further discussion of the D-norm (beyond the coincidence with the $C_{\alpha}$ norm) or robust optimization in general is beyond the scope of this thesis. Further discussions on the D-norm are given in \cite{bertsimas2004_Robust_Optimization_under_general_norms} and \cite{2014Yang_DRO}, while robust optimization is discussed in \cite[p. 292ff.]{2007_Cornuejols_OMF} or \cite{1998Ben-Tal_Robust_Convex_Optimization}.\footnote{This is only a selection of available literature on these topics.}

\begin{Proposition}[{\cite[p. 16]{pavlikov2014_CVaR_Norm_and_applications}}] \label{prop:CVaR_Norm_coincides_with_D_norm}
	For $\xinRn$, the CVaR norm $\CVaRnormX{\alpha}$ with parameter $\alpha \in [0, \frac{n-1}{n}]$ coincides with the D-norm $\DnormX{\kappa}$ with parameter $\kappa = n(1 - \alpha)$, i.e. $\CVaRnormX{\alpha} = \DnormX{\kappa}$.
\end{Proposition}

This is because the D-norm is an equivalent formulation to the CVaR norm given in \autoref{def:CVaR_based_on_Dual_CVaR_Problem}. Note that \autoref{prop:CVaR_Norm_coincides_with_D_norm} does \emph{not} hold for $\frac{n-1}{n} < \alpha \leq 1$, as for $\frac{n-1}{n} < \alpha \leq 1 \Rightarrow \kappa = n (1 - \alpha) < 1 \Rightarrow \kappa \not \in [1,n] $, so that the D-norm is not defined in this case \cite[p. 16]{pavlikov2014_CVaR_Norm_and_applications}.

Comparisons to $L_p$ norms are made more extensively in \autoref{chapter:Comparison_to_other_vector_norms}.


\section{Computational Efficiency} \label{sec:CVaR_Norms-Computational_Efficiency}

This section investigates how computationally efficient different algorithms are for calculating $\CVaRnormX[S]{\alpha}$ and $\CVaRnormX{\alpha}$. The definitions of $\CVaRnormX[S]{\alpha}$ and $\CVaRnormX{\alpha}$ in \autoref{def:Scaled_CVaR_Component_Wise} and \autoref{def:CVaR_Component_Wise}, respectively, naturally lead to simple algorithms for computing the norms. The algorithms that were implemented in MATLAB are printed in \autoref{code:Scaled_CVaR_Component_Wise} for $\CVaRnormX[S]{\alpha}$ and \autoref{code:CVaR_Component_Wise} for $\CVaRnormX{\alpha}$. Informally, they can be described as follows:
\begin{enumerate}
	\item Take the absolute values of the entries of $\xinRn$ and order them in ascending order.
	\item If $\al{} > \frac{n-1}{n}$, use \autoref{eqn:Scaled_CVaR_Component_Wise_3} or \autoref{eqn:CVaR_Component_Wise_3} to calculate $C^S_{\alpha}$ or $C_{\alpha}$, respectively.
	\item If $\al{} = \al{j}$, i.e., $\al{} = \frac{j}{n}$ for any $j = 0, 1, \dots, n-1$, use \autoref{eqn:Scaled_CVaR_Component_Wise_1} or \autoref{eqn:CVaR_Component_Wise_1} to calculate $C^S_{\alpha}$ or $C_{\alpha}$, respectively.
	\item Otherwise, find the closest $\al{j}$ and $\al{j+1}$, such that $\al{j} < \alpha < \al{j+1}$, calculate $\mu$ (for $C^S_{\alpha}$) or $\lambda$ (for $C^S_{\alpha}$), and use \autoref{eqn:Scaled_CVaR_Component_Wise_2} or \autoref{eqn:CVaR_Component_Wise_2} to calculate $C^S_{\alpha}$ or $C_{\alpha}$, respectively.
\end{enumerate} ~

To calculate $\CVaRnormX[S]{\alpha}$ and $\CVaRnormX{\alpha}$ using \autoref{prop:Scaled_CVaR_based_on_Definition} or \autoref{prop:CVaR_based_on_Definition}, respectively, the according optimization problem was written in MATLAB CVX (\cite{cvx2},\cite{cvx1}, for the code see \autoref{code:Scaled_CVaR_based_on_Definition} and \autoref{code:CVaR_based_on_Definition}). The algorithm that was used to solve the optimization problem was picked automatically by CVX with no further input by the author. When referring an ``optimization algorithm'' in the remainder of this section, the codes given in \autoref{prop:Scaled_CVaR_based_on_Definition} or \autoref{prop:CVaR_based_on_Definition} are meant.

To compare the computational efficiencies of the different algorithms, random vectors of dimensions $n \in \{2,3,10,10^2,10^3,10^4,10^5\}$ were generated, and each of the algorithms given in \autoref{code:Scaled_CVaR_Component_Wise} - \autoref{code:CVaR_based_on_Definition} was run 10 times to calculate $C^S_{\alpha}$ or $C_{\alpha}$, respectively. The average time taken over the 10 runs is the computation time stated in \autoref{table:CVaR_Norm_Computation_Time_for_different_n}, \autoref{table:CVaR_Norm_Computation_Time_for_different_alpha}, and \autoref{app_table:CVaR_Norm_Computation_Time_for_different_alpha_and_n}. These calculations were performed for values of $\alpha \in \{ 0, 0.1, 0.25, 0.5, 0.7, 0.9\}$

Summaries of the results are given in \autoref{table:CVaR_Norm_Computation_Time_for_different_n} and \autoref{table:CVaR_Norm_Computation_Time_for_different_alpha}; the complete results are displayed in \autoref{app_table:CVaR_Norm_Computation_Time_for_different_alpha_and_n}. \footnote{All calculations are performed on a PC with an Intel Core iS-2400S with 4 cores @ 2.5 GHz and 4 GB of memory.}
\begin{table}[H] \small
	\centering
	\begin{tabular}{| c | r || *{2}{r|}| *{2}{r|}  }
		\hline
		& & \multicolumn{4}{c|}{\textbf{Computation time in ms}} \\
		& & \multicolumn{2}{c||}{\textbf{Component-wise}} & \multicolumn{2}{c|}{\textbf{Optimization}}\\
$\boldsymbol\alpha$	&	$\boldsymbol n$	& \parbox{2.5cm}{\centering $\CVaRnormX[S]{\alpha}$ \\ (\autoref*{def:Scaled_CVaR_Component_Wise})} & \parbox{2.5cm}{\centering $\CVaRnormX{\alpha}$ \\ (\autoref*{def:CVaR_Component_Wise})} & \parbox{2.5cm}{\centering $\CVaRnormX[S]{\alpha}$ \\ (\autoref*{prop:Scaled_CVaR_based_on_Definition})} & \parbox{2.5cm}{\centering $\CVaRnormX{\alpha}$ \\ (\autoref*{prop:CVaR_based_on_Definition})}\\
\hline							
\hline											
\multirow{7}{*}{0.5}	&	2	&	0.13	&	0.08	&	178.59	&	174.96	\\
	&	3	&	0.18	&	0.12	&	180.96	&	179.34	\\
	&	10	&	0.13	&	0.08	&	184.33	&	181.49	\\
	&	100	&	0.15	&	0.10	&	217.66	&	213.11	\\
	&	1,000	&	0.19	&	0.14	&	323.36	&	239.72	\\
	&	10,000	&	1.00	&	0.92	&	571.45	&	551.93	\\
	&	100,000	&	5.64	&	5.00	&	5516.37	&	5128.19	\\
\hline											
	\end{tabular}
	\caption[Computations times of Scaled and Non-Scaled CVaR norms for different $n$.]{Computation times of $\CVaRnormX[S]{\alpha}$ and $\CVaRnormX{\alpha}$ at $\alpha = 0.5$ of a vector $\xinRn$ for different $n$ in milliseconds.}
	\label{table:CVaR_Norm_Computation_Time_for_different_n}
\end{table}

\begin{table}[H] \small
	\centering
	\begin{tabular}{| c | r@{.}l || *{2}{r|}| *{2}{r|}  }
		\hline
		& \multicolumn{2}{c||}{}& \multicolumn{4}{c|}{\textbf{Computation time in ms}} \\
		& \multicolumn{2}{c||}{}& \multicolumn{2}{c||}{\textbf{Component-wise}} & \multicolumn{2}{c|}{\textbf{Optimization}}\\
$\boldsymbol n$	&	\multicolumn{2}{c||}{$\boldsymbol \alpha$} & \parbox{2.5cm}{\centering $\CVaRnormX[S]{\alpha}$ \\ (\autoref*{def:Scaled_CVaR_Component_Wise})} & \parbox{2.5cm}{\centering $\CVaRnormX{\alpha}$ \\ (\autoref*{def:CVaR_Component_Wise})} & \parbox{2.5cm}{\centering $\CVaRnormX[S]{\alpha}$ \\ (\autoref*{prop:Scaled_CVaR_based_on_Definition})} & \parbox{2.5cm}{\centering $\CVaRnormX{\alpha}$ \\ (\autoref*{prop:CVaR_based_on_Definition})}\\
\hline							
\hline											
\multirow{7}{*}{1,000}	&	0&0	&	0.19	&	0.14	&	202.81	&	199.38	\\
&	0&1	&	0.19	&	0.14	&	244.86	&	236.01	\\
&	0&25	&	0.19	&	0.14	&	229.73	&	271.94	\\
&	0&5	&	0.19	&	0.14	&	323.36	&	239.72	\\
&	0&7	&	0.19	&	0.15	&	252.11	&	241.46	\\
&	0&9	&	0.19	&	0.14	&	289.31	&	249.22	\\

\hline											
	\end{tabular}
	\caption[Computations times of Scaled and Non-Scaled CVaR norms for different $\alpha$.]{Computation times of $\CVaRnormX[S]{\alpha}$ and $\CVaRnormX{\alpha}$ at different $\alpha$ of a vector $\xinRn$ for $n = 1000$ in milliseconds.}
	\label{table:CVaR_Norm_Computation_Time_for_different_alpha}
\end{table}

\autoref{table:CVaR_Norm_Computation_Time_for_different_n} indicates that for $n \leq 1,000$ the computing times for $\CVaRnormX[S]{\alpha}$ and $\CVaRnormX{\alpha}$ using the component-wise algorithms do not increase significantly with increasing $n$. For $n \geq 10,000$ there is a notable increase in computing time with increasing $n$, for both algorithms and both norms.

\autoref{table:CVaR_Norm_Computation_Time_for_different_alpha} shows that the value of $\alpha$ does not have any considerable effect on the computing time for the component-wise algorithm, whereas the computing times for the optimization algorithm fluctuate with $\alpha$.

Both tables clearly show that the component-wise algorithms (given in \autoref{code:Scaled_CVaR_Component_Wise} and \autoref{code:CVaR_Component_Wise}) outperform the optimization algorithms by several orders of magnitude. Hence, in the rest of this thesis only the component-wise algorithms will be used when comparing computational efficiencies against other norms. However, the component-wise algorithms cannot be used to solve any optimization problem involving the calculation of a CVaR norm as constraints cannot be included. Hence, the optimization algorithms to calculate $C^S_{\alpha}$ and $C_{\alpha}$ are the only choice when trying to solve optimization problems, e.g. model recovery problems discussed in \autoref{chapter:Model_Recovery_using_Atomic_Norms}.

\clearpage

%
%

\chapter{Comparisons to $L_p$ Vector Norms} \label{chapter:Comparison_to_other_vector_norms}

This chapter explores how the scaled CVaR norm $C^S_{\alpha}$ and CVaR norm $C_{\alpha}$ compare to several $L_p$ norms for different values of $\alpha$ and $p$, as investigated by \cite{Gotoh2015_Two_Pairs_Of_Polyhedral_Norms_Versus_Lp_Norms} and \cite{pavlikov2014_CVaR_Norm_and_applications}. First, in \autoref{sec:Comparison_to_other_vector_norms-Behaviour_of_Scaled_CVaR_Norm} a brief overview of the behaviour of $C^S_{\alpha}$ will be given following the examples of \cite{pavlikov2014_CVaR_Norm_and_applications}. Then, the focus will shift to the $C_{\alpha}$ norm: \autoref{subsec:Comparison_to_other_vector_norms-Relationship_between_alpha_and_p} illustrates how $\alpha$ and $p$ can be chosen so that $C_{\alpha}$ best approximates $L_p$. To conclude this chapter, \autoref{subsec:Comparison_to_other_vector_norms-Behaviour_of_CVaR_Norm} extends the numerical examples for $C_{\alpha}$ given in \cite{pavlikov2014_CVaR_Norm_and_applications} by the findings of \autoref{subsec:Comparison_to_other_vector_norms-Relationship_between_alpha_and_p}.


\section{Behaviour of Scaled CVaR Norm $C^S_{\alpha}$} \label{sec:Comparison_to_other_vector_norms-Behaviour_of_Scaled_CVaR_Norm}

To describe the behaviour of the scaled CVaR norm, Pavlikov and Uryasev use two examples \cite[p. 4 ff.]{pavlikov2014_CVaR_Norm_and_applications}. For each comparison, the scaled $L^S_p$ norm is used, which is defined by
\begin{align} \label{eqn:LSp_Norm}
	\LnormX[S]{p} &= \left( \frac{1}{n} \sum \limits_{i=1}^n \vert x_i \vert ^p \right)^{\frac{1}{p}} \,,
\end{align}
where $p \geq 1$. The actual examples used for the comparison are:
\begin{enumerate}
	\item Let $\mathbf{x} = (2,1,7,10,-12)^T$, calculate $\CVaRnormX[S]{\alpha}$ for $\alpha \in [0,1]$ and corresponding $\LnormX[S]{p}$ for $p = \frac{1}{(1 - \alpha)^2}$. This is shown in \autoref{fig:Scaled_CVaR_LSp_Norm_For_different_alpha}.
	\item Compare the unit disks for $C^S_{\alpha}$ and $L^S_{p}$, i.e. the sets $U^S_{\alpha} = \{ \mathbf{x} = (x_1,x_2) \mid \CVaRnormX[S]{\alpha} \leq 1 \}$ and $U^S_p = \{ \mathbf{x} = (x_1,x_2) \mid \LnormX[S]{p} \leq 1 \}$ for $\alpha \in \{ 0, 0.1, 1 - \frac{1}{\sqrt{2}}, 0.4, 1 \}$ and corresponding $p (\alpha) = \frac{1}{(1 - \alpha)^2}$. This comparison is shown in \autoref{fig:Scaled_CVaR_LSp_Norm_Unit_Disks}.
\end{enumerate}

\begin{figure}[H]
	\centering
	\includegraphics[width = 0.9 \textwidth]{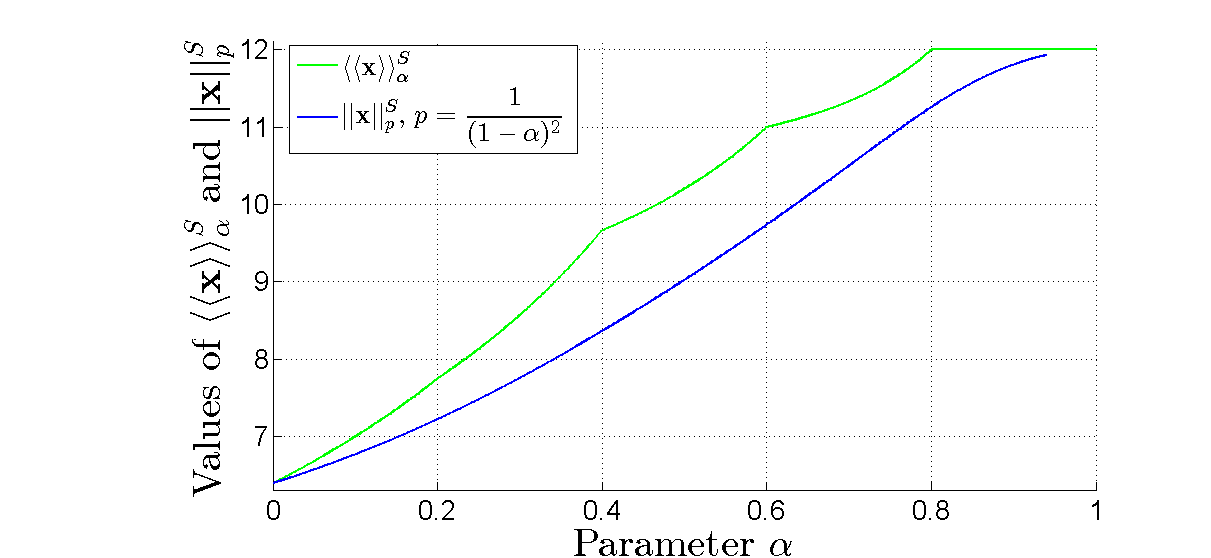}
	\caption{Reproduced from \cite[p. 6]{pavlikov2014_CVaR_Norm_and_applications}, $C^S_{\alpha}$ and $L^S_p$ Norms of $\mathbf{x}$ for different values of $\alpha$ and $p(\alpha)$.}
	\label{fig:Scaled_CVaR_LSp_Norm_For_different_alpha}
\end{figure}

\begin{figure}[H]
	\centering
	\includegraphics[width = 0.6 \textwidth]{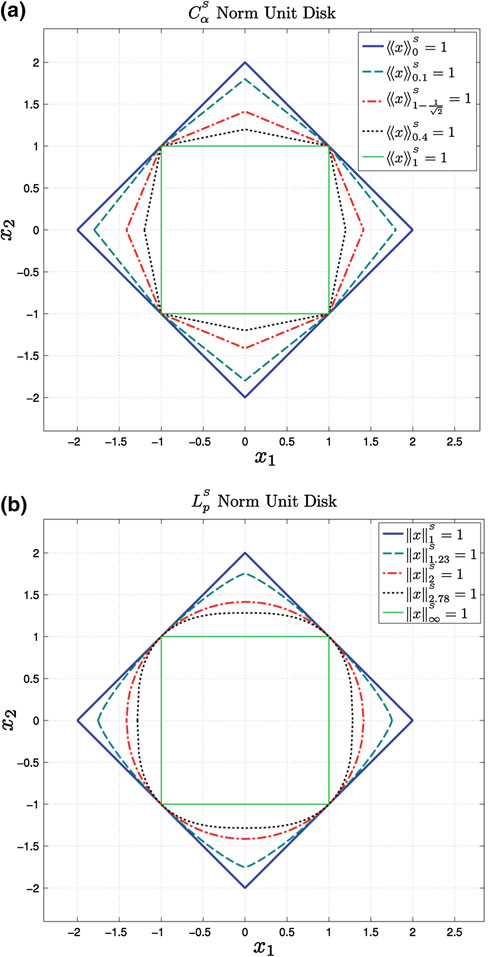}
	\caption{\cite[p. 5]{pavlikov2014_CVaR_Norm_and_applications} Norm unit disks of $C^S_{\alpha}$ and $L^S_p$ for different values of $\alpha$ and $p(\alpha)$.}
	\label{fig:Scaled_CVaR_LSp_Norm_Unit_Disks}
\end{figure}

As can be seen in \autoref{fig:Scaled_CVaR_LSp_Norm_Unit_Disks}, $\CVaRnormX[S]{0} = \LnormX[S]{1}$ and $\CVaRnormX[S]{\alpha} = \LnormX[S]{\infty} \text{~for~} \alpha \in [\frac{n-1}{n},1]$. This relationship follows from \autoref{def:Scaled_CVaR_Component_Wise} and \autoref{eqn:LSp_Norm}.


\section{Relationship between $\alpha$ and $p$ for $C_{\alpha}$ and $L_p$} \label{subsec:Comparison_to_other_vector_norms-Relationship_between_alpha_and_p}

In \cite{Gotoh2015_Two_Pairs_Of_Polyhedral_Norms_Versus_Lp_Norms}, Gotoh and Uryasev explored (among other things) the question: ``For what value of $\kappa \in [1, n]$ does the CVaR norm (or its dual \footnote{This thesis will not introduce or explain the \emph{dual CVaR norm}, but focus on the findings of \cite{Gotoh2015_Two_Pairs_Of_Polyhedral_Norms_Versus_Lp_Norms} regarding the CVaR norm (which was defined in \autoref{sec:CVaR_Norms-Non_scaled_CVaR_Norm}).}) give the best approximation of the $L_p$-norm, and in which sense is it the best'' \cite[p. 3]{Gotoh2015_Two_Pairs_Of_Polyhedral_Norms_Versus_Lp_Norms}?\footnote{Here, $k$ refers is the parameter used in \autoref*{def:D_norm} of the D-norm, which is related to $\alpha$ as $\kappa = n (1 - \alpha)$ (see \autoref*{prop:CVaR_Norm_coincides_with_D_norm}).}

Gotoh's and Uryasev's analysis consisted of finding \emph{tight} bounds on the ration $\frac{\CVaRnormX{\alpha}}{\LnormX{p}}$ - a lower bound $L$ and an upper bound $U$, such that $L \leq \frac{\CVaRnormX{\alpha}}{\LnormX{p}} \leq U$.\footnote{The term \emph{tight} means that there is some $\mathbf{x}$ which satisfies the equality.} Then they defined the ratio $U / L$ as a measure of proximity (i.e. the goodness of approximation of $\LnormX{p}$ by $\CVaRnormX{\alpha}$). Finally, they defined a quasi-convex function $f_{n,p}(\kappa) = U / L$ and analysed for with value of $\alpha (p)$ $f_{n,p}(\kappa)$ attains its minimum. This $\alpha^*$ then gives $\CVaRnormX{\alpha^*}$, which is is the best approximation of $\LnormX{p}$.

\begin{Proposition}[{\cite[p. 6]{Gotoh2015_Two_Pairs_Of_Polyhedral_Norms_Versus_Lp_Norms}}] \label{prop:Lower_Upper_Bounds_on_CVaR_Lp_Norm_Ratio}
	For any $p \in (1, \infty)$, $\alpha \in [0, \frac{n-1}{n}]$, and $\xinRn \setminus \{0\}$, it is valid
	\begin{equation} \label{eqn:Lower_Upper_Bounds_on_CVaR_Lp_Norm_Ratio}
		\min \{1, n^{1 - \frac{1}{p}} (1 - \alpha) \} \leq \frac{\CVaRnormX{\alpha}}{\LnormX{p}} \leq \left( \lfloor \kappa \rfloor + (\kappa - \lfloor \kappa \rfloor)^{\frac{p}{p-1}} \right)^{\frac{p-1}{p}} ,
	\end{equation}
	where $\kappa = n (1 - \alpha)$.
\end{Proposition}
The proof of \autoref{prop:Lower_Upper_Bounds_on_CVaR_Lp_Norm_Ratio} is given in Chapter A.1 of \cite{Gotoh2015_Two_Pairs_Of_Polyhedral_Norms_Versus_Lp_Norms}.

Based on \autoref{eqn:Lower_Upper_Bounds_on_CVaR_Lp_Norm_Ratio}, the ratio $U / L$, where $U = \left( \lfloor \kappa \rfloor + (\kappa - \lfloor \kappa \rfloor)^{\frac{p}{p-1}} \right)^{\frac{p-1}{p}}$ and $L = \min \{1, n^{1 - \frac{1}{p}} (1 - \alpha) \}$ defines a function, which evaluates the proximity of $\CVaRnormX{\alpha}$ to $\LnormX{p}$:
\begin{equation} \label{eqn:CVaR_Lp_Norm_Bounds_Ratio_function}
	f_{n,p} (\kappa) \defeq \frac{\left( \lfloor \kappa \rfloor + (\kappa - \lfloor \kappa \rfloor)^{\frac{p}{p-1}} \right)^{\frac{p-1}{p}}}{\min \{1, n^{1 - \frac{1}{p}} (1 - \alpha) \}}	 \,.
\end{equation}

\begin{Lemma}[{\cite[p. 9]{Gotoh2015_Two_Pairs_Of_Polyhedral_Norms_Versus_Lp_Norms}}] \label{lemma:f_np_kappa_is_continuous}
	The function $f_{n,p} (\kappa)$ is continuous at any $\kappa \in (1,n)$, and differentiable at any non-integer except $\kappa = n^{\frac{1}{p}}$, i.e. $\kappa \not \in \{ 1, \dots, n \} \cup \{ n^{\frac{1}{p}} \}$.
\end{Lemma}

\begin{Proposition} [{\cite[p. 9]{Gotoh2015_Two_Pairs_Of_Polyhedral_Norms_Versus_Lp_Norms}}] \label{prop:f_np_kappa_has_minimum_at}
	The function $f_{n,p} (\kappa)$ is decreasing for $\kappa \leq n^{\frac{1}{p}}$. The function $f_{n,p} (\kappa)$ is increasing for $\kappa \geq n^{\frac{1}{p}}$. Accordingly, $f_{n,p} (\kappa)$ uniquely attains its minimum value, $\left( \lfloor \kappa \rfloor + (\kappa - \lfloor \kappa \rfloor)^{\frac{p}{p-1}} \right)^{\frac{p-1}{p}}$, at $\kappa = n^{\frac{1}{p}}$.
\end{Proposition}

The proofs of \autoref{lemma:f_np_kappa_is_continuous} and \autoref{prop:f_np_kappa_has_minimum_at} are given in sections A.3 and A.4 of \cite{Gotoh2015_Two_Pairs_Of_Polyhedral_Norms_Versus_Lp_Norms}, respectively.

Using \autoref{prop:f_np_kappa_has_minimum_at} and substituting $\kappa = n (1 - \alpha)$ gives the values of $\alpha$ and $p$ for which $\CVaRnormX{\alpha}$ best approximates $\LnormX{p}$ \cite[p. 9]{Gotoh2015_Two_Pairs_Of_Polyhedral_Norms_Versus_Lp_Norms} as
\begin{align}
	\alpha^* &= 1 - n^{\frac{1}{p} - 1} \label{eqn:approximation_alpha_star} \text{~, and}\\
	p^* &= \frac{\ln(n)}{\ln(n (1 - \alpha))} \label{eqn:approximation_p_star} \,.
\end{align}

Gotoh and Uryasev also compared the proximity ratio $U/L = f_{n,p} (\kappa) $ given by \autoref{eqn:CVaR_Lp_Norm_Bounds_Ratio_function} for different combinations of $p$ and $n$, each with optimal $\kappa^* = n(1 - \alpha^*) = n^{\frac{1}{p}}$ (see \autoref{fig:f_np_kappa_ratio}). The ratio $f_{n,p} (\kappa^*)$ becomes largest at $p=2$, which indicates that $L_2$ is the hardest $L_p$ norm to approximate by the CVaR norm \cite[p. 11]{Gotoh2015_Two_Pairs_Of_Polyhedral_Norms_Versus_Lp_Norms}.
\begin{figure}[H]
	\centering
	\includegraphics[width = 0.7 \textwidth]{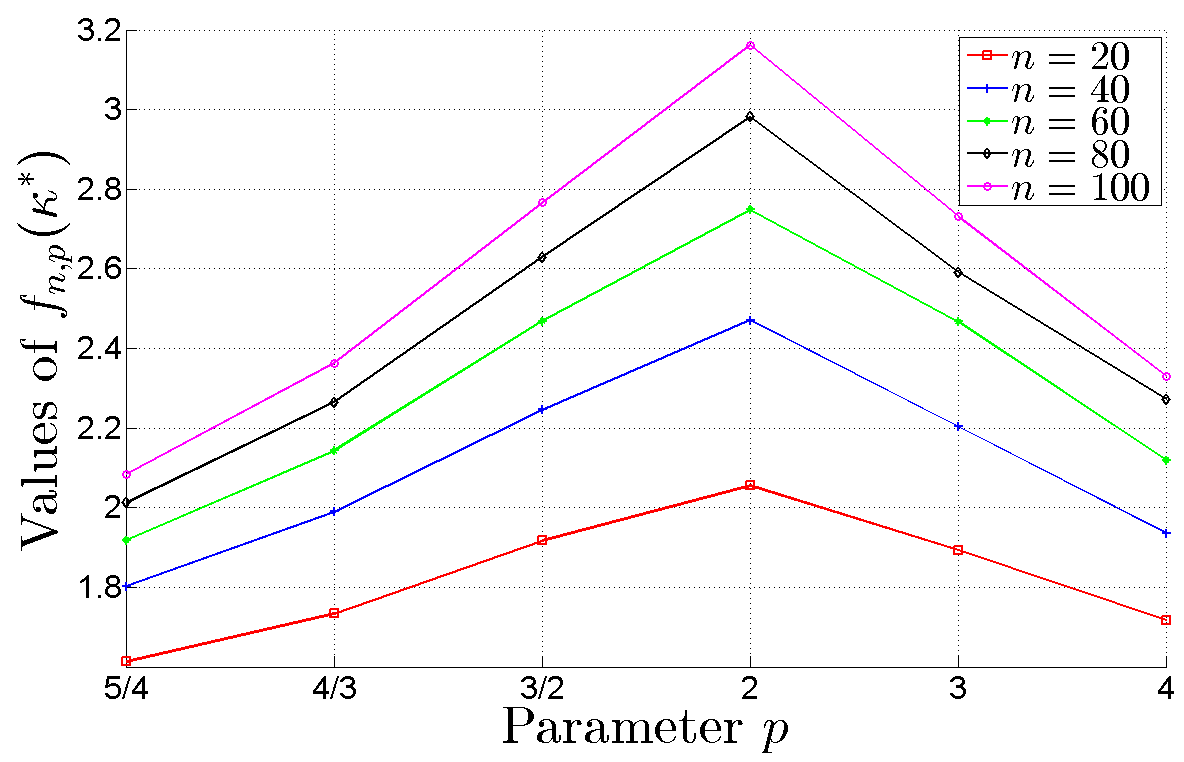}
	\caption{Reproduced from \cite[p. 11]{Gotoh2015_Two_Pairs_Of_Polyhedral_Norms_Versus_Lp_Norms}, $f_{n,p} (\kappa^*)$ for different values of $n$ and $p$, with $\kappa^* = n^{\frac{1}{p}}$.}
	\label{fig:f_np_kappa_ratio}
\end{figure}


\section{Behaviour of CVaR Norm $C_{\alpha}$} \label{subsec:Comparison_to_other_vector_norms-Behaviour_of_CVaR_Norm}

To see how $C_{\alpha}$ behaves for different values of $\alpha$, Pavlikov and Uryasev used the same examples as in the previous subsection, but compared $C_{\alpha}$ to standard $L_p$ norms
\begin{align} \label{eqn:Standard_Lp_Norm}
	\LnormX{p} &= \left( \sum \limits_{i=1}^n \vert x_i \vert ^p \right)^{\frac{1}{p}} \,,
\end{align}
where $p \geq 1$. Hence, using the same numerical examples the comparisons are
\begin{enumerate}
	\item Let $\mathbf{x} = (2,1,7,10,-12)^T$, calculate $\CVaRnormX{\alpha}$ for $\alpha \in [0,1]$ and corresponding $\LnormX{p}$ and $\LnormX{p^*}$, with $p = \frac{1}{(1 - \alpha)^2}$ and optimal\footnote{Here, \emph{optimal} means that for $p = p^*$, $\LnormX{p}$ best approximates $\CVaRnormX{\alpha}$} $p^* = \frac{\ln(n)}{\ln(n (1 - \alpha))}$. This is shown in \autoref{fig:CVaR_Lp_Norm_For_different_alpha}.
	\item Compare the unit disks for $C_{\alpha}$ and $L_{p}$, i.e. the sets $U_{\alpha} = \{ \mathbf{x} = (x_1,x_2) \mid \CVaRnormX{\alpha} \leq 1 \}$ and $U_p = \{ \mathbf{x} = (x_1,x_2) \mid \LnormX{p} \leq 1 \}$ for $\alpha \in \{ 0, 0.1, 1 - \frac{1}{\sqrt{2}}, 0.4, 0.5 \}$ and corresponding $p (\alpha) = \frac{1}{(1 - \alpha)^2}$. This comparison is shown in \autoref{fig:CVaR_Lp_Norm_Unit_Disks}.
\end{enumerate}

\begin{figure}[H]
	\centering
	\includegraphics[width = 0.7 \textwidth]{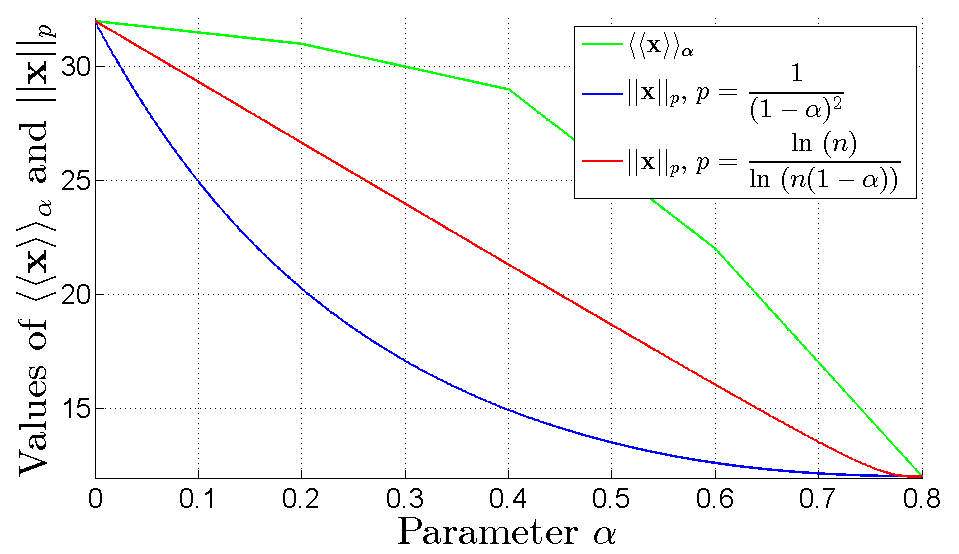}
	\caption{Reproduced from \cite[p. 10]{Gotoh2015_Two_Pairs_Of_Polyhedral_Norms_Versus_Lp_Norms}, $C_{\alpha}$ and $L_p$ Norms of $\mathbf{x}$ for different values of $\alpha$ and $p(\alpha)$.}
	\label{fig:CVaR_Lp_Norm_For_different_alpha}
\end{figure}

\begin{figure}[H]
	\centering
	\includegraphics[width = 0.6 \textwidth]{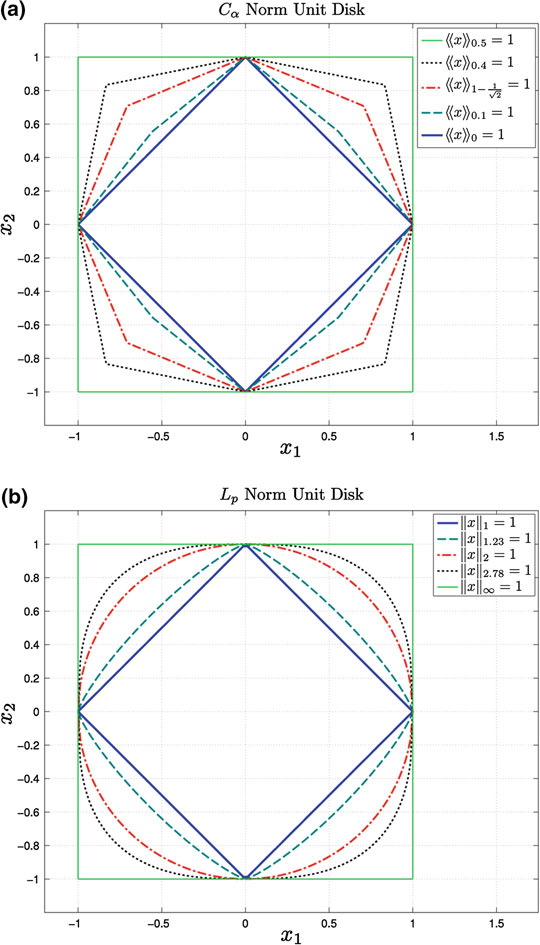}
	\caption{\cite[p. 17]{pavlikov2014_CVaR_Norm_and_applications} Norm unit disks of $C_{\alpha}$ and $L_p$ for different values of $\alpha$ and $p(\alpha)$.}
	\label{fig:CVaR_Lp_Norm_Unit_Disks}
\end{figure}

Again, there is a close relationship between $C_{\alpha}$ and $L_1$ / $L_{\infty}$. As is depicted in \autoref{fig:CVaR_Lp_Norm_Unit_Disks} and as can be shown from \autoref{eqn:CVaR_Component_Wise_1} and \autoref{eqn:Standard_Lp_Norm}, $\CVaRnormX{0} = \LnormX{1}$ and $\CVaRnormX{\frac{n-1}{n}} = \LnormX{\infty}$.

Letting $\mathbf{x} \in \mathbb{R}^2 : \vert x_1 \vert, \vert x_2 \vert \leq 10$ and producing surface plots of $\CVaRnormX{\alpha^*}$ and $\LnormX{p}$ for $p = 2$ and $\alpha^* = \frac{1}{1 - \sqrt{2}}$ gives the plots shown in \autoref{fig:CVaR_Lp_Norm_surface_plots}. Additional surface plots for varying values of $\alpha$ and $p^*$ are displayed in \autoref{app_diagrams:C_alpha_Lp_for_different_alpha_p}.
\begin{figure}[H]
	\centering
	\includegraphics[width = \textwidth]{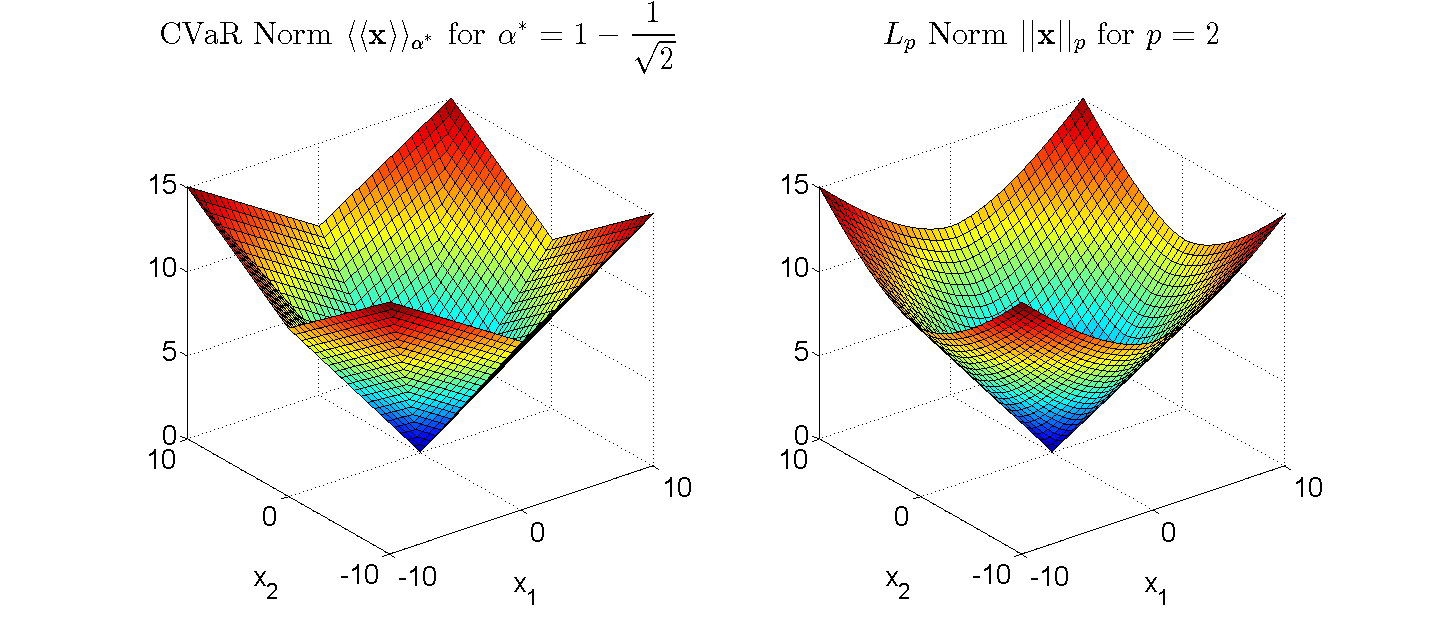}
	\caption{Norm surface plots ($C_{\alpha}$ and $L_p$) of $\mathbf{x}$ for $p = 2$ and $\alpha^* = \frac{1}{1 - \sqrt{2}}$.}
	\label{fig:CVaR_Lp_Norm_surface_plots}
\end{figure}

Comparing the projections of a circle $C = \{ \xinR{3} : x_1^2 + x_2^2 = 1, x_3 = 1 \}$ onto the unit ball $U = \{ \xinR{3} : \mathbf{x}^T \mathbf{x} = 1 \}$ using the $L_2$ norm and $C_{\alpha^*}$ norm, with $\alpha^* = 1 - \frac{1}{\sqrt{3}}$ is shown in \autoref{fig:projection_onto_unit_ball}. Further comparisons for different $\alpha$ are shown in \autoref{app_diagrams:Projection_onto_unit_ball}.
\begin{figure}[H]
	\centering
	\includegraphics[width = \textwidth]{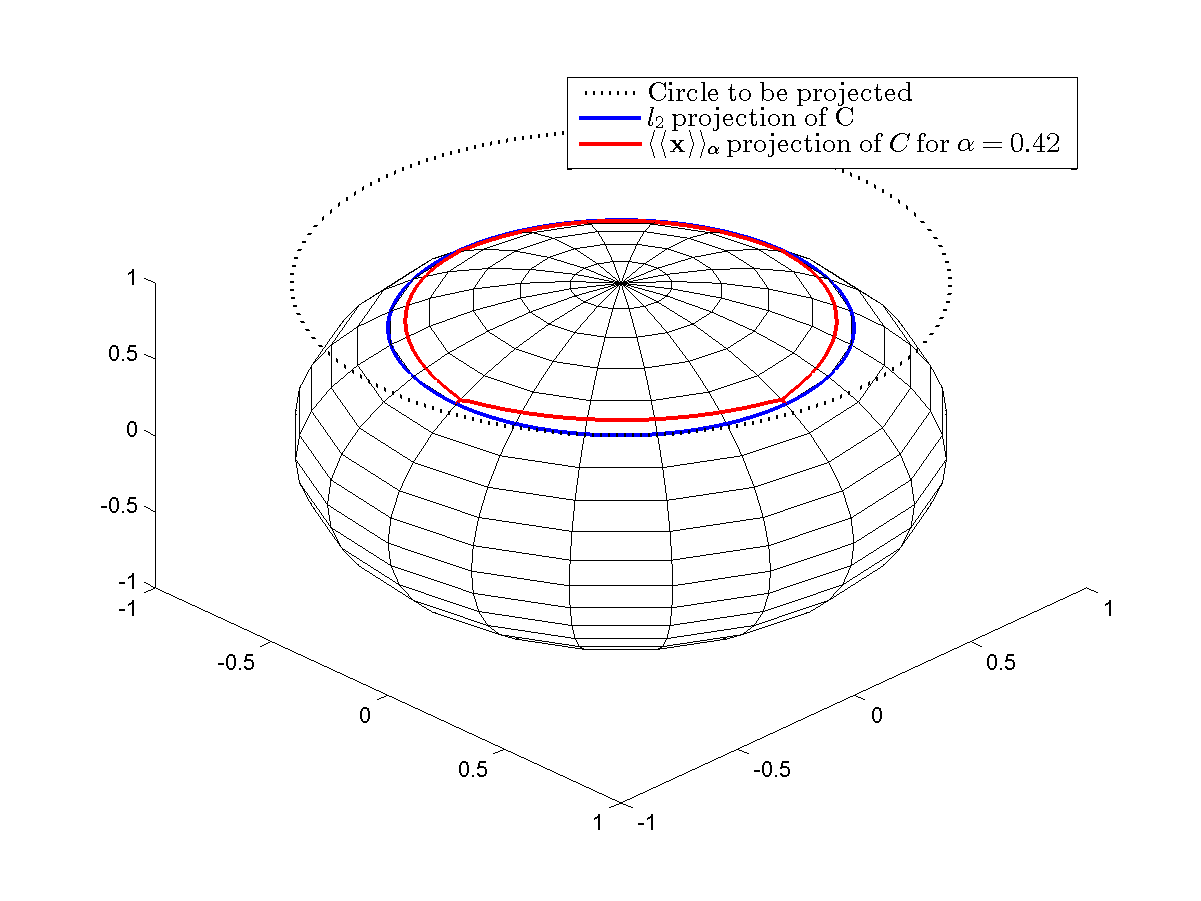}
	\caption[Projection of a circle onto the unit ball using different norms.]{Projection of a circle onto the unit ball in $\xinR{3}$ using $L_2$ and $C_{\alpha^*}$ norm, with $\alpha^* = 1 - \frac{1}{\sqrt{3}}$.}
	\label{fig:projection_onto_unit_ball}
\end{figure}

\clearpage

%
%

\chapter{Model Recovery Using Atomic Norms} \label{chapter:Model_Recovery_using_Atomic_Norms}

Many real world problems require solving an ill-posed inverse problem, in which the number of measurements is smaller than the dimension of the model to be estimated. But if the structure of the model is favourable, the original model can be recovered by the use of atomic norms, to be more precise, by minimizing the atomic norm, i.e. solving the problem \cite[p. 811]{Parillo2010_Convex_Geometry}
\begin{Problem}[prog:exact_recovery_primal]
\left.
\begin{array}{rll}
	\hat{\mathbf{x}} = \arg \min \limits_{\mathbf{x}} & \LnormX{\mathcal{A}} \\
	\text{s.t.} & \mathbf{y} = \mathbf{\Phi} \mathbf{x}
\end{array}
\right\} \,,
\end{Problem}
where $\LnormDot{\mathcal{A}}$ is the atomic norm. The candidate vector $\mathbf{x}^*$ can be formed from a set of atoms $\mathcal{A}$, i.e.  $\mathbf{x}^* = \sum_{i=1}^k c_i \mathbf{a}_i \text{~where~} \mathbf{a}_i \in \mathcal{A}, c_i \geq 0$ and information about a linear mapping $\mathbf{\Phi} : \mathbb{R}^p \rightarrow \mathbb{R}^n$ is available. Also, the measurement $\mathbf{y} = \mathbf{\Phi} \mathbf{x}^*$ is known. The goal is to reconstruct $\mathbf{x}^*$ given $\mathbf{y}$.

The following sections will discuss how atomic norms can be derived from a set of atoms and which conditions need to be satisfied to allow for recovery.


\section{Background on Atomic Norms and Convex Geometry} \label{sec:MRuAN-Background}

A model can be considered \emph{simple} if it can be expressed as a non-negative combination of atoms (i.e. basic building blocks of the model). More precisely, let $\xinR{p}$ be formed as \cite[p. 806]{Parillo2010_Convex_Geometry}
\begin{align}
	\mathbf{x} &= \sum \limits_{i=1}^k c_i \mathbf{a}_i \,, \label{eqn:x_as_sum_of_atoms}
\end{align}
for $\mathbf{a}_i \in \mathcal{A}, c_i \geq 0$, where $\mathcal{A}$ is the set of atoms.

The atomic norm of a set of atoms $\mathcal{A}$ is then derived by forming the convex hull of $\mathcal{A}$, i.e $ \text{conv}(\mathcal{A})$. \autoref{fig:Atomic_Norms_Explanation} displays the relation between different sets of atoms and their corresponding atomic norms in $\mathbb{R}^2$.
\begin{figure}[H]
	\centering
	\includegraphics[width = \textwidth]{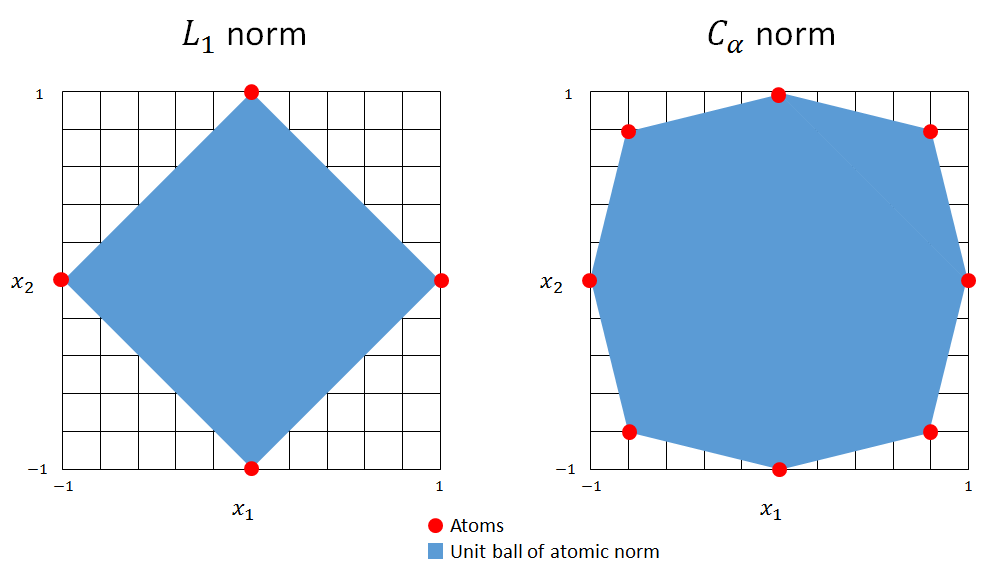}
	\caption{Atoms, their convex hull, and relation to the $L_1$ and $C_{\alpha}$ norms in $\mathbb{R}^2$.}
	\label{fig:Atomic_Norms_Explanation}
\end{figure}

Choosing the atoms as the unit vectors of $\mathbb{R}^2$ and forming the convex hull gives the unit ball of the $L_1$ norm. Hence, for $\mathcal{A}_{L_1} = \{ \pm \mathbf{e}_i \}_{i=1}^2$, the atomic norm is the $L_1$ norm (see left side of \autoref{fig:Atomic_Norms_Explanation}). If we extend then set of atoms to also include the points $\frac{1}{2 (1 - \alpha)} [ \pm 1, \pm 1 ]^T$, for $0 < \alpha < \frac{1}{2}$, i.e. $\mathcal{A}_1 = \{ \pm \mathbf{e}_i \}_{i=1}^2 \cup \frac{1}{2 (1 - \alpha)} [ \pm 1, \pm 1 ]^T \,, 0 < \alpha < \frac{1}{2}$, then the atomic norm of $\mathcal{A}_1$ is the $C_{\alpha}$ norm in $\mathbb{R}^2$, with $0 < \alpha < \frac{1}{2}$ (see right side of \autoref{fig:Atomic_Norms_Explanation} and \autoref{conj:CVaR_Norm_atoms1}). \\

A formal relation between $\text{conv} (\mathcal{A})$ and the atomic norm induced by $\mathcal{A}$ can be derived from different results of convex analysis:

\begin{Definition}[{\cite[p. 128]{2001Hiriart_Fundamentals_of_Convex_Analysis} Gauge of a set}] \label{def:gauge}
	Let $\mathcal{A}$ be a closed convex set containing the origin. The function defined by
	\begin{equation} \label{eqn:gauge}
		\gamma_{\mathcal{A}} (x) \defeq \inf \{ \lambda > 0 : x \in \lambda \emph{~conv}(\mathcal{A}) \}
	\end{equation}
	is called the \emph{gauge} of $\mathcal{A}$. If $\not \exists \lambda : x \in \lambda \emph{~conv}(\mathcal{A})$, then $\gamma_{\mathcal{A}} (x) = + \infty$.
\end{Definition}

\begin{Proposition}[{\cite[p. 10]{Bonsall1991_General_atomic_decomposition_theorem}}] \label{prop:gauge_alternative_definition}
	Assume that the centroid of $\emph{~conv}(\mathcal{A})$ is at the origin, which can be achieved by appropriate recentering. Then the gauge function can be rewritten as
	\begin{equation} \label{eqn:gauge_alternative_definition}
		\gamma_{\mathcal{A}} (x) = \inf \left\{ \sum \limits_{\mathbf{a} \in \mathcal{A}} c_{\mathbf{a}} : \mathbf{x} = \sum \limits_{\mathbf{a} \in \mathcal{A}} c_{\mathbf{a}} \mathbf{a} \,, \text{~~~~} c_{\mathbf{a}} \geq 0 \forall \mathbf{a} \in \mathcal{A} \right\} .
	\end{equation}
\end{Proposition}
Furthermore, if $\mathcal{A}$ is centrally symmetric about the origin (i.e. $\mathbf{a} \in \mathcal{A}$ if and only if $- \mathbf{a} \in \mathcal{A}$), then the gauge $\gamma_{\mathcal{A}}$ is a norm, which is called the \emph{atomic norm} induced by $\mathcal{A}$ \cite[p. 810]{Parillo2010_Convex_Geometry}. In this case, it will be denoted by $\LnormDot{\mathcal{A}}$. The support function of $\mathcal{A}$ is given below.
\begin{Definition}[{\cite[p. 134]{2001Hiriart_Fundamentals_of_Convex_Analysis}, \cite[p. 810]{Parillo2010_Convex_Geometry} Support Function}] \label{def:support_function}
	Let $\mathcal{A}$ be a non-empty set in $\mathbb{R}^n$. The function defined by
	\begin{equation} \label{eqn:support_function}
		\LnormX[*]{\mathcal{A}} \defeq \sup \left\{ \langle \mathbf{x}, \mathbf{a} \rangle : \mathbf{a} \in \mathcal{A} \right\}
	\end{equation}
	is called the \emph{support function} of $\mathcal{A}$.\footnote{$\langle \mathbf{x}, \mathbf{a} \rangle$  denotes the dot-product $\mathbf{x}^T \mathbf{a}$.}
\end{Definition}
If $\LnormDot{\mathcal{A}}$ is a norm, the support function $\LnormDot[*]{\mathcal{A}}$ is the dual norm of the atomic norm. This definition shows that the unit ball of $\LnormDot{\mathcal{A}}$ is equal to conv$(\mathcal{A})$ \cite[p. 810]{Parillo2010_Convex_Geometry}.\\

In addition to the above concepts, some background on cones is also necessary for the following sections:
\begin{Definition}[{\cite[p. 21]{2001Hiriart_Fundamentals_of_Convex_Analysis} Convex Cone}] \label{def:convex_cone}
	The set $K$ is a cone if $\forall t > 0\,, \mathbf{k} \in K \Rightarrow t \mathbf{k} \in K$. Furthermore, the cone is convex if the set $K$ is convex.
\end{Definition}
\begin{Definition}[{\cite[p. 814]{Parillo2010_Convex_Geometry} Polar Cone}] \label{def:polar_cone}
	The polar $K^*$ of a cone $K$ is the cone
	\begin{equation} \label{eqn:polar_cone}
		K^* \defeq \left\{ \xinR{p} \,\,\, : \,\,\, \langle \mathbf{x}, \mathbf{k} \rangle \leq 0 \,\,\, \forall \,\,\, \mathbf{k} \in K \right\} .
	\end{equation}
\end{Definition}
To provide a better understand of cones and polar cones, examples (taken from \cite[p. 35]{2011Acary_Nonsmooth_Modelling}) are shown in \autoref{fig:convex_cones_examples}.
\begin{figure}[H]
	\centering
	\includegraphics[width = 0.3\textwidth]{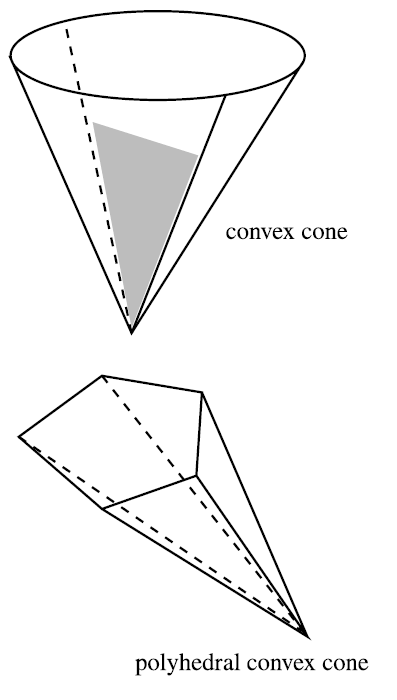}
	\includegraphics[width = 0.6\textwidth]{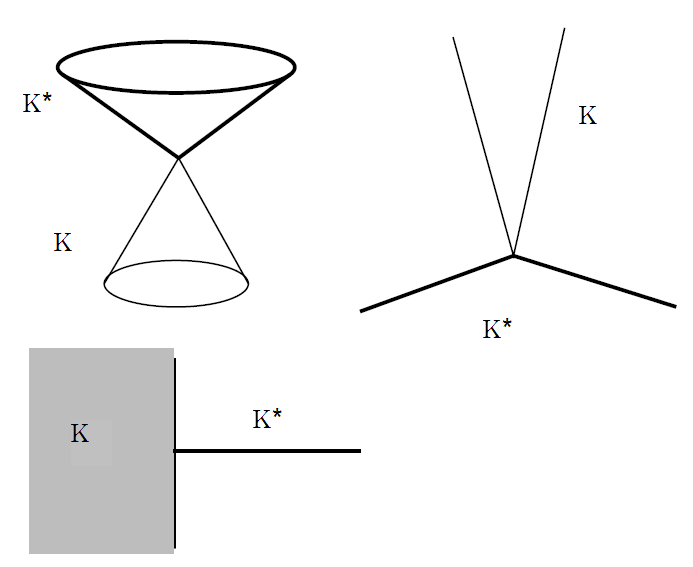}	
	\caption{\cite[p. 35]{2011Acary_Nonsmooth_Modelling} Examples of cones $K$ and polar cones $K^*$.}
	\label{fig:convex_cones_examples}
\end{figure}

\begin{Definition}[{\cite[p. 814]{Parillo2010_Convex_Geometry} Tangent Cone}] \label{def:tangent_cone}
	For some non-zero $\xinR{p}$, the \emph{tangent cone} at $\mathbf{x}$ with respect to the scaled unit ball $\LnormX{\mathcal{A}}$\emph{conv}$(\mathcal{A})$ is
	\begin{equation} \label{eqn:tangent_cone}
		T_{\mathcal{A}} (\mathbf{x}) \defeq \text{cone}\left\{ \mathbf{z} - \mathbf{x} \,\,\, : \,\,\, \vert \vert \mathbf{z} \vert \vert_{\mathcal{A}} \leq \LnormX{\mathcal{A}} \right\} .
	\end{equation}
\end{Definition}

\begin{Definition}[{\cite[p. 814]{Parillo2010_Convex_Geometry} Normal Cone}] \label{def:normal_cone}
	The \emph{normal cone} $N_{\mathcal{A}} (\mathbf{x})$ at $\mathbf{x}$ with respect to the scaled unit ball $\LnormX{\mathcal{A}}$\emph{conv}$(\mathcal{A})$ is the set of all directions that form obtuse angles with every descent direction of the atomic norm $\LnormDot{\mathcal{A}}$ at the point $\mathbf{x}$, i.e.
	\begin{equation} \label{eqn:normal_cone}
		N_{\mathcal{A}} (\mathbf{x}) \defeq \left\{ \mathbf{s} : \langle \mathbf{s}, \mathbf{z} - \mathbf{x} \rangle \leq 0 \,\,\, \forall \,\,\, \mathbf{z} \text{~s.t.~} \vert \vert \mathbf{z} \vert \vert_{\mathcal{A}} \leq \LnormX{\mathcal{A}} \right\} .
	\end{equation}
\end{Definition}
Examples of tangent and normal cones for a general convex set $C$ (again taken from \cite[p. 49]{2011Acary_Nonsmooth_Modelling}) are shown in \autoref{fig:tangent_normal_cones_examples} to provide a better understanding of these concepts.
\begin{figure}[H]
	\centering
	\includegraphics[width = 0.9\textwidth]{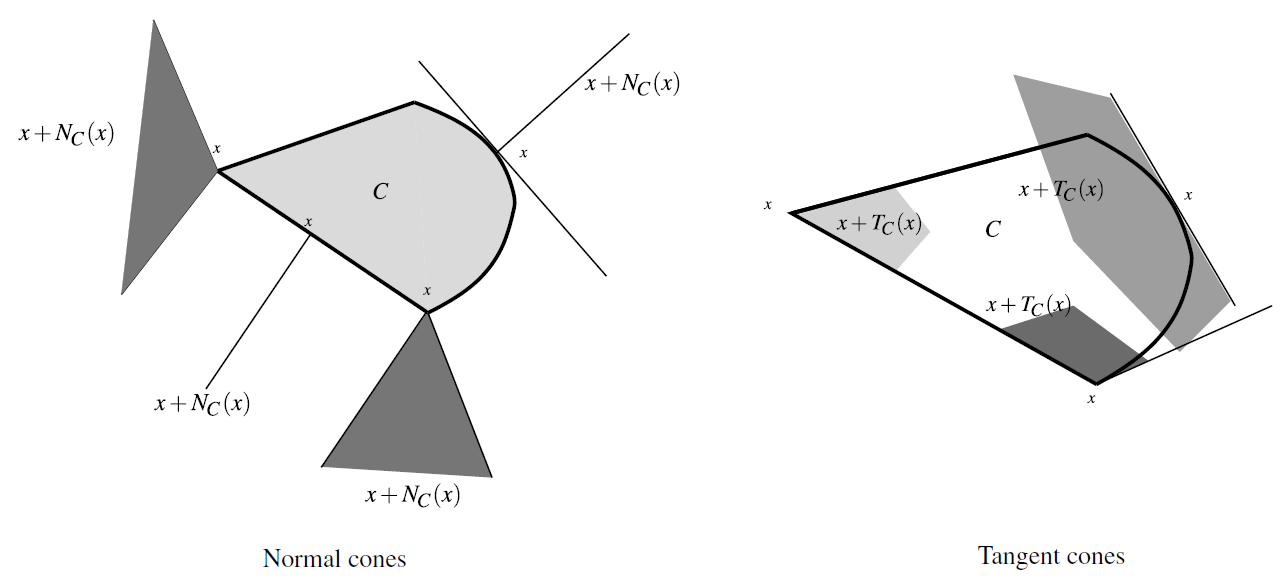}	
	\caption{\cite[p. 49]{2011Acary_Nonsmooth_Modelling} Examples of tangent and normal cones with respect to a set $C$.}
	\label{fig:tangent_normal_cones_examples}
\end{figure}

The tangent cone is equal to the set of descent directions of the atomic norm $\LnormDot{\mathcal{A}}$ at point $\mathbf{x}$, i.e. the set of all directions $\mathbf{d}$ such that the directional derivative is negative \cite[p. 814]{Parillo2010_Convex_Geometry}.

The normal cone is equal to the set of all normals of hyperplanes given by normal vectors $\mathbf{s}$ that support the scaled unit ball $\LnormX{\mathcal{A}}$conv$(\mathcal{A})$ at $\mathbf{x}$. Additionally, the tangent cone $T_{\mathcal{A}} (x)$ and normal cone $N_{\mathcal{A}} (x)$ are polar cones of each other. And finally, the normal cone $N_{\mathcal{A}} (x)$ is the conic hull of the subdifferential of the atomic norm at $\mathbf{x}$ \cite[p. 814]{Parillo2010_Convex_Geometry}.


\section{Recovery Conditions} \label{sec:MRuAN-Recovery_Conditions}

This section states the conditions that are necessary to recover a vector $\hat{\mathbf{x}}$ exactly (when the measurements $\mathbf{y} \in \mathbb{R}^n$ are noise free) or robustly (when the measurements are noisy). The concepts presented in \autoref{sec:MRuAN-Background} are used to derive the number of measurements $n$ needed to ensure exact (or robust) recovery. \\

Recall \autoref{prog:exact_recovery_primal}, which states
\begin{equation} 
\begin{array}{rll}
	\hat{\mathbf{x}} = \arg \min \limits_{\mathbf{x}} & \LnormX{\mathcal{A}} \\
	\text{s.t.} & \mathbf{y} = \mathbf{\Phi} \mathbf{x}
\end{array} \,.
\notag
\end{equation}

The dual problem of \ref{prog:exact_recovery_primal} is \cite[p. 811]{Parillo2010_Convex_Geometry}
\begin{Problem}[prog:exact_recovery_dual]
\left.
\begin{array}{rll}
	\max \limits_{\mathbf{z}} & \mathbf{y}^T \mathbf{z} \\
	\text{s.t.} & \vert \vert \mathbf{\Phi}^T \mathbf{z} \vert \vert \leq 1
\end{array}
\right\} \,.
\end{Problem}

Now suppose that the measurements $\mathbf{y}$ are noisy, i.e. $\mathbf{y}$ is formed as  $\mathbf{y} = \mathbf{\Phi} \mathbf{x}^* + \omega$, where $\omega$ is the noise term. If an upper bound on the noise term is known, i.e. $\vert \vert \omega \vert \vert \leq \delta$, the constraint in \autoref{prog:exact_recovery_primal} can be relaxed to give \cite[p. 811]{Parillo2010_Convex_Geometry}
\begin{Problem}[prog:robust_recovery_primal]
\left.
\begin{array}{rll}
	\hat{\mathbf{x}} = \arg \min \limits_{\mathbf{x}} & \LnormX{\mathcal{A}} \\
	\text{s.t.} & \vert \vert \mathbf{y} - \mathbf{\Phi} \mathbf{x} \vert \vert \leq \delta
\end{array}
\right\} \,.
\end{Problem}

In the noise free case, the solution to \autoref{prog:exact_recovery_primal} ($\hat{\mathbf{x}}$) is considered an \emph{exact recovery} so that $\hat{\mathbf{x}} = \mathbf{x}^*$. If the error $\vert \vert \hat{\mathbf{x}} - \mathbf{x}^* \vert \vert$ is small in \autoref{prog:robust_recovery_primal} then the recovery is considered \emph{robust}. The conditions for exact and robust recovery will be given below.

Let $\Ker (\mathbf{\Phi})$ denote the kernel or nullspace of the linear mapping $\mathbf{\Phi}$. Then the exact recovery condition is stated in \autoref{prop:exact_recovery_condition} below.
\begin{Proposition}[{\cite[p. 815]{Parillo2010_Convex_Geometry} Exact Recovery Condition}] \label{prop:exact_recovery_condition}
	$\hat{\mathbf{x}} = \mathbf{x}^*$ is the unique optimal solution of \autoref{prog:exact_recovery_primal} if and only if $\Ker (\mathbf{\Phi}) \cap T_{\mathcal{A}} (\mathbf{x}^*) = \{ 0 \}$.
\end{Proposition}

Given that the measurements of $\mathbf{y}$ are noisy, it is possible to give a condition for when $\mathbf{x}^*$ can be well approximated.
\begin{Proposition}[{\cite[p. 815]{Parillo2010_Convex_Geometry} Proximity of Robust Recovery}] \label{prop:robust_recovery_condition}
	Suppose that there are $n$ noisy measurements $\mathbf{y} = \mathbf{\Phi} \mathbf{x}^* + \omega$ where $\vert \vert \omega \vert \vert \leq \delta$ and $\mathbf{\Phi} : \mathbb{R}^p \rightarrow \mathbb{R}^n$. Let $\hat{\mathbf{x}}$ denote an optimal solution of \autoref{prog:robust_recovery_primal}. Further suppose that $\vert \vert \mathbf{\Phi} \mathbf{z} \vert \vert \geq \epsilon \vert \vert \mathbf{z} \vert \vert$ holds for all $\mathbf{z} \in T_{\mathcal{A}} (\mathbf{x}^*)$. Then $\vert \vert \hat{\mathbf{x}} - \mathbf{x}^* \vert \vert \leq \frac{2 \delta}{\epsilon}$.
\end{Proposition}

The proofs of \autoref{prop:exact_recovery_condition} and \autoref{prop:robust_recovery_condition} are given in \cite[p. 815]{Parillo2010_Convex_Geometry}. Hence the smaller the tangent cone at $\mathbf{x}^*$ with respect to conv$(\mathcal{A})$, the easier it is to satisfy the empty intersection condition of \autoref{prop:exact_recovery_condition} and to recover $\hat{\mathbf{x}}$ \cite[p. 816]{Parillo2010_Convex_Geometry}.

By \autoref{prop:exact_recovery_condition}, $\Ker (\mathbf{\Phi})$ must miss $T_{\mathcal{A}} (\mathbf{x}^*)$ for an exact recovery. Gordon (\cite{1998Gordon_Milmans_Inequality}) derived an expression for the probability that a uniformly distributed subspace of fixed dimension misses a cone and his findings form the basis of the analysis of Chandrasekaran et. al (\cite{Parillo2010_Convex_Geometry}). An important part in the analysis is the Gaussian width of a set.
\begin{Definition}[{\cite[p. 817]{Parillo2010_Convex_Geometry} Gaussian Width}] \label{def:gaussian_width}
	The \emph{Gaussian width} of a set $S \in \mathbb{R}^p$ is defined as
	\begin{equation} \label{eqn:gaussian_width}
		w(S) \defeq \E_{\mathbf{g}} \left[ \sup \limits_{\mathbf{z} \in S} \mathbf{g}^T \mathbf{z} \right] \,,
	\end{equation}
	where $\mathbf{g} \sim N(\mathbf{0},\mathbf{I})$ is a vector of independent zero-mean unit-variance Gaussians.
\end{Definition}
Gordon defined the likelihood that a random subspace misses a cone $K$ purely in terms of the dimension of the subspace and the Gaussian width $w(K \cap \mathbb{S}^{p-1})$, where $\mathbb{S}^{p-1} \subset \mathbb{R}^p$ is the unit sphere \cite[p. 817]{Parillo2010_Convex_Geometry}. To introduce the following results, the expected length of a $k$-dimensional Gaussian random vector (denoted $\lambda_k$) is needed. By integration and induction, it can be shown that $\lambda_k$ is tightly bounded as $\frac{k}{\sqrt{k+1}} \leq \lambda_k \leq \sqrt{k}$. With this notation, a bound on these quantities can be given. 
\begin{Theorem}[{\cite[p. 86]{1998Gordon_Milmans_Inequality}}] \label{theorem:recovery_theorem}
	Let $\Omega$ be a closed subset of $\mathbb{S}^{p-1}$ and let $\mathbf{\Phi} : \mathbb{R}^p \rightarrow \mathbb{R}^n$ be a random map with i.i.d. zero-mean Gaussian entries having variance one. Then
	\begin{equation} \label{eqn:recovery_theorem}
		\E \left[ \min \limits_{\mathbf{z} \in \Omega} \vert \vert \mathbf{\Phi} \mathbf{z} \vert \vert_2 \right] \geq \lambda_k - w(\Omega) \,.
	\end{equation}
\end{Theorem}

\autoref{theorem:recovery_theorem} then leads to the required number of measurements to give an exact or robust recovery with a given probability. Specifically, if the measurement map $\mathbf{\Phi} : \mathbb{R}^p \rightarrow \mathbb{R}^n$ consists of i.i.d. zero-mean Gaussian entries having variance $1/n$, then the required number of measurements is given in \autoref{corollary:recovery_measurements_n}, the proof of which is given in \cite[p. 818f.]{Parillo2010_Convex_Geometry}.
\begin{Corollary}[{\cite[p. 818]{Parillo2010_Convex_Geometry}}] \label{corollary:recovery_measurements_n}
	Let $\mathbf{\Phi} : \mathbb{R}^p \rightarrow \mathbb{R}^n$ be a random map with i.i.d. zero-mean Gaussian entries having variance $1/n$. Further let $\Omega =  T_{\mathcal{A}} (\mathbf{x}^*) \cap \mathbb{S}^{p-1}$ denote the spherical part of the tangent cone $T_{\mathcal{A}} (\mathbf{x}^*)$.
	\begin{enumerate}
		\item Suppose that there are measurements $\mathbf{y} = \mathbf{\Phi} \mathbf{x}^*$ to solve \autoref{prog:exact_recovery_primal}. Then $\mathbf{x}^*$ is the unique optimum of \autoref{prog:exact_recovery_primal} with probability at least $1 - \exp \left( - \frac{1}{2} \left[ \lambda_n - w(\Omega) \right]^2 \right)$ provided 
			\begin{equation} \label{eqn:recovery_measurements_n_exact}
				n \geq w(\Omega)^2 + 1 \,.
			\end{equation}
		\item Suppose that there are noisy measurements $\mathbf{y} = \mathbf{\Phi} \mathbf{x}^* + \omega$, with the noise bounded as $\vert \vert \omega \vert \vert \leq \delta$ to solve \autoref{prog:robust_recovery_primal}. Letting $\hat{\mathbf{x}}$ denote the optimal solution of \autoref{prog:robust_recovery_primal}, then $\vert \vert \mathbf{x}^* - \hat{\mathbf{x}} \vert \vert \leq \frac{2 \delta}{\epsilon}$ with probability at least $1 - \exp \left( - \frac{1}{2} \left[ \lambda_n - w(\Omega) - \sqrt{n} \epsilon \right]^2 \right)$ provided 
			\begin{equation} \label{eqn:recovery_measurements_n_robust}
				n \geq \frac{w(\Omega)^2 + 3/2}{(1 - \epsilon)^2} \,.
			\end{equation}		
	\end{enumerate}	
\end{Corollary}

Hence, to apply \autoref{corollary:recovery_measurements_n} for finding $n$ (the number of measurements needed to ensure recovery), one must calculate the Gaussian width of $\Omega =  T_{\mathcal{A}} (\mathbf{x}^*) \cap \mathbb{S}^{p-1}$. However, Gaussian widths are not easy to compute \cite[p. 819]{Parillo2010_Convex_Geometry}. Chandrasekaran et. al stated various well-known properties and derived new properties of Gaussian widths that can be used to calculate bounds on Gaussian widths in a variety of cases \cite[p. 819ff.]{Parillo2010_Convex_Geometry}. The most important of these properties within the scope of this dissertation are reproduced in the next section.


\section{Properties of Gaussian Widths} \label{sec:MRuAN-Gaussian_Widths}

This section states properties of Gaussian widths that might be useful\footnote{As a proof on the bounds of the Gaussian width of $T_{\mathcal{A}} (\mathbf{x}^*) \cap \mathbb{S}^{p-1}$ could not be proven within the scope of this dissertation, the author can only make assumptions on which properties might be useful in a proof.} for calculating the Gaussian width of $T_{\mathcal{A}} (\mathbf{x}^*) \cap \mathbb{S}^{p-1}$, where $\mathcal{A}$ are the atoms of the CVaR Norm.\footnote{For a more extensive list of properties see \cite[p. 819ff.]{Parillo2010_Convex_Geometry}.}

\begin{Proposition}[{\cite[p. 821]{Parillo2010_Convex_Geometry}}] \label{prop:GW_property1}
	Let $K$ be any non-empty convex cone in $\mathbb{R}^p$ and let $\mathbf{g} \sim N(\mathbf{0}, \mathbf{I})$ be a random Gaussian vector. Then
	\begin{equation} \label{eqn:GW_property1}
		w(K \cap \mathbb{S}^{p-1}) \leq \E_{\mathbf{g}} \left[ \text{\emph{dist}}(\mathbf{g},K^*) \right] \,,
	\end{equation}
	where \emph{dist} denotes the Euclidean distance between a point and a set.
\end{Proposition}

Since \autoref{corollary:recovery_measurements_n} requires $w(\Omega)^2$, Jensen's inequality is often useful to apply \autoref{prop:GW_property1} \cite[p. 822]{Parillo2010_Convex_Geometry}. Jensen's inequality states that if $\E [\xi]$ exists for a random variable $\xi$ and if $f(x)$ is a convex function, then \cite[p. 88]{2013Borovkov_Probability_Theory}
\begin{equation}
	f \left( \E [\xi] \right) \leq \E \left[ f(\xi) \right] \,. \notag
\end{equation}
Because $\mathbf{g}$ is a random vector, $\text{dist}(\mathbf{g},K^*)$ is a random variable. Also, $f(x) = x^2$ is a convex function. Hence, \cite[p. 822]{Parillo2010_Convex_Geometry}
\begin{equation} \label{eqn:jensens_inequality}
	\E_{\mathbf{g}} \left[ \text{dist}(\mathbf{g},K^*) \right]^2 \leq \E_{\mathbf{g}} \left[ \text{dist}(\mathbf{g},K^*)^2 \right]\,.
\end{equation}

By combining \autoref{eqn:GW_property1} and \autoref{eqn:jensens_inequality}, Chandrasekaran et. al derived the lemma below.
\begin{Lemma}[{\cite[p. 822]{Parillo2010_Convex_Geometry}}] \label{lemma:GW_property2}
	Let $K$ be any non-empty convex cone in $\mathbb{R}^p$. Then
	\begin{equation} \label{eqn:GW_property2}
		w(K \cap \mathbb{S}^{p-1})^2 + w(K^* \cap \mathbb{S}^{p-1})^2 \leq p \,.
	\end{equation}
\end{Lemma}

%
%

\chapter{Model Recovery Using the CVaR Norm} \label{chapter:Recovery_using_CVaR}

To use the CVaR norm for model recovery in the framework presented by Chandrasekaran et. al, some fundamental properties of the CVaR norm need to be derived. To recover $\hat{\mathbf{x}}$, the set of atoms $\mathcal{A}$ of the CVaR norm needs to be determined and a bound on the Gaussian width of the intersection of $T_{\mathcal{A}} (\hat{\mathbf{x}})$ with the unit sphere $\mathbb{S}^{p-1}$ needs to be established. The bound on the Gaussian width is needed to determine how many measurements $n$ are required to ensure recovery with a high probability.

To the best knowledge of the author, no research with this particular focus has been published. Hence, all results in this chapter are original. Unfortunately, due to limited scope of this thesis, only partial results are available. This being said, the following thoughts can be the basis for further research in this area.


\section{Atomic CVaR Norm} \label{section:RuCVaR-Atomic_CVaR_Norm}

In this section, the atoms of the CVaR norm $C_{\alpha}$ for $\al{p-2} < \alpha < \al{p-1}$ will be conjectured (the set of atoms will be called $\mathcal{A}_{p-1}$, see \autoref{subsec:RuCVaR-ACN-Atom_formulation}). It will be proposed and proven that $\mathcal{A}_{p-1}$ is a subset of the extreme points of the unit ball of $C_{\alpha}$ for $\al{p-2} < \alpha < \al{p-1}$, but due to the limited time of this thesis it cannot be proven that $\mathcal{A}_{p-1}$ is the exhaustive set of extreme points. It will also be shown in \autoref{subsec:RuCVaR-ACN-Similarity} that a subset of the extreme points of the unit ball of $C_{\alpha}$ for $\al{0} < \alpha < \al{1}$ (called $\mathcal{A}_1$) is similar to $\mathcal{A}_{p-1}$. But since some of the points of $\mathcal{A}_1$ are different,  the unit ball of $C_{\alpha}$ for $\al{0} < \alpha < \al{1}$ looks different (the respective unit balls of $C_{\alpha}$ in $\mathbb{R}^3$ are shown in \autoref{fig:Calpha_unit_balls_R3}). Finally, an experiment will be performed to numerically determine the extreme point of the unit ball of $C_{\alpha}$ for $\al{p-2} < \alpha < \al{p-1}$ in $\mathbb{R}^4$ and shown that the set of these extreme points is equal to $\mathcal{A}_{p-1}$.

\subsection{Formulation of the Atoms of the CVaR Norm} \label{subsec:RuCVaR-ACN-Atom_formulation}

The atoms of the CVaR norm for $C_{\alpha}$ for $\al{p-2} < \alpha < \al{p-1}$ are conjectured below.

\begin{Conjecture} \label{conj:CVaR_Norm_atoms1}
	Suppose that $\xinR{p}$ and $\al{p-2} < \alpha < \al{p-1}$ ,i.e., $\frac{p-2}{p} < \alpha < \frac{p-1}{p}$, and let the set of atoms $\mathcal{A}_{p-1}$ be such that
		\begin{align*}
			\mathcal{A}_{p-1} \defeq \{ \pm \mathbf{e}_i \}_{i=1}^p \cup \left\{ \frac{1}{p(1 - \alpha)} \mathbf{b} \right\} \,,
		\end{align*}
	where $\mathbf{e}_i$ is the unit vector with 1 as the $i$th component and 0 zeros elsewhere and $\{ \mathbf{b} \}$ is the set of all vectors in $\mathbb{R}^p$ that have either +1 or -1 as their components. Then the atomic norm induced by $\mathcal{A}_{p-1}$ is equivalent to the CVaR norm $\CVaRnormX{\alpha}$ for $\frac{p-2}{p} < \alpha < \frac{p-1}{p}$.
\end{Conjecture}

\begin{Proposition} \label{prop:CVaR_Norm_atoms1}
	The set $\mathcal{A}_{p-1}$ defined in \autoref{conj:CVaR_Norm_atoms1} is a subset of extreme points of the unit ball of $C_{\alpha}$ for $\al{p-2} < \alpha < \al{p-1}$ ,i.e., $\frac{p-2}{p} < \alpha < \frac{p-1}{p}$.
\end{Proposition}

\begin{proof}
	To prove \autoref{prop:CVaR_Norm_atoms1}, it needs to be shown that the points $\mathcal{A}_{p-1}$ lie on the unit ball of $\CVaRnormX{\alpha}$ for $\frac{p-2}{p} < \alpha < \frac{p-1}{p}$. To show this, an explicit expression for $\CVaRnormX{\alpha}$ will be derived first. By \autoref{eqn:CVaR_Component_Wise_2} and \autoref{eqn:CVaR_Component_Wise_1},
	\begin{align}
		\CVaRnormX{\alpha} =& \lambda \CVaRnormX{\alpha_{p-2}} + (1 - \lambda) \CVaRnormX{\alpha_{p-1}} \notag \\
		=& \lambda \sum \limits_{i = p-1}^{p} \vert x_{(i)} \vert + (1 - \lambda) \vert x_{(p)} \vert \notag \\
		=&  \vert x_{(p)} \vert + \left[ p (1 - \alpha) - 1 \right] \vert x_{(p-1)} \vert \,, \label{eqn:atomic_cvar_norm_explicit}
	\end{align}
	where $\vert x_{(p)} \vert$ is the largest of the absolute values of the components of $\mathbf{x}$ and $\vert x_{(p-1)} \vert$ is the second largest.
	
	Now, there are two types of vectors in $\mathcal{A}$, the unit vectors $\pm \mathbf{e}_i$ and the scaled $\mathbf{b}$ vectors. For both these types of vectors
	\begin{align*}
		\llangle \pm \mathbf{e}_i \rrangle_{\alpha} =& 1 + [p(1-\alpha) - 1] \times 0 = 1 \text{~, and}\\
		\left\llangle \frac{1}{p(1 - \alpha)} \mathbf{b} \right\rrangle_{\alpha} =& \frac{1}{p(1 - \alpha)} \left( 1 + [p(1-\alpha) - 1] \times 1 \right) = 1 \,.
	\end{align*}
	Hence all points in $\mathcal{A}_{p-1}$ lie on the unit ball of $C_{\alpha}$ for $\frac{p-2}{p} < \alpha < \frac{p-1}{p}$.
\end{proof}

\subsection{Similarity of Atoms for Two Different $\alpha$} \label{subsec:RuCVaR-ACN-Similarity}
 
Let the set of points $\mathcal{A}_{1} = \{ \pm \mathbf{e}_i \}_{i=1}^p \cup \left\{ \frac{1}{p(1 - \alpha)} \mathbf{b} \right\}$, with $0 < \alpha < \frac{1}{p}$. Then the points in $\mathcal{A}_1$ lie on the unit ball of $C_{\alpha}$ for $0 < \alpha < \frac{1}{p}$\footnote{Just as for $\mathcal{A}_{p-1}$, this is a conjecture that has yet to be proven.} and there is a close connection between $\mathcal{A}_1$ and $\mathcal{A}_{p-1}$. To show this, consider the explicit expression for $\CVaRnormX{\alpha}$, for $0 < \alpha < \frac{1}{p}$, which is $\CVaRnormX{\alpha} = \sum_{i=1}^p \vert x_{(i)} \vert - p \alpha \vert x_{(1)} \vert$. Then
\begin{align*}
	\llangle \pm \mathbf{e}_i \rrangle_{\alpha} =& 1 - p \alpha \times 0 = 1 \text{~, and}\\
	\left\llangle \frac{1}{p(1 - \alpha)} \mathbf{b} \right\rrangle_{\alpha} =& \frac{p}{p(1 - \alpha)} - \frac{p \alpha}{p(1 - \alpha)} = 1 \,.
\end{align*}

Hence, both sets contain the unit vectors $\pm \mathbf{e}_i$ and the scaled binary vectors $\frac{1}{p(1 - \alpha)} \mathbf{b}$. However, the scaling factor is different for the sets whenever $p>2$, as for $\mathcal{A}_{p-1}$, $\frac{p-2}{p} < \alpha < \frac{p-1}{p}$, and for $\mathcal{A}_{1}$, $0 < \alpha < \frac{1}{p}$. To show that the unit balls look different for these two $\alpha$, consider $\mathbf{x}_1 = \frac{1}{p(1 - \alpha)} [1, 1, \dots, 1]^T$ and $\mathbf{x}_2 = \frac{1}{p(1 - \alpha)} [1, 1, \dots, -1, \dots, 1]^T$, i.e., $\mathbf{x}_1 \in \mathbb{R}^p$ consists of all ones and $\mathbf{x}_2 \in \mathbb{R}^p$ consists of all ones except a $-1$ as the $i$th component, both scaled by $\frac{1}{p(1 - \alpha)}$. Then the vectors $\mathbf{y} = \frac{1}{2} \mathbf{x}_1 +  \frac{1}{2} \mathbf{x}_2 = \frac{1}{p(1 - \alpha)} [1, 1, \dots, 0, \dots, 1]^T$, $\mathbf{x}_1$, and $\mathbf{x}_2$, together with $0 < \alpha_1 < \frac{1}{p}$ and $\frac{p-2}{p} < \alpha_2 < \frac{p-1}{p}$ have the norms
\begin{align*}
	\llangle \mathbf{x}_1 \rrangle_{\alpha} =& 1 \,, &\text{~for~} \alpha = \alpha_1, \alpha_2 \,, \\
	\llangle \mathbf{x}_2 \rrangle_{\alpha} =& 1 \,, &\text{~for~} \alpha = \alpha_1, \alpha_2 \,, \\
	\llangle \mathbf{y} \rrangle_{\alpha_1} =& \frac{p-1}{p(1 - \alpha_1)} < 1 \,, &\text{~and~} \\
	\llangle \mathbf{y} \rrangle_{\alpha_2} =& 1 \,.
\end{align*}
Hence the point $\mathbf{y}$ lies on an edge of the unit ball of $C_{\alpha}$ for $\frac{p-2}{p} < \alpha < \frac{p-1}{p}$, but lies inside the unit ball of $C_{\alpha}$ for $0 < \alpha < \frac{1}{p}$. This can also be seen from \autoref{fig:Calpha_unit_balls_R3}.

\begin{figure}[H]
	\centering
	\includegraphics[width = 0.6\textwidth]{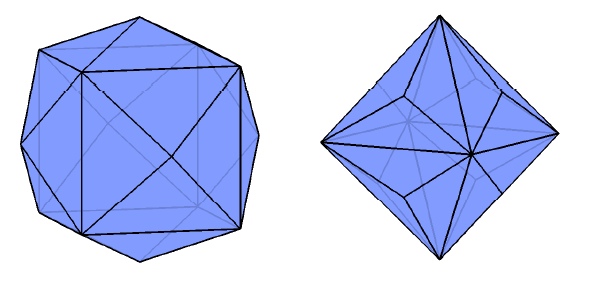}	
	\caption[{\cite[p. 13]{Gotoh2015_Two_Pairs_Of_Polyhedral_Norms_Versus_Lp_Norms} Unit balls of $C_{\alpha}$ in $\mathbb{R}^3$}.]{\cite[p. 13]{Gotoh2015_Two_Pairs_Of_Polyhedral_Norms_Versus_Lp_Norms} Unit balls of $C_{\alpha}$ in $\mathbb{R}^3$ for $\frac{1}{3} < \alpha < \frac{2}{3}$ (left) and $0 < \alpha < \frac{1}{3}$ (right).}
	\label{fig:Calpha_unit_balls_R3}
\end{figure}

\subsection{Numerically Determining $\mathcal{A}_{p-1}$ in $\mathbb{R}^4$} \label{subsec:RuCVaR-ACN-Numerical Experiment}

In this subsection, the atoms of $C_{\alpha}$ for $\al{p-2} < \alpha < \al{p-1}$ in $\mathbb{R}^4$ are determined in numerical experiments to provide more evidence that \autoref{conj:CVaR_Norm_atoms1} is true. To do this, 5,000 random hyperplanes in $\mathbb{R}^4$ are projected onto the unit ball of the CVaR norm. If the conjecture is true, all hyperplanes should be projected onto one of the points in $\mathcal{A}_{p-1}$.\footnote{The probability that a random hyperplane is projected onto an edge or surface of the unit ball is equal to zero.} Only if there are projections onto other points, \autoref{conj:CVaR_Norm_atoms1} is can be deemed false \cite{Richtarik2015_Discussion}.

To perform this experiment, a random hyperplane is generated by a zero-mean, unit variance Gaussian vector, i.e., the hyperplane satisfies $\mathbf{g}^T \mathbf{x} = 5$, where $\mathbf{g} \in \mathbb{R}^4 \sim N(\mathbf{0},\mathbf{I})$ and $\xinR{4}$.\footnote{The constant 5 is chosen arbitrarily.} The projection of the hyperplane onto the unit ball is given by 
\begin{align*}
	\mathbf{x}^U =& \frac{\arg \min \limits_{\mathbf{x}} \CVaRnormX{\alpha}}{\min \limits_{\mathbf{x}} \CVaRnormX{\alpha}} \,,
\end{align*}
with $\alpha = \frac{5}{8}$ and the constraint $\mathbf{g}^T \mathbf{x} = 5$.

Over the 5,000 trials, the hyperplane was projected onto a unit vector 5.86 \% of the time and onto a scaled binary vector 94.14 \% of the time, while no hyperplane was projected onto another point. The complete results of this experiment are shown in \autoref{app_table:projections ratio}.

This experiments provides evidence that \autoref{conj:CVaR_Norm_atoms1} is true, even though it could not be proven within the scope of this thesis. Repeating this experiment in higher dimensions or over more trials should yield the same results.


\section[Gaussian Width of a Tangent Cone with Respect to the $C_{\alpha}$ Norm]{Gaussian Width of a Tangent Cone with Respect to the Scaled Unit Ball of the $C_{\alpha}$ Norm} \label{section:RuCVaR-Gaussian_Width}

To find a bound on the measurements $n$ needed to recover $\hat{\mathbf{x}}$ using \autoref{prog:exact_recovery_primal} (for exact recovery) or \autoref{prog:robust_recovery_primal} (for robust recovery) with the CVaR norm, an expression for the tangent cone or the normal cone of a vector $\mathbf{x}^*$ with respect to $\mathcal{A}_{p-1}$ needs to be found. The derivation of expressions for these cones is beyond the scope of this thesis and could be an area for further research. Here, only an outline of the bounds will be given, if expressions for $T_{\mathcal{A}_{p-1}} (\mathbf{x}^*)$ or $N_{\mathcal{A}_{p-1}} (\mathbf{x}^*)$ are available. These bounds are derived using the properties described in \autoref{sec:MRuAN-Gaussian_Widths}.

\autoref{corollary:recovery_measurements_n} states that to guarantee recovery with high probability, the number of measurements $n$ needs to satisfy
\begin{align*}
	n \geq& \,\,w \left( T_{\mathcal{A}_{p-1}} (\mathbf{x}^*) \cap \mathbb{S}^{p-1} \right)^2 + 1 &\text{~in the exact case, or} \\
	n \geq& \,\,\frac{w \left( T_{\mathcal{A}_{p-1}} (\mathbf{x}^*) \cap \mathbb{S}^{p-1} \right)^2 + 3/2}{(1 - \epsilon)^2} &\text{~in the robust case.}
\end{align*}

Since the Gaussian width is difficult to calculate directly, the Euclidean distance between a cone and the point given by a random Gaussian vector could be used to provide a bound for $w \left( T_{\mathcal{A}_{p-1}} (\mathbf{x}^*) \cap \mathbb{S}^{p-1} \right)^2$. Using \autoref{eqn:GW_property1} and \autoref{eqn:jensens_inequality} gives 
\begin{align}
	w \left( T_{\mathcal{A}_{p-1}} (\mathbf{x}^*) \cap \mathbb{S}^{p-1} \right)^2 \leq& \,\, \E_{\mathbf{g}} \left[ \text{dist} \left( \mathbf{g}, N_{\mathcal{A}_{p-1}} (\mathbf{x}^*) \right) \right]^2 \notag \\
	\leq& \,\, \E_{\mathbf{g}} \left[ \text{dist} \left( \mathbf{g}, N_{\mathcal{A}_{p-1}} (\mathbf{x}^*) \right)^2 \right] \label{eqn:bound_width_tangent_CVaR_cone}
\end{align}

If an expression for $N_{\mathcal{A}_{p-1}} (\mathbf{x}^*)$ is available, \autoref{eqn:bound_width_tangent_CVaR_cone} could be used to determine the minimum number of measurements $n$ needed to recover $\hat{\mathbf{x}}$ as 
\begin{align*}
	n \geq& \,\, \E_{\mathbf{g}} \left[ \text{dist} \left( \mathbf{g}, N_{\mathcal{A}_{p-1}} (\mathbf{x}^*) \right)^2 \right] + 1 &\text{~in the exact case, or} \\
	n \geq& \,\,\frac{\E_{\mathbf{g}} \left[ \text{dist} \left( \mathbf{g}, N_{\mathcal{A}_{p-1}} (\mathbf{x}^*) \right)^2 \right] + 3/2}{(1 - \epsilon)^2} &\text{~in the robust case,}
\end{align*}
when the square of the Euclidean distance ($\text{dist} \left( \mathbf{g}, N_{\mathcal{A}_{p-1}} (\mathbf{x}^*) \right)^2$) can be calculated or bounded.

However, depending on the actual expressions of the tangent and normal cones, other properties of Gaussian widths (e.g. those stated in \cite[p. 819ff.]{Parillo2010_Convex_Geometry}) could be more useful to derive bounds on $n$.


\section{Numerical Recovery Experiments using the $C_{\alpha}$ Norm} \label{section:RuCVaR-Experiments}

This section explores the recovery probabilities of a vector given $n$ random measurements and using CVaR norm minimization. Since \autoref{section:RuCVaR-Gaussian_Width} could not provide a bound on the required number of measurements to ensure recovery, this section investigates under which circumstances recovery might be likely. However, the results are not promising.\\

For the following investigation, the goal was to recover two vectors in $\mathbb{R}^{100}$. The first vector $\mathbf{x}_1$ consists of 1 atom (either a unit vector or a scaled binary vector). The second vector $\mathbf{x}_2$ consists of 3 atoms, one positive unit vector, one negative unit vector, and one scaled binary vector. In both cases, the recovery probability was estimated by minimizing the CVaR norm of a candidate $\mathbf{x}^*$, with $n \leq 100$ random measurements (so that $\mathbf{\Phi} \in \mathbb{R}^{n \times 100}$ is a random map with i.i.d. zero mean Gaussian entries having variance $1/n$) and $\alpha = 0.985$ (so that $\frac{100-2}{100} < \alpha < \frac{100-1}{100}$).

For each $n$, \autoref{prog:exact_recovery_primal} was solved 50 times, each time with a new random map $\mathbf{\Phi}$. The probability of exact recovery (over the 50 random trials) was drawn versus the number of measurements $n$. This is shown in \autoref{fig:probability_of_recovery_CVaR}.
\begin{figure}[H]
	\centering
	\includegraphics[width = \textwidth]{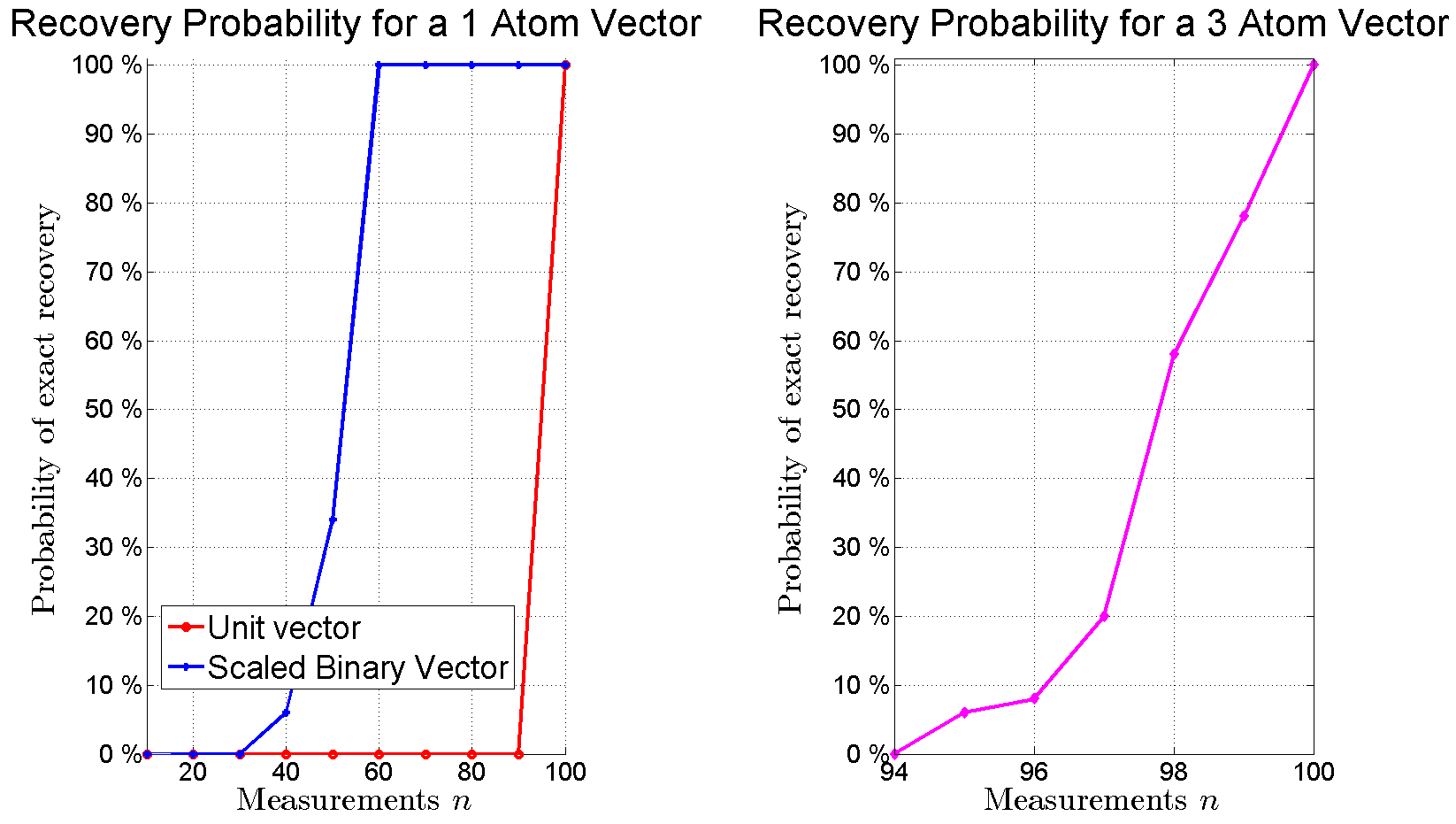}	
	\caption[Probability of exact recovery for a vector $\mathbf{x} \in \mathbb{R}^{100}$ using the CVaR norm as the atomic norm with $n$ measurements.]{Probability of exact recovery for a vector $\mathbf{x} \in \mathbb{R}^{100}$ using the CVaR norm as the atomic norm with $n$ measurements. Left: Recovery probability for $\mathbf{x}_1$ consisting of 1 atom (either a unit vector or a scaled binary vector). Right: Recovery probability for $\mathbf{x}_2$ consisting of 3 atoms.}
	\label{fig:probability_of_recovery_CVaR}
\end{figure}

\autoref{fig:probability_of_recovery_CVaR} shows that if $\mathbf{x}_1$ consists of a unit vector, at least 90 measurements are necessary to ensure recovery, while if $\mathbf{x}_1$ consists of a scaled binary vector, recovery could be ensured with 50-60 measurements. The second vector $\mathbf{x}_2$ could never be recovered for $n < 95$ and even for $n=99$, the recovery probability was just below 80 \%. Hence, it seems that if a vector $\mathbf{x}^*$ which is to be recovered consists of both types of atoms (i.e. unit vectors and scaled binary vectors), exact recovery cannot be guaranteed with high probability when $n < p$. This means that to recover $\mathbf{x}^*$, one would need as many observations as the dimension of the system. The reason for these unfavourable characteristics might be the tangent cone of $\mathbf{x}^*$ with respect to $\mathcal{A}_{p-1}$.\footnote{This assumption can only be confirmed if an expression for $T_{\mathcal{A}_{p-1}} (\mathbf{x}^*)$ can be derived.} \\

If $\mathbf{x}^*$ consists only of one type of atom, i.e., either of unit vectors or scaled binary vectors, the model recovery using the CVaR norm could be compared against the model recovery using the $L_1$ norm or $L_{\infty}$ norm, respectively. Depending on the type of atoms, the $C_{\alpha}$ norm shows two different characteristics when compared to the respective $L_p$ norm. When $\mathbf{x}^*$ is a k-sparse vector\footnote{A k-sparse vector is a vector where k components are not equal to zero.} the norm of choice for model recovery is the $L_1$ norm. By Proposition 3.10 of \cite[p. 823]{Parillo2010_Convex_Geometry}, to recover a k-sparse vector $\mathbf{x}^* \in \mathbb{R}^{100}$ using the $L_1$ norm, $2 \times k \times \ln \left( \frac{100}{k} \right) + \frac{5}{4} \times k + 1$  random Gaussian measurements suffice to recover $\mathbf{x}$ with high probability. Hence, for a 1-sparse vector approximately 12 measurements suffice, while for a 3-sparse vector approximately 26 measurements suffice. At the same time, more than 90 measurements are necessary to recover the same 1-sparse or 3-sparse vector $\mathbf{x}^*$ and same $\mathbf{\Phi}$ to ensure comparability (see \autoref{fig:probability_of_recovery_L1}).
\begin{figure}[H]
	\centering
	\includegraphics[width = 0.9\textwidth]{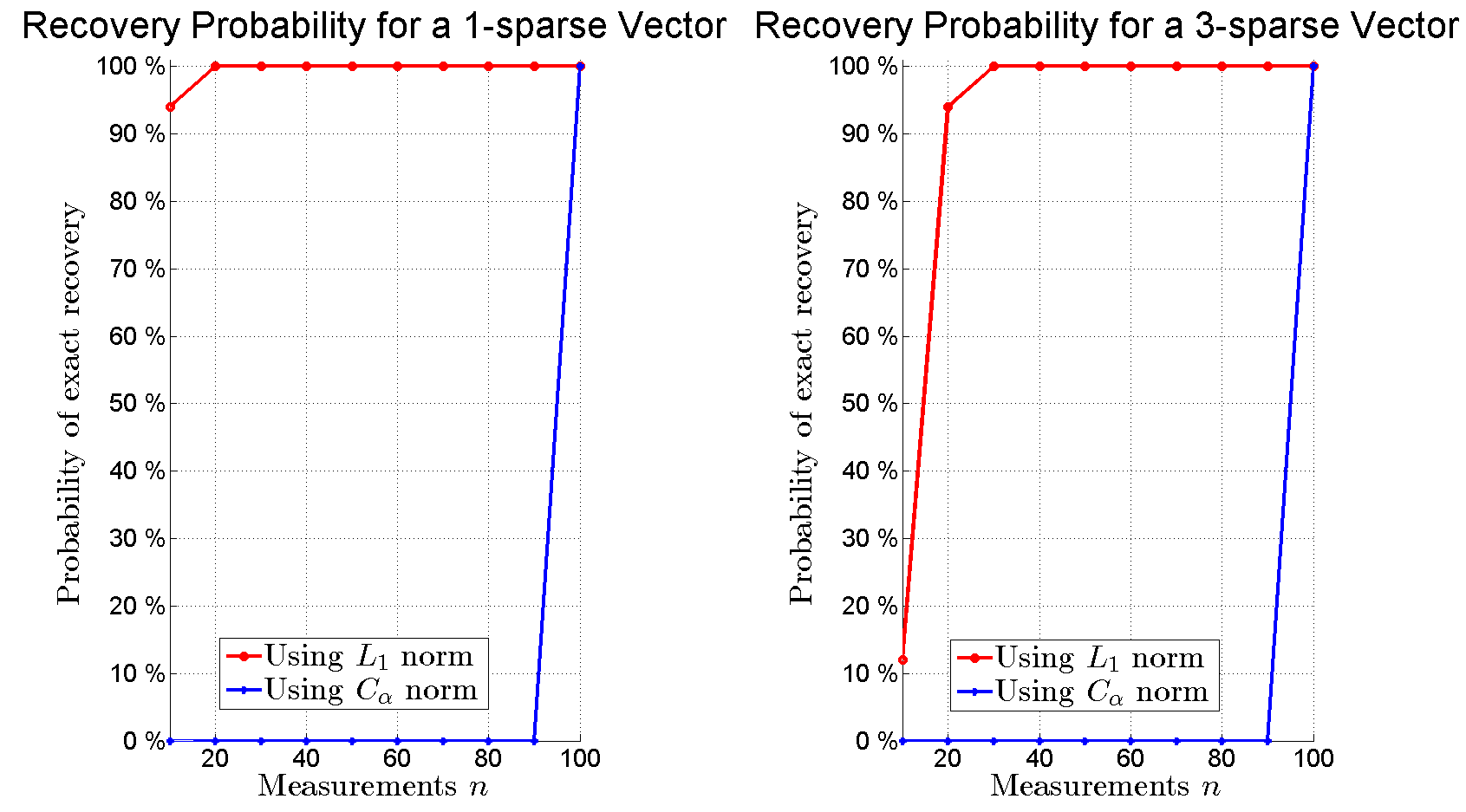}	
	\caption[Probability of exact recovery for a k-sparse vector $\mathbf{x} \in \mathbb{R}^{100}$ using either the $L_1$ norm or $C_{\alpha}$ norm as the atomic norm with $n$ measurements.]{Probability of exact recovery for a k-sparse vector $\mathbf{x} \in \mathbb{R}^{100}$ using the $L_1$ norm or $C_{\alpha}$ norm as the atomic norm with $n$ measurements. Left: Recovery probability for a 1-sparse vector. Right: Recovery probability for 3-sparse vector.}
	\label{fig:probability_of_recovery_L1}
\end{figure}

When $\mathbf{x}^*$ is the sum of $k$ scaled binary vectors the norm of choice for model recovery is the $L_{\infty}$ norm. When trying to recover a vector $\mathbf{x}^*$, that is either 1 scaled binary vector or the sum of 3 scaled binary vectors, the $C_{\alpha}$ norm is as good as the $L_{\infty}$ norm, and sometimes the $C_{\alpha}$ norm is even slightly better. Drawing the probability of exact recovery with the same $\mathbf{x}^*$ to be recovered and the same random measurement maps $\mathbf{\Phi}$ for $40 \leq n \leq 80$ shows that in certain cases the recovery probability of $\mathbf{x}^*$ was higher when using the $C_{\alpha}$ norm (see \autoref{fig:probability_of_recovery_Linfty}).

\begin{figure}[H]
	\centering
	\includegraphics[width = 0.9\textwidth]{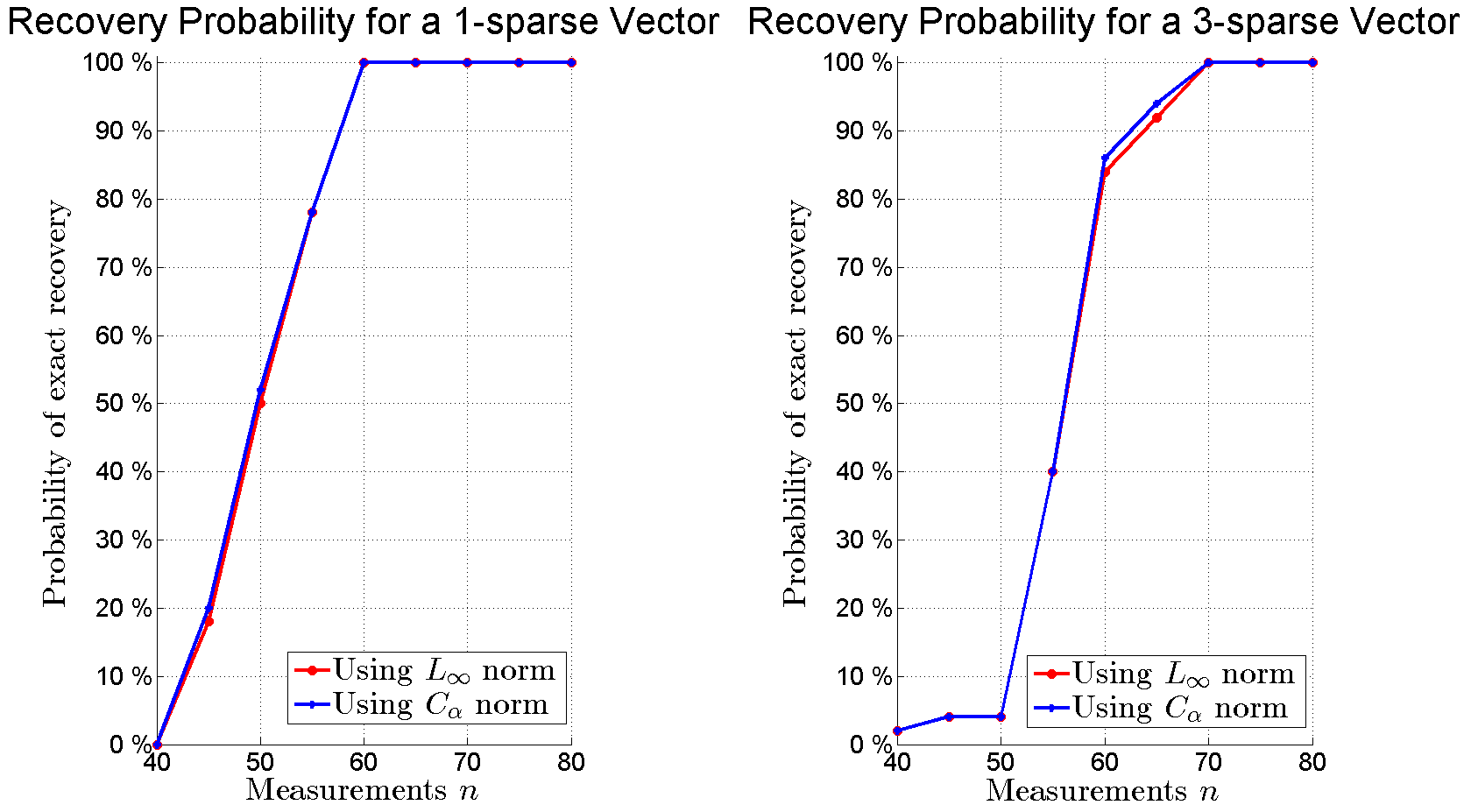}
	\caption[Probability of exact recovery for a vector $\mathbf{x} \in \mathbb{R}^{100}$ using either the $L_{\infty}$ norm or $C_{\alpha}$ norm as the atomic norm with $n$ measurements.]{Probability of exact recovery for a vector $\mathbf{x} \in \mathbb{R}^{100}$ that is the sum of $k$ scaled binary vectors using the $L_{\infty}$ norm or $C_{\alpha}$ norm as the atomic norm with $n$ measurements. Left: Recovery probability for $\mathbf{x}$ as 1 scaled binary vector. Right: Recovery probability for $\mathbf{x}$ as the sum of 3 scaled binary vectors.}
	\label{fig:probability_of_recovery_Linfty}
\end{figure}


\section{Concluding Remarks on Model Recovery Using the CVaR Norm} \label{section:RuCVaR-Conclusion}

Despite the incomplete proofs, this chapter could show some interesting properties of the CVaR norm regarding model recovery. It seems that the CVaR norm is not suitable to define an own type of signal to be recovered (i.e. a signal which consists of the atoms $\mathcal{A}_{p-1}$), but the CVaR norm could be an improvement over the $L_{\infty}$ norm for model recovery.

Since the unit balls of $C_{\alpha}$ differed for different choices of $\alpha$, it was suggested to take $C_{\alpha}$ with $\frac{p-2}{p} < \alpha < \frac{p-1}{p}$ as the atomic norm for recovering a vector $\mathbf{x}^* \in \mathbb{R}^p$. Then the set of atoms $\mathcal{A}_{p-1}$ (see \autoref{conj:CVaR_Norm_atoms1}) can be interpreted as the union of two sets of atoms of better known norms, namely the atoms of the $L_1$ norm and the atoms of the $L_{\infty}$ norm, scaled by $\frac{1}{p (1 - \alpha)}$.\footnote{The proof \autoref{conj:CVaR_Norm_atoms1} still needs to be completed.} The parameter $\alpha$ was chosen in the range $\left( \frac{p-2}{p}, \frac{p-1}{p} \right)$ for these investigations, however, when choosing $0 < \alpha < \frac{1}{p}$, the results might be different. This could be an area for further research.

Unfortunately, a bound on the number of random measurements $n$ could not be established, as it was not possible to derive expressions for the tangent or normal cones with respect to $\mathcal{A}_{p-1}$ in the scope of this thesis. As a remedy, numerical experiments were performed to gain insight into exact recovery probabilities using the CVaR norm.

The numerical experiments in \autoref{section:RuCVaR-Experiments} suggest that it is not possible to recover an arbitrary $\mathbf{x}^*$ with a high probability when $n < p$, i.e. when the number of observations is smaller than the dimension of the model. Hence, it would not make sense to use the CVaR norm for the recovery of a signal consisting of the atoms of $\mathcal{A}_{p-1}$.\footnote{A real world occurrence of this type of signal (or model) could not be identified during this thesis.} It was also shown that the CVaR norm is not suitable to recover a k-sparse vector. However, the CVaR norm showed a slight improvement over the $L_{\infty}$ norm in the experiments, when trying to recover signals $\mathbf{x}^*$ that are formed as the sum of $k$ scaled binary vectors. The reason for this is probably that the tangent cone with respect to $\mathcal{A}_{p-1}$ at $\mathbf{x}^*$ is smaller than the tangent cone with respect to the atoms of the $L_{\infty}$ norm. This would need to be confirmed in further research, as it was not possible to derive an expression for $T_{\mathcal{A}_{p-1}} (\mathbf{x}^*)$ in the scope of this thesis. Also, the practical implications of this need to be considered, as the gains of a smaller tangent cone might be offset by the greater effort to calculate the CVaR norm compared to the $L_{\infty}$ norm.

Again, it should be stressed that the numerical experiments were done by choosing $\alpha$ as $\frac{p-2}{p} < \alpha < \frac{p-1}{p}$. Choosing a different $\alpha$ gives a different unit ball and therefore different characteristics for the model recovery problem. This could all be evaluated in further research.

\clearpage

%
%
\chapter{Conclusion}

This thesis covered a wide range of theory on CVaR, both as a risk measure and a vector norm. It was shown how the CVaR is defined for a univariate loss distribution and how this definition can be extended to define the CVaR of a portfolio of assets, i.e. for multivariate loss distributions. The CVaR concept was then abstracted to define a new family of vector norms in $\mathbb{R}^n$, which were then analysed in detail. In the last part of the thesis, model recovery problems were introduced and it was shown how the new CVaR norm could be used in the context of model recovery problems. \\

\autoref{chapter:CVaR_as_risk_measure} started by introducing Value-at-Risk, and showed how the Conditional Value-at-Risk can be derived from VaR in the case of a continuous random variable. Then, the notion of a coherent risk measure was introduced and it was explained why VaR fails to be coherent, whereas CVaR is. After this intuitive introduction, CVaR was properly defined and analysed in \autoref{sec:CVaR_RM-Closer_Analysis_of_CVaR}. CVaR can be calculated as the expectation of the generalized $\alpha$ tail distribution. Alternatively, CVaR can be calculated as a weighted average of VaR and $\text{CVaR}^+$ by the Convex Combination Formula (see \autoref{eqn:CVaR_Convex_Combination_Formula}). Another possibility to calculate CVaR is to use Acerbi's Integral Formula (presented in \autoref{sec:CVaR_RM-Acerbi}), for which a novel proof for continuous loss distributions was given in \autoref{sec:CVaR_RM-Acerbi_new_proof}.

\autoref{chapter:CVaR_Portfolio_Optimization} then extended the ideas developed in \autoref{chapter:CVaR_as_risk_measure} to multivariate loss distributions which arise in portfolio selection. To introduce portfolio optimization problems, \autoref{sec:CVaR_PO_Markowitz} presented the first model that was developed to minimize portfolio risk, i.e. the Markowitz Model (see \autoref{problem:MVO_optimization}). It was also shown that it is always favourable to diversify a portfolio in order to reduce risk. The optimal risk/return combinations that can be achieved in a portfolio were drawn to explain the efficient frontier. Motivated by some shortcomings of the Markowitz Model, the Rockafellar and Uryasev Model was presented in \autoref{sec:CVaR_PO_Min_CVaR} to demonstrate how a portfolio can be optimized with regards to minimizing the portfolio's tail risk. The model and associated linear optimization programme that has been developed in \cite{Rockafellar2000_Optimization_of_CVaR} was analysed in detail, before establishing a connection between the Markowitz Model and the Rockafellar and Uryasev Model. \autoref{sec:CVaR_PO_Examples} concluded the chapter by providing two numerical examples. The first example showed that in certain cases, Mean-Variance and CVaR optimization indeed give the same optimal portfolio, while the second example showed that for skewed loss distributions CVaR optimization is preferable over Mean-Variance optimization.

For situations in which a portfolio has already been formed, but for which the investor wishes to hedge risks, a procedure was presented in \autoref{chapter:Portfolio_Hedging}. Since the example was a trader's portfolio consisting of stock options, the financial background on options was presented in \autoref{sec:Portfolio_Hedging-Background_on_Options}, while \autoref{sec:Portfolio_Hedging-Background_on_Risk_Management} showed how a risk managers can estimate the daily asset volatilities to properly manage the risk on a daily basis. The trader's portfolio was described in \autoref{sec:Portfolio_Hedging-Forming_Strangle} and the hedging procedure was outlined in detail in \autoref{sec:Portfolio_Hedging-Hedging_Strangle}. The original contribution of \autoref{sec:Portfolio_Hedging-Hedging_Strangle} was the explicit formulation of the linear programme to minimize the CVaR of the portfolio. \\
~\\

Next, the focus shifted away from financial applications of CVaR. The fairly new concept of CVaR norms was introduced in \autoref{chapter:CVaR_Norms}. The first one, the Scaled CVaR norm, was presented in \autoref{sec:CVaR_Norms-Scaled}, with its definition and alternative characterization given by Pavlikov and Uryasev in \cite{pavlikov2014_CVaR_Norm_and_applications}. A novel contribution was an alternative proof for the equivalence of the two characterizations. Next, the Non-Scaled CVaR norm (or simply CVaR norm) was presented in \autoref{sec:CVaR_Norms-Non_scaled_CVaR_Norm}, by showing how it can be derived from the Scaled CVaR norm. Also, it was shown how the CVaR norm can be interpreted as the optimal value of the knapsack problem. To provide a better understanding of these new norms, \autoref{sec:CVaR_Norms-Properties} stated some of the quite different properties that the two CVaR norms have. A new property of the Scaled CVaR norm, i.e. piecewise convexity, was proposed and proven, which was again an original contribution of this thesis. Finally, the computational efficiencies of the different characterizations of the CVaR norms were investigated in \autoref{sec:CVaR_Norms-Computational_Efficiency}. This comparison of computing times was another original contribution.

After introducing the Scaled CVaR norm and CVaR norm, comparisons to the more familiar family of $L_p$ norms were drawn in \autoref{chapter:Comparison_to_other_vector_norms}. The main goal of this chapter was to show how $C_{\alpha}^S$ and $C_{\alpha}$ behave in comparison to $L_p^S$ and $L_p$ for different combinations of $\alpha$ and $p$. Also, in \autoref{subsec:Comparison_to_other_vector_norms-Relationship_between_alpha_and_p} it was analysed how to choose $\alpha$ in relation to $p$ so that the $C_{\alpha}$ most closely approximates the $L_p$ norm.

A possible application of the CVaR norm was investigated for model recovery problems. The theoretical background for model recovery problems was presented in \autoref{chapter:Model_Recovery_using_Atomic_Norms}. The aim of these problems is to recover models or signals of dimension $p$ with $n < p$ random measurements. Atomic norms and important concepts from convex geometry, such as tangent and normal cones, were introduced in \autoref{sec:MRuAN-Background}. The recovery conditions (which are based on atomic norms and convex geometry) were presented in \autoref{sec:MRuAN-Recovery_Conditions}. For these conditions, the Gaussian width of a set plays a crucial role, but it is generally difficult to determine the Gaussian width of arbitrary sets. Therefore, \autoref{sec:MRuAN-Gaussian_Widths} presented selected properties of Gaussian widths, which might be useful in calculating bounds on Gaussian widths relating to the CVaR norm.

The final chapter, \autoref{chapter:Recovery_using_CVaR}, contained completely original work. The goal of this chapter was to show how the CVaR norm could be used for model recovery problems. Due to the limited scope of this thesis, only partial results could be presented so that this chapter might form a basis for further research in this area. \autoref{section:RuCVaR-Atomic_CVaR_Norm} gave a conjecture on the set of atoms relating to the CVaR norm for $\frac{p-2}{p} < \alpha < \frac{p-1}{p}$ (\autoref{conj:CVaR_Norm_atoms1}), which was partially proven. A comparison of unit balls of the $C_{\alpha}$ norm for $\frac{p-2}{p} < \alpha < \frac{p-1}{p}$ and $0 < \alpha < \frac{1}{p}$ was given, and a numerical experiment was performed in $\mathbb{R}^4$ to provide evidence for \autoref{conj:CVaR_Norm_atoms1}. The final section, \autoref{section:RuCVaR-Experiments}, then performs numerical experiments to show the recovery rate for different $\mathbf{x}^*$ using the CVaR norm as the atomic norm. From these experiments, it appears that the CVaR norm is not suitable to recover an own type of signal, as recovery could not be guaranteed with high probability for $n < p$. For other types of $\mathbf{x}^*$ (i.e. $k$-sparse vectors and vectors that are the sum of $k$ binary vectors), model recovery using the CVaR norm was compared to using the $L_1$ norm and $L_{\infty}$ norm, respectively. While the CVaR norm performed considerably worse than the $L_1$ norm for recovering $k$-sparse vectors, the CVaR norm was marginally better than the $L_{\infty}$ norm for recovering vectors that are the sum of $k$ binary vectors. As these experiments were carried out with a particular choice of $\alpha$, different $\alpha$ might yield different results, as the unit balls of the CVaR are quite different depending on $\alpha$. Hence, it might be promising to conduct further research in this area.

\clearpage

\addcontentsline{toc}{chapter}{Bibliography}
\bibliography{literature}
\clearpage

\appendix
\addcontentsline{toc}{chapter}{Appendices}
\renewcommand{\thepage}{\Roman{page}}
\setcounter{page}{1}

\clearpage

%
%

\chapter{Matlab Code}\label{app:Matlab_Code}

\section{List of Matlab Code Developed During this Dissertation}
\label{code:List_of_codes}
\begin{table}[H] \scriptsize
	\centering
	\begin{tabular}{| C{0.3cm} | L{4.5cm} | L{6.8cm} | L{2.5cm} |}
		\hline
		\textbf{\#} & \textbf{Filename} & \textbf{Purpose of Code} & \textbf{Used for} \\
		\hline
		1 & CVaR\_Norm\_Component.m & Calculate the CVaR norm of $\xinRn$ at a given $\alpha$ using \autoref{def:CVaR_Component_Wise} (see \autoref{code:CVaR_Component_Wise}) & CVaR norm calculations \\
		\hdashline
		2 & CVaR\_Norm\_Optimization.m & Calculate the CVaR norm of $\xinRn$ at a given $\alpha$ using \autoref{prop:CVaR_based_on_Definition} (see \autoref{code:CVaR_based_on_Definition}) & CVaR norm calculations \\
		\hdashline
		3 & Scaled\_CVaR\_Norm \_Component.m & Calculate the Scaled CVaR norm of $\xinRn$ at a given $\alpha$ using \autoref{def:Scaled_CVaR_Component_Wise} (see \autoref{code:Scaled_CVaR_Component_Wise}) & Scaled CVaR norm calculations \\
		\hdashline
		4 & Scaled\_CVaR\_Norm \_Optimization.m & Calculate the Scaled CVaR norm of $\xinRn$ at a given $\alpha$ using \autoref{prop:Scaled_CVaR_based_on_Definition} (see \autoref{code:Scaled_CVaR_based_on_Definition}) & Scaled CVaR norm calculations \\
		\hdashline
		5 & Experiment01\_CVaR\_Norms \_Computing\_Times.m & Compare computing times of codes 1-4 & \autoref{table:CVaR_Norm_Computation_Time_for_different_n}, \autoref{table:CVaR_Norm_Computation_Time_for_different_alpha}, \autoref{app_table:CVaR_Norm_Computation_Time_for_different_alpha_and_n} \\
		\hdashline
		6 & Experiment03\_CVaR\_Norm\_on\_2D\_grid.m & Draw surface plots of $C_{\alpha}$ and $L_p$ of $\xinR{2}$ for different $\alpha$ and $p$ & \autoref{fig:CVaR_Lp_Norm_surface_plots}, \autoref{app_diagrams:C_alpha_Lp_for_different_alpha_p} \\
		\hdashline
		7 & Experiment05\_CVaR\_Lp\_Norm\_as \_functions\_of\_alpha\_p.m & Calculate $C_{\alpha}^S$, $C_{\alpha}$ and corresponding $L_p$, $L_p^S$ for $\alpha \in [0,1]$ & \autoref{fig:Scaled_CVaR_LSp_Norm_For_different_alpha}, \autoref{fig:CVaR_Lp_Norm_For_different_alpha} \\
		\hdashline
		8 & Experiment06\_Projecting\_Points \_onto\_unit\_ball.m & Project a circle in $\mathbb{R}^3$ onto the unit ball $x_1^2 + x_2^2 = 1, x_3 =1$ using $L_2$ norm and $C_{\alpha}$ norm minimization for different $\alpha$ & \autoref{fig:projection_onto_unit_ball}, \autoref{app_diagrams:Projection_onto_unit_ball}\\
		\hdashline
		9 & Experiment07\_UL\_ratio\_for\_Lp \_approximation\_by\_CVaR\_norm.m & Calculate and draw proximity ratio of $C_{\alpha}$ and $L_p$ for different $p$ & \autoref{fig:f_np_kappa_ratio} \\
		\hdashline
		10 & Experiment10\_MVO\_CVaR \_Optimization\_Normal\_Dist.m & Compute Mean-Variance and CVaR optimal portfolios for normally distributed losses & \autoref{table:MVO_CVaR_Optimal_Portfolios_different_R}\\
		\hdashline
		11 & Experiment11\_MVO\_CVaR \_Optimization\_Skewed\_Dist.m &  Compute Mean-Variance and CVaR optimal portfolios for skewed loss distributions, draw histogram of simulated portfolio losses, give risk metrics of optimal portfolios & \autoref{table:Performance_of_Optimal_Portfolios_scenario2}, \autoref{table:MVO_CVaR_Optimal_Portfolios_scenario2}, \autoref{app_diagrams:Loss_distributions_portfolios_scenario2}\\
		\hdashline
		12 & Experiment12\_Hedging.m & Perform Hedging procedure described in \autoref{sec:Portfolio_Hedging-Hedging_Strangle}, draw option payoff profiles before / after hedging, draw loss distribution before / after hedging, give risk metrics of portfolio before / after hedge & \autoref{fig:strangle_googl_yhoo_payoff_profile_before}, \autoref{fig:strangle_googl_yhoo_payoff_profile_after}, \autoref{fig:Trader_Loss_Distribution_before}, \autoref{fig:Trader_Loss_Distribution_after}, \autoref{table:Risk_Metrics_unhedged_hedged}, \autoref{app_table:trader_positions_after_yhoo}, \autoref{app_table:trader_positions_after_googl} \\
		\hdashline
		13 & Experiment13\_VaR \_CVaR\_pdf\_cdf.m & Draw pdf and cdf of a normal random variable to explain VaR and CVaR & \autoref{fig:VaR_CVaR_Explanation} \\
		\hdashline
		14 & Experiment14\_MVO \_Efficient\_Frontier.m & Calculate Mean-Variance optimal portfolio for different required expected returns $R$ and draw efficient frontier & \autoref{fig:Efficient_Frontier} \\
		\hdashline
		15 & Experiment15\_Find \_CVaR\_Graphically.m & Draw $\phi_{\alpha}(c)$ (\autoref{eqn:phi_alpha})for different $c$ & \autoref{fig:phic_vs_c} \\
		\hdashline
		16 & Experiment16a\_Scaled \_CVaR\_own\_examples.m & Draw unit balls of $C_{\alpha}^S$ for different values of $\alpha$ & \autoref{fig:CSalpha_unit_balls_own_example} \\
		\hdashline
		17 & Experiment16b\_CVaR \_own\_examples.m & Draw unit balls of $C_{\alpha}$ for different values of $\alpha$ & \autoref{fig:Calpha_unit_balls_own_example} \\
		\hdashline		
		18 & Experiment17\_Show\_Piecewise \_Convexity\_CSalpha.m & Draw $C_{\alpha}^S$ of 4 different $\mathbf{x}$ versus $\alpha$ & \autoref{fig:piecewise_convexity_example} \\
		\hline
	\end{tabular}
\end{table}
Continued on next page...
\clearpage
... continued from previous page
\begin{table}[H] \scriptsize
	\centering
	\begin{tabular}{| C{0.3cm} | L{4.5cm} | L{6.8cm} | L{2.5cm} |}
		\hline
		\textbf{\#} & \textbf{Filename} & \textbf{Purpose of Code} & \textbf{Used for} \\
		\hline
		19 & Experiment20\_CVaR \_Model\_Recovery.m & Test model recovery of different $\mathbf{x}^*$ using CVaR norm & \autoref{fig:probability_of_recovery_CVaR} \\
		\hdashline
		20 & Experiment20a\_L1 \_Model\_Recovery.m & Compare recovery probability of different $\mathbf{x}^*$ using CVaR norm versus $L_1$ norm & \autoref{fig:probability_of_recovery_L1} \\
		\hdashline
		21 & Experiment20b\_Linfty \_Model\_Recovery.m & Compare recovery probability of different $\mathbf{x}^*$ using CVaR norm versus $L_{\infty}$ norm & \autoref{fig:probability_of_recovery_Linfty} \\
		\hdashline
		22 & Experiment21\_CVaR\_Atoms\_R4.m & Project random hyperplanes onto unit ball of $C_{0.625}$ in $\mathbb{R}^4$ & \autoref{app_table:projections ratio} \\
		\hline
	\end{tabular}
\end{table}

\section{Scaled CVaR Calculation based on \autoref*{def:Scaled_CVaR_Component_Wise}}
\label{code:Scaled_CVaR_Component_Wise}
\begin{lstlisting}
% Author:
% Jakob Kisiala, June 2015
% Computes the scaled CVaR norm of a vector at a given alpha, using
% componentwise definition

% INPUT:
% x = n-by-1 vector of values 
% alpha = scalar between 0 and 1
% OUTPUT:
% C_S_alpha = << x >>^S_{alpha}

function C_S_alpha = Scaled_CVaR_Norm_Component(x, alpha)
C_S_alpha = 0;
% check if alpha is admissible
if (alpha < 0 || alpha > 1)
    display('Please put in an alpha such that 0 <= alpha <= 1 - Scaled CVaR could not be calculated');
    return
end

% check if x is a vector
size_x = size(x);
dim_x = length(size_x);

if (dim_x > 2) % x has more than 2 dimensions
    display('Please only input vectors x - Scaled CVaR could not be calculated');
    return
end
if (size_x(1) > 1 && size_x(2) > 1) % x is a matrix
    display('Please only input vectors x - Scaled CVaR could not be calculated');
    return
end

n = length(x);

% check four cases:
% 0: alpha = 0
% 1: alpha > (n-1)/n
% 2: alpha equal to some alpha_j
% 3: alpha between alpha_j and alpha_{j+1}

% case 0: alpha = 0
if(alpha == 0)
    C_S_alpha = sum(abs(x))/n;
    return
end

% for the remaining three cases additional vectors are needed:
alpha_j_vector  = ([0:n-1]')/n;

% case 1: alpha > (n-1)/n
if (alpha > alpha_j_vector(n))
    C_S_alpha = max(abs(x));
    return
end

% sort vector x by magnitude of components
x_abs_sorted = sort(abs(x));

epsilon = 1e-10;
temp_vector = alpha_j_vector - alpha;

% case 2: alpha equal to some alpha_j

if (any(abs(temp_vector) < epsilon))
    C_S_alpha = calculate_Norm_for_alpha_j(x_abs_sorted,alpha);
    return
end

% case 3: alpha between alpha_j and alpha_{j+1}
    % find alpha_j
    temp_index = temp_vector < 0;
    alpha_j = max(alpha_j_vector(temp_index));
    % find alpha_{j+1}
    temp_index = temp_vector > 0;
    alpha_jPlus1 = min(alpha_j_vector(temp_index));
    
    mu = ((alpha_jPlus1 - alpha)*(1 - alpha_j)) / ((alpha_jPlus1 - alpha_j)*(1 - alpha));
    
    C_aj = calculate_Norm_for_alpha_j(x_abs_sorted,alpha_j);
    C_ajPlus1 = calculate_Norm_for_alpha_j(x_abs_sorted,alpha_jPlus1);
    
    C_S_alpha = mu*C_aj + (1 - mu)*C_ajPlus1;

    % function to calculate the C^S_{alpha} for alpha_j
    function C_S_alpha1 = calculate_Norm_for_alpha_j(vector,alpha_j) 
        j = find(abs(alpha_j_vector - alpha_j) < 1e-10) - 1;
        C_S_alpha1 = (1 / (n - j)) * sum(vector(j+1:n));
    end
end
\end{lstlisting}

\section{Scaled CVaR Calculation based on \autoref*{prop:Scaled_CVaR_based_on_Definition}}
\label{code:Scaled_CVaR_based_on_Definition}
\begin{lstlisting}
% Author:
% Jakob Kisiala, June 2015
% Computes the scaled CVaR norm of a vector at a given alpha, using
% CVaR optimization

% INPUT:
% x = n-by-1 vector of values 
% alpha = scalar between 0 and 1
% OUTPUT:
% C_S_alpha = << x >>^S_{alpha}

function C_S_alpha = Scaled_CVaR_Norm_Optimization(x, alpha)
C_S_alpha = 0;
% check if alpha is admissible
if (alpha < 0 || alpha > 1)
    display('Please put in an alpha such that 0 <= alpha <= 1 - Scaled CVaR could not be calculated');
    return
end

% check if x is a vector
size_x = size(x);
dim_x = length(size_x);

if (dim_x > 2) % x has more than 2 dimensions
    display('Please only input vectors x - Scaled CVaR could not be calculated');
    return
end
if (size_x(1) > 1 && size_x(2) > 1) % x is a matrix
    display('Please only input vectors x - Scaled CVaR could not be calculated');
    return
end

x_abs = abs(x);

% special case: alpha = 1
if (alpha == 1)
    C_S_alpha = max(x_abs);
    return
end

% use CVaR optimization to calculate norm
n = length(x);
e = ones(n,1);

cvx_begin
    cvx_quiet(true) % supresses cvx's output
    variables z(n) c
    minimize (c + (1/(n*(1-alpha)))*(e'*z) )
    subject to
        z >= x_abs - c;
        z >= 0;
cvx_end

C_S_alpha = cvx_optval;

end
\end{lstlisting}

\section{CVaR Calculation based on \autoref*{def:CVaR_Component_Wise}}
\label{code:CVaR_Component_Wise}
\begin{lstlisting}
% Author:
% Jakob Kisiala, June 2015
% Computes the (non-scaled) CVaR norm of a vector at a given alpha, using
% componentwise definition

% INPUT:
% x = n-by-1 vector of values 
% alpha = scalar between 0 and 1
% OUTPUT:
% C_alpha = << x >>_{alpha}

function C_alpha = CVaR_Norm_Component(x, alpha)
C_alpha = 0;
% check if alpha is admissible
if (alpha < 0 || alpha >= 1)
    display('Please put in an alpha such that 0 <= alpha < 1 - CVaR could not be calculated');
    return
end

% check if x is a vector
size_x = size(x);
dim_x = length(size_x);

if (dim_x > 2)
    display('Please only input vectors x - CVaR could not be calculated');
    return
end
if (size_x(1) > 1 && size_x(2) > 1)
    display('Please only input vectors x - CVaR could not be calculated');
    return
end

% check four cases:
% 0: alpha = 0
% 1: alpha > (n-1)/n
% 2: alpha equal to some alpha_j
% 3: alpha between alpha_j and alpha_{j+1}

% case 0: alpha = 0
if(alpha == 0)
    C_alpha = sum(abs(x));
    return
end

% for the remaining three cases additional vectors are needed:
n = length(x);
alpha_times_n = alpha*n;

% case 1: alpha > (n-1)/n
if (alpha_times_n > n-1)
    C_alpha = n*(1-alpha)*max(abs(x));
    return
end

% x vector, in aboslute values sorted in ascending order
x_abs_sorted = sort(abs(x));

epsilon = 1e-10;

% case 2: alpha equal to some alpha_j

if (mod(alpha_times_n,1) < epsilon)
    %j = find(abs(alpha_j_vector - alpha) < 1e-10) - 1;
    %C_S_alpha = (1 / (n - j)) * sum(x_abs_sorted(j+1:n));
    C_alpha = calculate_Norm_for_alpha_j(x_abs_sorted,round(alpha_times_n));
    return
end

% case 3: alpha between alpha_j and alpha_{j+1}
    % find alpha_j
    j = floor(alpha_times_n);
    alpha_j = j/n;
    % find alpha_{j+1}
    jPlus1 = ceil(alpha_times_n);
    alpha_jPlus1 = jPlus1/n;
    
    lambda = (alpha_jPlus1 - alpha) / (alpha_jPlus1 - alpha_j);
    
    C_aj = calculate_Norm_for_alpha_j(x_abs_sorted,j);
    C_ajPlus1 = calculate_Norm_for_alpha_j(x_abs_sorted,jPlus1);
    
    C_alpha = lambda*C_aj + (1 - lambda)*C_ajPlus1;

    % function to calculate the C^S_{alpha} for alpha_j
    function C_alpha1 = calculate_Norm_for_alpha_j(vector,j)         
        C_alpha1 = sum(vector(j+1:n));
    end
end
\end{lstlisting}

\section{CVaR Calculation based on \autoref*{prop:CVaR_based_on_Definition}}
\label{code:CVaR_based_on_Definition}
\begin{lstlisting}
% Author:
% Jakob Kisiala, June 2015
% Computes the (non-scaled) CVaR norm of a vector at a given alpha, using
% CVaR optimization

% INPUT:
% x = n-by-1 vector of values 
% alpha = scalar between 0 and 1
% OUTPUT:
% C_alpha = << x >>_{alpha}

function C_alpha = CVaR_Norm_Optimization(x, alpha)
C_alpha = 0;
% check if alpha is admissible
if (alpha < 0 || alpha >= 1)
    display('Please put in an alpha such that 0 <= alpha < 1 - CVaR could not be calculated');
    return
end

% check if x is a vector
size_x = size(x);
dim_x = length(size_x);

if (dim_x > 2)
    display('Please only input vectors x - CVaR could not be calculated');
    return
end
if (size_x(1) > 1 && size_x(2) > 1)
    display('Please only input vectors x - CVaR could not be calculated');
    return
end

x_abs = abs(x);

% use CVaR optimization to calculate norm
n = length(x);
e = ones(n,1);

cvx_begin
    cvx_quiet(true) % supresses cvx's output
    variables z(n) c
    minimize (n*(1 - alpha)*c + e'*z )
    subject to
        z >= x_abs - c;
        z >= 0;
cvx_end

C_alpha = cvx_optval;
end
\end{lstlisting}

\clearpage

%
%

\chapter{Extended Tables}\label{app:tables}

\section[Option Prices on NASDAQ:YHOO]{Option Prices on NASDAQ:YHOO on 22 July 2015, 9:00 a.m. New York Time}
\label{app_table:option_prices_yhoo}
\begin{table}[H] \scriptsize
	\centering
	\begin{tabular}{ *{2}{| l | l | r| r |}}
		\hline
\textbf{Underlying}	&	\textbf{Option}	&	\textbf{Strike}	&	\textbf{Price}	&	\textbf{Underlying}	&	\textbf{Option}	&	\textbf{Strike}	&	\textbf{Price}	\\
\hline
Yahoo	&	Call	&	31.5	&	7.050	&	Yahoo	&	Put	&	31.5	&	0.170	\\
Yahoo	&	Call	&	34.0	&	4.625	&	Yahoo	&	Put	&	34.0	&	0.020	\\
Yahoo	&	Call	&	35.0	&	3.650	&	Yahoo	&	Put	&	35.0	&	0.025	\\
Yahoo	&	Call	&	35.5	&	3.125	&	Yahoo	&	Put	&	35.5	&	0.030	\\
Yahoo	&	Call	&	36.0	&	2.520	&	Yahoo	&	Put	&	36.0	&	0.040	\\
Yahoo	&	Call	&	36.5	&	2.305	&	Yahoo	&	Put	&	36.5	&	0.045	\\
Yahoo	&	Call	&	37.0	&	1.790	&	Yahoo	&	Put	&	37.0	&	0.060	\\
Yahoo	&	Call	&	37.5	&	1.330	&	Yahoo	&	Put	&	37.5	&	0.080	\\
Yahoo	&	Call	&	38.0	&	0.905	&	Yahoo	&	Put	&	38.0	&	0.130	\\
Yahoo	&	Call	&	38.5	&	0.575	&	Yahoo	&	Put	&	38.5	&	0.285	\\
Yahoo	&	Call	&	39.0	&	0.305	&	Yahoo	&	Put	&	39.0	&	0.480	\\
Yahoo	&	Call	&	39.5	&	0.155	&	Yahoo	&	Put	&	39.5	&	0.880	\\
Yahoo	&	Call	&	40.0	&	0.085	&	Yahoo	&	Put	&	40.0	&	1.260	\\
Yahoo	&	Call	&	40.5	&	0.060	&	Yahoo	&	Put	&	40.5	&	1.740	\\
Yahoo	&	Call	&	41.0	&	0.040	&	Yahoo	&	Put	&	41.0	&	2.195	\\
Yahoo	&	Call	&	41.5	&	0.025	&	Yahoo	&	Put	&	41.5	&	2.715	\\
Yahoo	&	Call	&	42.0	&	0.030	&	Yahoo	&	Put	&	42.0	&	3.225	\\
Yahoo	&	Call	&	42.5	&	0.035	&	Yahoo	&	Put	&	42.5	&	3.725	\\
Yahoo	&	Call	&	43.0	&	0.015	&	Yahoo	&	Put	&	43.0	&	4.225	\\
Yahoo	&	Call	&	43.5	&	0.065	&	Yahoo	&	Put	&	43.5	&	4.650	\\
Yahoo	&	Call	&	44.0	&	0.025	&	Yahoo	&	Put	&	44.0	&	5.275	\\
Yahoo	&	Call	&	44.5	&	0.170	&	Yahoo	&	Put	&	44.5	&	5.675	\\
Yahoo	&	Call	&	45.0	&	0.015	&	Yahoo	&	Put	&	45.0	&	6.150	\\
Yahoo	&	Call	&	46.5	&	0.010	&	Yahoo	&	Put	&	46.5	&	7.700	\\
Yahoo	&	Call	&	49.5	&	0.010	&	Yahoo	&	Put	&	49.5	&	10.500	\\
Yahoo	&	Call	&	50.0	&	0.010	&	Yahoo	&	Put	&	50.0	&	11.025	\\
Yahoo	&	Call	&	50.5	&	0.010	&	Yahoo	&	Put	&	50.5	&	11.525	\\
		\hline
	\end{tabular}
\end{table}

\section[Option Prices on NASDAQ:GOOGL]{Option Prices on NASDAQ:GOOGL on 22 July 2015, 9:00 a.m. New York Time}
\label{app_table:option_prices_googl}
\begin{table}[H] \scriptsize
	\centering
	\begin{tabular}{ *{2}{| l | l | r| r |}}
		\hline
\textbf{Underlying}	&	\textbf{Option}	&	\textbf{Strike}	&	\textbf{Price}	&	\textbf{Underlying}	&	\textbf{Option}	&	\textbf{Strike}	&	\textbf{Price}	\\
\hline
Google	&	Call	&	510.0	&	194.200	&	Google	&	Put	&	510.0	&	0.030	\\
Google	&	Call	&	535.0	&	169.400	&	Google	&	Put	&	535.0	&	0.155	\\
Google	&	Call	&	545.0	&	158.950	&	Google	&	Put	&	545.0	&	0.180	\\
Google	&	Call	&	550.0	&	153.950	&	Google	&	Put	&	550.0	&	0.030	\\
Google	&	Call	&	560.0	&	144.200	&	Google	&	Put	&	560.0	&	0.055	\\
Google	&	Call	&	565.0	&	139.200	&	Google	&	Put	&	565.0	&	0.130	\\
Google	&	Call	&	570.0	&	134.200	&	Google	&	Put	&	570.0	&	0.130	\\
Google	&	Call	&	580.0	&	124.200	&	Google	&	Put	&	580.0	&	0.205	\\
Google	&	Call	&	590.0	&	114.200	&	Google	&	Put	&	590.0	&	0.180	\\
Google	&	Call	&	597.5	&	106.500	&	Google	&	Put	&	597.5	&	0.155	\\
Google	&	Call	&	600.0	&	104.250	&	Google	&	Put	&	600.0	&	0.030	\\
Google	&	Call	&	615.0	&	88.950	&	Google	&	Put	&	615.0	&	0.155	\\
Google	&	Call	&	620.0	&	83.950	&	Google	&	Put	&	620.0	&	0.155	\\
Google	&	Call	&	630.0	&	74.000	&	Google	&	Put	&	630.0	&	0.155	\\
Google	&	Call	&	650.0	&	54.350	&	Google	&	Put	&	650.0	&	0.150	\\
Google	&	Call	&	652.5	&	52.100	&	Google	&	Put	&	652.5	&	0.275	\\
Google	&	Call	&	655.0	&	49.500	&	Google	&	Put	&	655.0	&	0.275	\\
Google	&	Call	&	657.5	&	46.850	&	Google	&	Put	&	657.5	&	0.275	\\
Google	&	Call	&	660.0	&	44.550	&	Google	&	Put	&	660.0	&	0.300	\\
Google	&	Call	&	665.0	&	39.550	&	Google	&	Put	&	665.0	&	0.425	\\
Google	&	Call	&	667.5	&	36.900	&	Google	&	Put	&	667.5	&	0.525	\\
Google	&	Call	&	670.0	&	34.650	&	Google	&	Put	&	670.0	&	0.600	\\
Google	&	Call	&	675.0	&	29.950	&	Google	&	Put	&	675.0	&	0.800	\\
Google	&	Call	&	677.5	&	27.600	&	Google	&	Put	&	677.5	&	0.950	\\
Google	&	Call	&	680.0	&	25.400	&	Google	&	Put	&	680.0	&	1.150	\\
Google	&	Call	&	682.5	&	23.150	&	Google	&	Put	&	682.5	&	1.375	\\
Google	&	Call	&	685.0	&	20.900	&	Google	&	Put	&	685.0	&	1.700	\\
Google	&	Call	&	687.5	&	18.650	&	Google	&	Put	&	687.5	&	2.075	\\
Google	&	Call	&	690.0	&	16.800	&	Google	&	Put	&	690.0	&	2.600	\\
Google	&	Call	&	692.5	&	14.750	&	Google	&	Put	&	692.5	&	3.175	\\
Google	&	Call	&	695.0	&	12.850	&	Google	&	Put	&	695.0	&	3.850	\\
Google	&	Call	&	697.5	&	11.350	&	Google	&	Put	&	697.5	&	4.700	\\
Google	&	Call	&	700.0	&	9.900	&	Google	&	Put	&	700.0	&	5.600	\\
Google	&	Call	&	702.5	&	8.450	&	Google	&	Put	&	702.5	&	6.750	\\
Google	&	Call	&	705.0	&	7.250	&	Google	&	Put	&	705.0	&	8.100	\\
Google	&	Call	&	710.0	&	5.050	&	Google	&	Put	&	710.0	&	10.950	\\
Google	&	Call	&	712.5	&	4.250	&	Google	&	Put	&	712.5	&	12.550	\\
Google	&	Call	&	715.0	&	3.450	&	Google	&	Put	&	715.0	&	14.250	\\
Google	&	Call	&	717.5	&	2.875	&	Google	&	Put	&	717.5	&	16.100	\\
Google	&	Call	&	720.0	&	2.425	&	Google	&	Put	&	720.0	&	18.200	\\
Google	&	Call	&	725.0	&	1.675	&	Google	&	Put	&	725.0	&	22.550	\\
Google	&	Call	&	730.0	&	1.175	&	Google	&	Put	&	730.0	&	27.200	\\
Google	&	Call	&	735.0	&	0.775	&	Google	&	Put	&	735.0	&	31.350	\\
		\hline
	\end{tabular}
\end{table}

\section[Trader's positions before hedging]{Trader's positions on 22 July 2015, 9:00 a.m. New York Time before hedging}
\label{app_table:trader_positions_before}
\begin{table}[H] \scriptsize
	\centering
	\begin{tabular}{ | l | l | r| r | r | }
		\hline
\textbf{Underlying}	&	\textbf{Option}	&	\textbf{Strike}	&	\textbf{Position}	&	\textbf{Cost of Position (USD)}	\\
\hline
Yahoo	&	Call	&	31.5	&	 35 	&	 24,675 	\\
Yahoo	&	Call	&	34.0	&	 40 	&	 18,500 	\\
Yahoo	&	Call	&	35.0	&	 25 	&	 9,125 	\\
Yahoo	&	Call	&	35.5	&	 30 	&	 9,375 	\\
Yahoo	&	Call	&	36.0	&	 45 	&	 11,340 	\\
Yahoo	&	Call	&	37.0	&	 35 	&	 6,265 	\\
Yahoo	&	Call	&	38.0	&	 40 	&	 3,620 	\\
Yahoo	&	Call	&	38.5	&	 50 	&	 2,875 	\\
Yahoo	&	Call	&	39.0	&	-50 	&	-1,525 	\\
Yahoo	&	Call	&	40.0	&	 10 	&	 85 	\\
Yahoo	&	Call	&	40.5	&	-10 	&	-60 	\\
Yahoo	&	Call	&	41.5	&	 50 	&	 125 	\\
Yahoo	&	Call	&	42.0	&	-1,100 	&	-3,300 	\\
Yahoo	&	Call	&	42.5	&	-50 	&	-175 	\\
Yahoo	&	Call	&	43.0	&	-40 	&	-60 	\\
Yahoo	&	Call	&	43.5	&	-40 	&	-260 	\\
Yahoo	&	Call	&	44.5	&	-35 	&	-595 	\\
Yahoo	&	Call	&	45.0	&	-45 	&	-68 	\\
Yahoo	&	Put	&	31.5	&	-10 	&	-170 	\\
Yahoo	&	Put	&	37.5	&	-1,050 	&	-8,400 	\\
Yahoo	&	Put	&	38.0	&	 6 	&	 78 	\\
Yahoo	&	Put	&	39.0	&	 50 	&	 2,400 	\\
Yahoo	&	Put	&	39.5	&	 49 	&	 4,312 	\\
Yahoo	&	Put	&	40.0	&	 50 	&	 6,300 	\\
Yahoo	&	Put	&	41.5	&	 50 	&	 13,575 	\\
Yahoo	&	Put	&	42.0	&	-50 	&	-16,125 	\\
Yahoo	&	Put	&	42.5	&	-50 	&	-18,625 	\\
Yahoo	&	Put	&	43.0	&	-50 	&	-21,125 	\\
Yahoo	&	Put	&	45.0	&	 50 	&	 30,750 	\\
Yahoo	&	Put	&	49.5	&	 50 	&	 52,500 	\\
Yahoo	&	Put	&	50.0	&	 50 	&	 55,125 	\\
Yahoo	&	Put	&	50.5	&	 50 	&	 57,625 	\\
Google	&	Call	&	730.0	&	-100 	&	-11,750 	\\
Google	&	Put	&	665.0	&	-100 	&	-4,250 	\\
\hline
\hline
\multicolumn{4}{| l |}{\textbf{Total}}	& \textbf{222,163} \\
		\hline
	\end{tabular}
\end{table}

\section[Trader's positions in Yahoo Options after hedging]{Trader's positions in Yahoo Options on 22 July 2015, 9:00 a.m. New York Time after hedging}
\label{app_table:trader_positions_after_yhoo}
\begin{table}[H] \scriptsize
	\centering
	\begin{tabular}{ | l | l | r| r | r | r | r |}
		\hline
\textbf{Underlying}	&	\textbf{Strike}	&	\parbox{1.2cm}{\centering ~ \\ \textbf{Call} \\ \textbf{Position} \\ ~}	& \parbox{2cm}{\centering \textbf{Cost of Call} \\ \textbf{Position (USD)}} & \parbox{1.2cm}{\centering ~ \\ \textbf{Put} \\ \textbf{Position} \\ ~}	&	\parbox{2cm}{\centering \textbf{Cost of Put} \\ \textbf{Position (USD)}}& \parbox{2cm}{\centering \textbf{Net Cost of} \\ \textbf{Position (USD)}} \\
\hline
Yahoo	&	31.5	&	85	&	 59,925 	&	-60	&	-1,020 	&	 58,905 	\\
Yahoo	&	34	&	90	&	 41,625 	&	-50	&	-100 	&	 41,525 	\\
Yahoo	&	35	&	75	&	 27,375 	&	-50	&	-125 	&	 27,250 	\\
Yahoo	&	35.5	&	80	&	 25,000 	&	-50	&	-150 	&	 24,850 	\\
Yahoo	&	36	&	95	&	 23,940 	&	-50	&	-200 	&	 23,740 	\\
Yahoo	&	36.5	&	-50	&	-11,525 	&	-50	&	-225 	&	-11,750 	\\
Yahoo	&	37	&	85	&	 15,215 	&	-50	&	-300 	&	 14,915 	\\
Yahoo	&	37.5	&	50	&	 6,650 	&	-1100	&	-8,800 	&	-2,150 	\\
Yahoo	&	38	&	90	&	 8,145 	&	-44	&	-572 	&	 7,573 	\\
Yahoo	&	38.5	&	49	&	 2,818 	&	-50	&	-1,425 	&	 1,393 	\\
Yahoo	&	39	&	-100	&	-3,050 	&	100	&	 4,800 	&	 1,750 	\\
Yahoo	&	39.5	&	50	&	 775 	&	49	&	 4,312 	&	 5,087 	\\
Yahoo	&	40	&	60	&	 510 	&	100	&	 12,600 	&	 13,110 	\\
Yahoo	&	40.5	&	40	&	 240 	&	50	&	 8,700 	&	 8,940 	\\
Yahoo	&	41	&	50	&	 200 	&	50	&	 10,975 	&	 11,175 	\\
Yahoo	&	41.5	&	100	&	 250 	&	100	&	 27,150 	&	 27,400 	\\
Yahoo	&	42	&	-1150	&	-3,450 	&	-100	&	-32,250 	&	-35,700 	\\
Yahoo	&	42.5	&	-100	&	-350 	&	-100	&	-37,250 	&	-37,600 	\\
Yahoo	&	43	&	-90	&	-135 	&	-100	&	-42,250 	&	-42,385 	\\
Yahoo	&	43.5	&	-90	&	-585 	&	50	&	 23,250 	&	 22,665 	\\
Yahoo	&	44	&	-50	&	-125 	&	-50	&	-26,375 	&	-26,500 	\\
Yahoo	&	44.5	&	-85	&	-1,445 	&	50	&	 28,375 	&	 26,930 	\\
Yahoo	&	45	&	-95	&	-143 	&	100	&	 61,500 	&	 61,358 	\\
Yahoo	&	46.5	&	-50	&	-50 	&	50	&	 38,500 	&	 38,450 	\\
Yahoo	&	49.5	&	-50	&	-50 	&	100	&	 105,000 	&	 104,950 	\\
Yahoo	&	50	&	-50	&	-50 	&	100	&	 110,250 	&	 110,200 	\\
Yahoo	&	50.5	&	-50	&	-50 	&	100	&	 115,250 	&	 115,200 	\\
\hline
\hline
\multicolumn{6}{| l |}{\textbf{Total}}	& \textbf{591,280} \\
		\hline
	\end{tabular}
\end{table}

\section[Trader's positions in Google Options after hedging]{Trader's positions in Google Options on 22 July 2015, 9:00 a.m. New York Time after hedging}
\label{app_table:trader_positions_after_googl}
\begin{table}[H] \scriptsize
	\centering
	\begin{tabular}{ | l | l | r| r | r | r | r |}
		\hline
\textbf{Underlying}	&	\textbf{Strike}	&	\parbox{1.2cm}{\centering ~ \\ \textbf{Call} \\ \textbf{Position} \\ ~}	& \parbox{2cm}{\centering \textbf{Cost of Call} \\ \textbf{Position (USD)}} & \parbox{1.2cm}{\centering ~ \\ \textbf{Put} \\ \textbf{Position} \\ ~}	&	\parbox{2cm}{\centering \textbf{Cost of Put} \\ \textbf{Position (USD)}}& \parbox{2cm}{\centering \textbf{Net Cost of} \\ \textbf{Position (USD)}} \\
\hline
Google	&	510	&	-5	&	-97,100 	&	-5	&	-15 	&	-97,115 	\\
Google	&	535	&	-5	&	-84,700 	&	-5	&	-78 	&	-84,778 	\\
Google	&	545	&	5	&	 79,475 	&	-5	&	-90 	&	 79,385 	\\
Google	&	550	&	5	&	 76,975 	&	-5	&	-15 	&	 76,960 	\\
Google	&	560	&	-5	&	-72,100 	&	-5	&	-28 	&	-72,128 	\\
Google	&	565	&	-5	&	-69,600 	&	-5	&	-65 	&	-69,665 	\\
Google	&	570	&	-5	&	-67,100 	&	-5	&	-65 	&	-67,165 	\\
Google	&	580	&	-5	&	-62,100 	&	-5	&	-103 	&	-62,203 	\\
Google	&	590	&	-5	&	-57,100 	&	-5	&	-90 	&	-57,190 	\\
Google	&	597.5	&	5	&	 53,250 	&	-5	&	-78 	&	 53,173 	\\
Google	&	600	&	-5	&	-52,125 	&	-5	&	-15 	&	-52,140 	\\
Google	&	615	&	5	&	 44,475 	&	-5	&	-78 	&	 44,398 	\\
Google	&	620	&	5	&	 41,975 	&	-5	&	-78 	&	 41,898 	\\
Google	&	630	&	5	&	 37,000 	&	-5	&	-78 	&	 36,923 	\\
Google	&	650	&	-5	&	-27,175 	&	5	&	 75 	&	-27,100 	\\
Google	&	652.5	&	-5	&	-26,050 	&	-5	&	-138 	&	-26,188 	\\
Google	&	655	&	-5	&	-24,750 	&	-5	&	-138 	&	-24,888 	\\
Google	&	657.5	&	5	&	 23,425 	&	5	&	 138 	&	 23,563 	\\
Google	&	660	&	1	&	 4,455 	&	5	&	 150 	&	 4,605 	\\
Google	&	665	&	5	&	 19,775 	&	-95	&	-4,038 	&	 15,738 	\\
Google	&	667.5	&	5	&	 18,450 	&	5	&	 263 	&	 18,713 	\\
Google	&	670	&	5	&	 17,325 	&	5	&	 300 	&	 17,625 	\\
Google	&	675	&	5	&	 14,975 	&	5	&	 400 	&	 15,375 	\\
Google	&	677.5	&	5	&	 13,800 	&	5	&	 475 	&	 14,275 	\\
Google	&	680	&	5	&	 12,700 	&	5	&	 575 	&	 13,275 	\\
Google	&	682.5	&	4	&	 9,260 	&	5	&	 688 	&	 9,948 	\\
Google	&	685	&	-5	&	-10,450 	&	5	&	 850 	&	-9,600 	\\
Google	&	687.5	&	5	&	 9,325 	&	-5	&	-1,038 	&	 8,288 	\\
Google	&	690	&	-5	&	-8,400 	&	5	&	 1,300 	&	-7,100 	\\
Google	&	692.5	&	5	&	 7,375 	&	-5	&	-1,588 	&	 5,788 	\\
Google	&	695	&	5	&	 6,425 	&	-5	&	-1,925 	&	 4,500 	\\
Google	&	697.5	&	5	&	 5,675 	&	-5	&	-2,350 	&	 3,325 	\\
Google	&	700	&	-5	&	-4,950 	&	5	&	 2,800 	&	-2,150 	\\
Google	&	702.5	&	-5	&	-4,225 	&	5	&	 3,375 	&	-850 	\\
Google	&	705	&	4	&	 2,900 	&	-4	&	-3,240 	&	-340 	\\
Google	&	710	&	5	&	 2,525 	&	-5	&	-5,475 	&	-2,950 	\\
Google	&	712.5	&	-5	&	-2,125 	&	5	&	 6,275 	&	 4,150 	\\
Google	&	715	&	-5	&	-1,725 	&	5	&	 7,125 	&	 5,400 	\\
Google	&	717.5	&	-5	&	-1,438 	&	5	&	 8,050 	&	 6,613 	\\
Google	&	720	&	-3	&	-728 	&	5	&	 9,100 	&	 8,373 	\\
Google	&	725	&	5	&	 838 	&	5	&	 11,275 	&	 12,113 	\\
Google	&	730	&	-105	&	-12,338 	&	-5	&	-13,600 	&	-25,938 	\\
Google	&	735	&	-5	&	-388 	&	5	&	 15,675 	&	 15,288 	\\
\hline
\hline
\multicolumn{6}{| l |}{\textbf{Total}}	& \textbf{-149,800} \\
		\hline
	\end{tabular}
\end{table}

\section{Computation times of Scaled and (non-scaled) CVaR Norm in ms}
\label{app_table:CVaR_Norm_Computation_Time_for_different_alpha_and_n}
\begin{table}[H] \scriptsize
	\centering
	\begin{tabular}{| c | r || *{2}{r|}| *{2}{r|}  }
		\hline
		& & \multicolumn{4}{c|}{\textbf{Computation time in ms}} \\
		& & \multicolumn{2}{c||}{\textbf{Component-wise}} & \multicolumn{2}{c|}{\textbf{Optimization}}\\
$\boldsymbol\alpha$	&	$\boldsymbol n$	& \parbox{2.5cm}{\centering $\CVaRnormX[S]{\alpha}$ \\ (\autoref*{def:Scaled_CVaR_Component_Wise})} & \parbox{2.5cm}{\centering $\CVaRnormX{\alpha}$ \\ (\autoref*{def:CVaR_Component_Wise})} & \parbox{2.5cm}{\centering $\CVaRnormX[S]{\alpha}$ \\ (\autoref*{prop:Scaled_CVaR_based_on_Definition})} & \parbox{2.5cm}{\centering $\CVaRnormX{\alpha}$ \\ (\autoref*{prop:CVaR_based_on_Definition})}\\
\hline							
\hline				
\multirow{7}{*}{0}	&	2	&	0.62	&	0.50	&	220.44	&	197.76	\\
	&	3	&	0.11	&	0.03	&	211.03	&	179.15	\\
	&	10	&	0.11	&	0.03	&	181.00	&	173.62	\\
	&	100	&	0.12	&	0.03	&	196.84	&	194.83	\\
	&	1000	&	0.21	&	0.04	&	202.81	&	199.38	\\
	&	10000	&	1.05	&	0.05	&	455.95	&	435.50	\\
	&	100000	&	4.94	&	0.27	&	3766.36	&	3497.11	\\
\hline											
\multirow{7}{*}{0.1}	&	2	&	0.18	&	0.12	&	216.77	&	188.06	\\
	&	3	&	0.19	&	0.12	&	189.60	&	182.71	\\
	&	10	&	0.12	&	0.08	&	199.62	&	186.78	\\
	&	100	&	0.14	&	0.10	&	229.93	&	226.96	\\
	&	1000	&	0.19	&	0.14	&	244.86	&	236.01	\\
	&	10000	&	1.00	&	0.94	&	625.06	&	599.35	\\
	&	100000	&	5.25	&	5.03	&	6175.45	&	5843.76	\\
\hline											
\multirow{7}{*}{0.25}	&	2	&	0.20	&	0.12	&	181.25	&	175.68	\\
	&	3	&	0.18	&	0.12	&	181.29	&	184.35	\\
	&	10	&	0.19	&	0.13	&	265.65	&	242.76	\\
	&	100	&	0.14	&	0.10	&	214.34	&	217.70	\\
	&	1000	&	0.19	&	0.14	&	229.73	&	271.94	\\
	&	10000	&	1.06	&	0.98	&	600.00	&	584.77	\\
	&	100000	&	5.61	&	5.02	&	5772.24	&	5277.92	\\
\hline											
\multirow{7}{*}{0.5}	&	2	&	0.13	&	0.08	&	178.59	&	174.96	\\
	&	3	&	0.18	&	0.12	&	180.96	&	179.34	\\
	&	10	&	0.13	&	0.08	&	184.33	&	181.49	\\
	&	100	&	0.15	&	0.10	&	217.66	&	213.11	\\
	&	1000	&	0.19	&	0.14	&	323.36	&	239.72	\\
	&	10000	&	1.00	&	0.92	&	571.45	&	551.93	\\
	&	100000	&	5.64	&	5.00	&	5516.37	&	5128.19	\\
\hline											
\multirow{7}{*}{0.7}	&	2	&	0.05	&	0.04	&	179.90	&	176.37	\\
	&	3	&	0.05	&	0.03	&	184.08	&	188.00	\\
	&	10	&	0.13	&	0.08	&	187.39	&	189.06	\\
	&	100	&	0.14	&	0.10	&	250.00	&	267.46	\\
	&	1000	&	0.19	&	0.15	&	252.11	&	241.46	\\
	&	10000	&	0.97	&	0.92	&	624.84	&	612.85	\\
	&	100000	&	5.57	&	5.06	&	6201.42	&	5965.34	\\
\hline											
\multirow{7}{*}{0.9}	&	2	&	0.05	&	0.04	&	177.20	&	178.02	\\
	&	3	&	0.05	&	0.04	&	182.70	&	183.60	\\
	&	10	&	0.12	&	0.08	&	177.95	&	180.54	\\
	&	100	&	0.14	&	0.10	&	231.81	&	231.11	\\
	&	1000	&	0.19	&	0.14	&	289.31	&	249.22	\\
	&	10000	&	0.98	&	0.91	&	749.50	&	713.43	\\
	&	100000	&	5.26	&	5.02	&	8122.91	&	7767.68	\\
		\hline
	\end{tabular}
\end{table}
\clearpage

\section{Ratio of Projections of Random Hyperplanes onto $C_{\alpha}$ Unit Ball in $\mathbb{R}^4$ over 5,000 Trials}
\label{app_table:projections ratio}
\begin{table}[H] 
	\centering
	\begin{tabular}{ | >{$}l<{$} | r |}
		\hline
\textbf{Projected onto}	&	\textbf{Ratio} \\
\hline
\mathbf{x}=[1,0,0,0]^T	&	0.62	\%	\\
\mathbf{x}=[0,1,0,0]^T	&	0.88	\%	\\
\mathbf{x}=[0,0,1,0]^T	&	0.70	\%	\\
\mathbf{x}=[0,0,0,1]^T	&	0.72	\%	\\
\mathbf{x}=[-1,0,0,0]^T	&	0.66	\%	\\
\mathbf{x}=[0,-1,0,0]^T	&	0.80	\%	\\
\mathbf{x}=[0,0,-1,0]^T	&	0.86	\%	\\
\mathbf{x}=[0,0,0,-1]^T	&	0.62	\%	\\
\mathbf{x}=(2/3) \times [1,1,1,1]^T	&	5.64	\%	\\
\mathbf{x}=(2/3) \times [1,1,1,-1]^T	&	6.14	\%	\\
\mathbf{x}=(2/3) \times [1,1,-1,1]^T	&	6.24	\%	\\
\mathbf{x}=(2/3) \times [1,1,-1,-1]^T	&	5.84	\%	\\
\mathbf{x}=(2/3) \times [1,-1,1,1]^T	&	5.76	\%	\\
\mathbf{x}=(2/3) \times [1,-1,1,-1]^T	&	6.08	\%	\\
\mathbf{x}=(2/3) \times [1,-1,-1,1]^T	&	5.44	\%	\\
\mathbf{x}=(2/3) \times [1,-1,-1,-1]^T	&	5.04	\%	\\
\mathbf{x}=(2/3) \times [-1,1,1,1]^T	&	5.42	\%	\\
\mathbf{x}=(2/3) \times [-1,1,1,-1]^T	&	6.16	\%	\\
\mathbf{x}=(2/3) \times [-1,1,-1,1]^T	&	6.16	\%	\\
\mathbf{x}=(2/3) \times [-1,1,-1,-1]^T	&	5.86	\%	\\
\mathbf{x}=(2/3) \times [-1,-1,1,1]^T	&	6.22	\%	\\
\mathbf{x}=(2/3) \times [-1,-1,1,-1]^T	&	6.00	\%	\\
\mathbf{x}=(2/3) \times [-1,-1,-1,1]^T	&	6.28	\%	\\
\mathbf{x}=(2/3) \times [-1,-1,-1,-1]^T	&	5.86	\%	\\
\text{other~} \mathbf{x}	&	0.00	\%	\\
\hline
	\end{tabular}
\end{table}

%
%

\chapter{Extended Diagrams}\label{app:diagrams}

\section[Monte Carlo simulated loss distributions of single assets]{Monte Carlo simulated loss distributions of single assets (Scenario 2 of \autoref*{sec:CVaR_PO_Examples})}
\label{app_diagrams:Loss_distributions_single_assets_scenario2}
\begin{figure}[H]
	\centering
	\includegraphics[width = \textwidth]{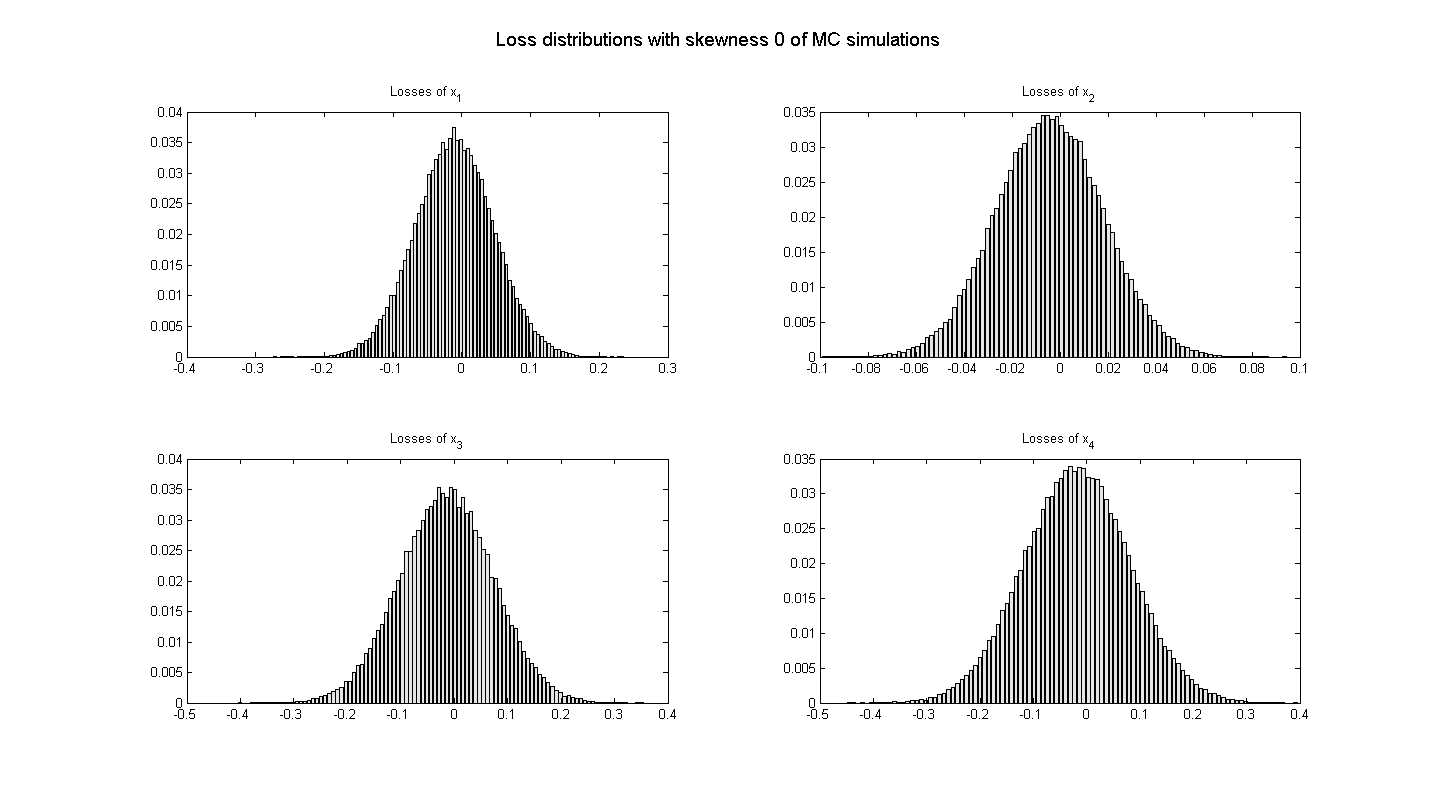}
\end{figure}
\begin{figure}[H]
	\centering
	\includegraphics[width = \textwidth]{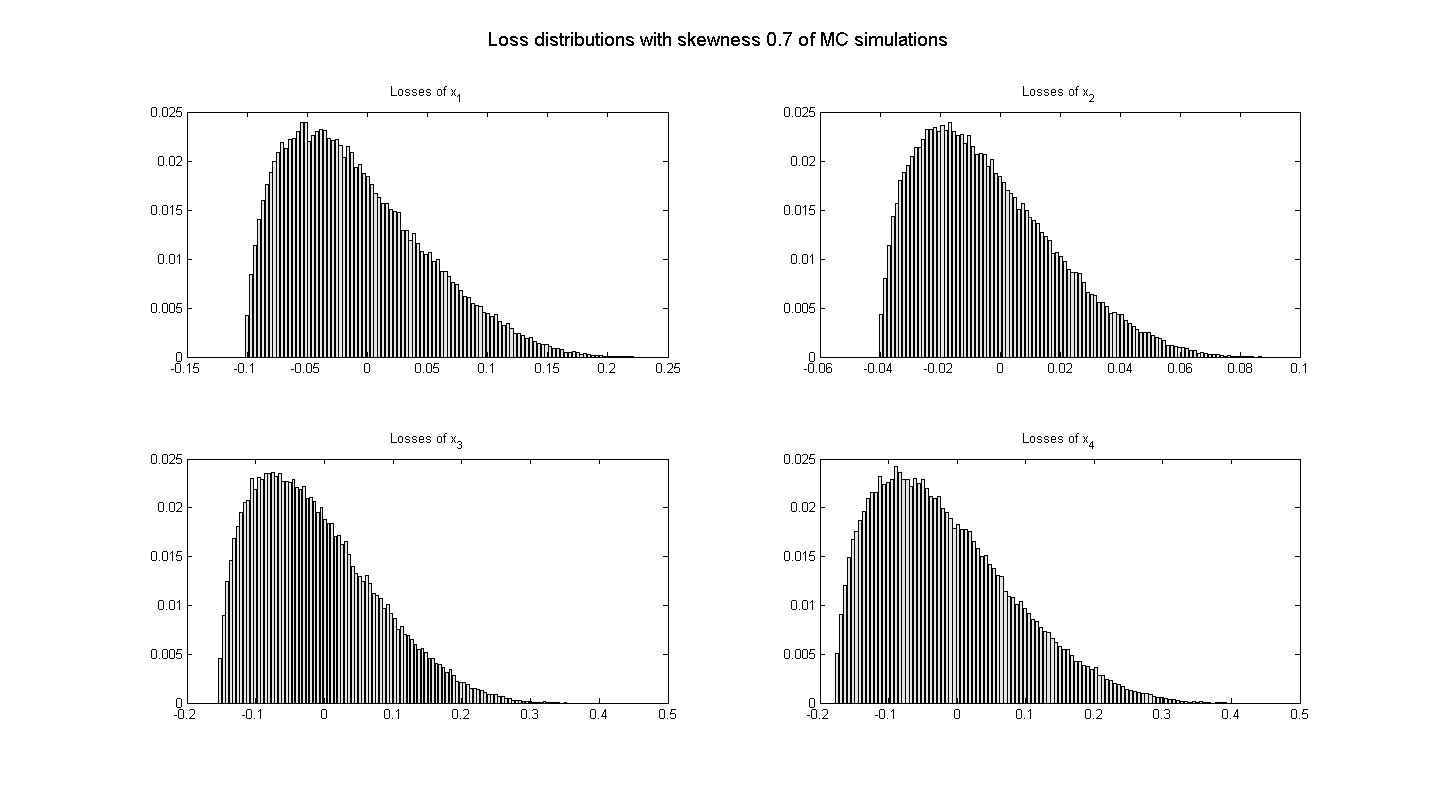}
\end{figure}

\section[Monte Carlo simulated loss distributions of optimal portfolios]{Monte Carlo simulated loss distributions of optimal portfolios (Scenario 2 of \autoref*{sec:CVaR_PO_Examples})}
\label{app_diagrams:Loss_distributions_portfolios_scenario2}
\begin{figure}[H]
	\centering
	\includegraphics[width = \textwidth]{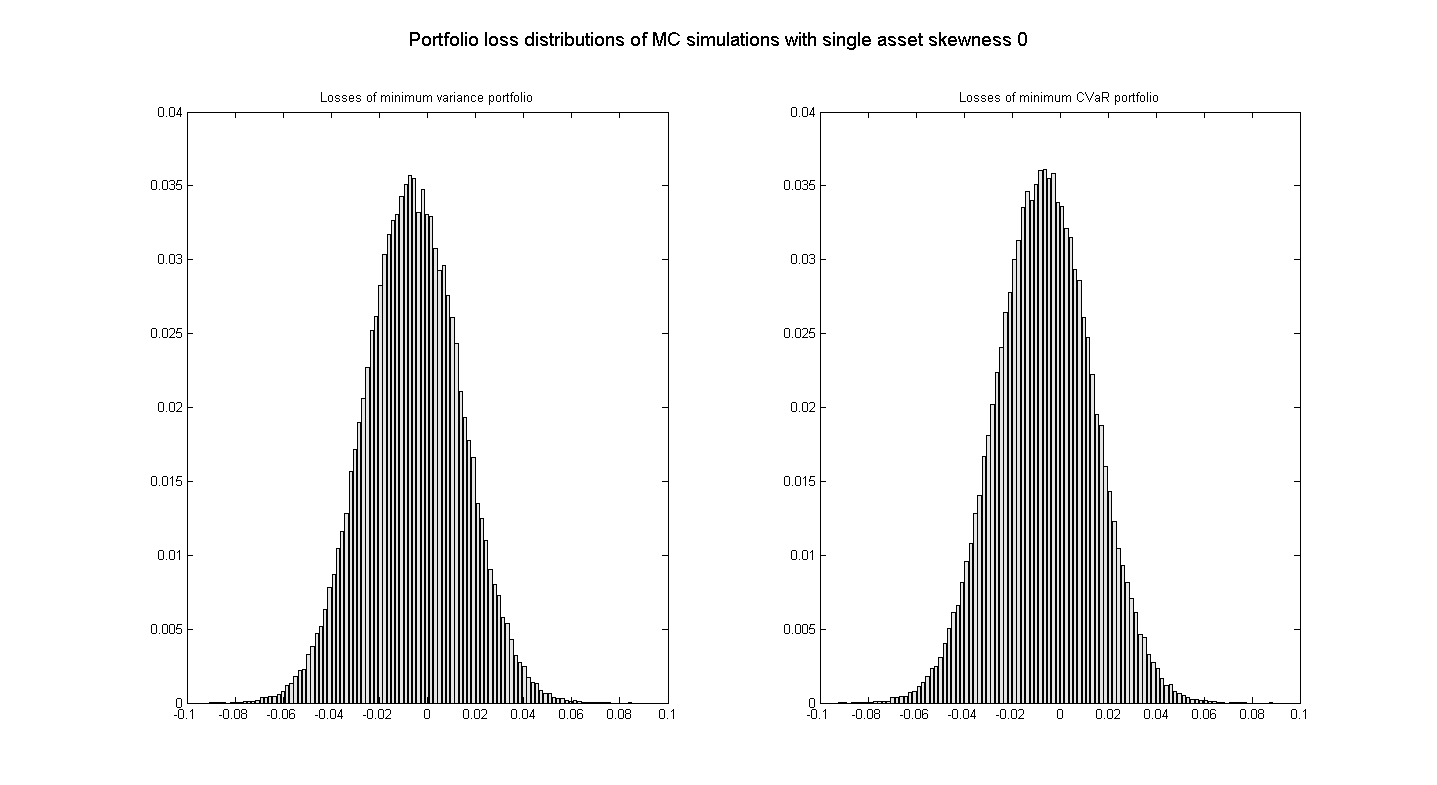}
\end{figure}
\begin{figure}[H]
	\centering
	\includegraphics[width = \textwidth]{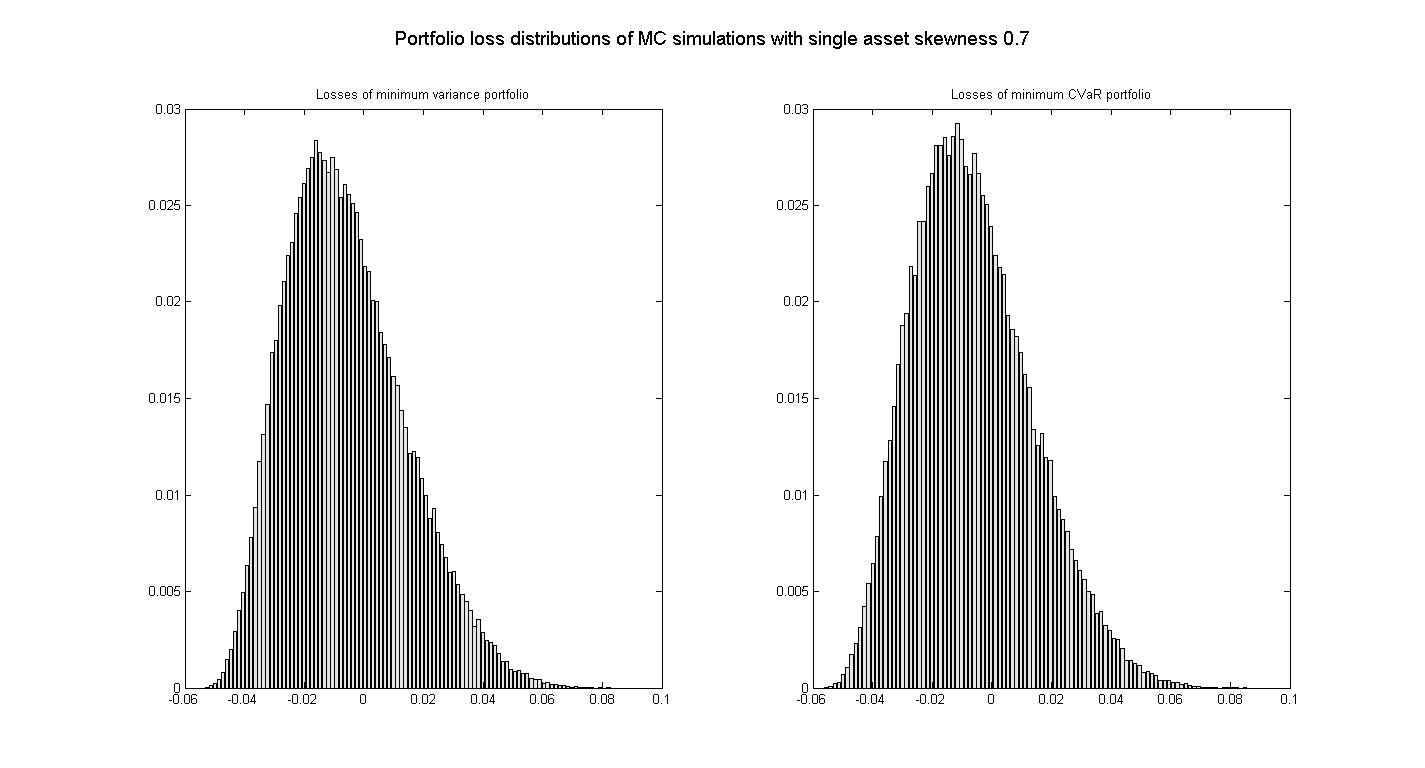}
\end{figure}

\section{$C_{\alpha}$ and $L_{p^*}$ norm surface plots of $ \mathbf{x} \in \mathbb{R}^n$ for different $\alpha$ and $p^*$}
\label{app_diagrams:C_alpha_Lp_for_different_alpha_p}
\begin{figure}[H]
	\centering
	\includegraphics[width = \textwidth]{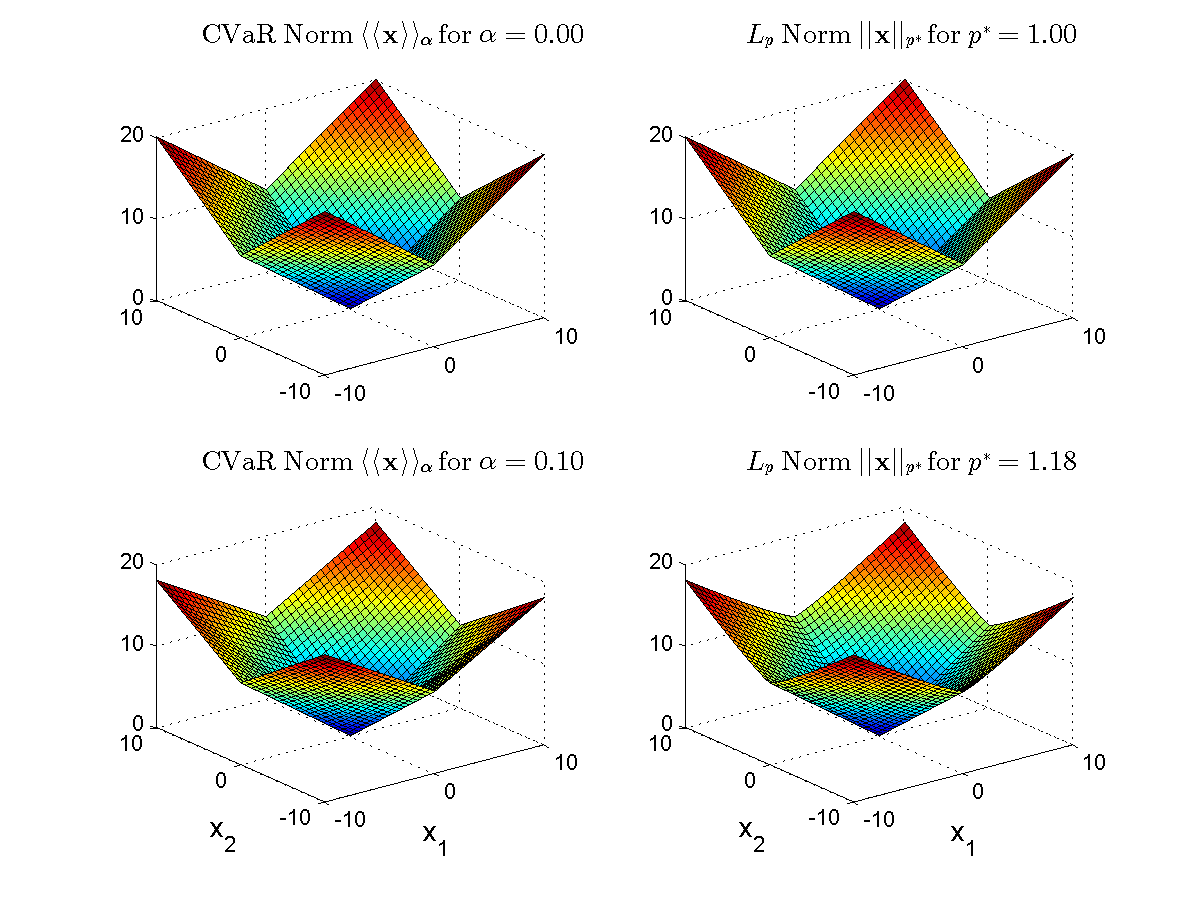}
\end{figure}
\begin{figure}[H]
	\centering
	\includegraphics[width = \textwidth]{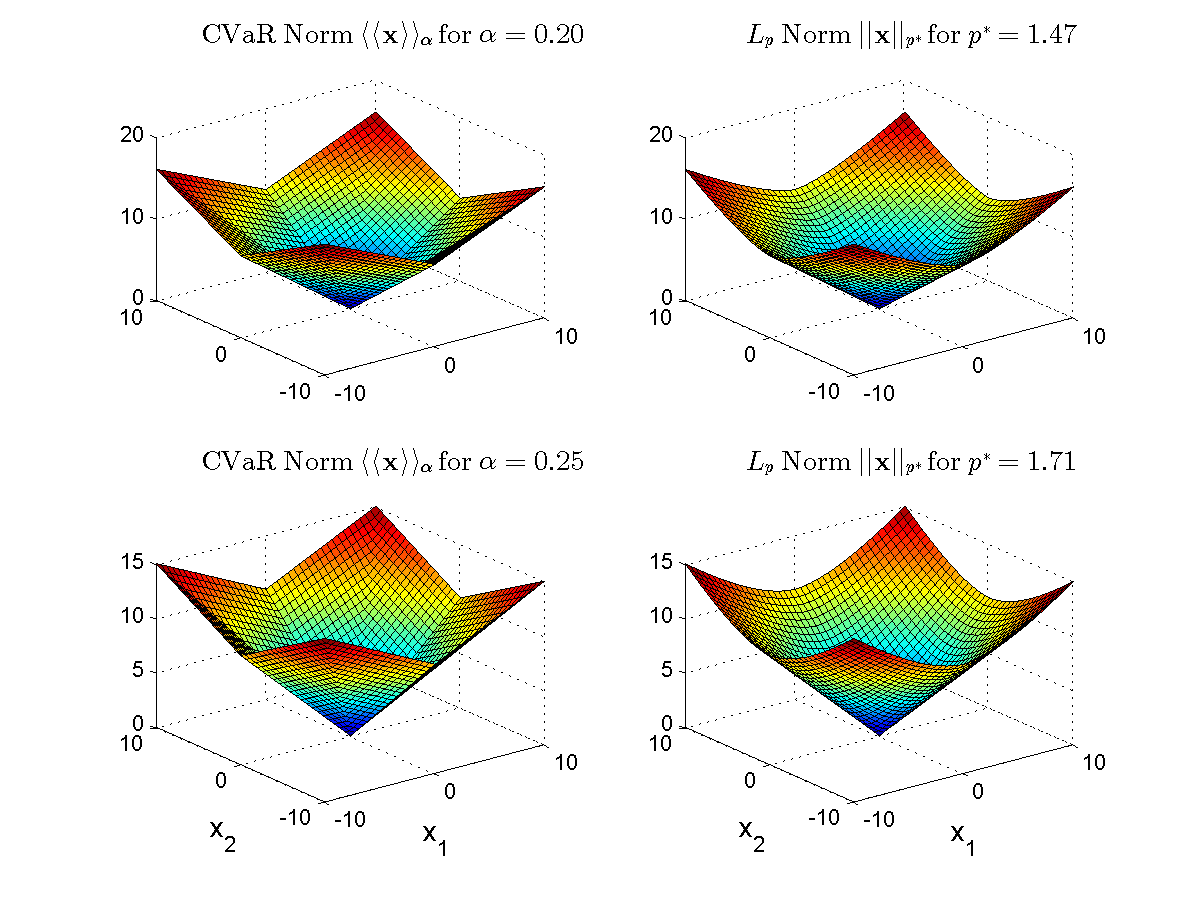}
\end{figure}
\begin{figure}[H]
	\centering
	\includegraphics[width = \textwidth]{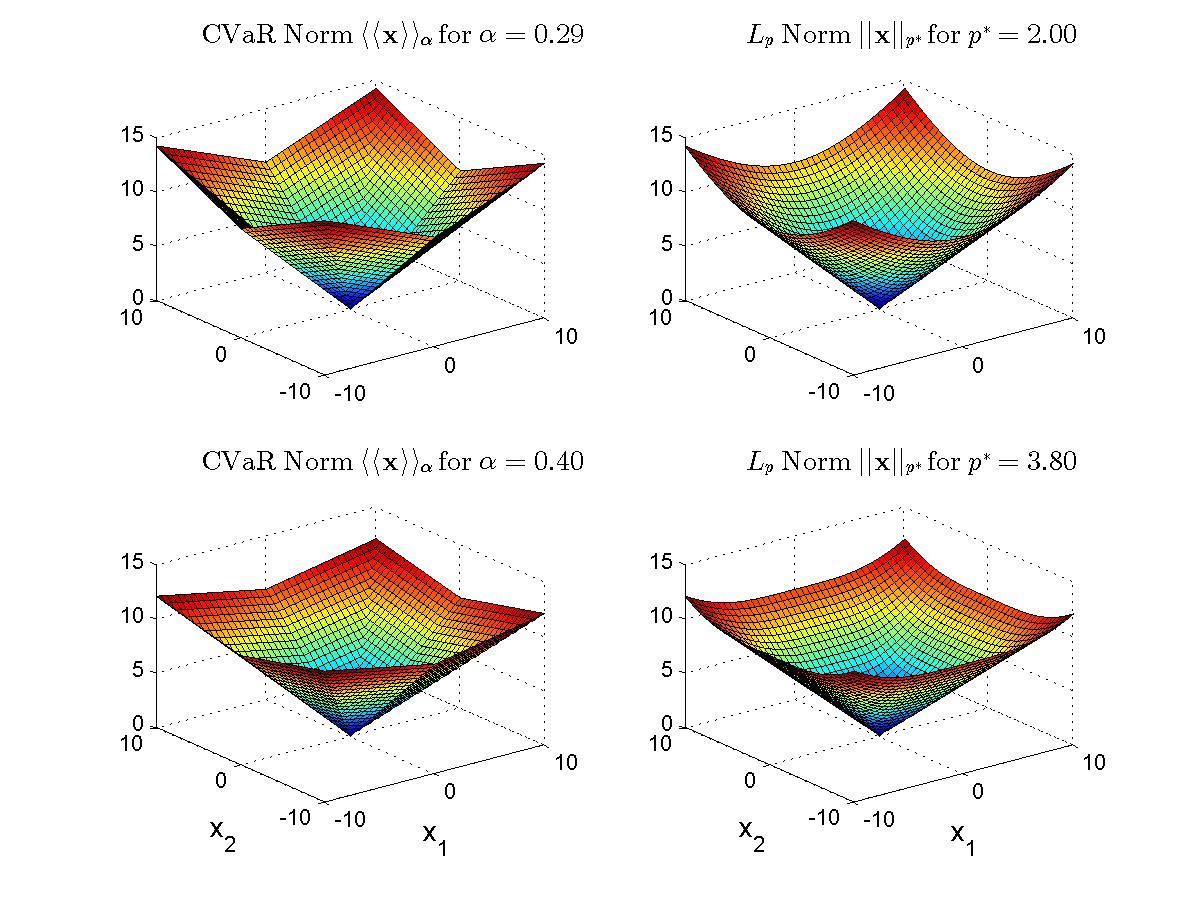}
\end{figure}
\begin{figure}[H]
	\centering
	\includegraphics[width = \textwidth]{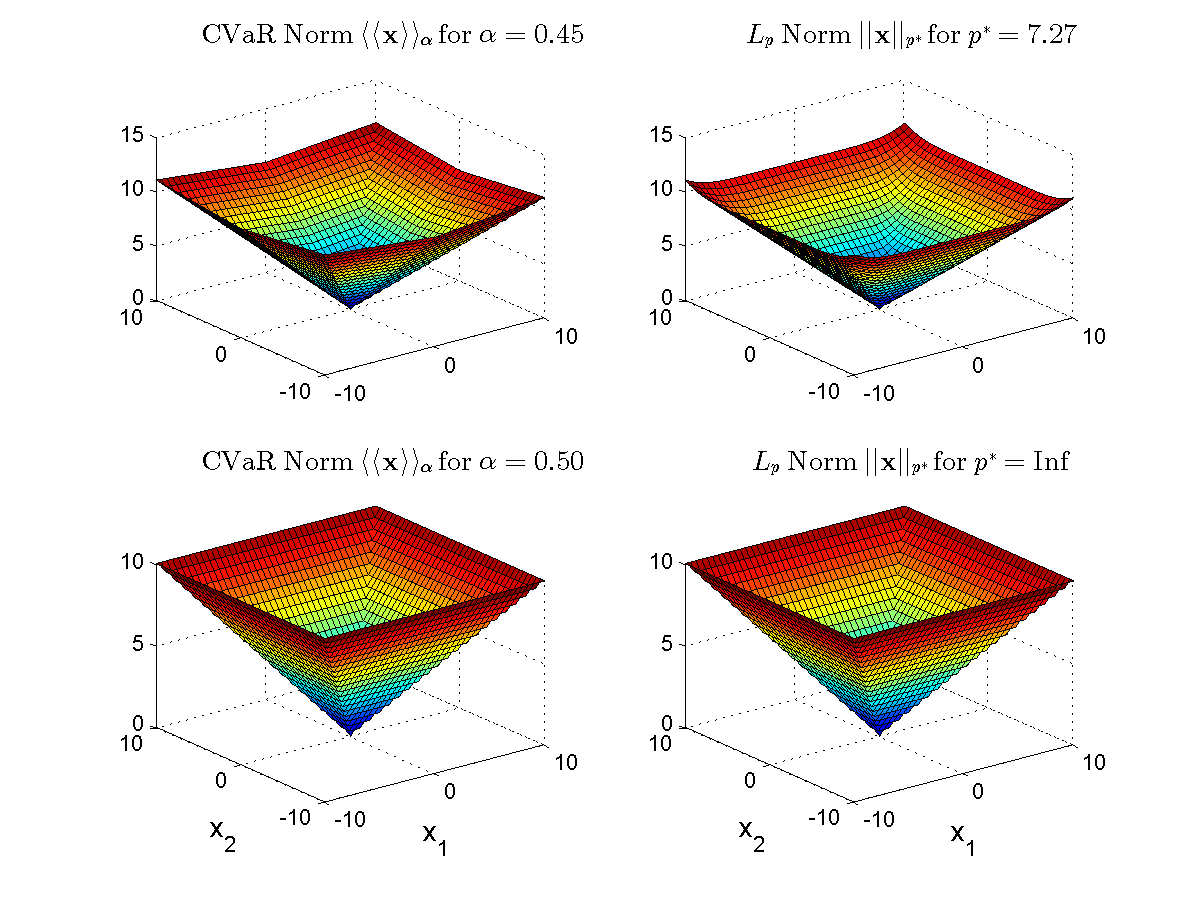}
\end{figure}

\section{Projection of a circle onto the unit ball in $\mathbb{R}^3$ using $L_2$ and $C_{\alpha}$ norms}
\label{app_diagrams:Projection_onto_unit_ball}
\begin{figure}[H]
	\centering
	\includegraphics[width = 0.9\textwidth]{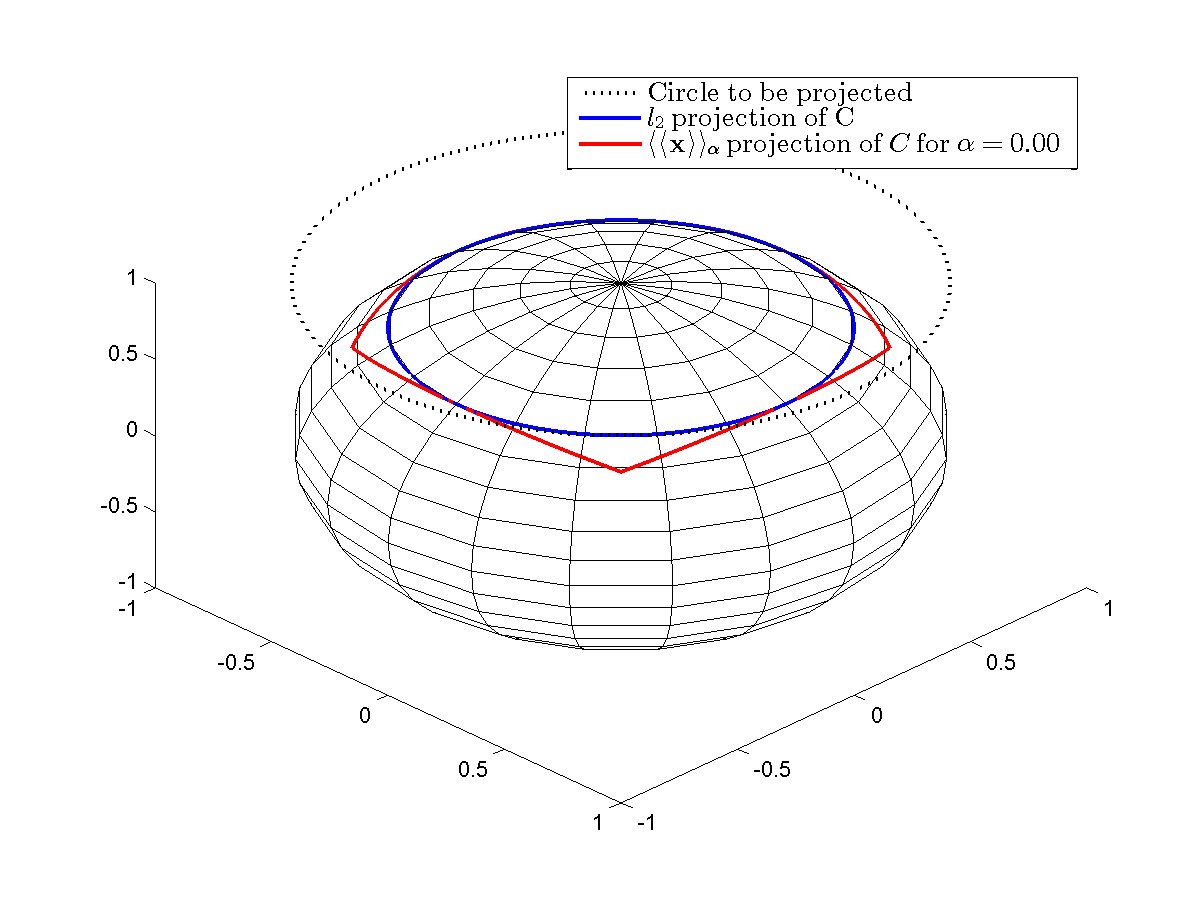}
\end{figure}	
\begin{figure}[H]
	\centering
	\includegraphics[width = \textwidth]{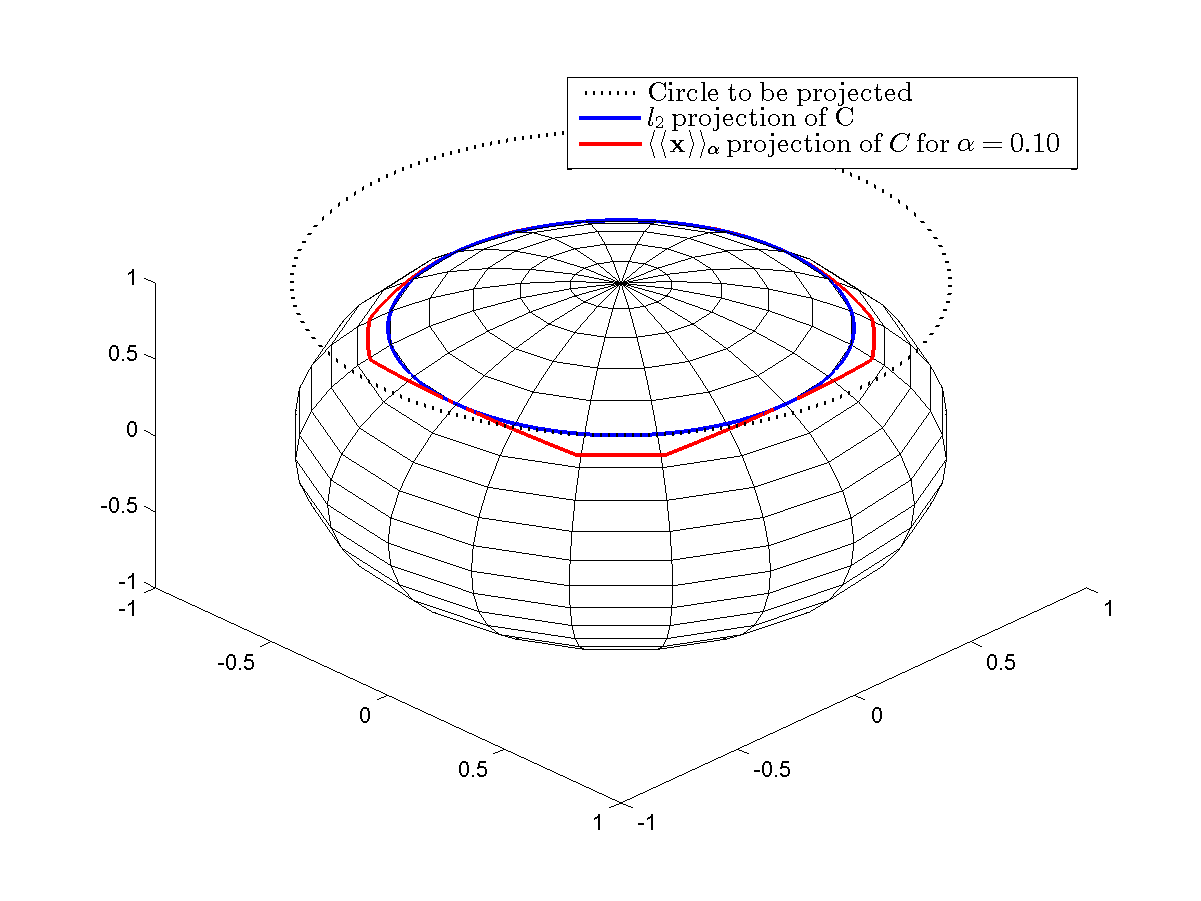}
\end{figure}	
\begin{figure}[H]
	\centering
	\includegraphics[width = \textwidth]{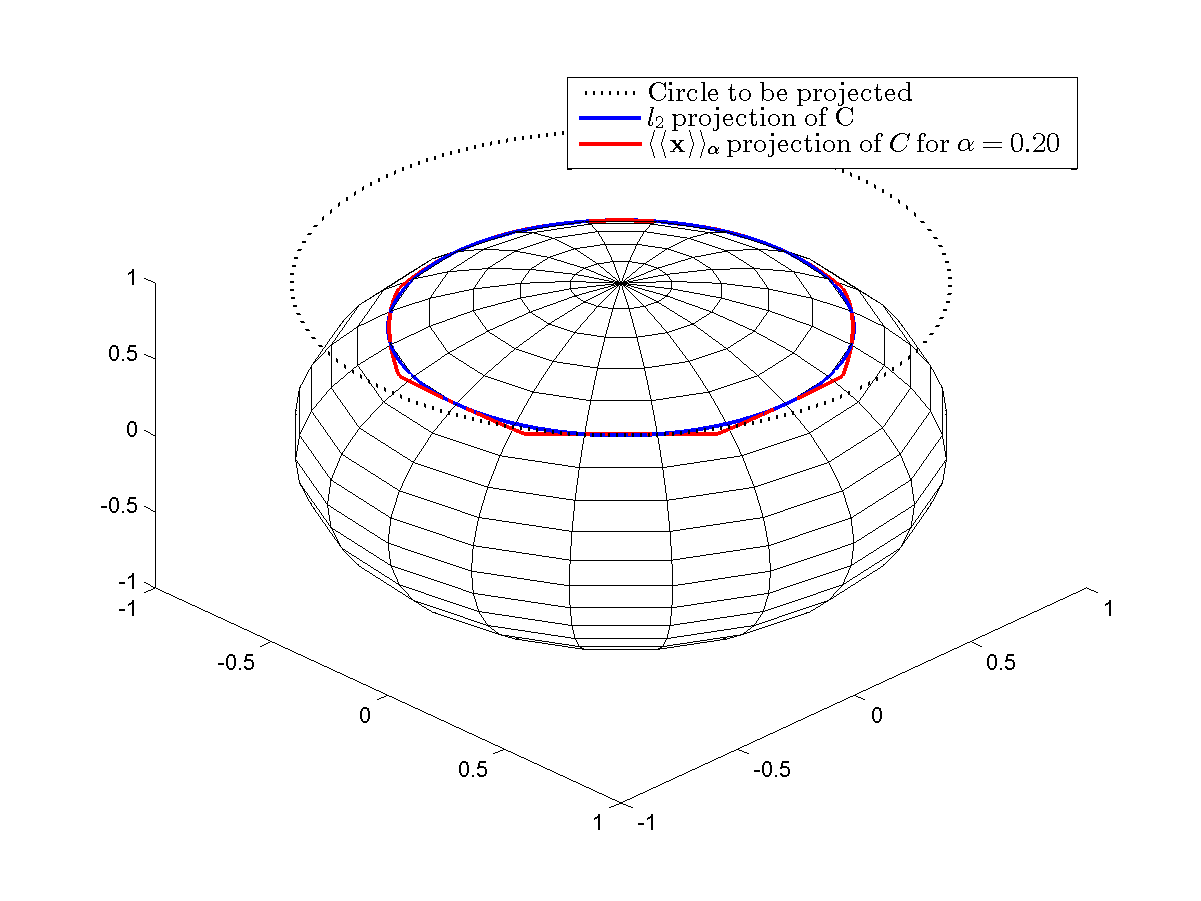}
\end{figure}	
\begin{figure}[H]
	\centering
	\includegraphics[width = \textwidth]{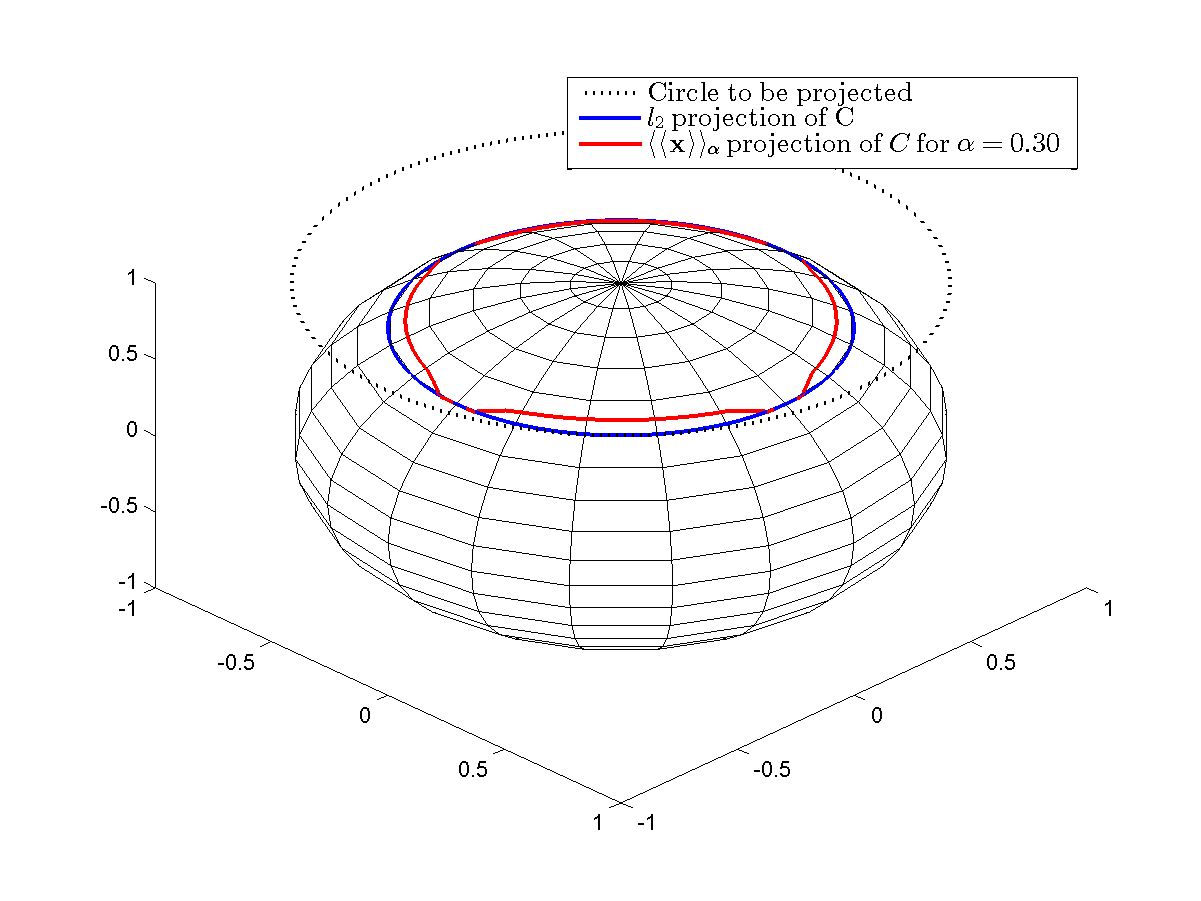}
\end{figure}	
\begin{figure}[H]
	\centering	
	\includegraphics[width = \textwidth]{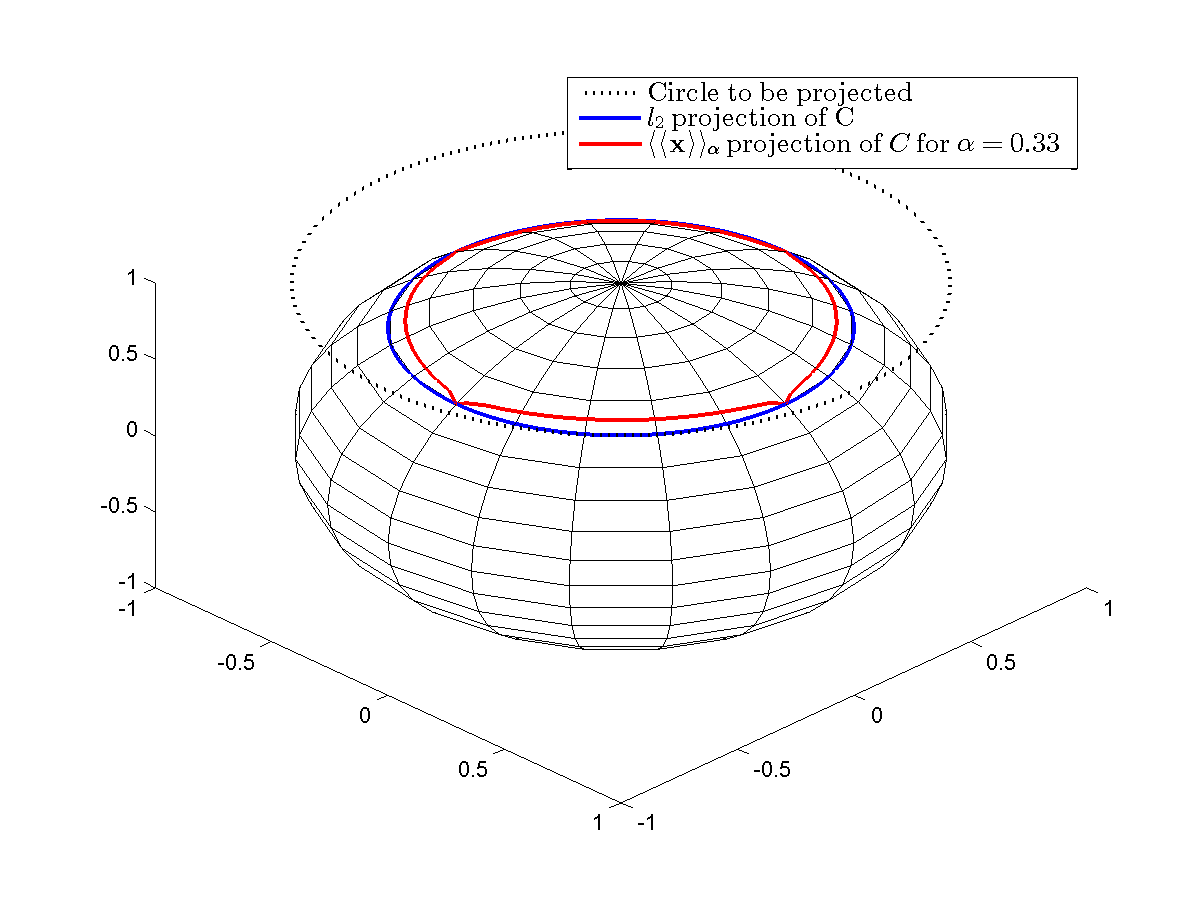}
\end{figure}
\begin{figure}[H]
	\centering	
	\includegraphics[width = \textwidth]{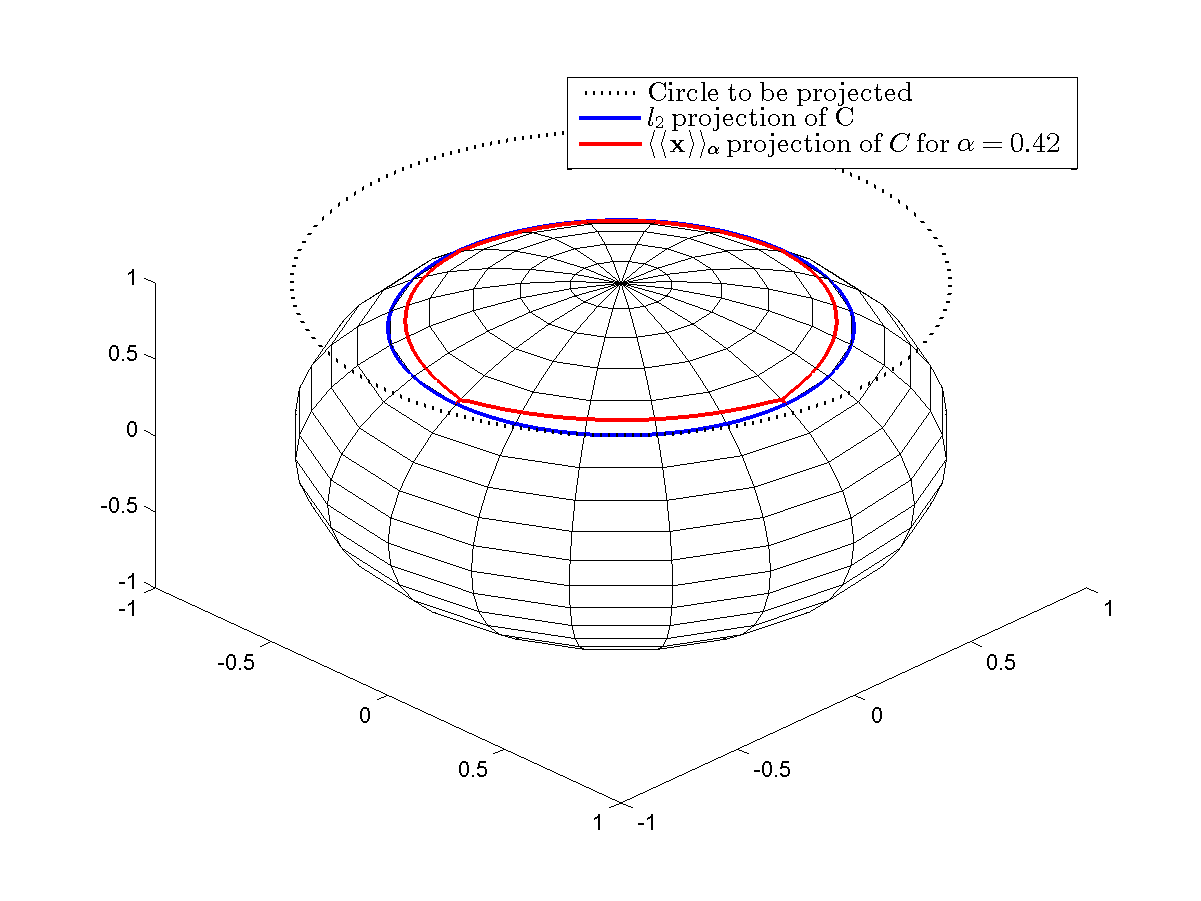}	
\end{figure}
\begin{figure}[H]
	\centering	
	\includegraphics[width = \textwidth]{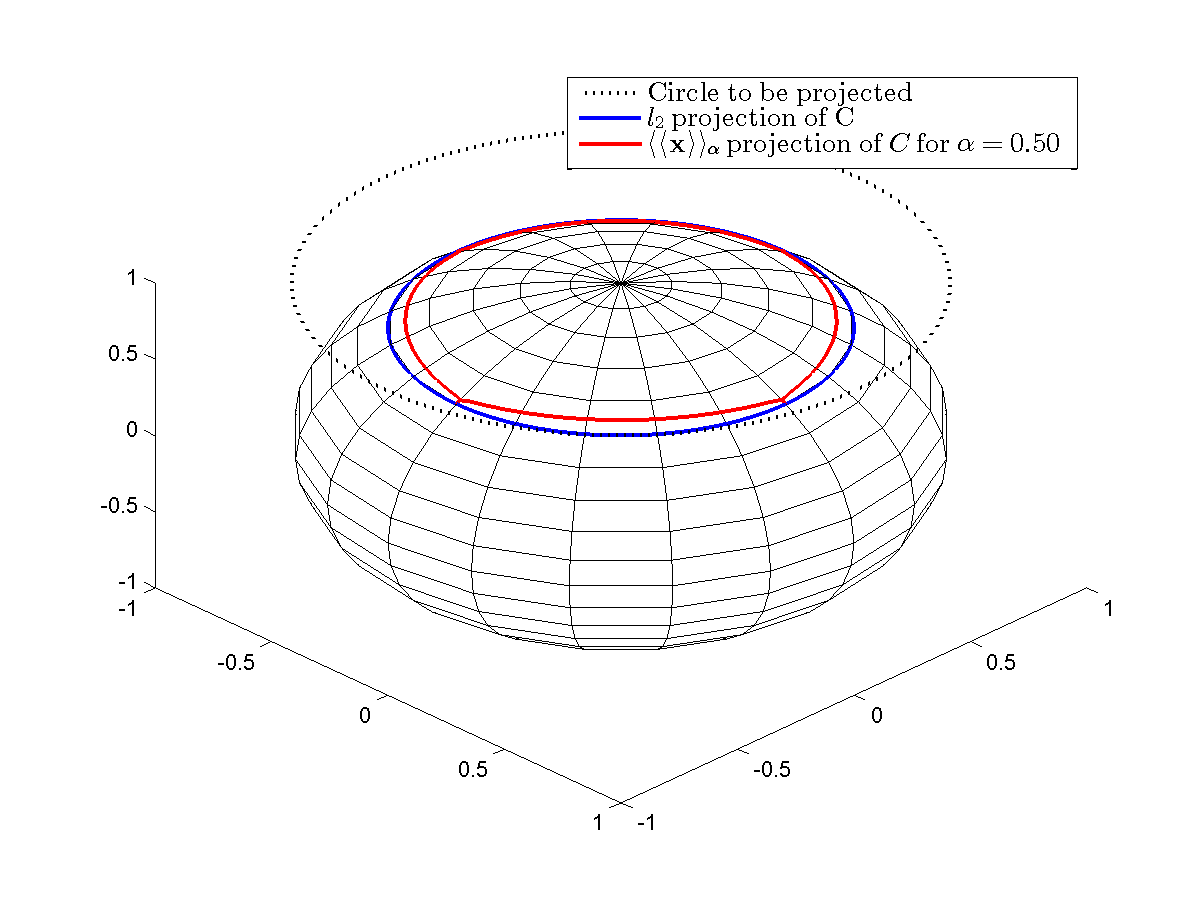}
\end{figure}
\begin{figure}[H]
	\centering
	\includegraphics[width = \textwidth]{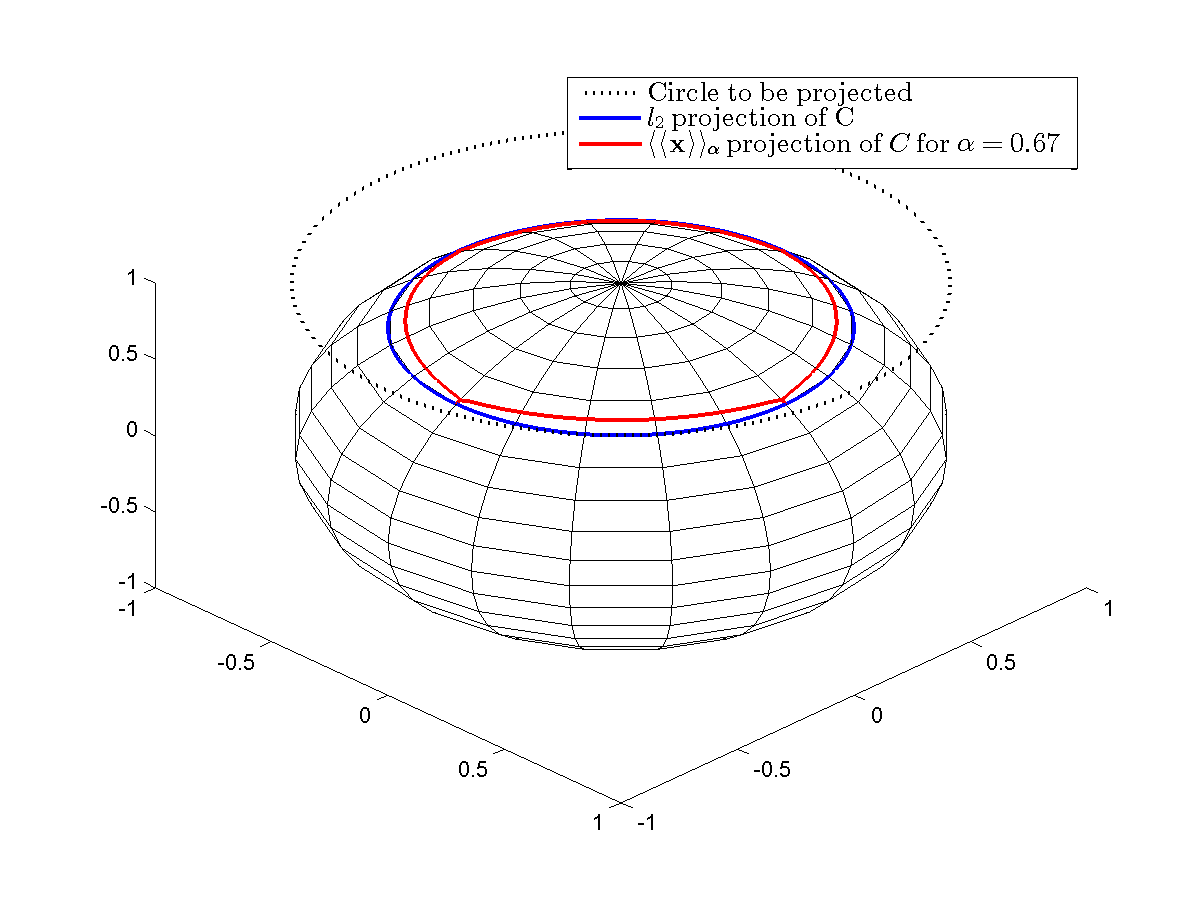}	
\end{figure}

\end{document}